\documentclass{article}
\usepackage{float}
\usepackage{cite}
\usepackage{shapepar}
\usepackage{listings}
\usepackage{amsmath}
\usepackage{amsthm}
\usepackage{mathtools}
\usepackage{graphics}
\usepackage{amssymb}
\usepackage{mathrsfs}
\usepackage{cite}
\usepackage{framed}
\usepackage{diagbox}
\usepackage{booktabs}
\usepackage{fancybox}
\usepackage{geometry}
\usepackage{multirow}
\usepackage{enumerate}
\usepackage{caption}
\usepackage{subcaption}
\usepackage{hyperref}
\usepackage{boldline}
\usepackage[table]{xcolor}
\usepackage{slashbox}
\usepackage{tabularray}
\usepackage{vcell}
\usepackage{titlesec}
\graphicspath{{Figures_All/}}

\def\bn{\begin{enumerate}} \def\en{\end{enumerate}}

  \def\ss{\smallskip}

\titleformat{\paragraph}
{\normalfont\normalsize\bfseries}{\theparagraph}{1em}{}
\titlespacing*{\paragraph}
{0pt}{3.25ex plus 1ex minus .2ex}{1.5ex plus .2ex}

\newtheorem{theorem}{Theorem}
\usepackage[ruled,linesnumbered]{algorithm2e}
\newtheorem{remark}{Remark}

\SetCommentSty{mycommfont}

\title{A Direct Sampling-Based Deep Learning Approach \\for Inverse Medium Scattering Problems}
\date{}
\author{Jianfeng Ning \thanks{School of Mathematics and Statistics, Wuhan University, Wuhan, China. ({ningjf@whu.edu.cn}).} \and Fuqun Han \thanks{Department of Mathematics, The Chinese University of Hong Kong, Shatin, N.T., Hong Kong.  ({fqhan@math.cuhk.edu.hk}).}  \and Jun Zou \thanks{Department of Mathematics, The Chinese University of Hong Kong, Shatin, N.T., Hong Kong. The work of this author was substantially supported by Hong Kong RGC General Research Fund (Projects 14306921 and 14306719). ({zou@math.cuhk.edu.hk}).}}

\begin{document}
	\maketitle
	\begin{abstract}
		In this work, we focus on the inverse medium scattering problem (IMSP), which aims to recover unknown scatterers based on measured scattered data. Motivated by the efficient direct sampling method (DSM) introduced in \cite{ito2012direct}, we propose a novel direct sampling-based deep learning approach (DSM-DL)  for reconstructing inhomogeneous scatterers. In particular, we use the U-Net neural network to learn the relation between the index functions and the true contrasts. Our proposed DSM-DL is computationally efficient, robust to noise, easy to implement, and able to naturally incorporate multiple measured data to achieve high-quality reconstructions. Some representative tests are carried out with  varying numbers of incident waves and different noise levels to evaluate the performance of the proposed method. The results demonstrate the promising benefits of combining deep learning techniques with the DSM for IMSP.
	\end{abstract}
 
	\section{Introduction}
	Inverse medium scattering problems (IMSP) aim to recover the geometry and physical properties of unknown scatterers in the domain of interest from the measured scattered fields. This inverse problem has wide practical applications such as biomedical imaging, geophysical exploration, remote sensing, nondestructive evaluation, and so on (see, e.g., \cite{buchanan2004marine,kirsch2011introduction,zhdanov2002geophysical}). Suppose that a bounded domain of interest $\Omega$ is occupied by some inhomogeneous medium scatterers in the homogeneous background space $\mathbb{R}^{N}$ ($N=2$ or $3$). With the incident field $u^{i}$, the total field $u=u^{i}+u^s$ satisfies the following Helmholtz equation \cite{colton1998inverse}
	\begin{equation}
		\Delta u +k^{2}\varepsilon_{r}(x)u(x)=0 \quad \text{in}\quad \mathbb{R}^{N},
	\end{equation}
	where $k$ is the wavenumber,  and the scattered field can be expressed asymptotically as \cite{colton1998inverse}:
	\begin{equation}
		u^{s}(x) = \frac{\exp(\mathbf{i}k|x|)}{|x|^{(N-1)/2}}(u^{\infty}(\hat{x}) + O(1/|x|)), \quad |x|\rightarrow\infty,
		\label{asymptic_us}
	\end{equation}
	where $\hat{x}=\frac{x}{\Vert x\Vert_{2}}\in \mathbb{S}^{N-1}$ and $u^{\infty}$ is called the far-field pattern of $u^{s}$. The induced current density $I(x)$ is defined as $I(x)=\eta(x)u(x)$ with the contrast $\eta(x)=\varepsilon_{r}(x)-1$, where $\varepsilon_{r}(x)=1$ in the homogeneous medium. Then the scattered field $u^s$ on the measurement surface $S$ satisfies \cite{colton1998inverse}
	\begin{equation}
		u^{s}(x)= k^{2}\int_{\Omega}G(x,y)I(y)dy, \quad x \in S,
		\label{data equation}
	\end{equation}
	where $G(x,y)$ is the free-space Green's function for the scattering problem which is given by
	\begin{equation}
		G(x,y)=\left\{
		\begin{aligned}
			\frac{\mathbf{i}}{4}H_{0}^{(1)}(k|x-y|), \quad N=2,\\
			\frac{\exp(-\mathbf{i}k|x-y|)}{4\pi|x-y|}, \quad N=3,
		\end{aligned}
		\right.
	\end{equation}
	where the function $H_{0}^{(1)}$ refers to the Hankel function of the first kind with order $0$. 
	The IMSP is then to recover $\varepsilon_{r}$ from noisy measurements of the scattered field $u^{s}$ or the far-field pattern $u^{\infty}$, corresponding to some incident fields.
	
The IMSP is known to be severely ill-posed and highly non-linear, especially very challenging in its numerical solutions 
when the measurement data is available only at a limited aperture. Various numerical algorithms have been developed for solving inverse medium scattering problems, 
such as iterative methods like recursive linearization methods \cite{bao2015inverse}, contrast source-type inversion methods \cite{van1997contrast,van2021forward}, iteratively regularized Gauss-Newton method\cite{langer2010investigation}, least-squares methods \cite{ito2022least}, and subspace optimization methods \cite{chen2009subspace}, which can obtain satisfactory reconstructions while they are generally very time-consuming. Non-iterative methods such as reverse time migrations \cite{chen2013reverse}, singular sources methods \cite{potthast2005singular}, factorization methods \cite{kirsch2002music,qu2019locating} and sampling-type methods \cite{cheney2001linear,cakoni2011linear,potthast2006survey,harris2020orthogonality,ito2012direct} are fast 
while their reconstructions may not be accurate. We refer to \cite{colton1998inverse, cakoni2014qualitative, xudong_book_2018} 
for more extensive reviews on both theoretical analyses and numerical algorithms on inverse scattering problems.
	
	In recent years, many deep learning methods have already been applied to solve inverse scattering problems. Using conventional convolutional neural networks to directly approximate the relationship between the near scattered fields and the true contrasts was shown to be only feasible for simple scatterers \cite{wei2018deep}, hence, incorporating deep learning methods with conventional numerical methods for inverse scattering problems has become a popular direction. By exploiting the inherent low-rank property of the scattering problems, a sophisticated network architecture called SwitchNet was developed in \cite{khoo2019switchnet} for both forward and inverse scattering problems. The backpropagation, and dominant current CNN schemes were proposed and compared in \cite{wei2018deep} for solving inverse medium scattering problems. A Two-step enhanced deep learning approach invented in \cite{yao2019two} consists of an initial guess step and a refinement step by introducing two convolutional neural networks. An iterative reconstruction algorithm called Learned Combined Algorithm was proposed in \cite{li2022reconstruction} to reconstruct the inhomogeneous media from the far-field data by repeatedly applying the Landweber approach, the iteratively regularized Gauss-Newton method, and the deep neural projector. In \cite{gao2022artificial}, by employing a few parameters to represent the star-shaped scatterers, the authors designed a fully connected neural network for the inverse obstacle problem of recovering an impenetrable scatterer from measured scattered data with only a single incident field. For more deep learning-based approaches to inverse scattering problems, we refer to the review article \cite{chen2020review} and references therein for some recent developments. On the other hand, some deep learning-based approaches have also been developed for solving general inverse problems. A partially learned approach was proposed in \cite{adler2017solving}, where the “gradient” component in an iterative scheme is learned using a convolutional network. The NETT \cite{li2020nett} combines deep learning with a Tikhonov regularization scheme, where a regularizer is defined and learned by a deep neural network. Several sophisticated neural operators, e.g.,\cite{lu2021learning,li2020fourier,tripura2022wavelet}, have been mainly designed to learn forward operators, and their extensions for solving inverse problems by using reversed input-output, as well as applying Tikhonov regularization to a trained forward model were studied in a recent survey \cite{tanyu2022deep}. A novel and influential framework called Physics-informed neural network (PINN) was invented in \cite{raissi2019physics} for solving forward and inverse problems involving nonlinear PDEs, where the solutions are approximated by neural networks and the loss function is to enforce the neural networks to satisfy the strong PDE equations.  
We refer to the surveys \cite{arridge2019solving,mccann2017convolutional} for more deep learning concepts on various inverse problems. 
	
	Among the sampling-type methods, the direct sampling method (DSM) proposed in \cite{ito2012direct} can directly reconstruct the shapes and locations of unknown scatterers using one or a few incident waves, which is computationally efficient, easy to implement, and highly stable to large noise. The method is also independent of a priori knowledge of the geometry and physical properties of the unknown scatterers. The main idea of the DSM is to design an index function that takes large values for points inside the inhomogeneous scatterers, thus providing a shape estimate.  The DSMs have also been developed for some other important inverse problems, such as inverse scattering problems using far-field data \cite{li2013direct}, inverse electromagnetic scattering \cite{ito2013direct,chow2022direct}, electrical impedance tomography \cite{chow2014direct}, inverse elastic scattering problems \cite{ji2018direct}, diffusive optical tomography \cite{chow2015direct}, inversion of the radon transform \cite{chow2021direct}. However, the DSMs have some limitations. Firstly, like most non-iterative methods, they can not provide very accurate geometric reconstructions, especially for complicated scatterers. Secondly, they are mainly focused on cases of one or two incident waves, which may limit their applications to more complicated cases. Thirdly, the DSMs can only estimate the locations and shapes of the scatterers, while it is not able to estimate physical properties (e.g., the permittivity and conductivity), which is also partly due to very limited measurement data used by DSMs. To address these limitations, we combine the deep learning method with the DSM proposed in \cite{ito2012direct}, in which the network is used to approximate the relation between the index functions and the true contrasts. The numerical results show that this new approach can fully make use of the measurement data from multiple incident fields so that it can robustly get better reconstruction and even recover the physical properties of complicated scatterers. The trained networks also have some generosity to recover some scatterers with geometries and physical properties that are quite different from those used for the training data. In addition, the proposed DSM-DL also inherits the advantages of the classical DSM, i.e., it is computationally efficient, robust to noise, and it can provide reasonable reconstruction with one or two incident waves.
	
	Combining deep learning and sampling-type methods for solving inverse scattering problems was also studied recently in \cite{le2022sampling} by using the orthogonality sampling method. However, it only considered one incident wave and is mainly focused on recovering the shapes and locations of the scatterers. In addition, it requires both the Cauchy data and the scattered field data. As a comparison, our proposed method can make use of a general number of incident waves, with only scattered field data or far-field pattern data available, to recover more complicated scatterers. As we shall see from the numerical results, it can also recover the physical properties of the scatterers. The deep learning methods were also combined with sampling methods to tackle electrical impedance tomography \cite{guo2021construct} and diffusive optical tomography \cite{jiang2021learn}, in which the Cauchy difference functions are the inputs of the neural network and computation of several forward problems is required, while in our method, the inputs of the neural network are the index functions that can be computed with simple scalar product. We shall point out that the general framework of the proposed DSM-DL can also be directly applied to many other important inverse problems where the DSMs have been developed. 
	
	The rest of this paper is organized as follows. In Section \ref{secDSM}, we review the original DSM from \cite{ito2012direct} for inverse medium scattering problems and present some analysis of properties of index functions. The development of the proposed DSM-DL and the architecture of the neural network, as well as several strategies for data augmentations are introduced in Section \ref{secDSMDL}. Several representative numerical experiments are carried out to validate the advantages of the DSM-DL in Section \ref{secNumeri}, and some concluding remarks are given in Section \ref{secConcl}.

	\section{Direct sampling method}\label{secDSM}
	
	The direct sampling method\cite{ito2012direct} is a very simple technique for estimating the shapes and locations of the unknown medium scatterers when the measurement data are available from limited incident fields without any iteration. Specifically, denote $B(0,R)$ as the open ball with radius $R$ in $\mathbb{R}^2$ or $\mathbb{R}^3$ and let $S=\partial B(0,R)$ be the circular curve in $\mathbb{R}^{2}$ or spherical surface in $\mathbb{R}^{3}$, and $\Omega$ be the bounded domain of interest that is occupied by some inhomogeneous medium inclusions, the direct sampling method employs an index function of the form:
	\begin{equation}
 		\Phi(x):=|\langle u^{s},G(x,\cdot)\rangle_{L^{2}(S)}|
		\label{indexfunction},
	\end{equation}
	for any point $x\in \Omega$ or its normalized form
	\begin{equation}
		\hat{\Phi}(x):=\dfrac{|\langle u^{s},G(x,\cdot)\rangle_{L^{2}(S)}|}{\Vert u^{s}\Vert_{L^{2}(S)}\Vert G(x,\cdot) \Vert_{L^{2}(S)} }.
		\label{norIndex}
	\end{equation}
	If $\Phi(x)$ attains an extreme value at a point $x$, the point lies most likely within the support of the scatterers, whereas if $\Phi(x)$ takes small values, the point $x$ likely lies outside the support. The detailed derivation of the index function can be found in \cite{ito2012direct}.  An example of the true contrast and the computed index function is presented in Fig.\,\ref{fig:IndFun}. The computation of index function $\Phi$ only involves evaluating the scalar product of the measured data with the Green's function; thus the computation is very cheap. Furthermore, the highly noisy data is expected to be represented by high Fourier modes, while the Green's function $G(x,y)$ is a very smooth function on the circular curve/spherical surface and is dominated by the low-frequency Fourier modes. Thus, the noises are roughly orthogonal to the Green's function $G(x,y)$ and have little contribution to the index function. Consequently, the DSM is expected to be very robust to highly noisy data. 
	\begin{figure}[htp]
		\centering
		\includegraphics[width=0.65\linewidth]{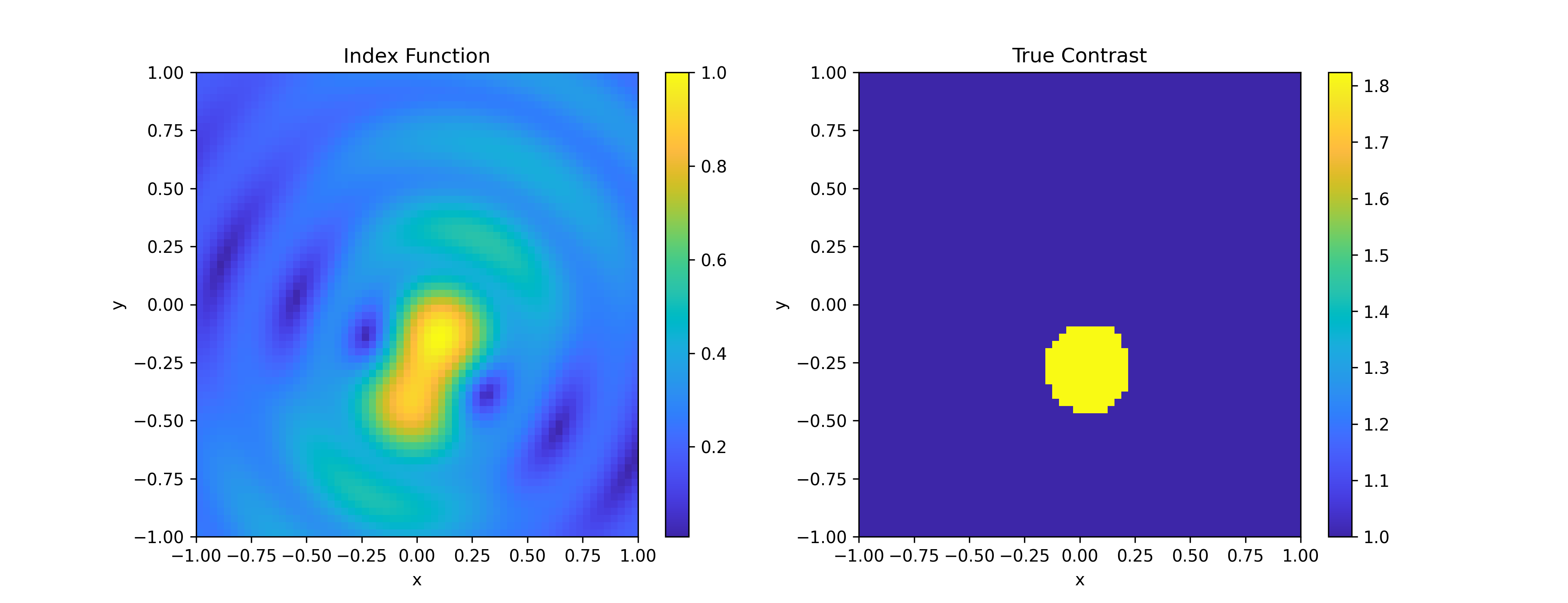}
		\caption{An example of the index functions by DSM (left) and the true contrast (right)}
		\label{fig:IndFun}
	\end{figure}

	We are now going to present some properties of the index functions, which can be regarded as some motivation of 
	our subsequent strategy to combine the deep learning strategy with the DSM.
	\begin{theorem}
		Let $S=\partial B(0,R)$ be the circular curve in $\mathbb{R}^{2}$ or the spherical surface in $\mathbb{R}^{3}$, and  $\Omega$ be an open domain that is contained in the interior of $S$. Define the operator $\mathcal{T}: L^{2}(S)\rightarrow C^{\infty}(\Omega)$ as
		\begin{equation}
			\phi(x) = \mathcal{T}(u)(x):=\langle u,G(x,\cdot)\rangle_{L^{2}(S)}=\int_{S}u(y)\overline{G(x,y)}dy \quad x \in \Omega,
			\label{operator}
		\end{equation}
		then $\mathcal{T}$ is an injective operator.
		\label{theorem}
	\end{theorem}
	\begin{proof}
		For $u_{1}, u_{2}\in L^2(S)$, suppose that $\mathcal{T}(u_{1})=\mathcal{T}(u_{2})$, i.e.,
		\begin{equation}
			\int_{S}u_{1}(y)\overline{G(x,y)}dy=\int_{S}u_{2}(y)\overline{G(x,y)}dy\,, \quad x \in \Omega.
		\end{equation}
		Then since $G(x,y)$ is analytic for $x\ne y$, we have
		\begin{equation}
			\int_{S}u_{1}(y)\overline{G(x,y)}dy=\int_{S}u_{2}(y)\overline{G(x,y)}dy\,, \quad x \in B(0,R),
		\end{equation}
		thus by continuity of the function $\mathcal{T}(u)$ in $\overline{B(0,R)}$, we have
		\begin{equation}
			\int_{S}u_{1}(y)\overline{G(x,y)}dy=\int_{S}u_{2}(y)\overline{G(x,y)}dy\,, \quad x \in S.
		\end{equation}
		We recall that the above single layer potential is an injective operator if $S$ is a circular curve or a spherical surface \cite{ammari2009layer}, we then have $u_{1}=u_{2}$.
	\end{proof}

The above theorem shows that the output $\mathcal{T}(u)$ does not lose any information of $u$. Next, we study the relation between the index functions associated with the far-field and near-field patterns.
	
For the reconstruction using the far-field pattern, the DSM with the index function \eqref{norIndex} for the near-field pattern can be generalized to a DSM with the following index function for the far-field pattern \cite{li2013direct}:
	\begin{equation}
		\Phi^{\infty}(x):=|\langle u^{\infty},G^{\infty}(x,\cdot)\rangle_{L^{2}(\mathbb{S}^{N-1})}|,
		\label{Ind_Far}
	\end{equation}
	or its normalized form
	\begin{equation}
		\hat{\Phi}^{\infty}(x):=\dfrac{|\langle u^{\infty},G^{\infty}(x,\cdot)\rangle_{L^{2}(\mathbb{S}^{N-1})}|}{\Vert u^{\infty}\Vert_{L^{2}(\mathbb{S}^{N-1})}\Vert G^{\infty}(x,\cdot) \Vert_{L^{2}(\mathbb{S}^{N-1})}},
		\label{Ind_Far_Norm}
	\end{equation}
	where $u^{\infty}$ is the far-field data and $G^{\infty}(x,\hat{y})$ is given by\cite{colton1998inverse}
	\begin{equation}
		G^{\infty}(x,\hat{y})=\left\{
		\begin{aligned}
			&\frac{\exp(\mathbf{i}\pi/4)}{\sqrt{8k\pi}}\exp(-\mathbf{i}kx\cdot \hat{y}), \quad &\text{in}\quad \mathbb{R}^2, \\
			&\frac{1}{4\pi}\exp(-\mathbf{i}kx\cdot \hat{y}), \quad &\text{in}\quad \mathbb{R}^3.
		\end{aligned}\right.
	\end{equation}
By the asymptotic behavior of the scattered field (\ref{asymptic_us}) and the following asymptotic behavior of the Green's function:
   \begin{equation}
   	  G(x,y)  =\frac{\exp(\mathbf{i}k|y|)}{|y|^{(N-1)/2}}(G^{\infty}(x,\hat{y}) + O(1/|y|)), \quad |y|\rightarrow\infty,
   \end{equation}
   we can derive that, when $S = \partial B(0,R)$, the index function will be
   \begin{equation}
   	\begin{split} 
   	\Phi^{R}(x)&:=\bigg|\int_{S}u^{s}(y)\overline{G(x,y)}dy\bigg| \\
   	&=\bigg|\frac{\exp(\mathbf{i}2kR)}{R^{(N-1)}}\int_{S}u^{\infty}(\hat{y})\overline{G^{\infty}(x,\hat{y})}dy\bigg| + O\bigg(\frac{1}{R}\bigg)\\
   	&=C_N\Phi^{\infty}(x)+O\bigg(\frac{1}{R}\bigg),
   	\end{split} 	 
   \label{conver_Ind}
   \end{equation}
where the constant $C_N$ only depends on the dimension $N$. Thus, the index function for the far-field data is the limit of the index function for the scattered field as the radius $R$ tends to infinity.
	
However, despite the effectiveness and robustness of the DSM, we note that it also has the potential to be improved and generalized in a few aspects. Firstly, it focuses mainly on the reconstruction with one or two incident waves which may limit its applications to more complicated cases. Studying efficient strategies to make use of multiple data fully is important since data from more incident waves usually contain more inherent information on the distribution of the scatterers which can improve the robustness of the reconstruction. Secondly, the direct sampling method can only roughly predict the shapes and locations of the scatterers and is unable to recover the physical properties which are essential in many applications.
	
\section{Direct sampling-based deep learning approach (DSM-DL)}\label{secDSMDL}
To obtain reconstruction with high accuracy and strong robustness, one can either design more effective algorithms or make use of more measurement data. The latter is usually necessary for scatterers with complicated shapes. Suppose there are $N_{i}$ incident waves, then we can  compute $N_{i}$ index functions $\Phi_{p}(x)$ defined in \eqref{indexfunction} with $p=1,2,\cdots,N_{i}$. Based on Theorem \ref{theorem} it is reasonable to expect that more index functions contain more information about the scatterers. To fully utilize the multiple index functions from many incidences, we should better understand the relation between the index functions and the true contrasts. For fixed $N_{i}$ incident fields $\{u_{p}^{i}\}_{p=1}^{N_{i}}$, we define an operator $\mathcal{A}$ as 
\begin{equation}
		\mathcal{A}(\varepsilon_{r}) :=\{\Phi_{p}\}_{p=1}^{N_{i}}  \in [C(\Omega)]^{N_i}, \quad \varepsilon_{r}\in W,
\end{equation}
where $W \subset L^2(\Omega)$ is the set of all possible $\varepsilon_{r}$, i.e., $\varepsilon_{r}(x)\geq 0 $ for $x\in \Omega$. The above operator is well-defined and well-posed as it is simply a combination of the forward scattering problem and the convolution in \eqref{indexfunction}. Now, we are ready to introduce a new inverse problem that we would like to solve: 
\textit{For fixed incident fields $\{u_{p}^{i}\}_{p=1}^{N_{i}}$ with corresponding index functions $\{\Phi_{p}\}_{p=1}^{N_{i}}$ computed by \eqref{indexfunction}, we aim at recovering the true contrast $\varepsilon_{r}$.}
In other words, our goal is to find an operator $\mathcal{A}_{inv}$ such that
\begin{equation}
	\mathcal{A}_{inv}(\{\Phi_{p}\}_{p=1}^{N_{i}}) \approx \varepsilon_{r} \quad : \quad [C(D)]^{N_{i}}\rightarrow W.
\end{equation}
    However, it is extremely difficult to theoretically derive an explicit form that can systematically and effectively incorporate multiple index functions to approximate the true contrasts. Although some simple strategies such as taking the maximum, average, or product of all index functions can be applied, they may not always give satisfactory results. In this work, we employ a deep learning method to approximate the relationship between the index functions and the true contrasts which is considered as a black-box. Next, we introduce our new direct sampling-based deep learning approach (DSM-DL), for which the $N_{i}$ index functions $\Phi_{p}(x)$ for $p=1,2,\cdots,N_{i}$ defined in \eqref{indexfunction} are the inputs of the neural network while the outputs are the approximate contrast $\varepsilon_{r}$, i.e., 
	\begin{equation}
		\begin{split}
			\varepsilon_{r}&\approx \mathcal{G}_{\Theta}(\Phi_{1},\Phi_{2},\cdots,\Phi_{N_{i}}) \quad : \quad [C(\Omega)]^{N_{i}}\rightarrow W,\\
			& = \mathcal{G}_{\Theta}(|\mathcal{T}(u^{s}_{1})|,|\mathcal{T}(u^{s}_{2})|,\cdots,|\mathcal{T}(u^{s}_{N_{i}})|) \quad : \quad [C(S)]^{N_{i}}\rightarrow W,\\
			& = \mathcal{G}_{\Theta}\circ \mathcal{DSM}(u_{1}^{s},u_{2}^{s},\cdots,u^{s}_{N_{i}})\quad : \quad [C(S)]^{N_{i}}\rightarrow W, 
		\end{split}		
	\end{equation} 
	where $u_{1}^{s},u_{2}^{s},\cdots,u_{N_{i}}^{s}$ are the scattered fields corresponding to $N_{i}$ incident fields, $\mathcal{G}_{\Theta}$ denotes the neural network, $\mathcal{DSM}$ refers to the output of the classical DSM, and $\Theta$ denotes the free parameters to be learned in the network. By introducing a metric $Loss(x,y)$ to measure the distance between any two images $x$ and $y$, we can iteratively update the free parameters $\Theta$ as follows:
	\begin{equation}
		\Theta\leftarrow\Theta-\frac{1}{|\mathcal{I}|}\sum_{n\in \mathcal{I}}\tau\nabla_{\Theta}Loss\bigg(\mathcal{G}_\Theta(\{\Phi_{p,n}\}_{p=1}^{N_{i}}),\varepsilon_{r,n}\bigg),
  \label{eq:loss_func}
	\end{equation} 
where $\tau$ is the learning rate and $\mathcal{I}$ denote the index set for the training data. We summarize the DSM-DL scheme in Algorithm \ref{AlgorithmDLDSM}. The explicit loss function will be given in equation (\ref{loss_func}) in the numerical experiments.
    
    Specifically, if we discretize the index functions and the true contrasts on $N_{d}\times N_{d}$ grid points in $\Omega$ for 2D problems, then the input to the neural network is a tensor with the dimension $N_{i}\times N_{d}\times N_{d}$, while the output of the neural network is a tensor with the dimension $1\times N_{d}\times N_{d}$, i.e., the index functions are put into different input channels of the network. We choose the well-known U-Net \cite{ronneberger2015u} as the neural network in our numerical experiments, the detailed descriptions of the U-Net and the reasons for choosing it will be given in subsection\,\ref*{unet}. Overall, the DSM-DL has the architecture as presented in Fig.\,\ref{fig:arc-DLDSM}, which consists of a DSM and a neural network that takes the index functions as the inputs and outputs an approximation of the true contrast.
    
    In the DSM-DL scheme, we do not use the normalized index function as it requires additional computation and does not provide better reconstructions based on our numerical experience. Instead, we scale all the inputs by multiplying a constant $2/C$ to improve the network stability in the numerical experiments, where $C=\max\{\Phi_{p,n}(x)|\,x\in\Omega\,;\,\, p=1,2,\cdots,N_i\, ;\,\,n\in\mathcal{I}\}$.

    \begin{remark}
    	Based on Theorem \ref{theorem}, it is also reasonable to use the phased index functions as the inputs of the neural network which we have observed can indeed yield some satisfactory results. However, our numerical experiments indicate that using index functions in absolute-valued form can often achieve better reconstructions. 
    \end{remark}
    \begin{remark}
    	The average index function $\Phi_{ave}:=\frac{1}{N_i}\sum_{p=1}^{N_i}\Phi_{p}$ may provide better reconstruction estimates than the single incident field index function, and one might suggest using it as the input to the neural network. However, the average index function is likely to lose some of the information of the scattered fields as the average operation is not injective. And our numerical experience shows that using multi-channel input $\{\Phi_{p}\}_{p=1}^{N_i}$ can significantly outperform using $\Phi_{ave}$ as input.
    \end{remark}
   \begin{algorithm}[t]
    	\caption{Direct sampling-based deep learning approach for IMSP (Training)}
    	\label{AlgorithmDLDSM}
    	\KwIn{ \\
    		 \quad$\bullet$ Fixed $N_i$ and incident fields $\{u_{p}^{i}\}_{p=1}^{N_i}$.\\
   		\quad$\bullet$ $N$ true contrasts$\{\varepsilon_{r,n}\}_{n=1}^{N}$ and their corresponding scattered fields $\{u_{p,n}^{s}, p=1,\cdots,N_{i}\}_{n=1}^{N}$.\\
   		\quad$\bullet$ The Green's function $\{G(x,y),x\in \Omega, y\in S\}$.\\
   		\quad$\bullet$ Batch number: $L$; Number of epochs: $Q$; Learning rates: $\{\tau_{q}\}_{q=1}^{Q}$.
   	    }
   		Compute the index functions $\Phi_{p,n}(x)=|\langle u^{s}_{p,n},G(x,\cdot)\rangle_{L^{2}(S)}|$, $x\in \Omega; n=1,\cdots,N;p=1,\cdots,N_{i}$.  
   		\\Introduce a loss function $Loss$.
   		\\
   		Initialize the network $\mathcal{G}_{\Theta}$.\\
   		\For{$q=1,2,\cdots,Q$}{
   			\For{$l=1,2,\cdots,L$}{
   				Update $\Theta\leftarrow\Theta-\frac{1}{|\mathcal{I}_l|}\sum_{n\in \mathcal{I}_l}\tau_q\nabla_{\Theta}Loss\big(\mathcal{G}_\Theta(\{\Phi_{p,n}\}_{p=1}^{N_{i}}),\varepsilon_{r,n}\big)$. 
   				\\
   				\tcc{ $\mathcal{I}_{l}$ refers to the index set in the $l_{th}$ batch, $\tau_{q}$ refers to learning rate in the $q_{th}$ epoch.}
   			}
   		}
   \end{algorithm}
	\begin{figure}[htp]
		\centering
		\includegraphics[width=1.0\linewidth]{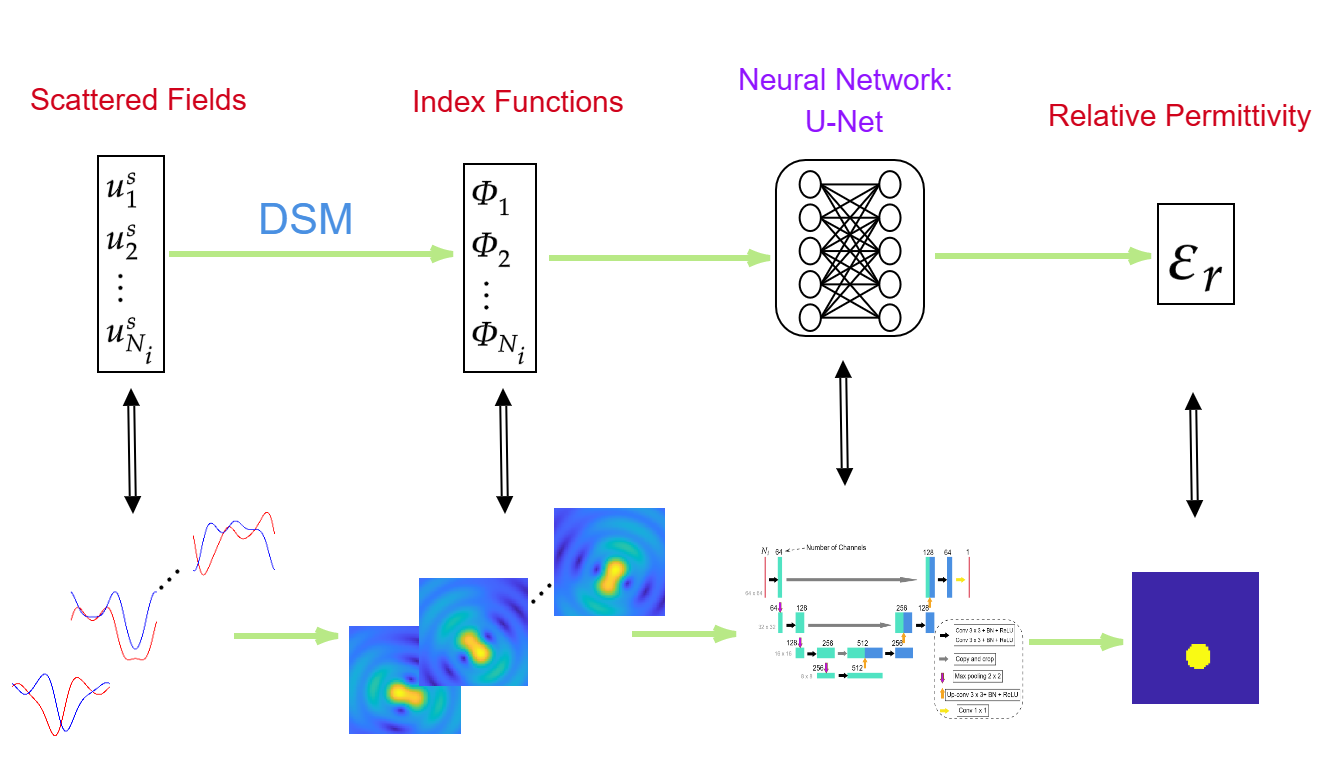}
		\caption{The architecture of the direct sampling-based deep learning approach}
		\label{fig:arc-DLDSM}
	\end{figure}

\subsection{Data augmentations based on the properties of index functions}
\subsubsection{Rotation and symmetry properties of index functions}

We are now going to present some nice properties of the index functions, including the rotation and symmetry properties,
which can be used not only for data augmentations \cite{shorten2019survey}, i.e., increasing the amount of data from the existing observation data, but also provide a technique to help improve the reconstruction quality. For simplicity, we now focus on the two-dimensional problems with incident fields as plane waves, the following discussion can be directly applied to three-dimensional problems and cases with point source incident fields. Firstly, for any $\phi\in [0,2\pi)$, we define the rotation map $R_{\phi}$ on $\mathbb{R}^{2}$ as
\begin{equation}
	R_{\phi}(x)=R_{\phi}(r\cos(\alpha),r\sin(\alpha)):=(r\cos(\alpha+\phi),r\sin(\alpha+\phi)), \quad x=(r\cos(\alpha),r\sin(\alpha))\in \mathbb{R}^{2},
\end{equation}
and the rotation operator $\mathcal{R}_{\phi}$ on $L^{2}(\mathbb{R}^{2})$ as
\begin{equation}
	\mathcal{R}_{\phi}(f)(x)=f(R_{-\phi}(x)), \quad f\in L^{2}(R^{2}).
\end{equation}
Furthermore, for an incident direction $d=(\cos(\theta_d),\sin(\theta_d))$ and a point $x=(r\cos(\alpha),r\sin(\alpha))\in \mathbb{R}^{2} $, we define $S_d(x)$ as the symmetric point of $x$ with respect to the incident direction $d$ (which should be understood as a line across the origin), i.e.,
\begin{equation}
    S_d(x) :=(r\cos(2\theta_d-\alpha),r\sin(2\theta_d-\alpha)), \quad x=(r\cos(\alpha),r\sin(\alpha))\in \mathbb{R}^{2}.
\end{equation}
And the symmetry operator on $L^2(\mathbb{R}^2)$ is defined as $\mathcal{S}_d(f)(x):= f(S_d(x))$ for any $x\in \mathbb{R}^2$ and $f\in L^2(\mathbb{R}^2)$. 

Suppose that $\Phi(x;\varepsilon_{r},d)$ is the index function corresponding to the incident field $e^{\mathbf{i}kx\cdot d}$ and the inhomogeneous media $\varepsilon_{r}$. We emphasize that for a fixed incident direction $d$, if $\varepsilon_{r_2}=\mathcal{R}_{\phi}(\varepsilon_{r_1})$ (or $\varepsilon_{r_2}=\mathcal{S}_{q}(\varepsilon_{r_1})$ with $q\ne d$), we usually can not derive the identity $\Phi(x;\mathcal{R}_{\phi}(\varepsilon_{r}),d)=\mathcal{R}_{\phi}(\Phi(x;\varepsilon_{r},d))$ (or $\Phi(x;\mathcal{S}_{q}(\varepsilon_{r}),d)=\mathcal{S}_{q}(\Phi(x;\varepsilon_{r},d))$) for their corresponding index functions. Thus, unlike problems such as image denoising, we can not apply arbitrary rotation and symmetry operators on existing data to increase the amount of observation data. Instead, we have the following important results. 
\begin{theorem}
	Suppose that $\Phi(x;\varepsilon_{r},d)$ is the index function corresponding to the incident field $e^{\mathbf{i}kx\cdot d}$ and the inhomogeneous medium $\varepsilon_{r}$. Then it holds that
	\begin{equation}
		\Phi(x;\mathcal{S}_{d}(\varepsilon_{r}),d)=\mathcal{S}_{d}(\Phi(x;\varepsilon_{r},d)),\quad x\in \mathbb{R}^2,
		\label{sym_equa}
	\end{equation}
	and for any $\phi\in [0,2\pi)$,
	\begin{equation}
		\Phi(x;\mathcal{R}_{\phi}(\varepsilon_{r}),d)= \mathcal{R}_{\phi}(\Phi(x;\varepsilon_{r},R_{-\phi}(d))),
		\quad x\in \mathbb{R}^2.
		\label{roata_qua}
	\end{equation}
	\label{thm_Rota}
\end{theorem}
\begin{proof}
	We only prove the identity (\refeq{roata_qua}), as the argument is the same for identity (\ref{sym_equa}). 
	By considering the rotation operator to the coordinate system, we note that applying a plane wave with direction $d$ to the medium $\mathcal{R}_{\phi}(\varepsilon_{r})$ is equivalent to applying a plane wave with direction $R_{-\phi}(d)$ to the medium $\varepsilon_{r}$, and we have the following relation:
	\begin{equation}
		\Phi(x;\mathcal{R}_{\phi}(\varepsilon_{r}),d)= \Phi(R_{-\phi}(x);\varepsilon_{r},R_{-\phi}(d)),
		\quad x\in \mathbb{R}^2.
	\end{equation}
Thus by definition of $\mathcal{R}_{\phi}$, we have
\begin{equation}
	\Phi(x;\mathcal{R}_{\phi}(\varepsilon_{r}),d)= \mathcal{R}_{\phi}(\Phi(x;\varepsilon_{r},R_{-\phi}(d))),
	\quad x\in \mathbb{R}^2.
\end{equation}

\end{proof}

\subsubsection{Data augmentations in DSM-DL}
\label{sec_Data_aug}
We now present our strategies for data augmentations based on Theorem \refeq{thm_Rota}. Let us first consider the case of $N_{i}=1$, i.e., only a single incident field is used. The identity (\ref{roata_qua}) in Theorem \ref{thm_Rota} shows that if we know the index function $\Phi(x;\varepsilon_{r},d)$ for the medium $\varepsilon_{r}$ with incident direction $d$, we can obtain the index function $\Phi(x;\mathcal{S}_{d}(\varepsilon_{r}),d)$ for the medium $\mathcal{S}_{d}(\varepsilon_{r})$. Thus, we can apply symmetry operator $\mathcal{S}_{d}$ on the medium $\varepsilon_{r}$ to increase the amount of observation data.

For $N_i>1$, suppose that the $N_i$ incident fields $e^{\mathbf{i}x\cdot d_i},i=1,2,\cdots,N_i$ are equally distributed with $d_i=(\cos(\theta_i), \sin(\theta_i))$, $\theta_i = 2(i-1)\pi/N_i$, $i=1,2,\cdots ,N_i$. For a given $\phi_j= 2\pi j/N_i, j=1,2,\cdots,N_i-1$, we have
\begin{equation}
	\begin{split}
		\Phi(x;\mathcal{R}_{\phi_j}(\varepsilon_{r}),d_i)&=\mathcal{R}_{\phi}(\Phi(x;\varepsilon_{r},R_{-\phi_j}(d_i)))\\
		&=\mathcal{R}_{\phi}(\Phi(x;\varepsilon_{r},d_{i-j})),
		\quad x\in \mathbb{R}^2.
	\end{split}
\label{equa26}
\end{equation} 
Thus, if we know the index functions \{$\Phi(x;\varepsilon_{r},d_{i})| i=1,2,\cdots,N_i$\} for the medium $\varepsilon_{r}$, we can obtain the index functions $\{\Phi(x;\mathcal{R}_{\phi_j}(\varepsilon_{r}),d_{i})| i=1,2,\cdots,N_i\}$ for the medium $\mathcal{R}_{\phi_j}(\varepsilon_{r}), j=1,2,\cdots,N_{i}-1$, which means that we can get $N_i-1$ new data.

The rotation and symmetry properties of the index functions can not only be used for data augmentations but can also provide a new way to compute the approximation of the true contrasts. Note that the input of the neural network $\mathcal{G}_{\Theta}$ has the following form:
\begin{equation}
    \mathbf{\Phi}^{\mathbf{d}}_{\varepsilon_{r}}:=[\Phi(x,\varepsilon_{r},d_1),\Phi(x,\varepsilon_{r},d_2),\cdots,\Phi(x,\varepsilon_{r},d_{N_i})], \quad \mathbf{d}=[d_1,d_2,\cdots,d_{N_i}],
\end{equation}
and 
\begin{equation}
	\varepsilon_{r}\approx\mathcal{G}_{\Theta}(\mathbf{\Phi}^{\mathbf{d}}_{\varepsilon_{r}}).
	\label{no_average}
\end{equation}
Thus
\begin{equation}
	\varepsilon_{r}\approx \mathcal{R}_{- \phi_{j}}\bigg(\mathcal{G}_{\Theta}(\mathbf{\Phi}^{\mathbf{d}}_{\mathcal{R}_{\phi_j}(\varepsilon_{r})})\bigg) \quad \text{for}\quad \phi_j=\frac{2\pi j}{N_i} \quad  j\in \mathbb{N}\,.
\end{equation}
Then by taking an average over the reconstruction with respect to different $\phi_j$, we have
\begin{equation}
	\varepsilon_{r}\approx\frac{1}{N_i}\sum_{j=1}^{N_i}\mathcal{R}_{- \phi_{j}}\bigg(\mathcal{G}_{\Theta}(\mathbf{\Phi}^{\mathbf{d}}_{\mathcal{R}_{\phi_j}(\varepsilon_{r})})\bigg).
	\label{average}
\end{equation}
where $\mathbf{\Phi}^{\mathbf{d}}_{\mathcal{R}_{\phi_j}(\varepsilon_{r})}$ can be computed from $\mathbf{\Phi}^{\mathbf{d}}_{\varepsilon_{r}}$ by equation (\ref{equa26}). 

From our numerical experience, we have observed that 
the averaging approximation (\ref{average}) can generally produce more accurate 
reconstructions of the true medium $\varepsilon_{r}$ than the scheme (\ref{no_average}).

\subsection{U-Net and advantages of the DSM-DL algorithm}\label{unet}
The U-Net\cite{ronneberger2015u} is a popular convolutional neural network that was originally designed for biomedical image segmentation. The architecture of the U-Net that we shall use in our DSM-DL algorithm is shown in Fig.\,\ref{fig:UNet}. We can see that the architecture of the U-Net consists of two paths, forming a U-shaped structure. The left-hand contracting path is designed to extract feature maps at different levels, whereas the right-hand expansive path reconstructs the image by deconvolutions of the corresponding features. The main compositions in the U-Net consist of convolutional layers, batch normalization, ReLU, max pooling, and skip connections. We now provide some more detailed descriptions of these components, from which we can also see the advantages and motivations for using U-Net as the network in our DSM-DL 
algorithm. 
 \begin{itemize}
\item The convolutional layers process the data only for its receptive field with convolution operations. Two important main features of the convolutional layers are weight sharing and local operation. Compared with fully connected layers, the convolutional layers have fewer free parameters and are more efficient in terms of memory and complexity. Furthermore, the convolutional layers take into account the spatial relationships between individual features, making them very appropriate for data with a gridded topology. Since the size of inputs (index functions) in the DSM-DL scheme is usually large, especially for the cases of multiple incidence fields and 3D problems, and the index functions are close to the true contrasts, which indicates that their relations are likely to be (mainly) local. Therefore, the convolutional layers would be a very appropriate choice for the DSM-DL scheme.
    \item 
 The batch normalization proposed in\cite{ioffe2015batch} is a method used to accelerate and stabilize the training of the deep neural networks through normalization of the layers' inputs by re-scaling and re-centering. The batch normalization transfers the layer's input $X$ to $Y$ as
\begin{equation}
	Y=\alpha\frac{X-E[X]}{\sqrt{\text{Var}[X] + \epsilon}}+\beta. 
\end{equation}  
The mean $E[X]$ and standard deviation $\text{Var}[X]$ above are calculated for each dimension, while $\alpha$ and $\beta$ are free parameters to be learned.

\item The activation functions enable the neural networks to approximate nonlinear mappings. The ReLU activation function (rectified linear unit) is defined as the piecewise linear function
\begin{equation}
	f(x)=\max(0,x).
\end{equation}
Since the true contrasts are images that are usually piece-wise constants, and the relative permittivity of the homogeneous background is equal to 1, thus the ReLU is very suitable for the proposed DSM-DL.

\item The max poolings reduce the spatial information while they can significantly increase the effective receptive field of the neural network, and the up-sampling convolutions can increase the resolution of the outputs.

\item Skip connections are able to transfer features from the contracting path to the expansive path so that the lost spatial information during down-sampling can be recovered, which is very suited to the DSM-DL since the index functions and the true contrasts have similar shapes. In addition, the skip connections can also stabilize the training and convergence as it alleviates the vanishing gradient problem. 
\end{itemize}

\begin{figure}[htp]
	\centering
	\includegraphics[width=1.0\linewidth]{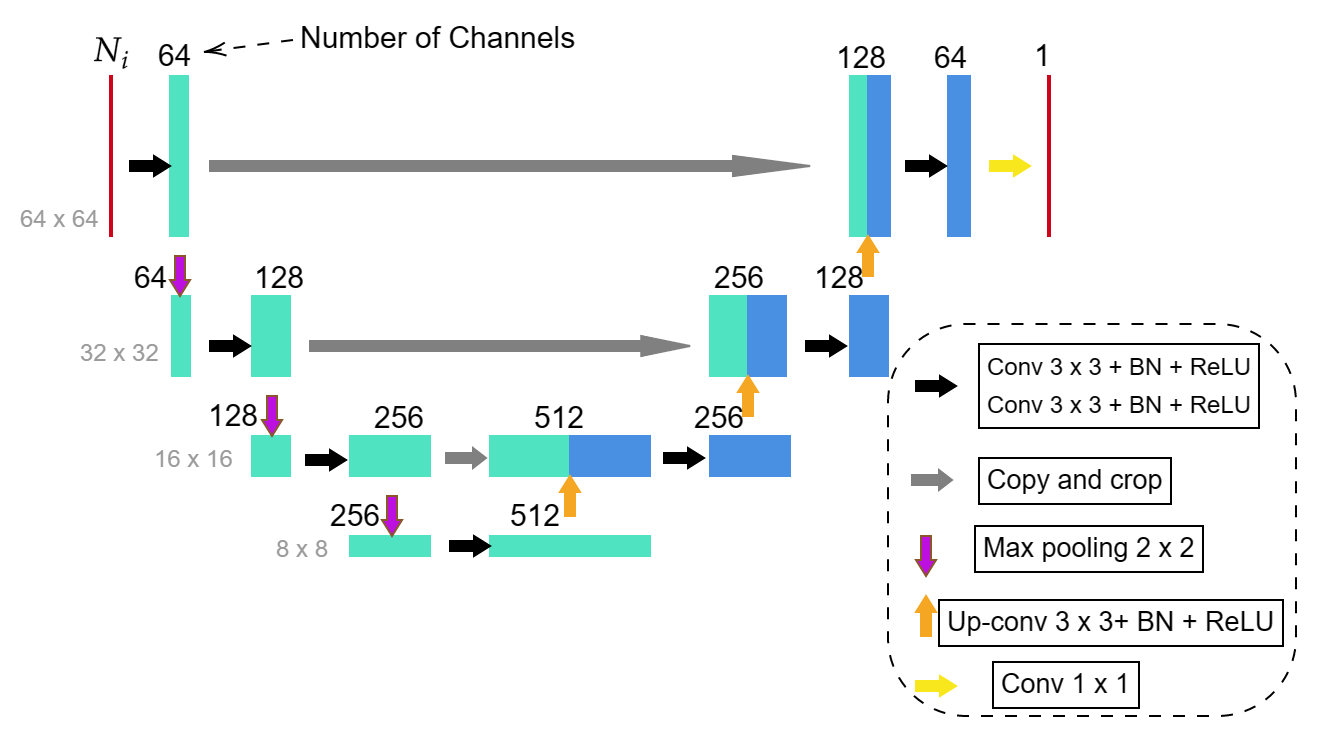}
	\caption{The architecture of the U-Net used in the numerical experiments}
	\label{fig:UNet}
\end{figure}
	
\ss 
Before we present our numerical results, we shall list main advantages of the proposed DSM-DL due to the DSM and the structure of the network that we use. 
\begin{itemize}
	\item The implementation is simple and rather straightforward. The evaluation of the index functions is very cheap and naturally parallel. The forward pass of the neural network is also highly parallel since the main compositions of the U-Net are CNNs. In our numerical experiments, each test is completed within 1 second, after the network is trained.
 
	\item The DSM-DL is highly robust to noise. This is mainly because the index functions (the inputs of the neural network) are robust to noise, i.e., the noise is smoothed in the DSM process.
 
	\item The neural network is easy to train since the index functions and the true contrasts which are the input and output of the network have similar shapes and are defined in the same domain. In this sense, the forward pass of the neural network can be considered as a refinement process. 

    \item Data from multiple incident fields can be fully incorporated. From the numerical results, we can see that our method can be applied to the case of an arbitrary number of incident fields, and using data from more incident fields can significantly improve the accuracy, robustness, and generalization capacity.

    \item As we shall see from the upcoming numerical tests for reconstructing Latin letters “D”, “S” and “M”, the number of the receivers and the radius of the measurement surface for the testing data are not necessary to be consistent with those for the training data. This is due to the convergence property (\ref{conver_Ind}) of the index functions.
 
    \item The DSM-DL is applicable and robust for the reconstruction with limited aperture data. 
Our numerical results indicate that the DSM-DL can still work satisfactorily if the scattered fields are 
only available on part of $S$.
\end{itemize}
	
	\section{Numerical results} \label{secNumeri}
	In this section, we present numerical results using a circle dataset and a modified MNIST dataset for inverse medium scattering problems in 2D to illustrate the accuracy and robustness of the DSM-DL scheme in reconstructing some unknown scatterers from near-field scattered fields. The proposed scheme can be directly applied to 3D problems and far-field problems by using the index functions (\ref{Ind_Far}). 
	
	\subsection{Numerical setup}
	In the numerical tests, we consider a square sampling region $\Omega=[-1,1]^2$ and discretize it into $64\times64$ pixels. We take $N_{i}$ sources and $N_r=32$ receivers that are equally distributed on a circle of radius $3$ centered at the origin and will consider $N_i=1, 2, 4, 8$ and $16$. The scattered fields are generated numerically using the method of moments. The noisy measured scattered field is then generated pointwise by the formula
	\begin{equation}
		u^{s}_{\delta}(x)=u^{s}(x)+\delta\Vert u^{s}\Vert_{2}\frac{\zeta_{r}(x)+i\zeta_{i}(x)}{\sqrt{2}} \quad \text{for} \quad x \in S,
	\end{equation}
    where $\delta$ refers to the relative noise level,
    $\zeta_{r}(x)$ and $\zeta_{i}(x)$ refer to the real and imaginary parts of the noise which follow the standard normal distribution.
  The wavelength is $\lambda=0.75$ and $k = 2\pi/\lambda$. To use the data augmentation strategy that we discussed in section \ref{sec_Data_aug}, we consider the following loss function to train the neural networks:
	\begin{equation}
		\begin{split}	Loss =& \frac{1}{2}\bigg(\Vert \mathcal{G}_\Theta(\mathbf{\Phi}_{\varepsilon_{r}})-\varepsilon_{r}\Vert_{2}^{2} + \alpha_{1} \text{TV}(\mathcal{G}_\Theta(\mathbf{\Phi}_{\varepsilon_{r}}))  + \alpha_{2}(1-\text{SSIM}(\mathcal{G}_\Theta(\mathbf{\Phi}_{\varepsilon_{r}}),\varepsilon_{r}))\\
		&+\Vert \mathcal{G}_\Theta(\mathbf{\Phi}_{\mathcal{R}_{\pi}(\varepsilon_{r})})-\mathcal{R}_{\pi}(\varepsilon_{r})\Vert_{2}^{2} + \alpha_{1} \text{TV}(\mathcal{G}_\Theta(\mathbf{\Phi}_{\mathcal{R}_{\pi}(\varepsilon_{r})}))  + \alpha_{2}(1-\text{SSIM}(\mathcal{G}_\Theta(\mathbf{\Phi}_{\mathcal{R}_{\pi}(\varepsilon_{r})}),\mathcal{R}_{\pi}(\varepsilon_{r})))\bigg)
		\end{split}
	\label{loss_func}
	\end{equation} 
for $N_i>1$, where $\mathbf{\Phi}_{\mathcal{R}_{\pi}(\varepsilon_{r})}$ can be computed from $\mathbf{\Phi}_{\varepsilon_{r}}$ by (\ref{equa26}), and the rotation operator $\mathcal{R}_\pi$ should be replaced by the symmetry operator $\mathcal{S}_{d_1}$ for $N_i=1$. We only use $\mathcal{R}_\pi$ and $\mathcal{S}_{d_1}$ for data augmentations in the numerical experiments, so that the amount of augmentation is the same for all $N_i$, although we can increase more data for $N_i>2$. The TV refers to the total variation which is commonly used as the regularizer in image reconstruction, and SSIM \cite{wang2003multiscale} is the structural similarity function which is a common metric for measuring the perceptual distance between two images and was also added to the loss function in \cite{huang2020deep}. $\alpha_{1},\alpha_{2}$ are the corresponding weights, and we set $\alpha_{1}=0.01, \alpha_{2}=0.05$ for the circle example, $\alpha_{1}=0.05, \alpha_{2}=0.05$ for the MNIST example. The loss function is computed averagely on the training data for each iteration, and the Adam optimizer is used to update the learnable parameters. $5\%$ Gaussian noise is imposed on the scattered field for the training data, while we test the trained network with different high noise levels. We choose the network which achieves the smallest validation error during training as the final model. In the training stage, a server with GeForce GTX 1080Ti and 256 GiB of system RAM is used. It takes about 0.5 hours for the training of each example, and it takes less than 1.0 seconds for reconstructing the scatterers using the trained networks. 
    
    To reconstruct the true medium $\varepsilon_r$, we employ the following average form:
    \begin{equation}
    	\tilde{\varepsilon}_{r}=\frac{1}{2}\bigg(\mathcal{G}_\Theta(\mathbf{\Phi}_{\varepsilon_{r}})+\mathcal{R}_{\pi}(\mathcal{G}_\Theta(\mathbf{\Phi}_{\mathcal{R}_{\pi}(\varepsilon_{r})}))\bigg)\,.
    \end{equation} 
    To quantify the quality of the reconstruction, we compute the relative L2 error via
    \begin{equation}
    	R_{e} = \frac{\Vert\tilde{\varepsilon}_{r}-\varepsilon_{r} \Vert_2}{\Vert \varepsilon_{r}\Vert_2}.
    \end{equation}
	
	\subsection{Circle dataset example}
	In the first example, we simulate the scatters by employing a dataset of circles. The number of circles in each image is randomly set between 1 and 3, the radius is taken from the uniform distribution $U(0.15,0.4)$, and the relative permittivity is taken from $U(1.5,2.0)$. The overlapping is allowed between circles. We use $3000$ images as the training data, $200$ images as the validation data, and $200$ images as the testing data to compute the testing error. The batch size is taken as $6$ and we use a total of $30$ epochs to train the neural networks. As we use the loss function (\ref*{loss_func}), we actually have $6000$ training data and the batch size is $12$. The learning rate starts at $0.001$ and decreases by a factor of $0.5$ every $3$ epochs. The trained networks are then used to test the following tests.
	
	\subsubsection{Tests with testing data from the circle dataset}
	
	In Fig.\,\ref{tab:fig-Circle}, reconstructed images of four tests from the testing data with noise levels $\delta=15\%$ and $\delta=40\%$ are presented. From the numerical results, we can see that when the number of incidences is small, the trained networks can provide very satisfactory results if the circles are well separated, while it can only estimate the locations and the rough shape of scatterers if they are overlapped. This is reasonable as the interaction between scatterers is strong if they are very close to each other. The reconstructions can be further improved if we use more incident fields, and the trained networks with $N_{i}=8$ or $16$ can obtain very high-quality results for all tests. The average relative error and SSIM for the testing data are presented in Table\,\ref{tab:error_circle} and Fig.\,\ref{fig:Error_Circle}. From Fig.\,\ref{fig:Error_Circle}, we can intuitively observe the impacts of the noise level and the number of incidences on the performance. An important finding is that as $N_{i}$ increases, the error gaps between different noise levels become smaller, which means that multiple measurement data can make the reconstruction more robust to noise.
		\begin{figure}[htp]\small
		\begin{center}
			\begin{tblr}
				{colspec = {X[-1]X[c]X[c,h]X[c,h]X[c,h]X[c,h]X[c,h]},
					stretch = 0,
					rowsep = 0pt,}
				Noise Level& Ground truth& $N_{i}$=1& $N_{i}$=2 &$N_{i}$=4&$N_{i}$=8& $N_{i}$=16\\
				15\%&\SetCell[r=2]{c}\includegraphics[width=0.15\textwidth]{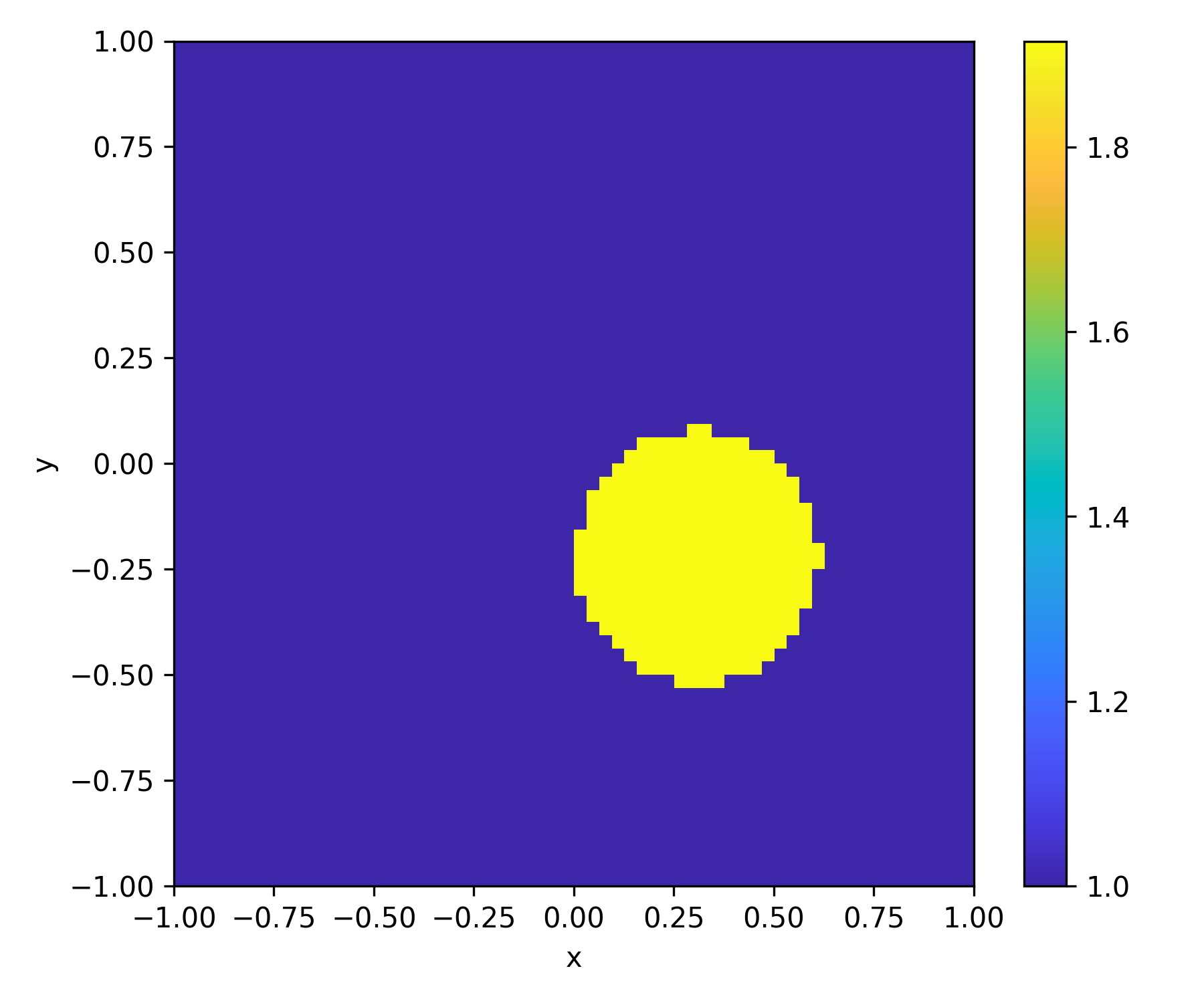}&
				\includegraphics[width=0.15\textwidth]{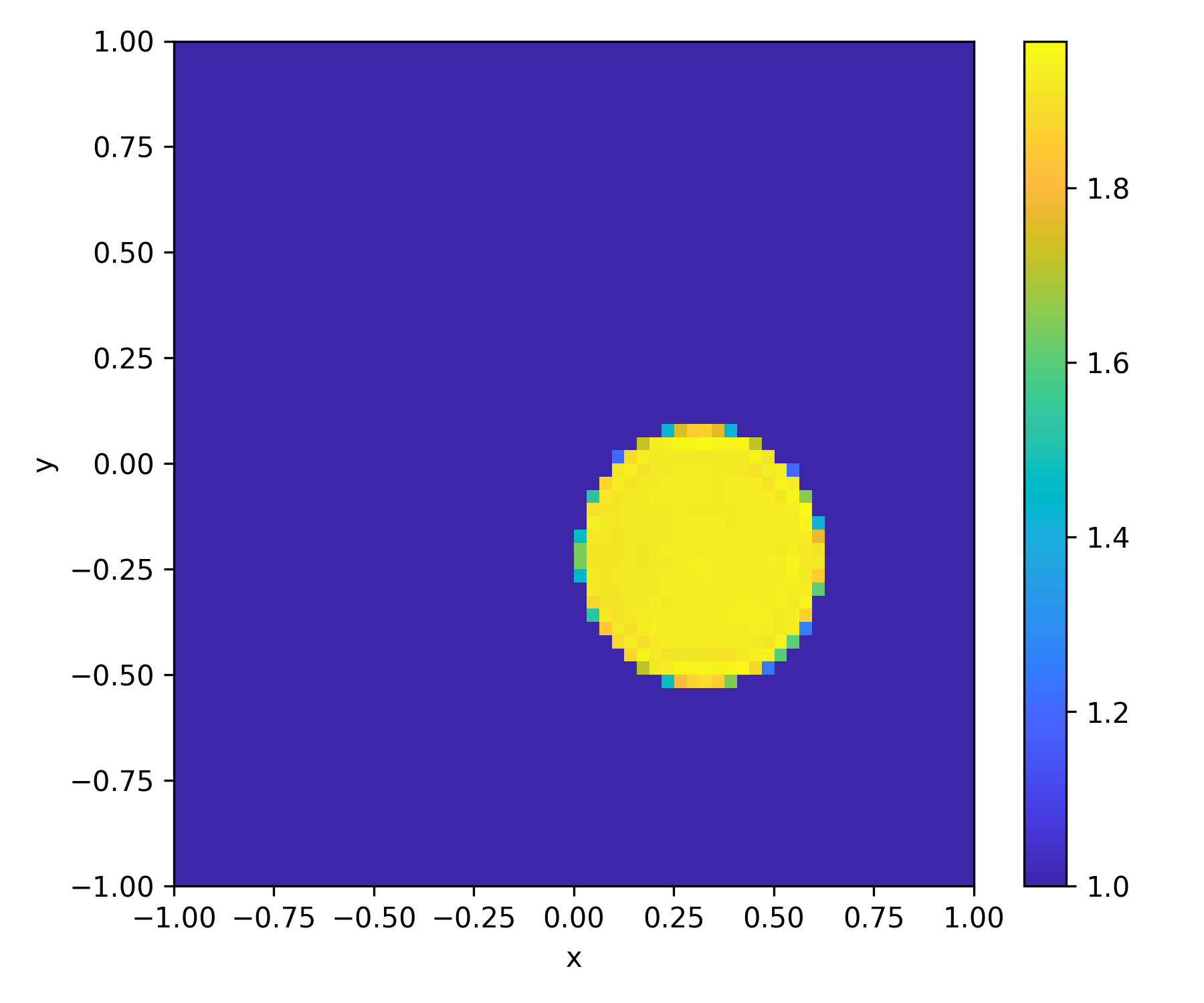}&
				\includegraphics[width=0.15\textwidth]{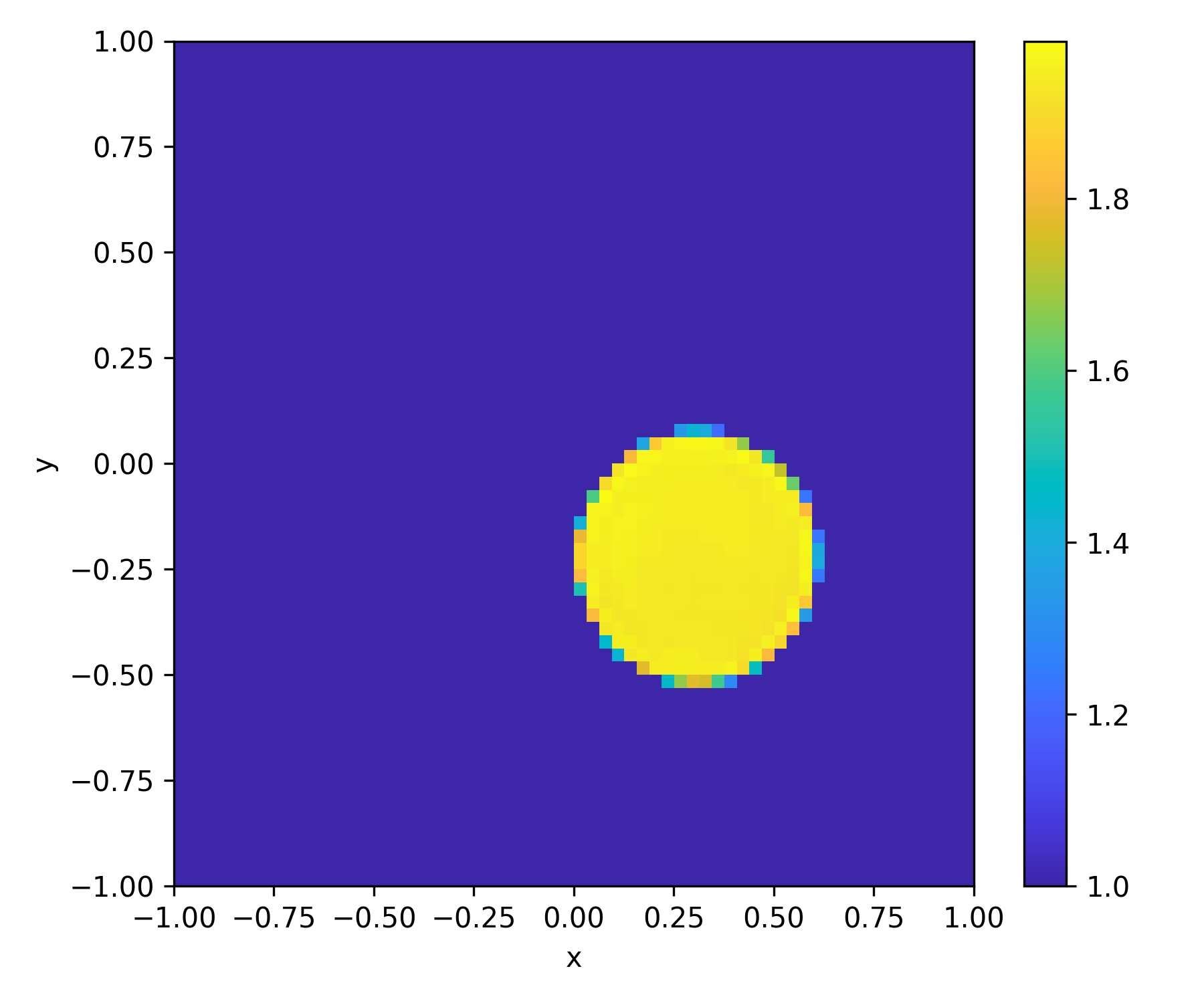}&
				\includegraphics[width=0.15\textwidth]{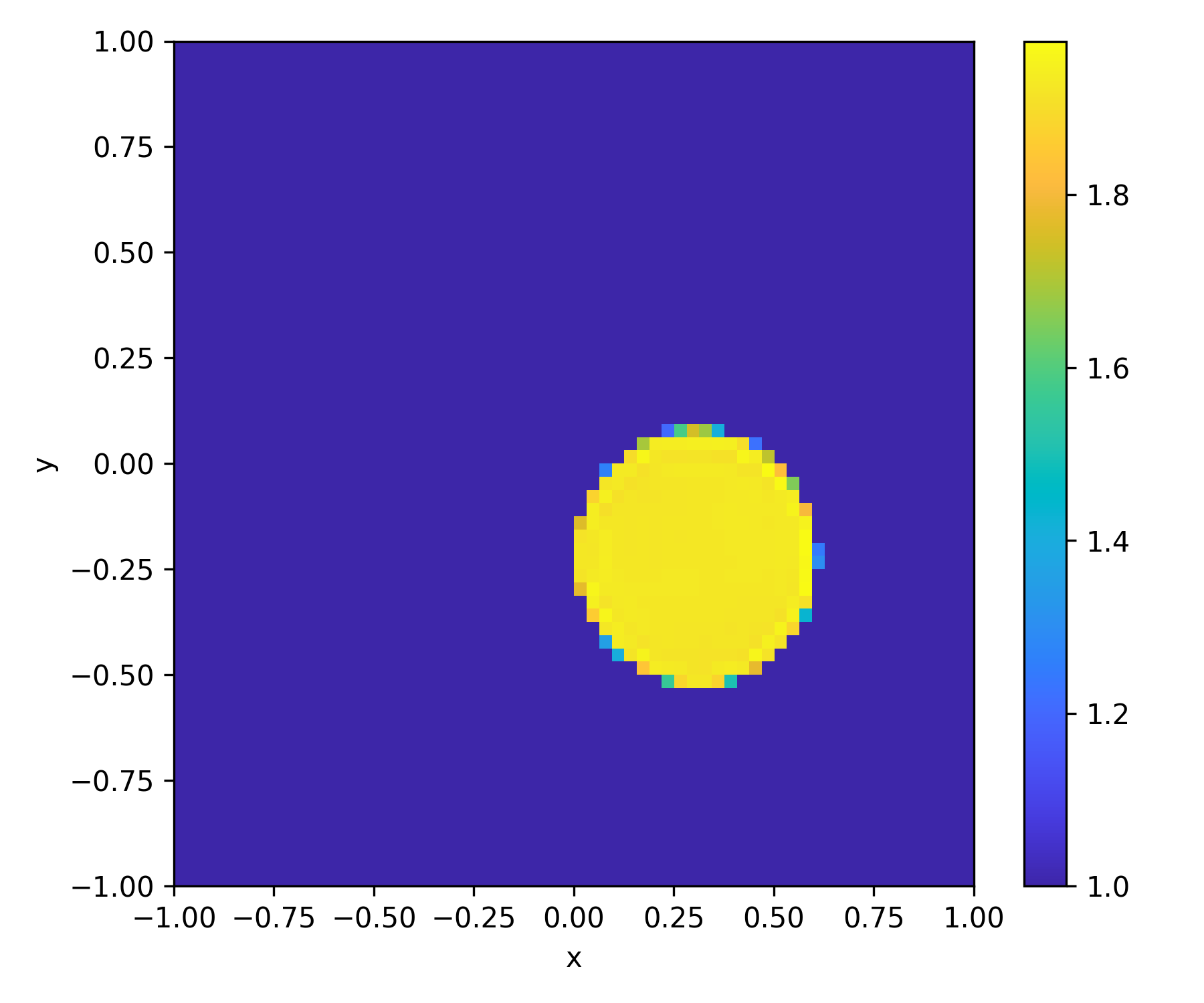}&
				\includegraphics[width=0.15\textwidth]{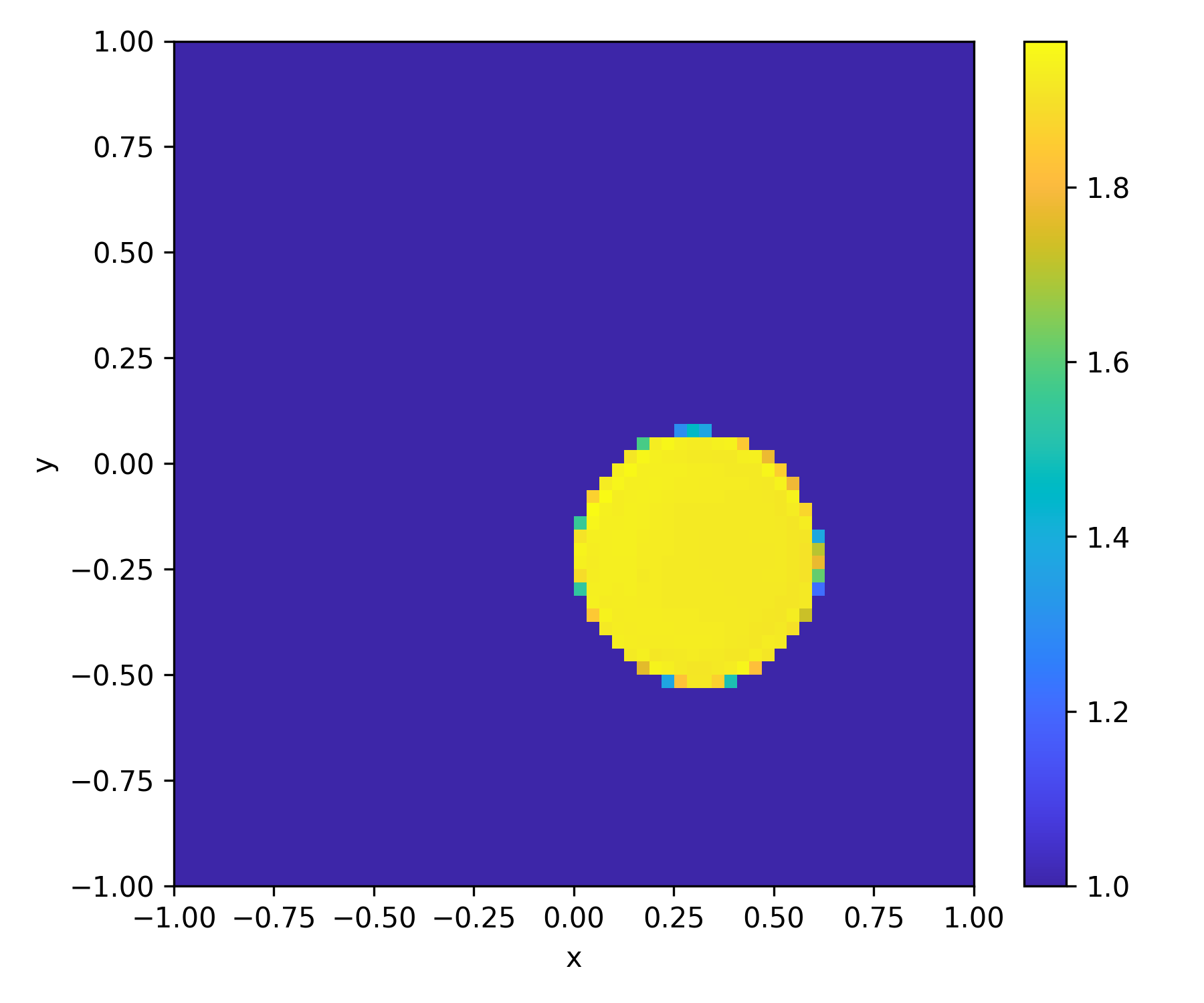}&
				\includegraphics[width=0.15\textwidth]{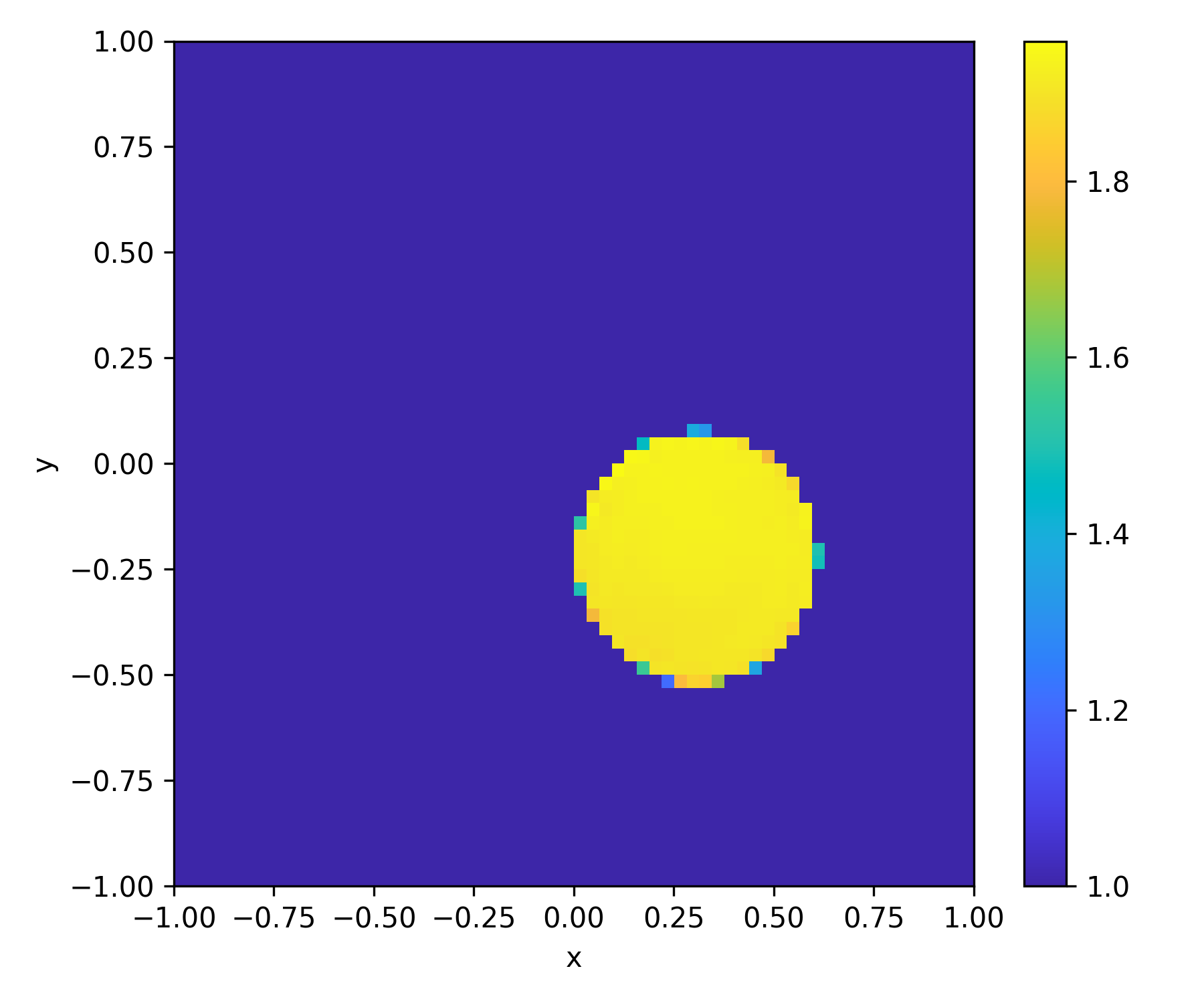}
				\\
				40\%& &
				\includegraphics[width=0.15\textwidth]{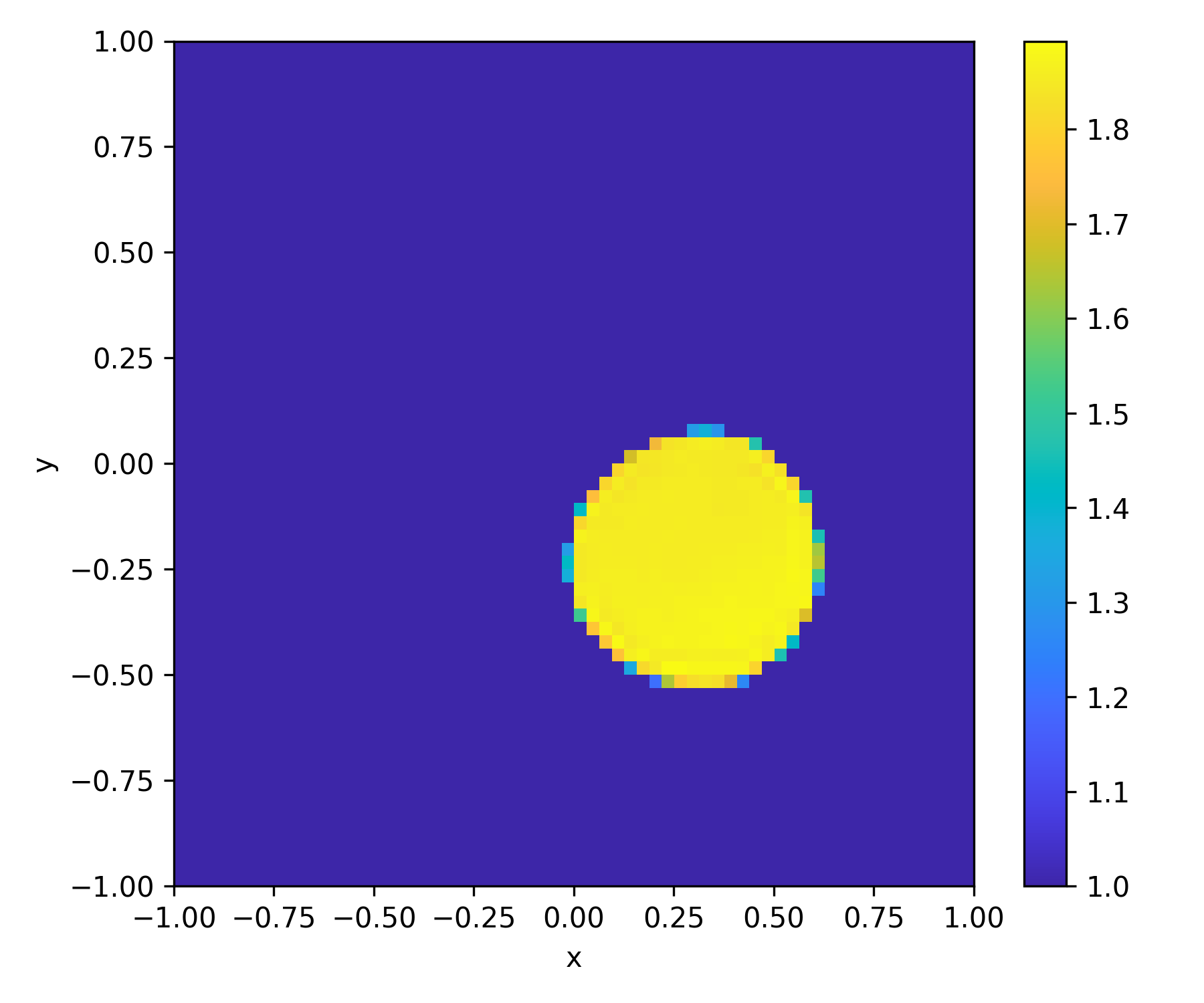}&
				\includegraphics[width=0.15\textwidth]{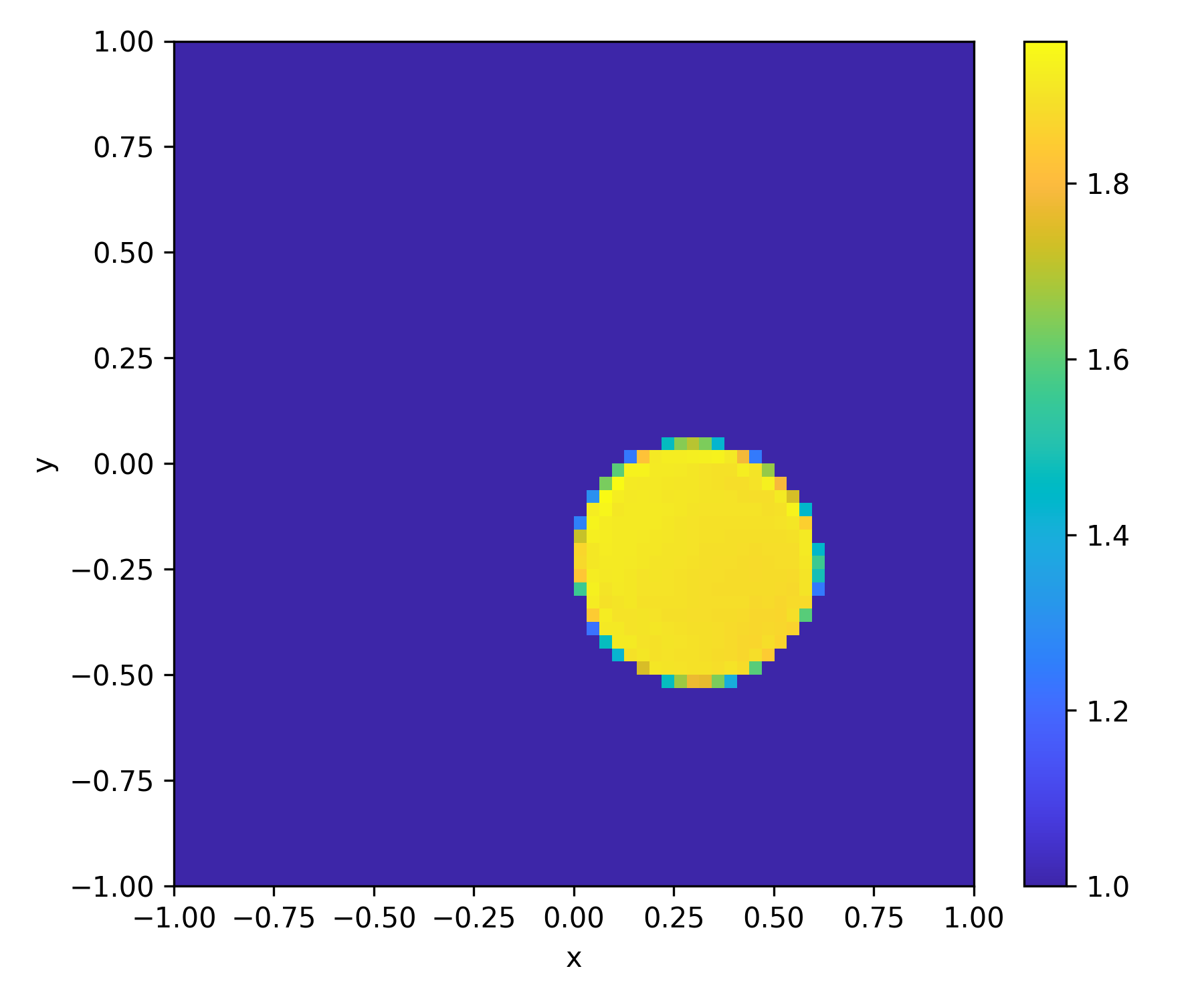}&
				\includegraphics[width=0.15\textwidth]{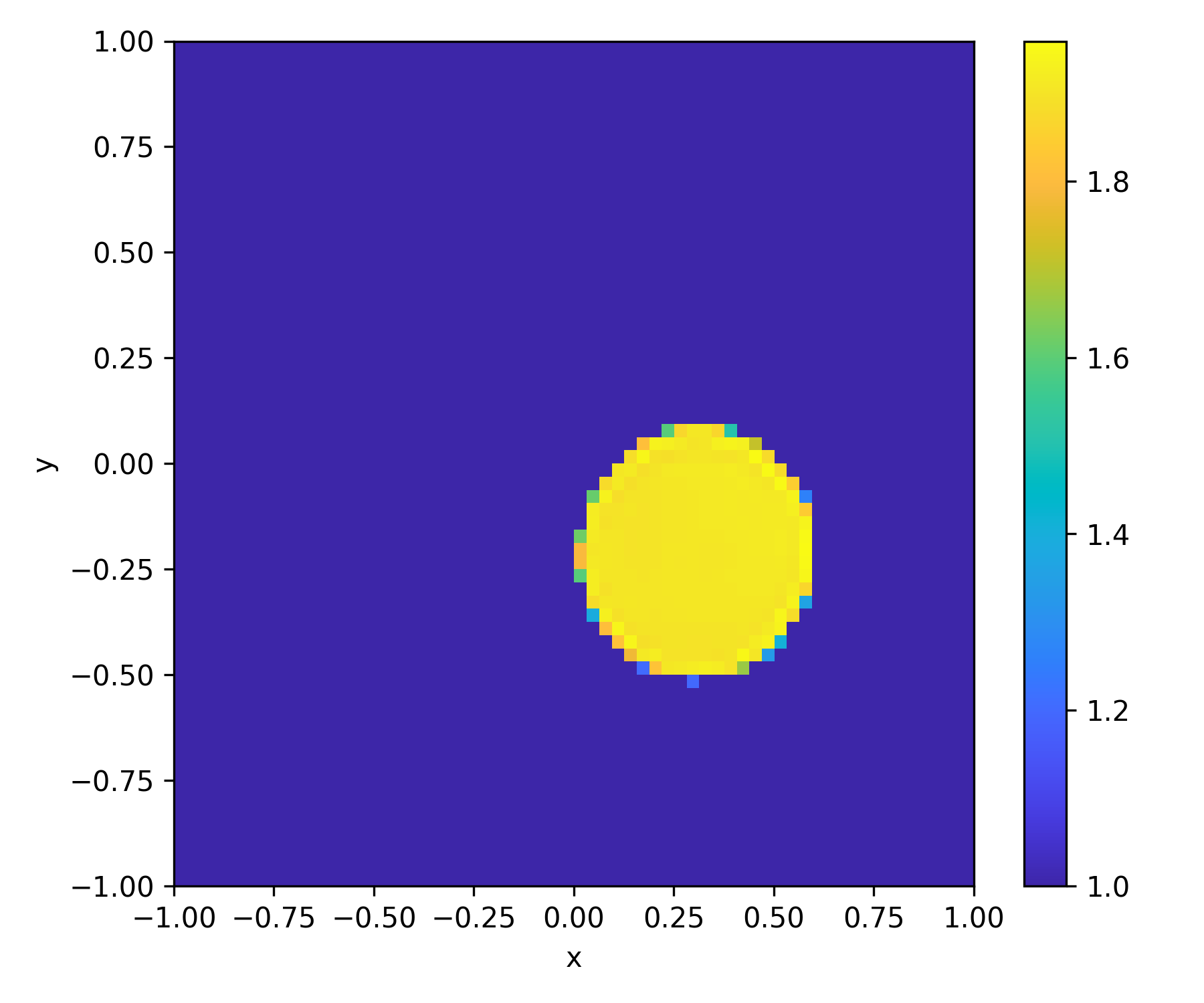}&
				\includegraphics[width=0.15\textwidth]{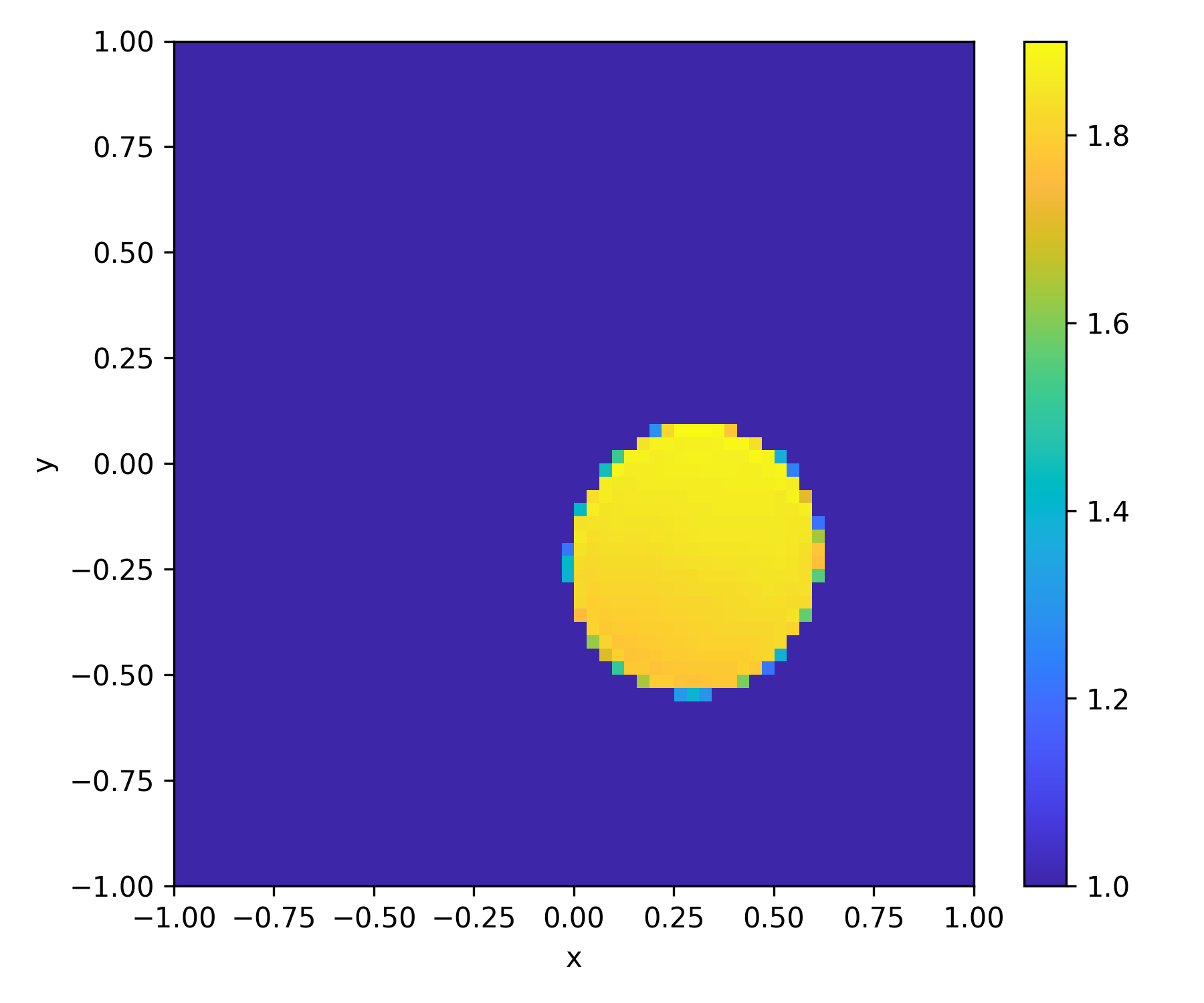}&
				\includegraphics[width=0.15\textwidth]{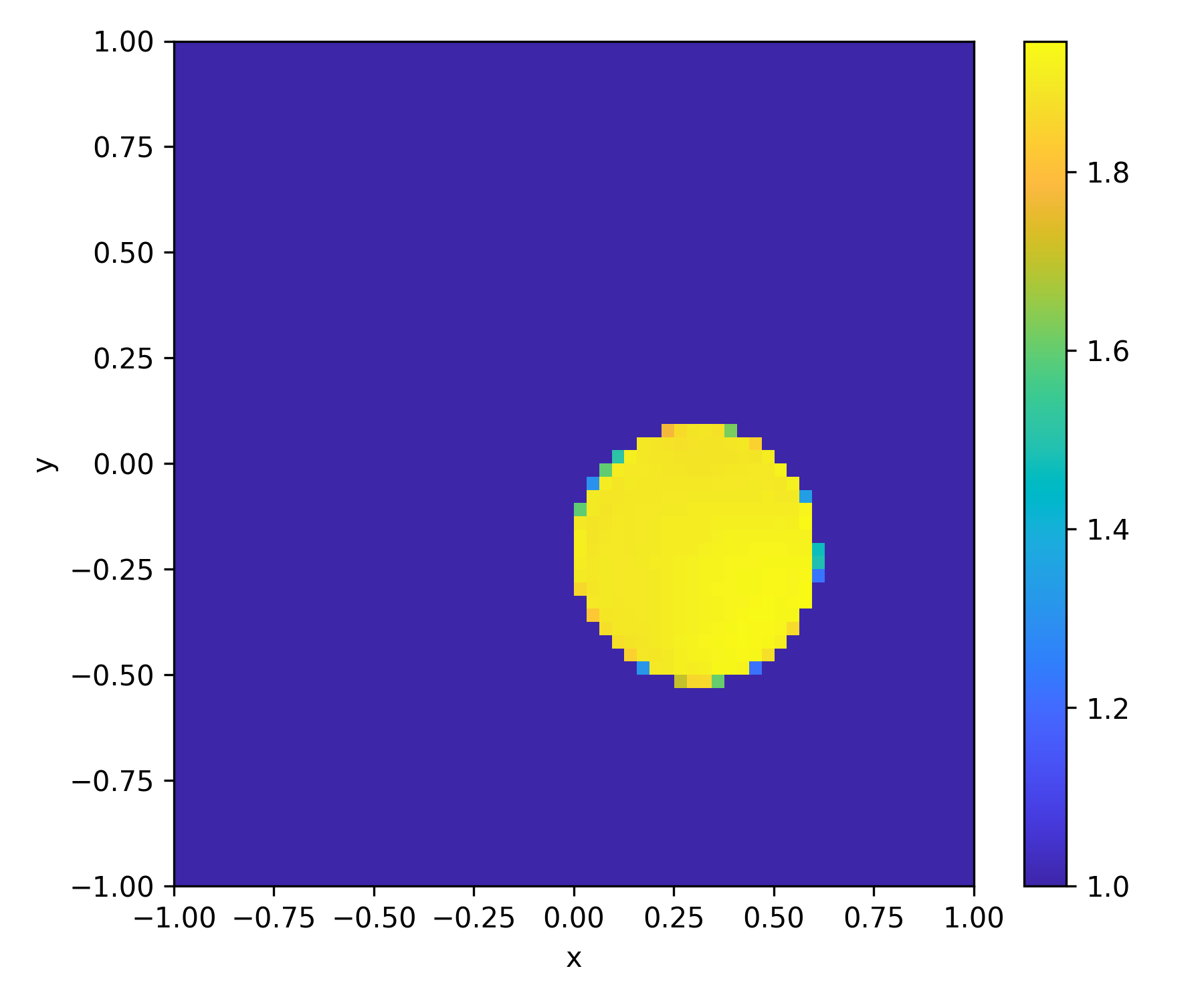}
				\\
				
				15\%&\SetCell[r=2]{c}\includegraphics[width=0.15\textwidth]{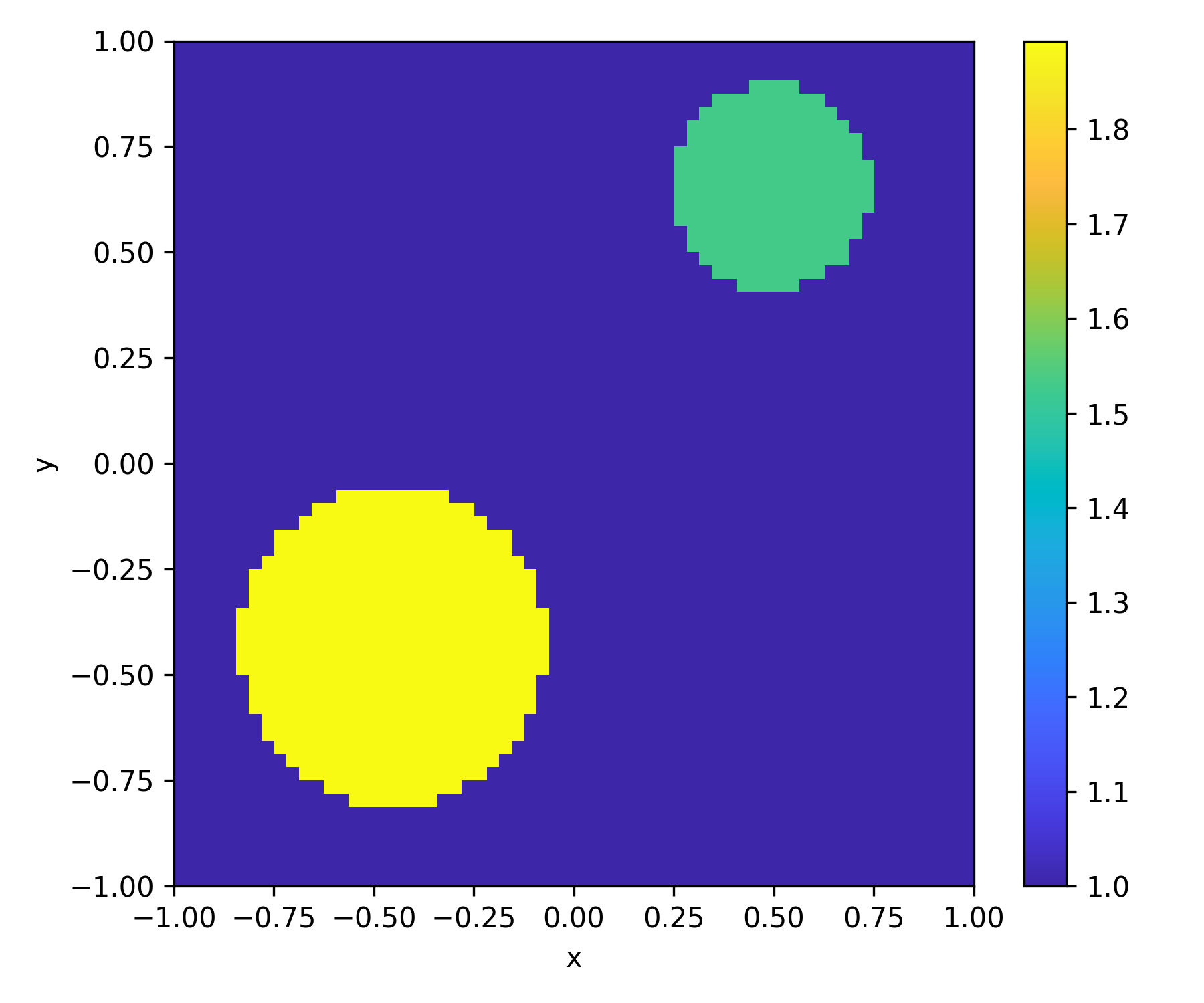} &
				\includegraphics[width=0.15\textwidth]{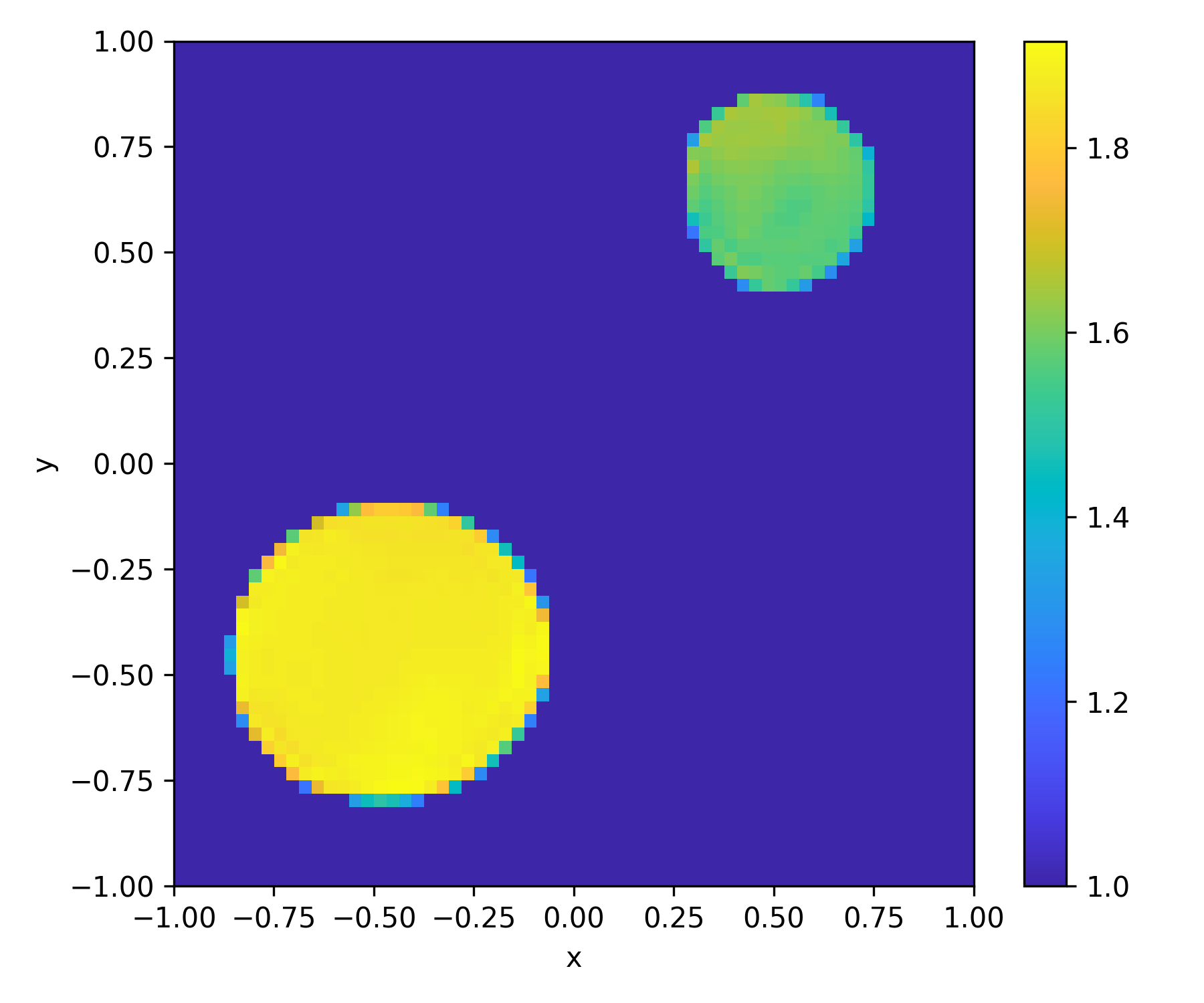}&
				\includegraphics[width=0.15\textwidth]{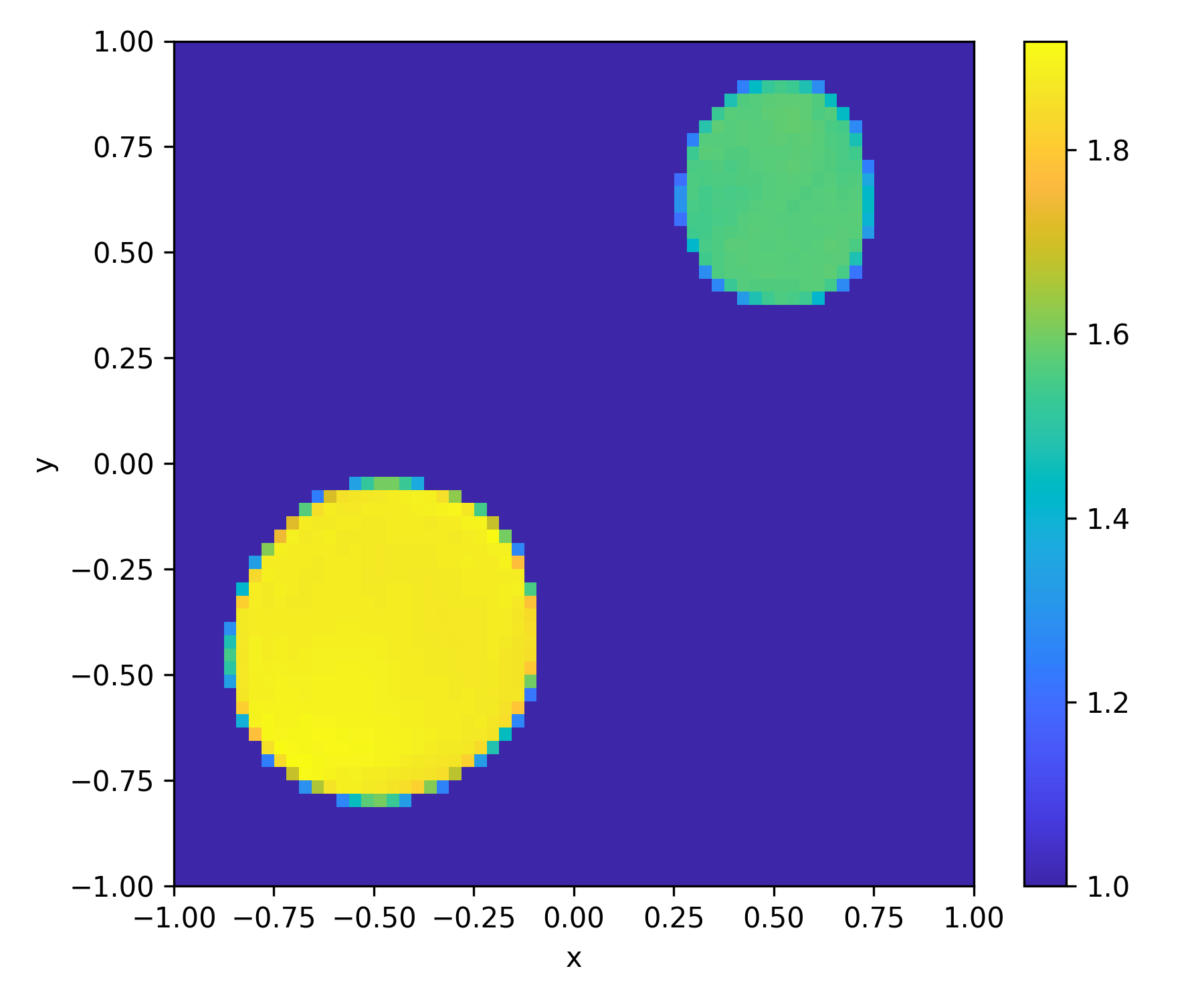}&
				\includegraphics[width=0.15\textwidth]{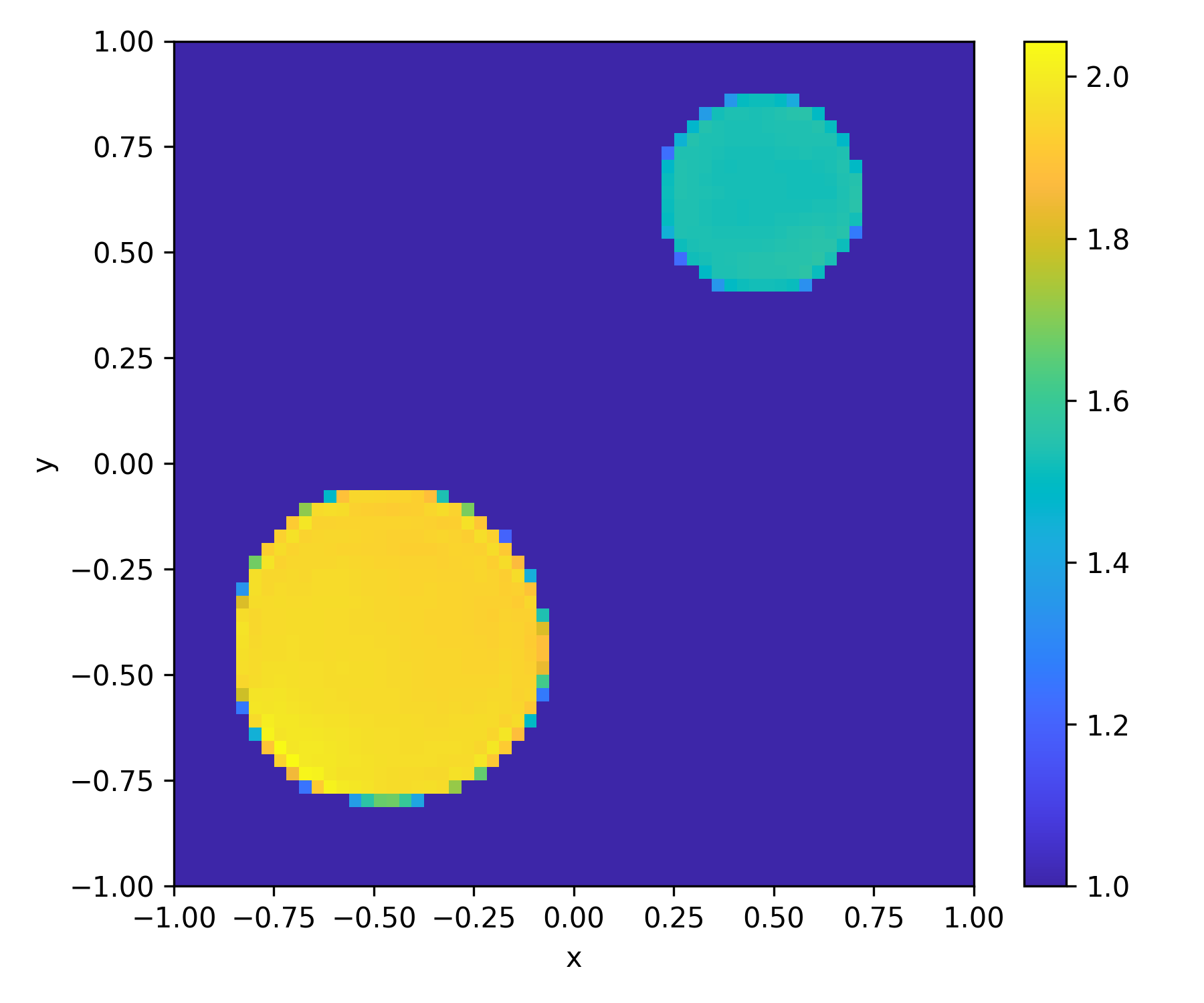}&
				\includegraphics[width=0.15\textwidth]{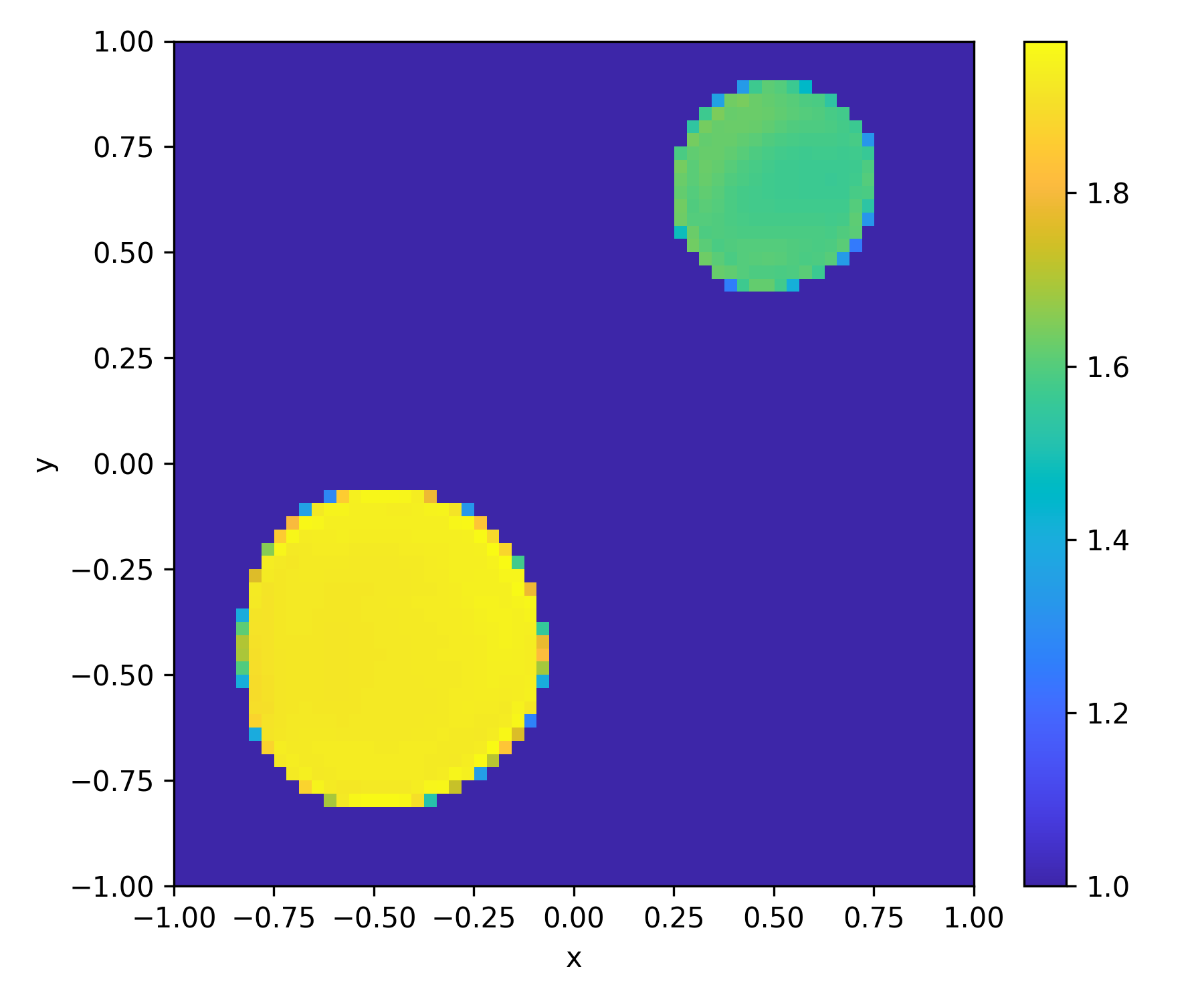}&
				\includegraphics[width=0.15\textwidth]{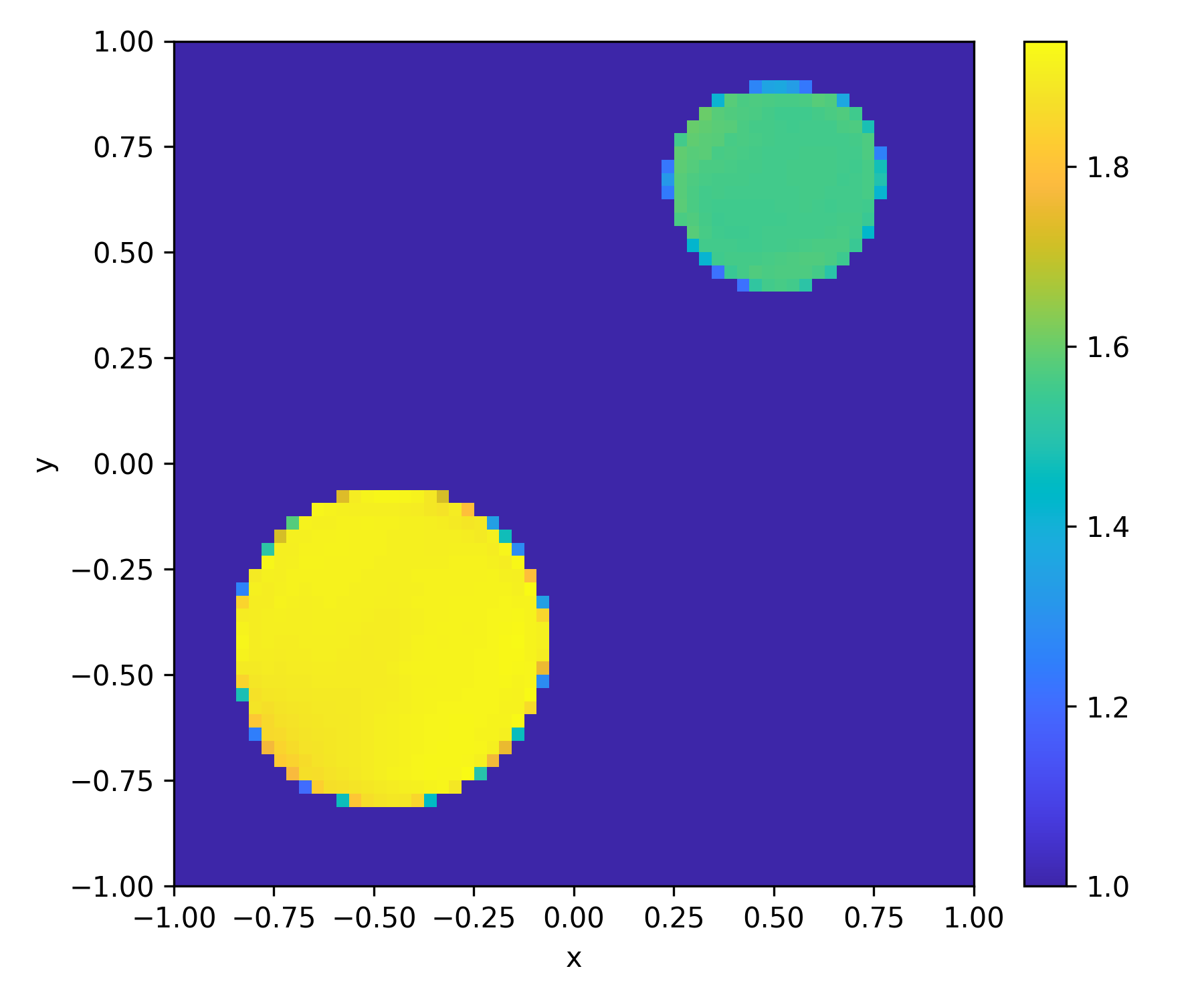}
				\\
				40\%& &
				\includegraphics[width=0.15\textwidth]{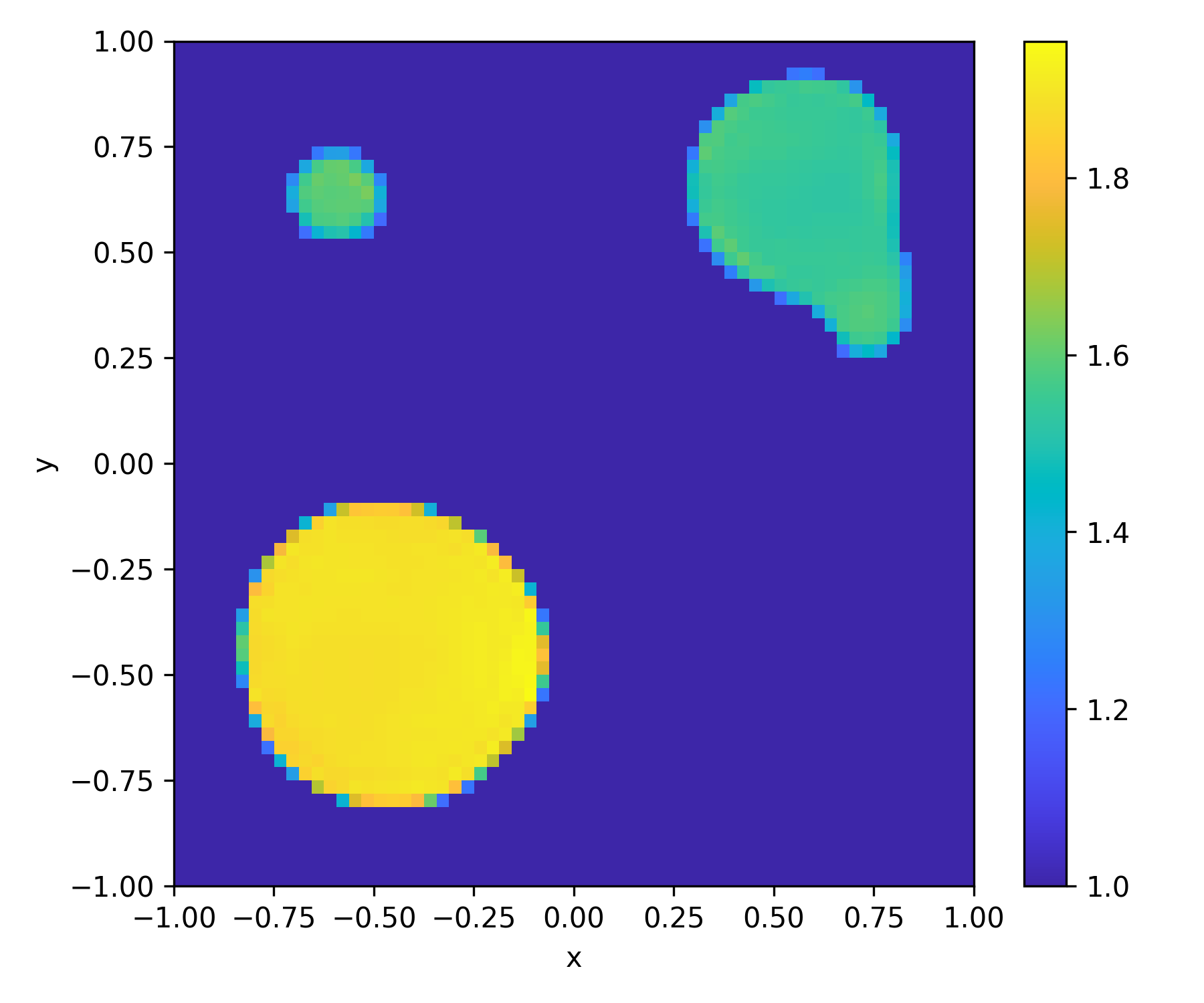}&
				\includegraphics[width=0.15\textwidth]{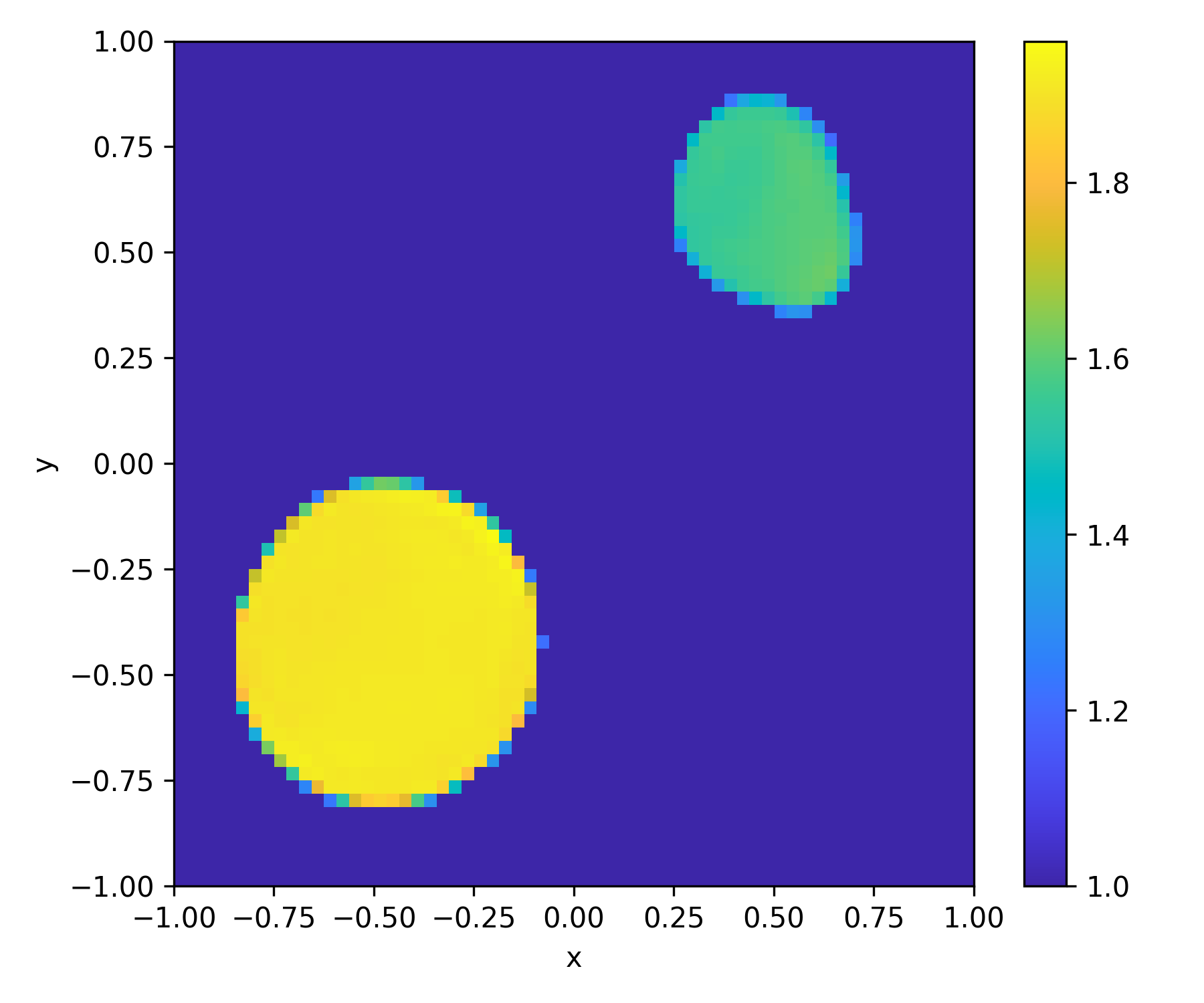}&
				\includegraphics[width=0.15\textwidth]{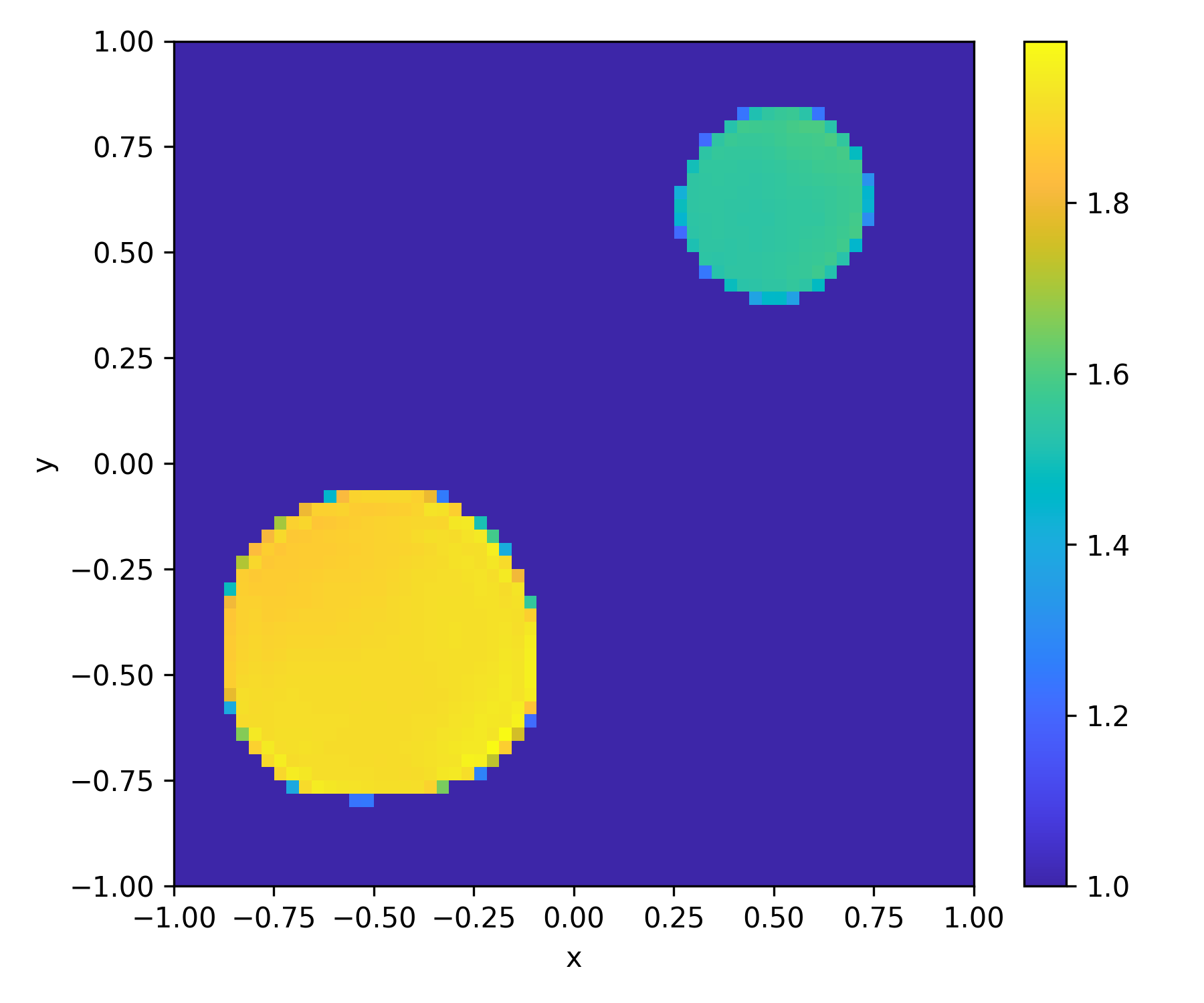}&
				\includegraphics[width=0.15\textwidth]{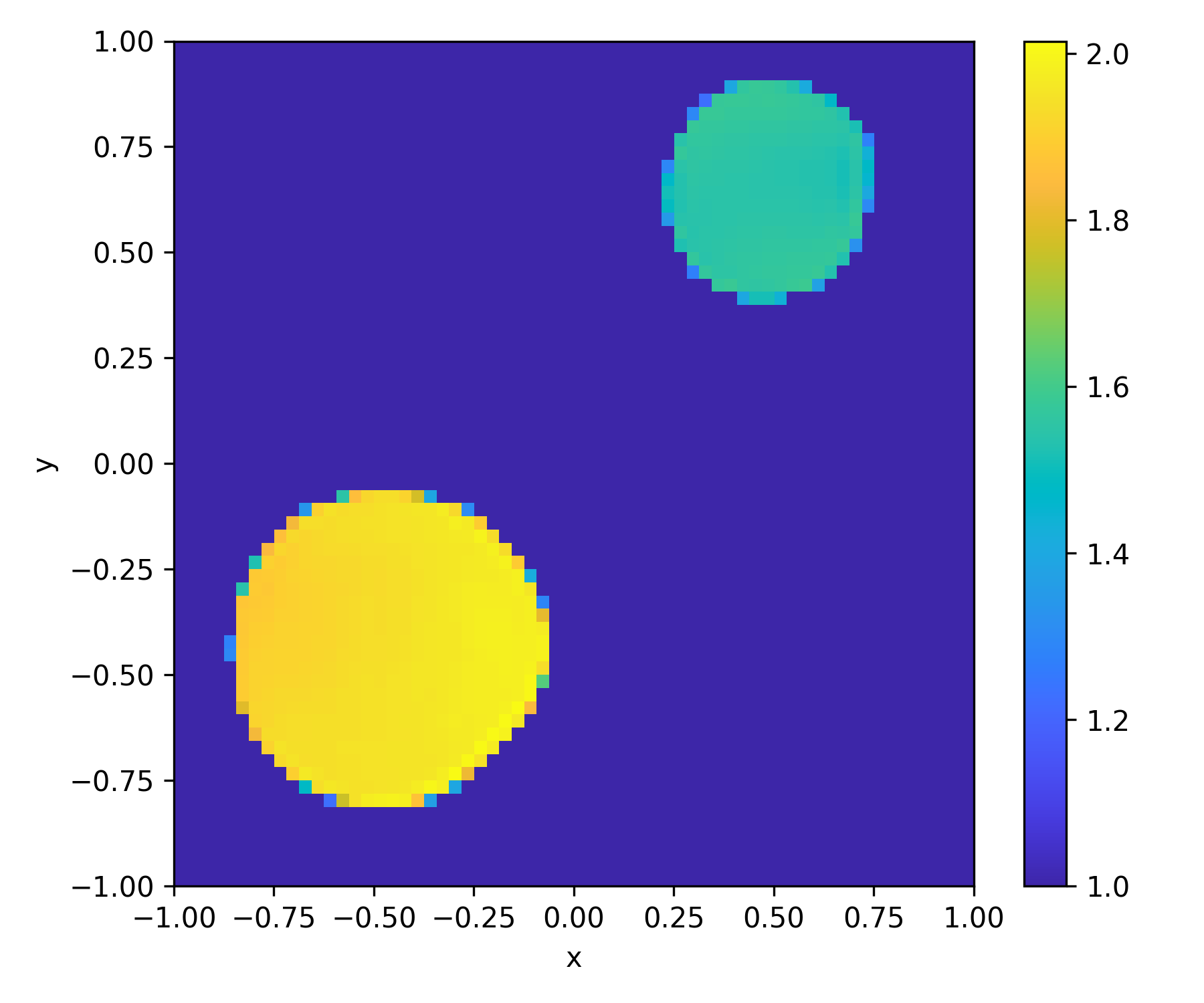}&
				\includegraphics[width=0.15\textwidth]{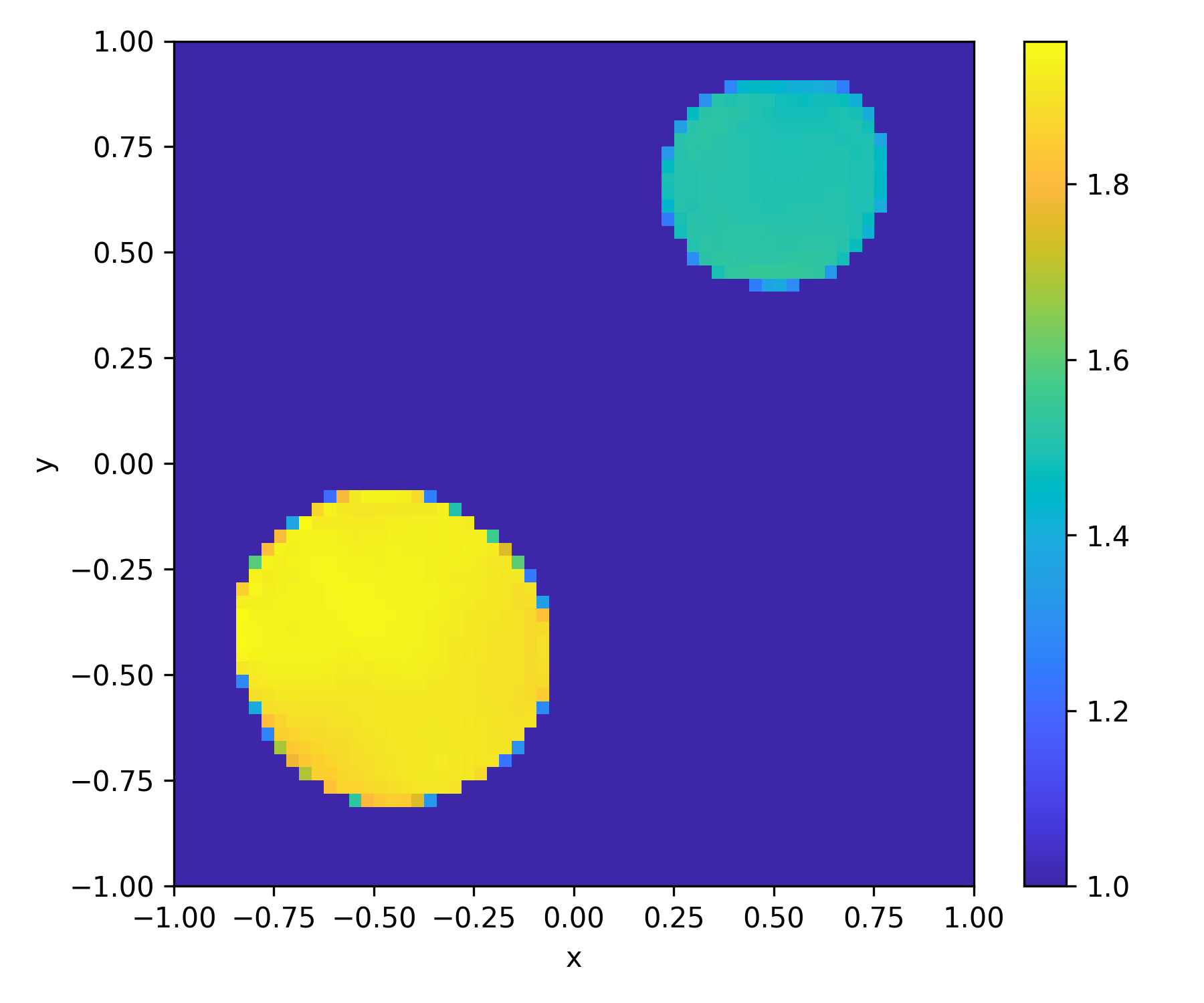}
				\\
				
				15\%&\SetCell[r=2]{c}\includegraphics[width=0.15\textwidth]{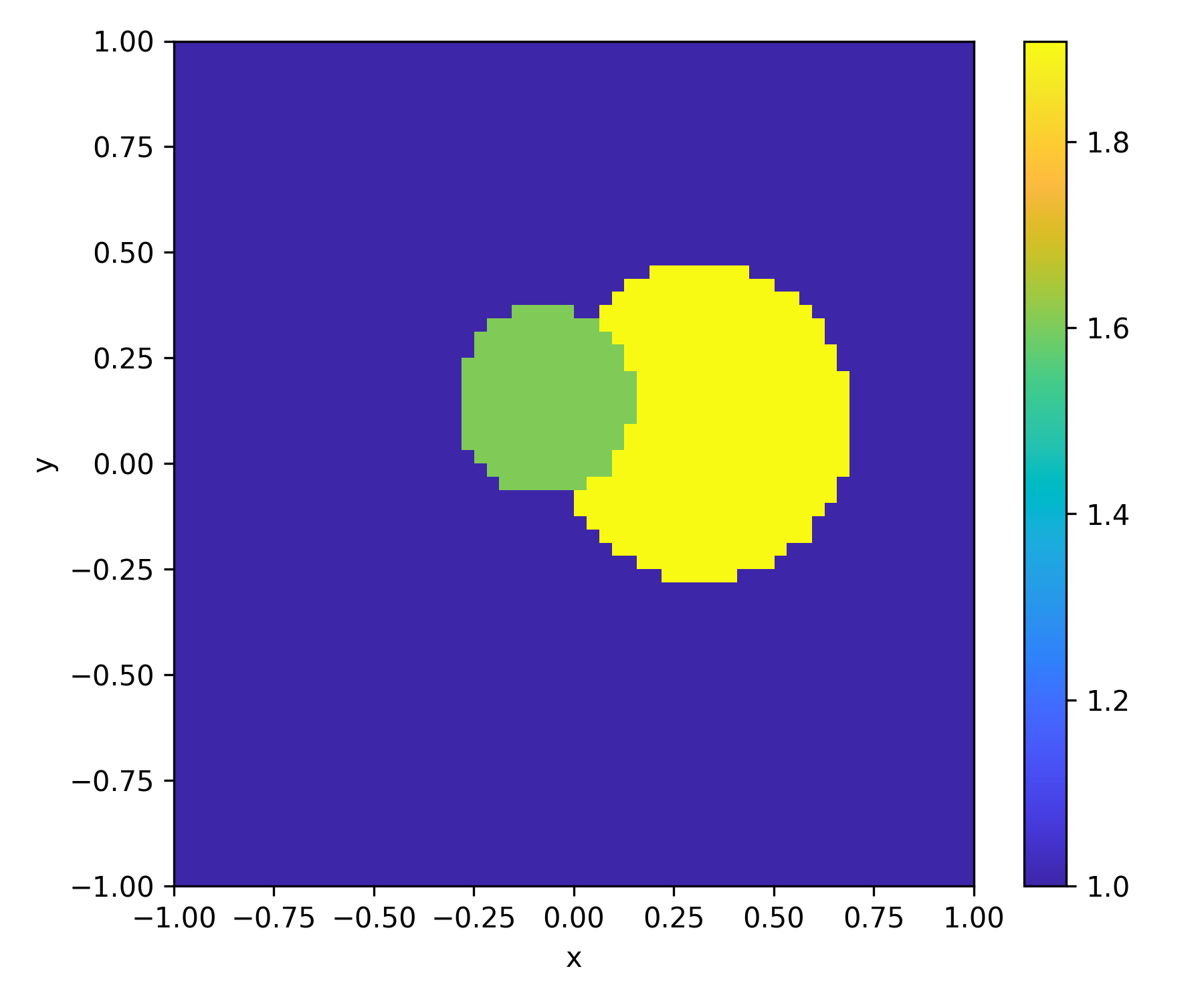}&
				\includegraphics[width=0.15\textwidth]{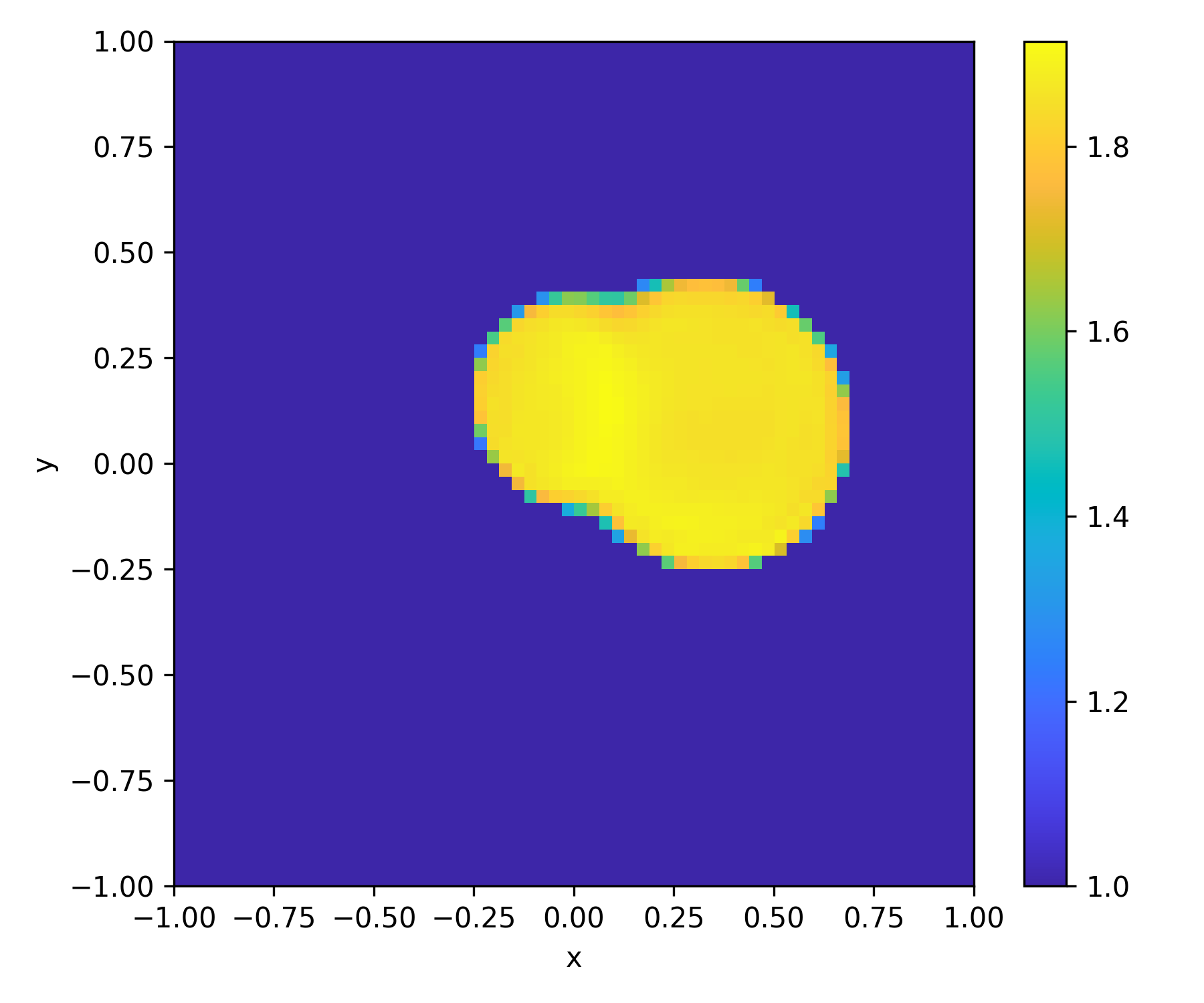}&
				\includegraphics[width=0.15\textwidth]{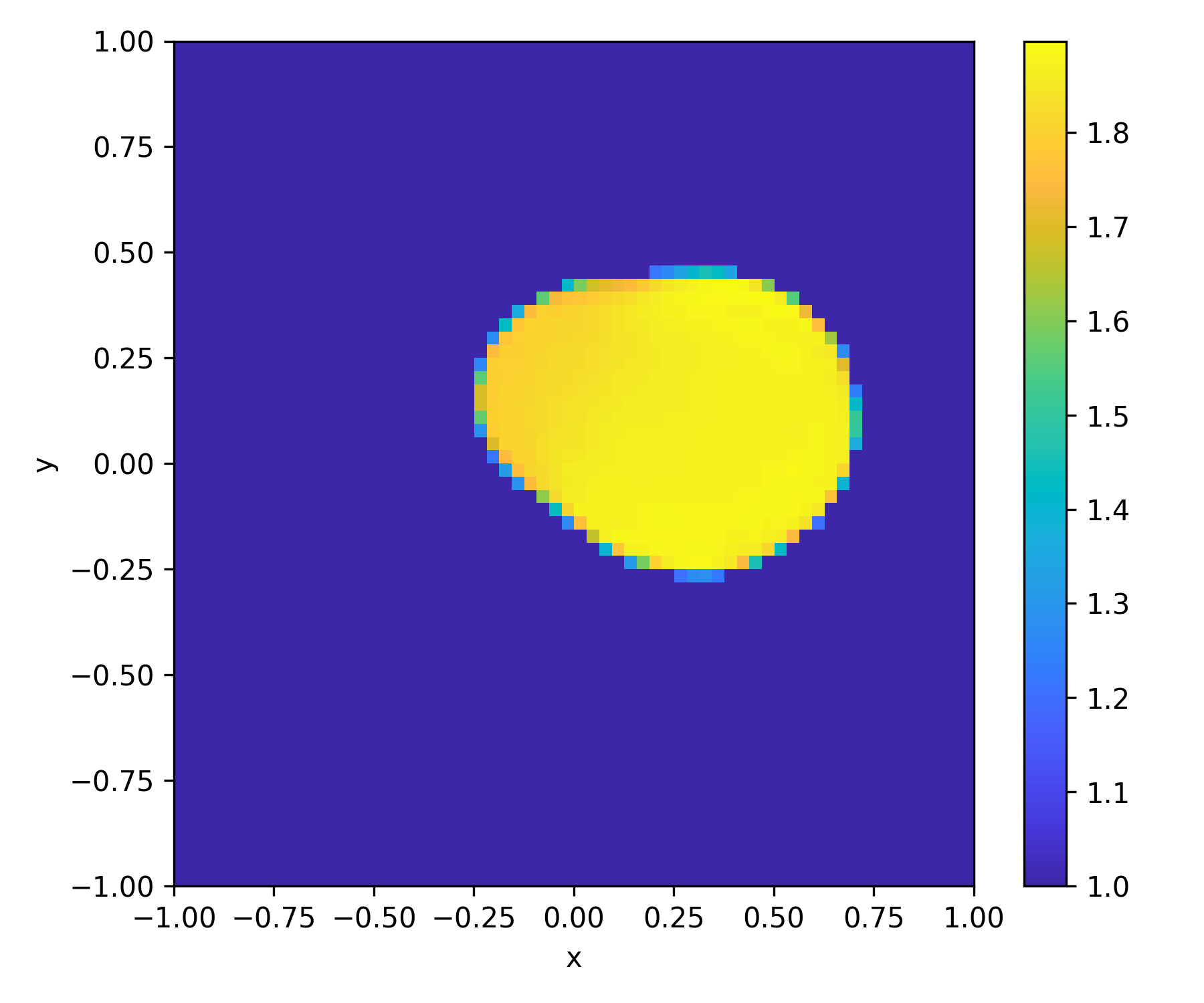}&
				\includegraphics[width=0.15\textwidth]{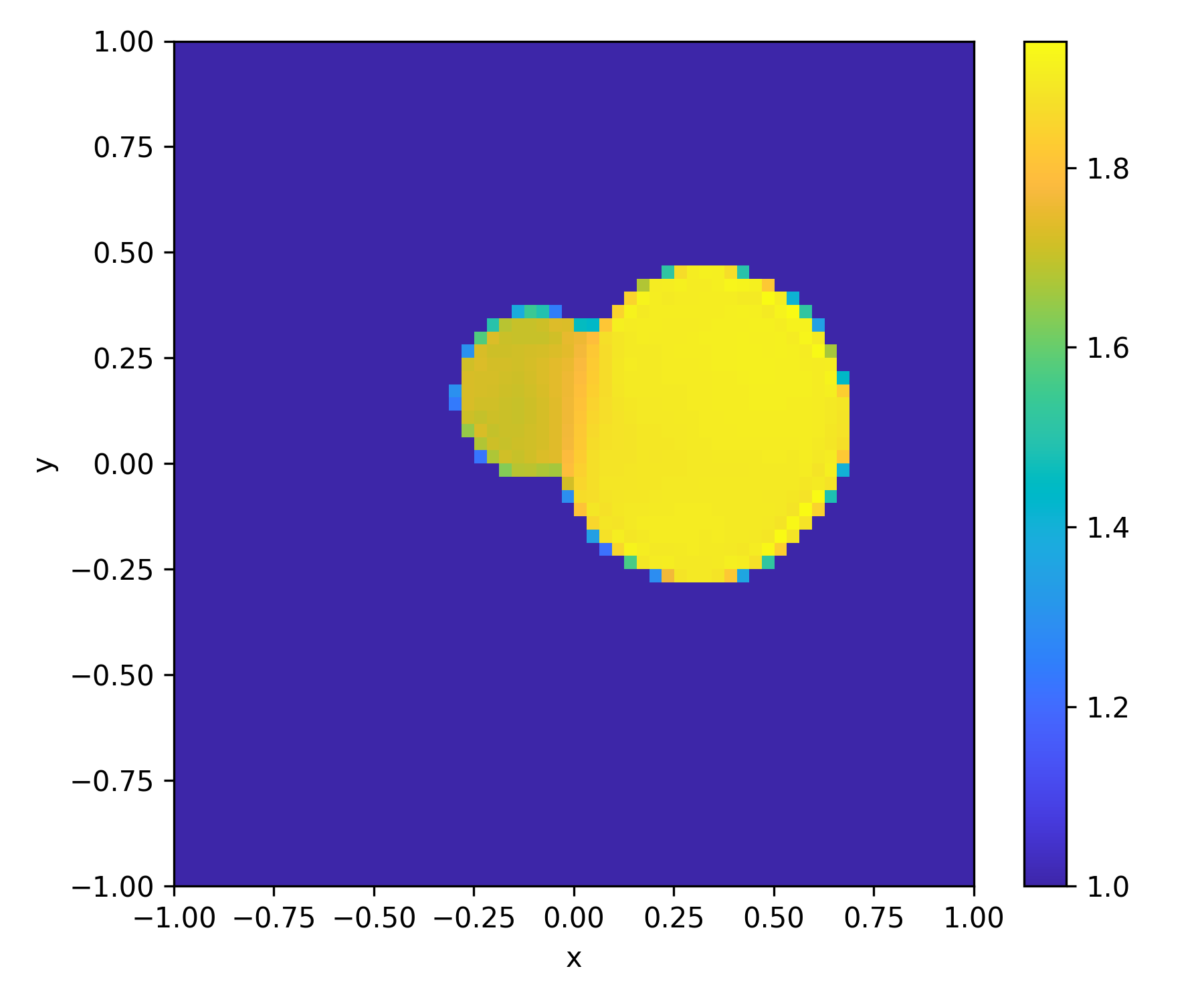}&
				\includegraphics[width=0.15\textwidth]{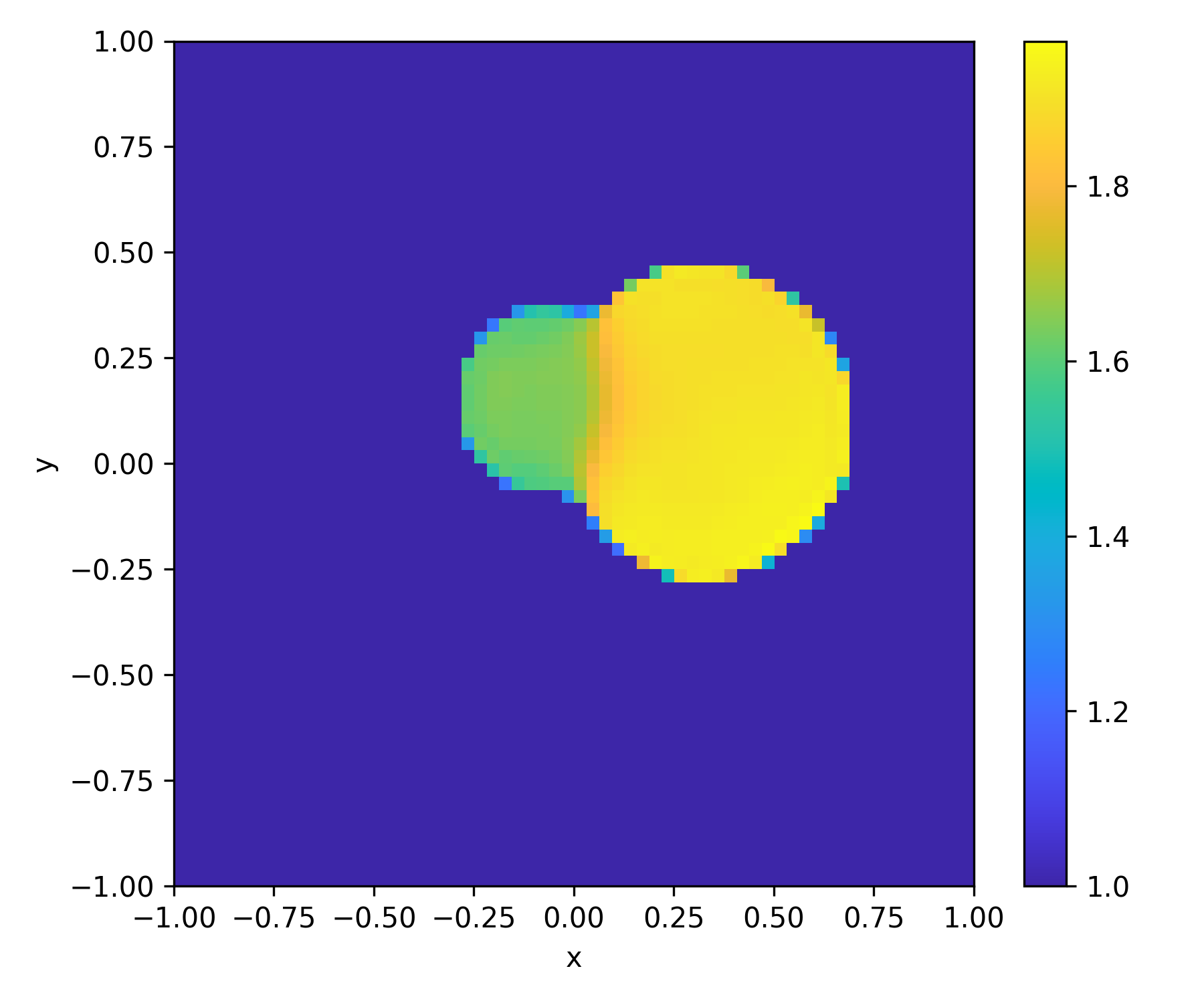}&
				\includegraphics[width=0.15\textwidth]{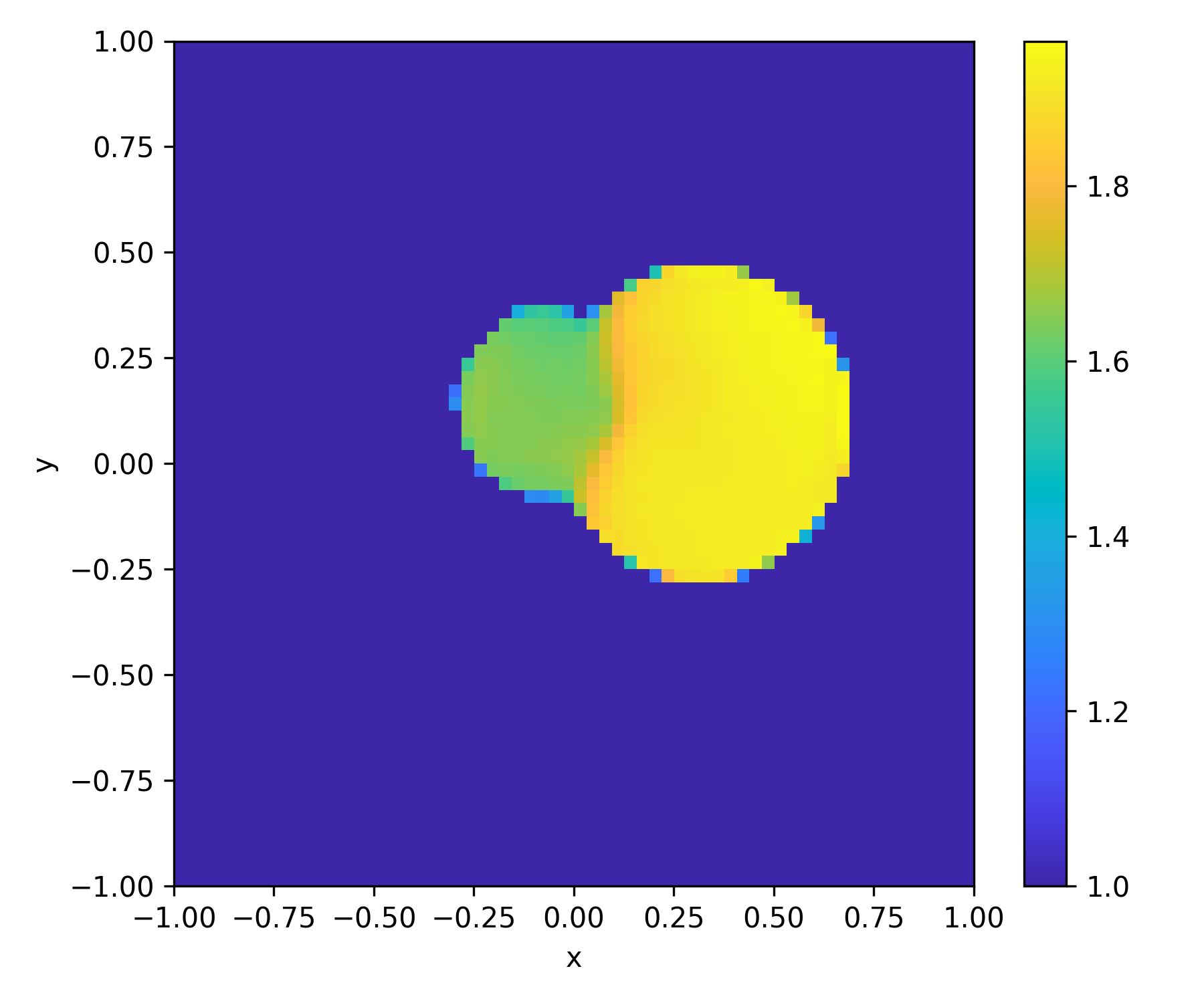}
				\\
				40\%& &
				\includegraphics[width=0.15\textwidth]{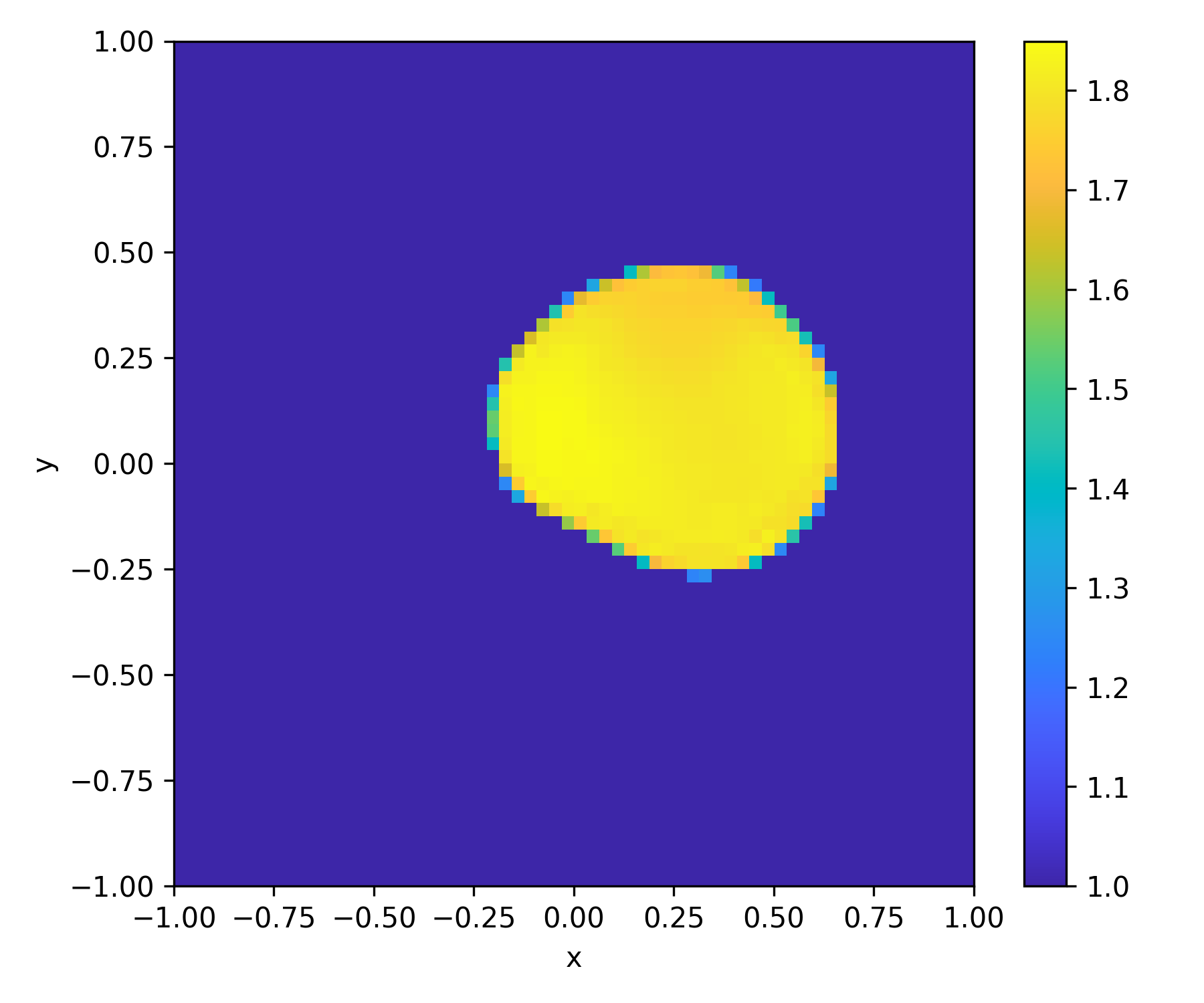}&
				\includegraphics[width=0.15\textwidth]{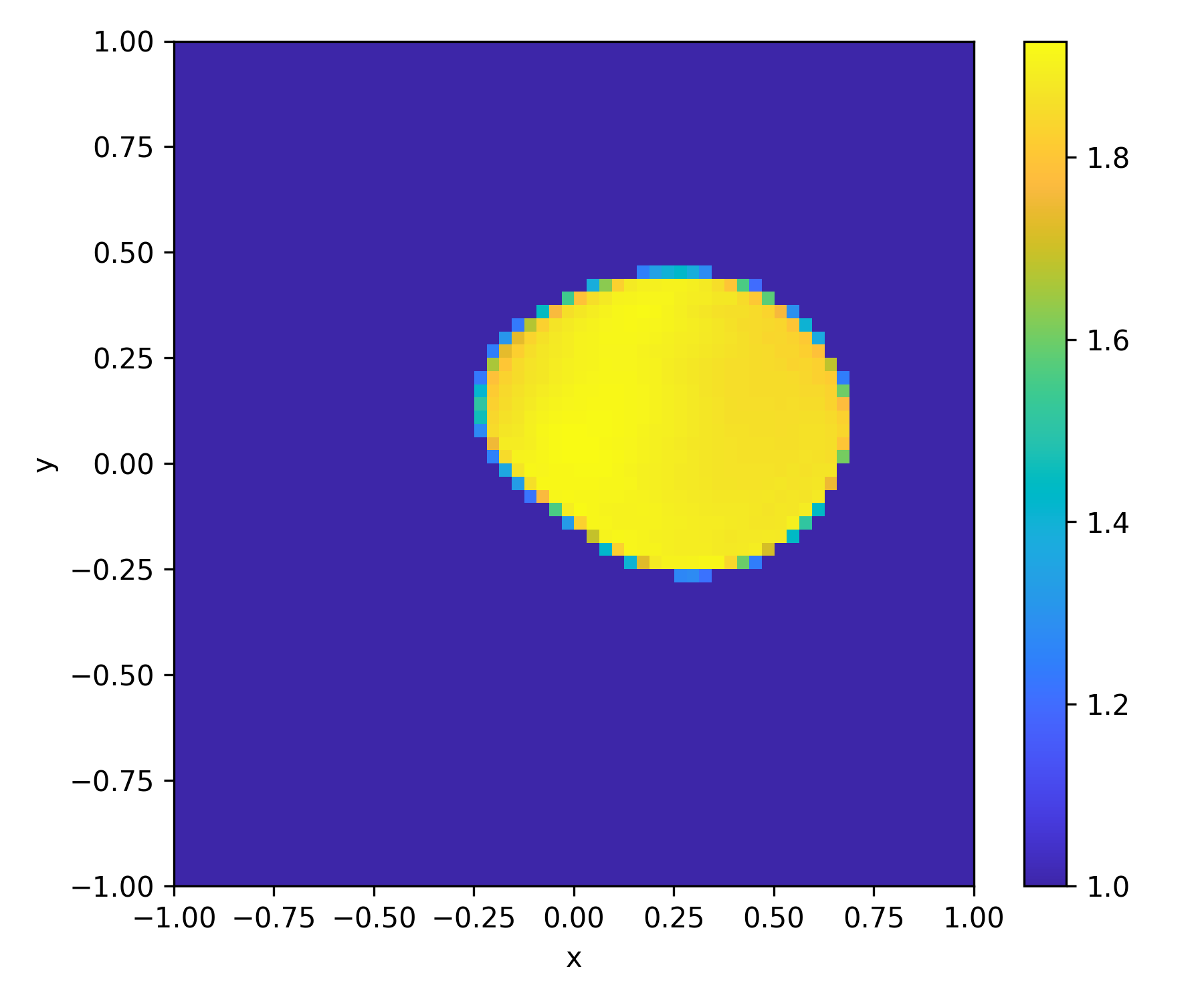}&
				\includegraphics[width=0.15\textwidth]{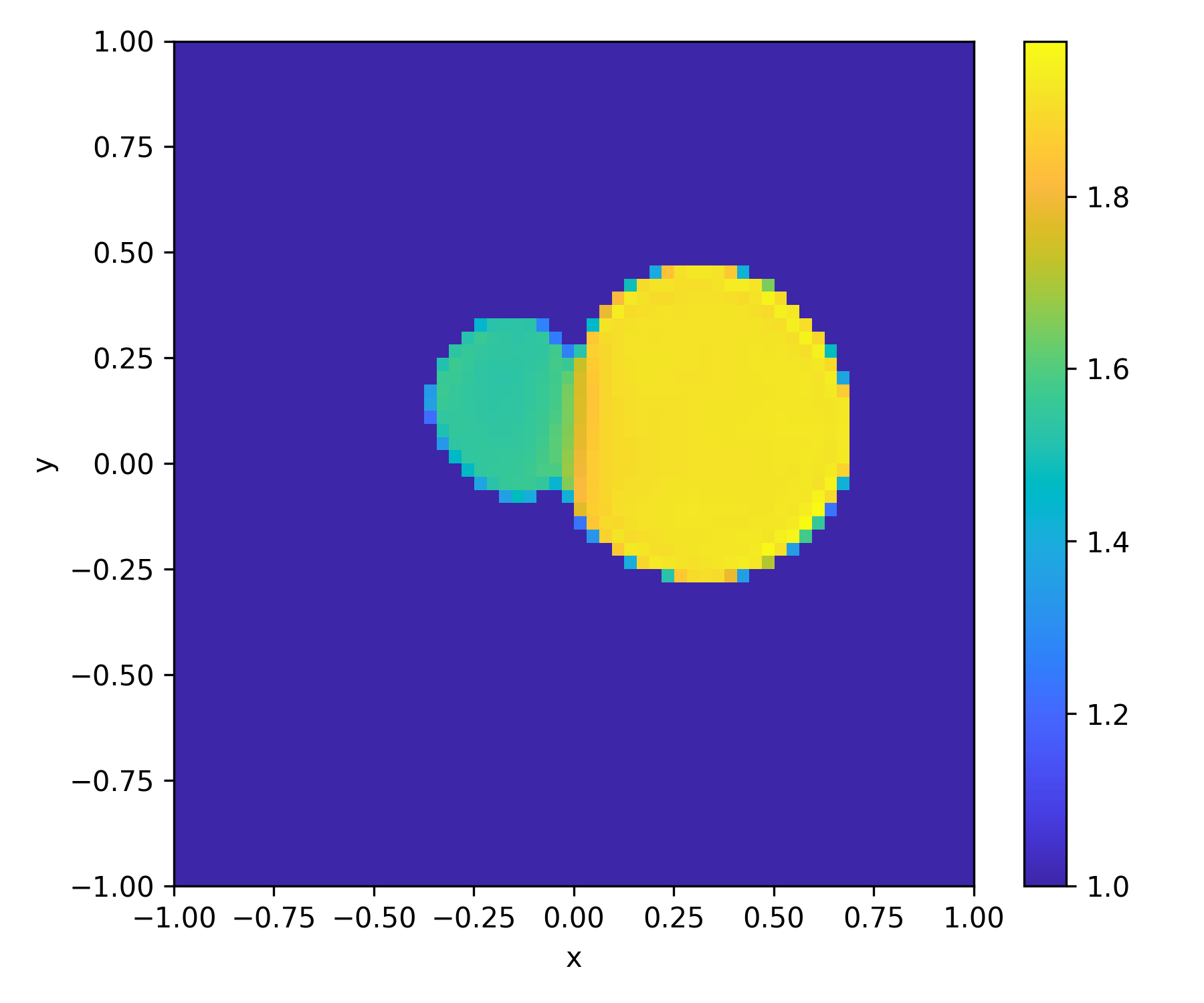}&
				\includegraphics[width=0.15\textwidth]{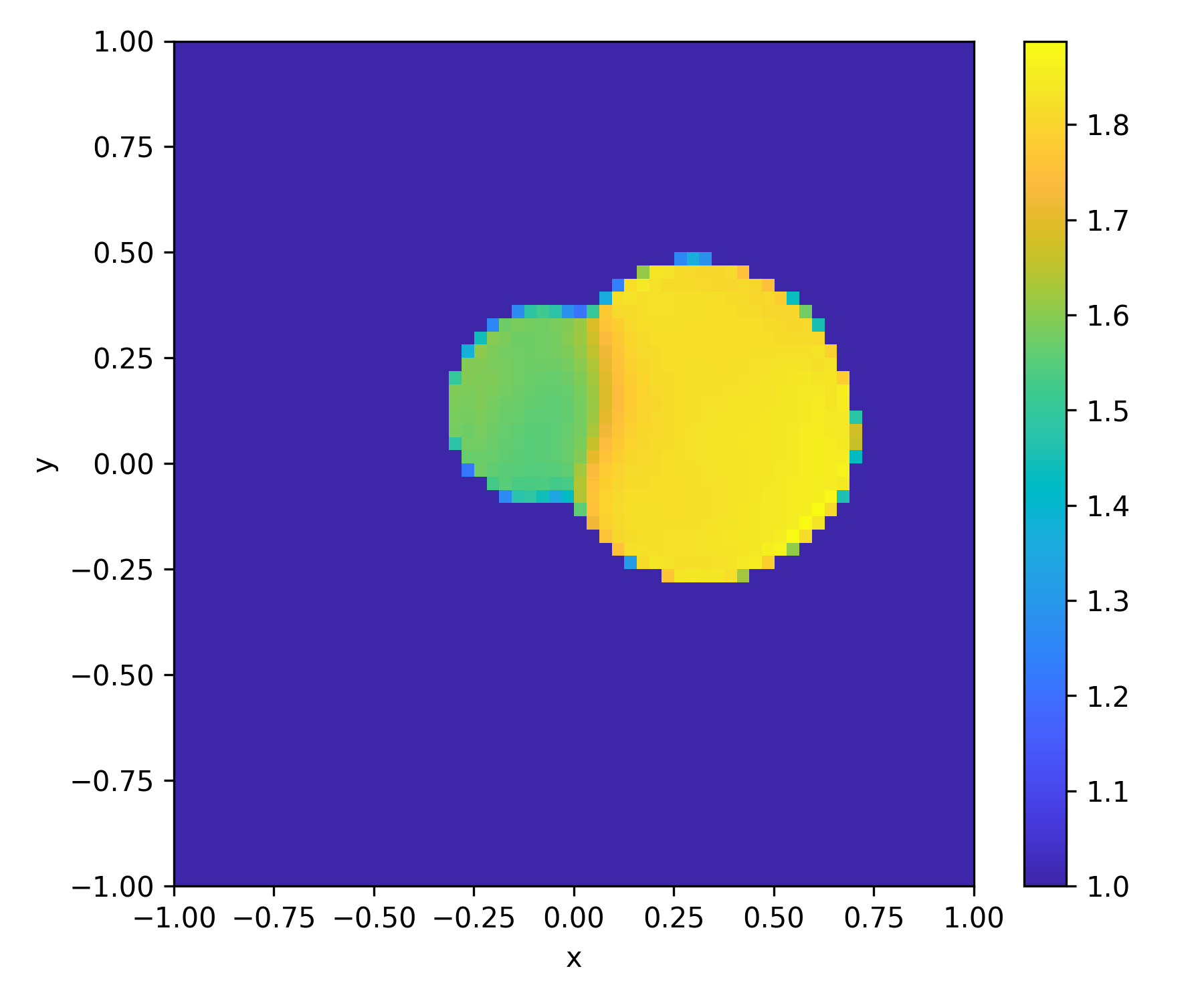}&
				\includegraphics[width=0.15\textwidth]{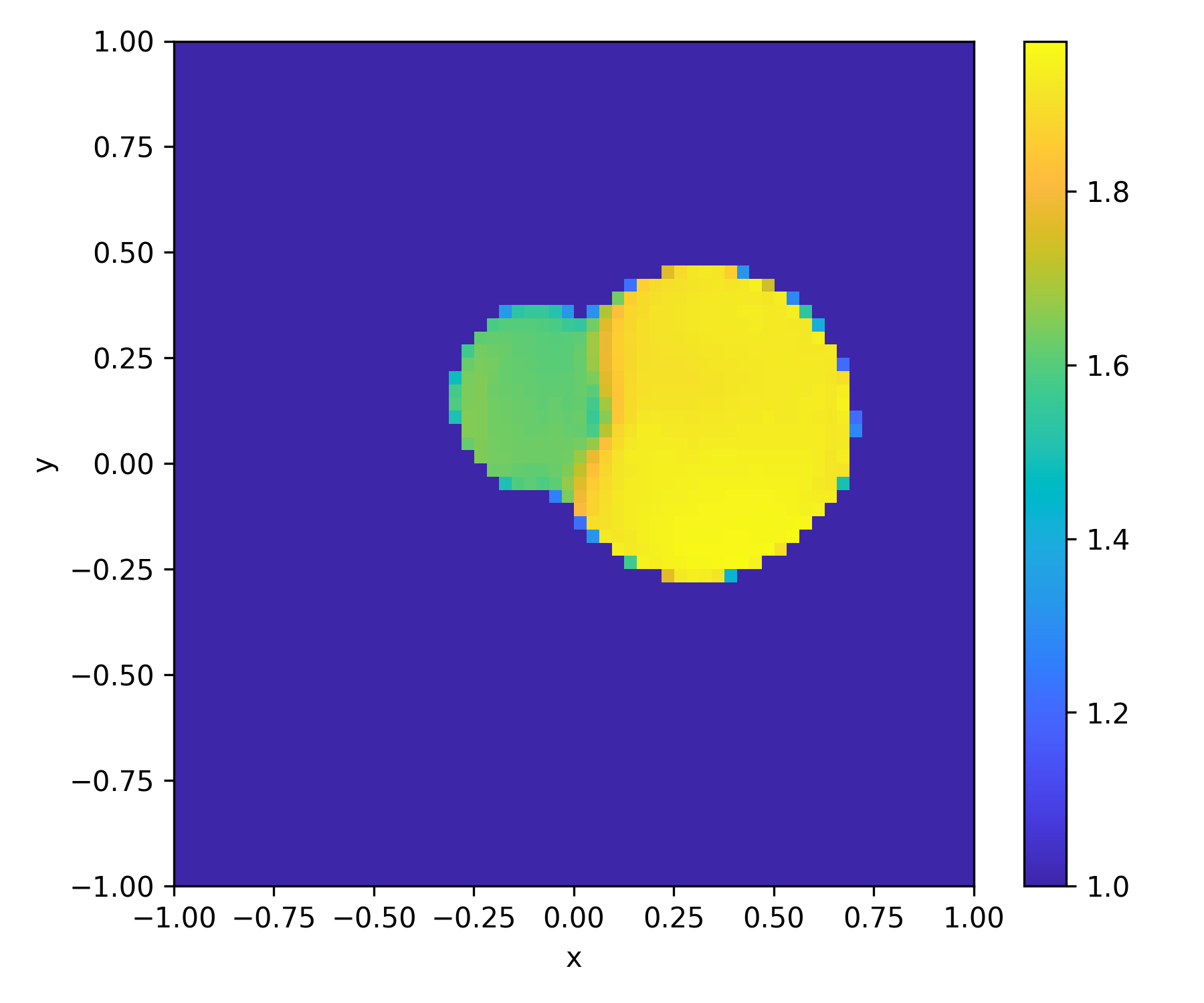}
				\\ 	
				15\%&\SetCell[r=2]{c}\includegraphics[width=0.15\textwidth]{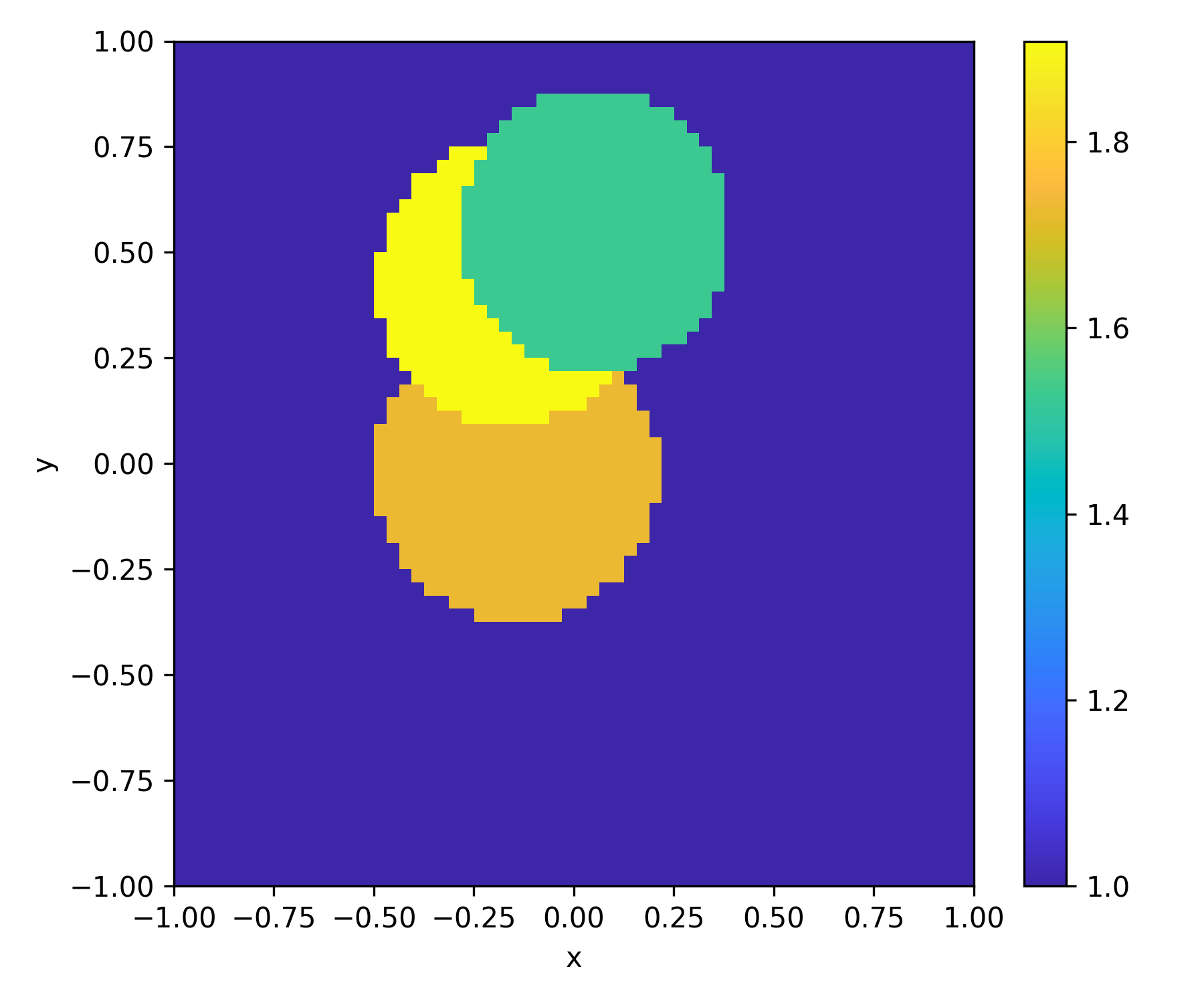}&
				\includegraphics[width=0.15\textwidth]{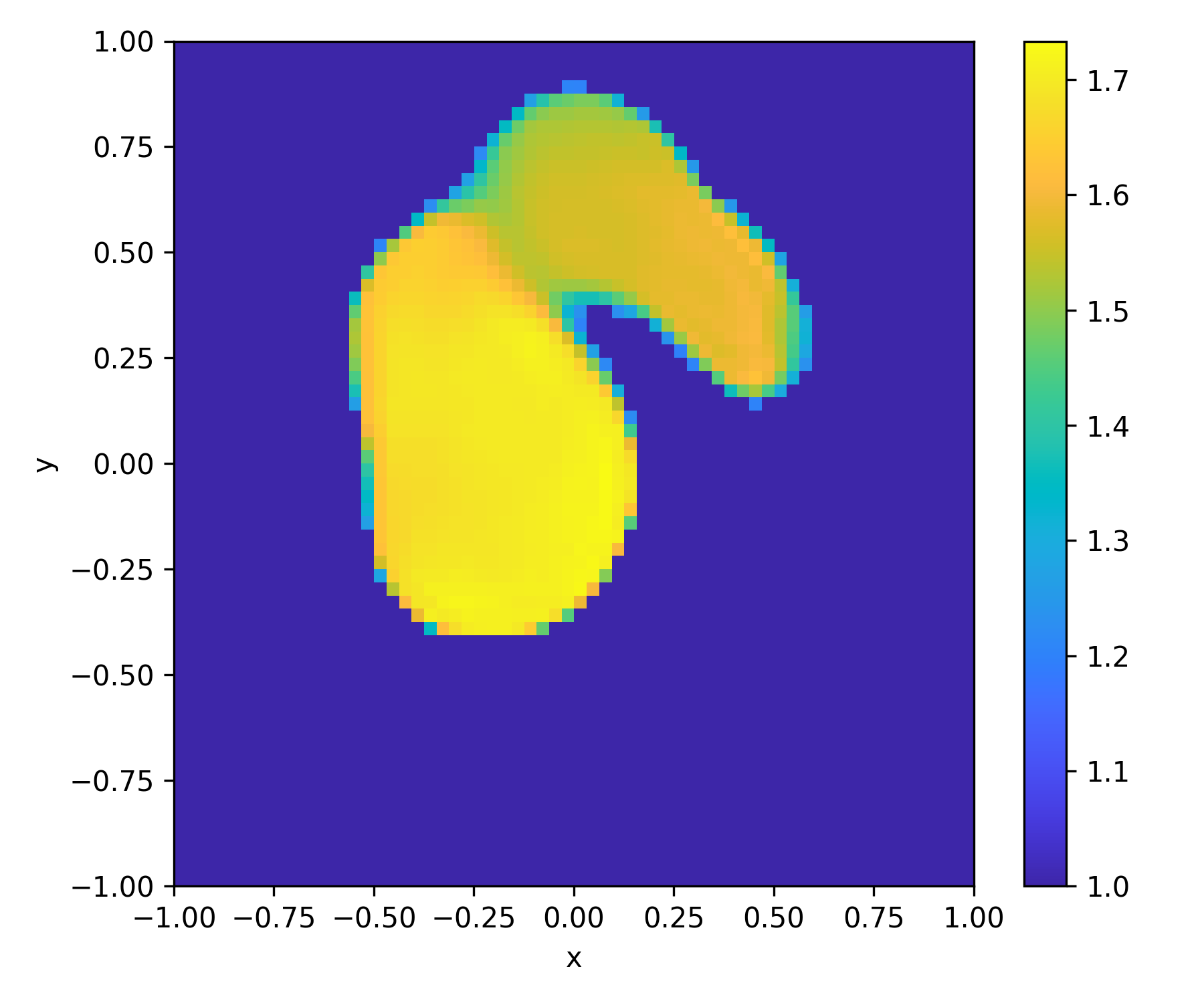}&
				\includegraphics[width=0.15\textwidth]{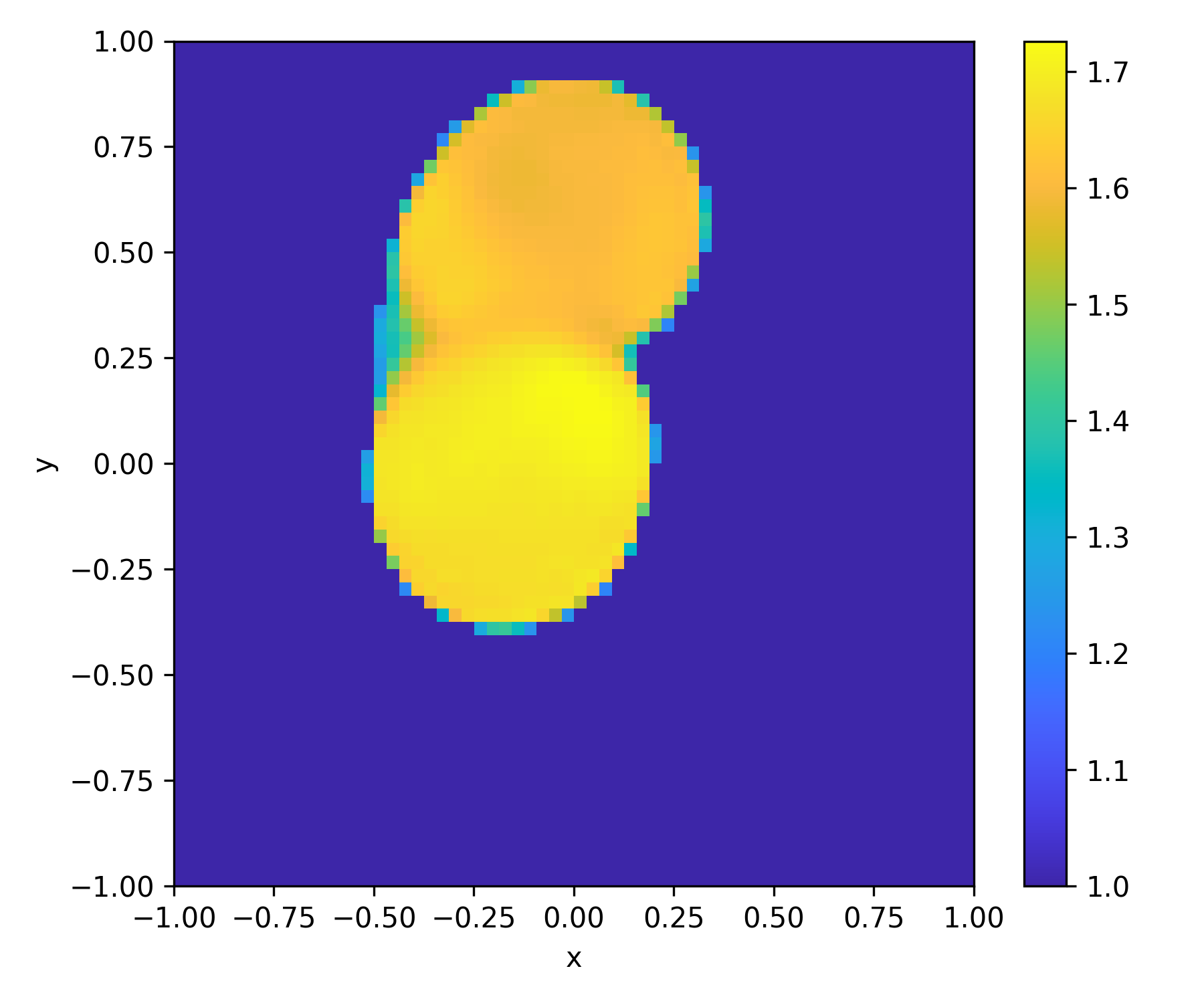}&
				\includegraphics[width=0.15\textwidth]{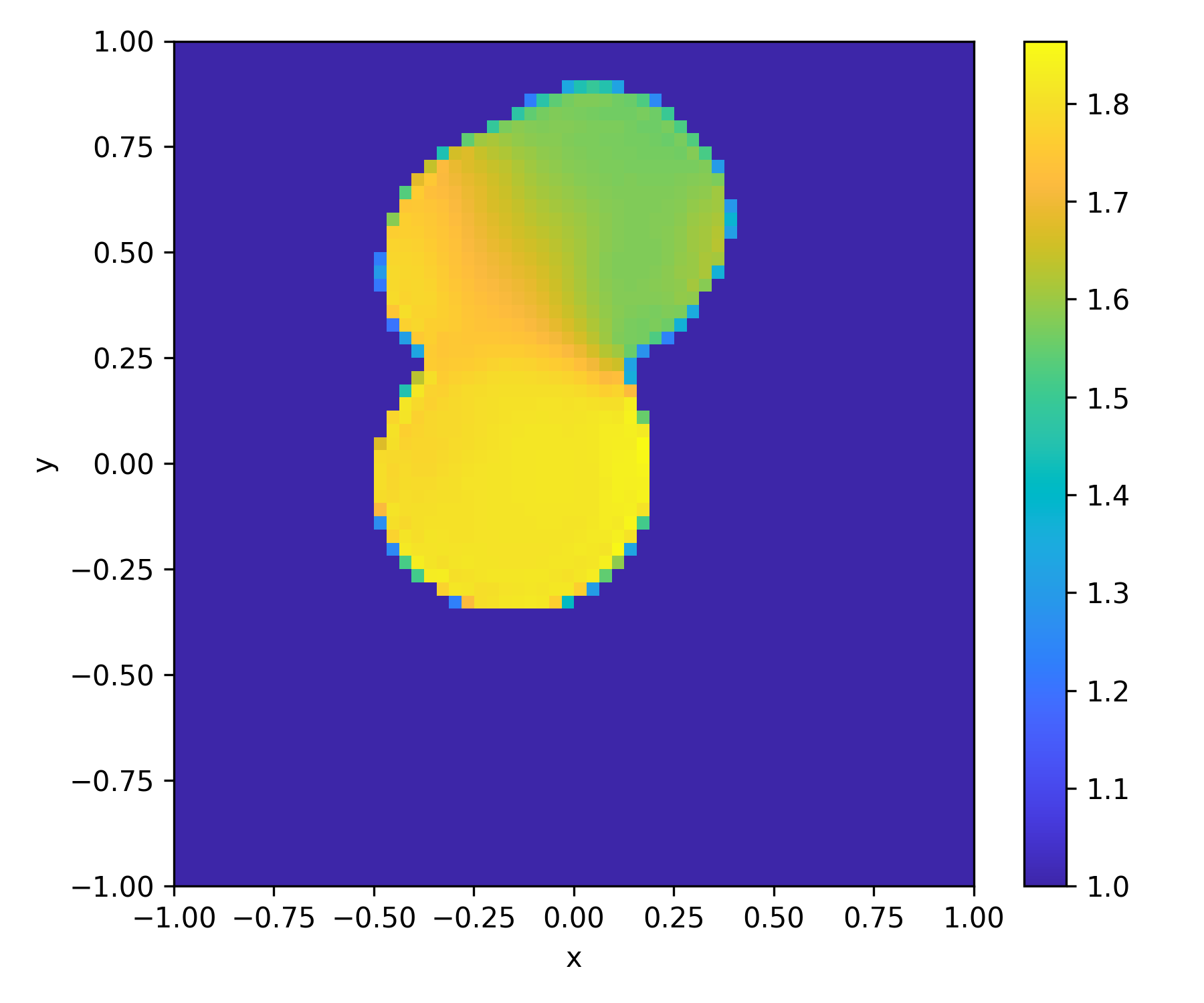}&
				\includegraphics[width=0.15\textwidth]{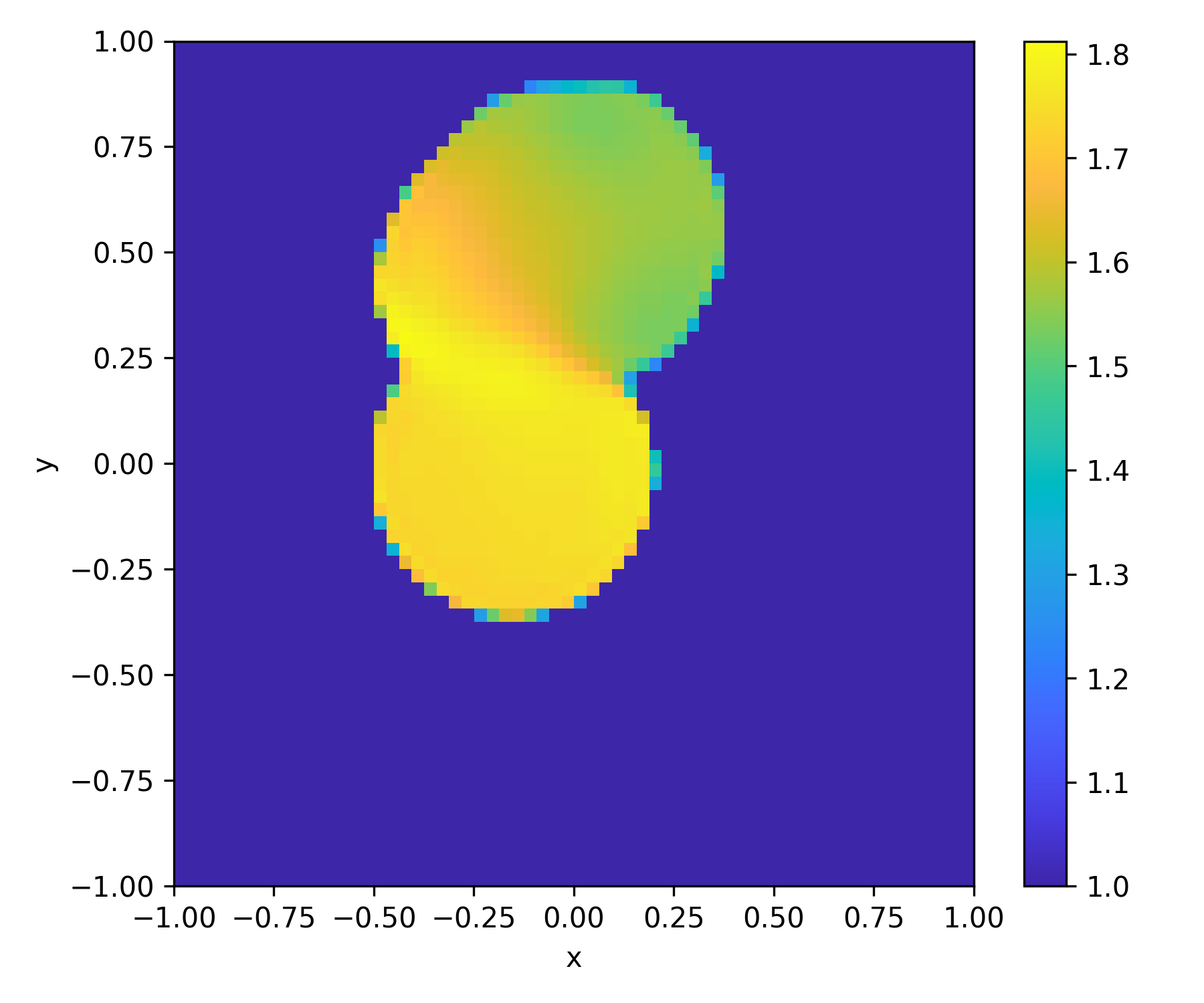}&
				\includegraphics[width=0.15\textwidth]{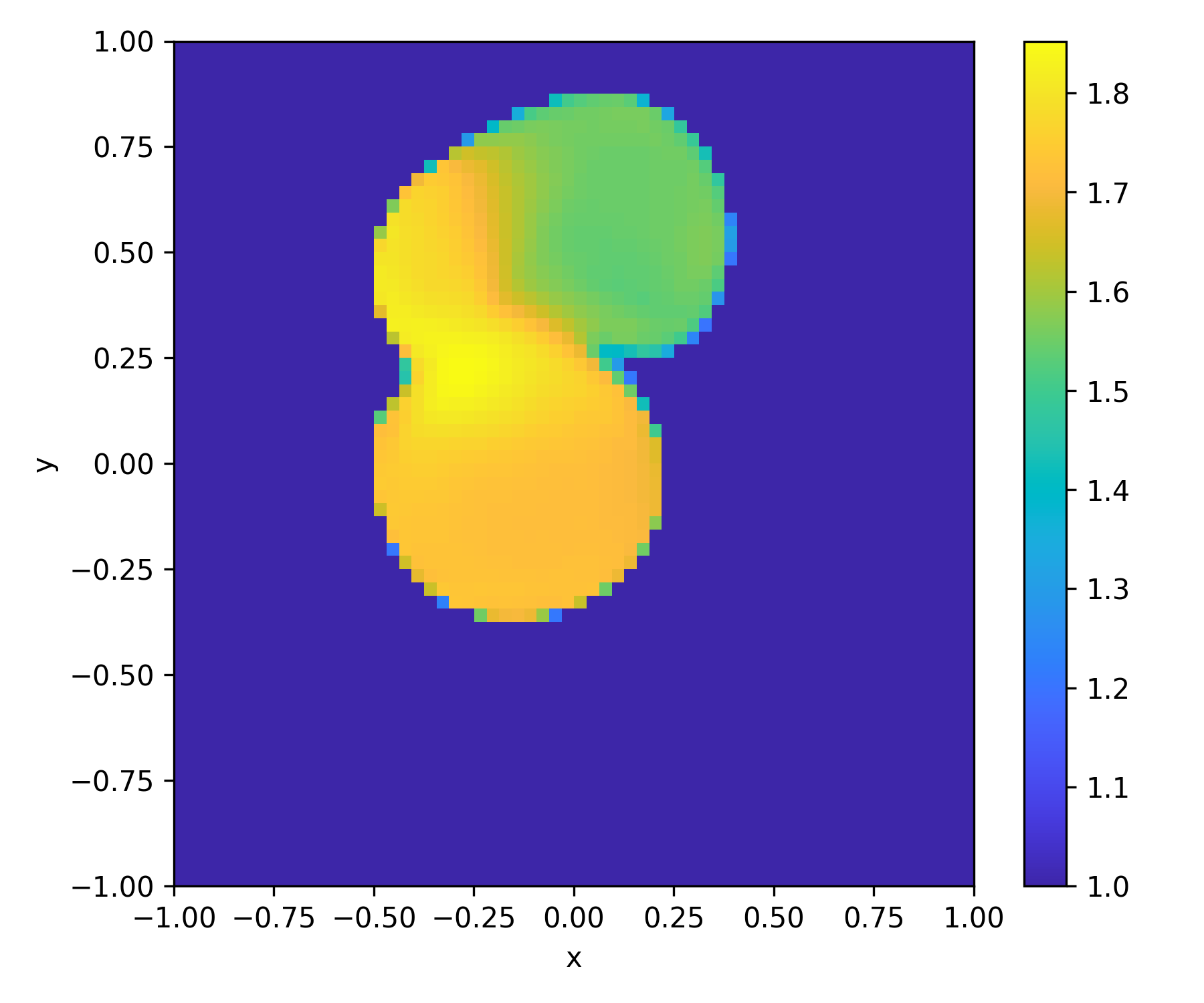} 	
				\\ 	
				40\%& &
				\includegraphics[width=0.15\textwidth]{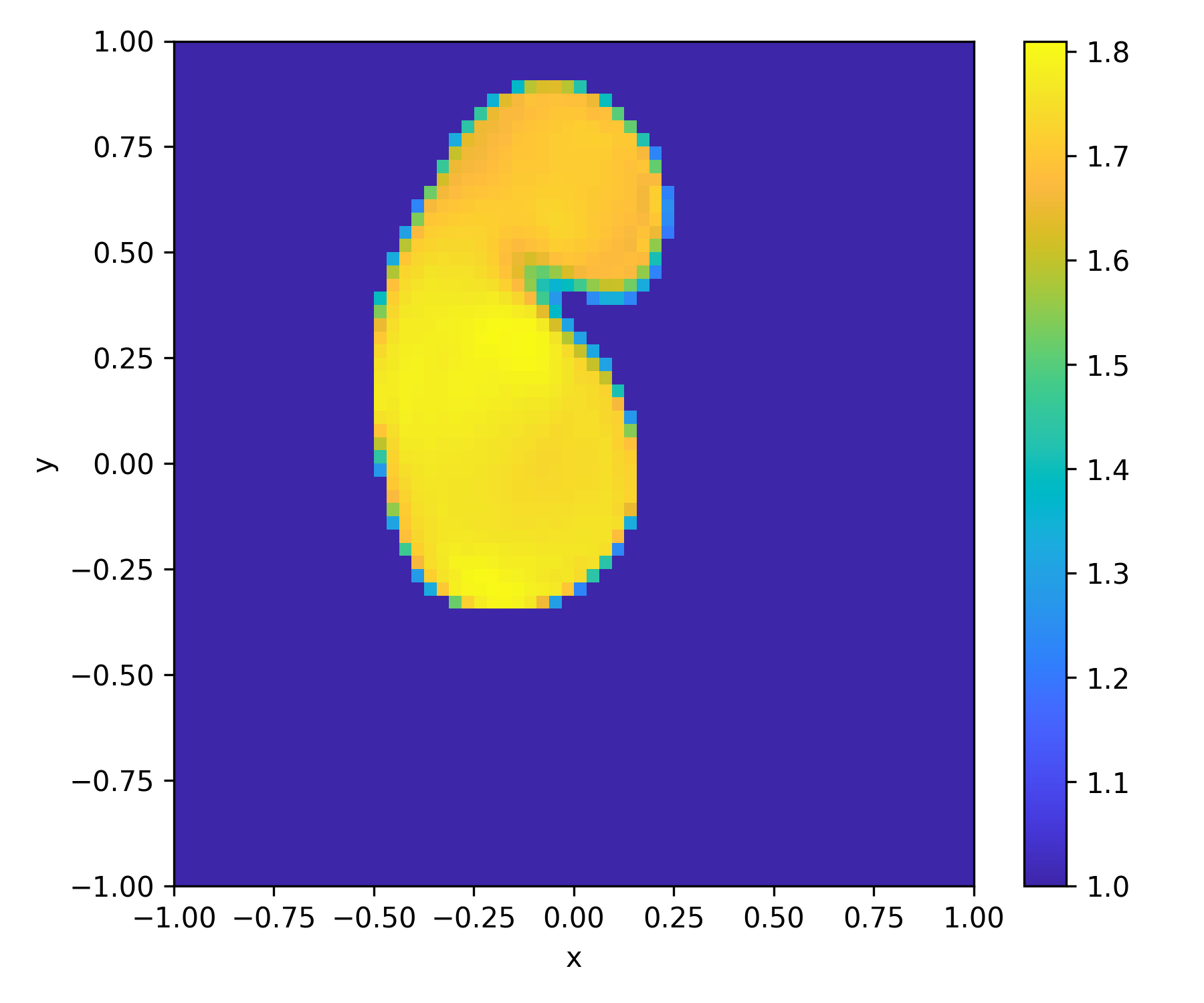}&
				\includegraphics[width=0.15\textwidth]{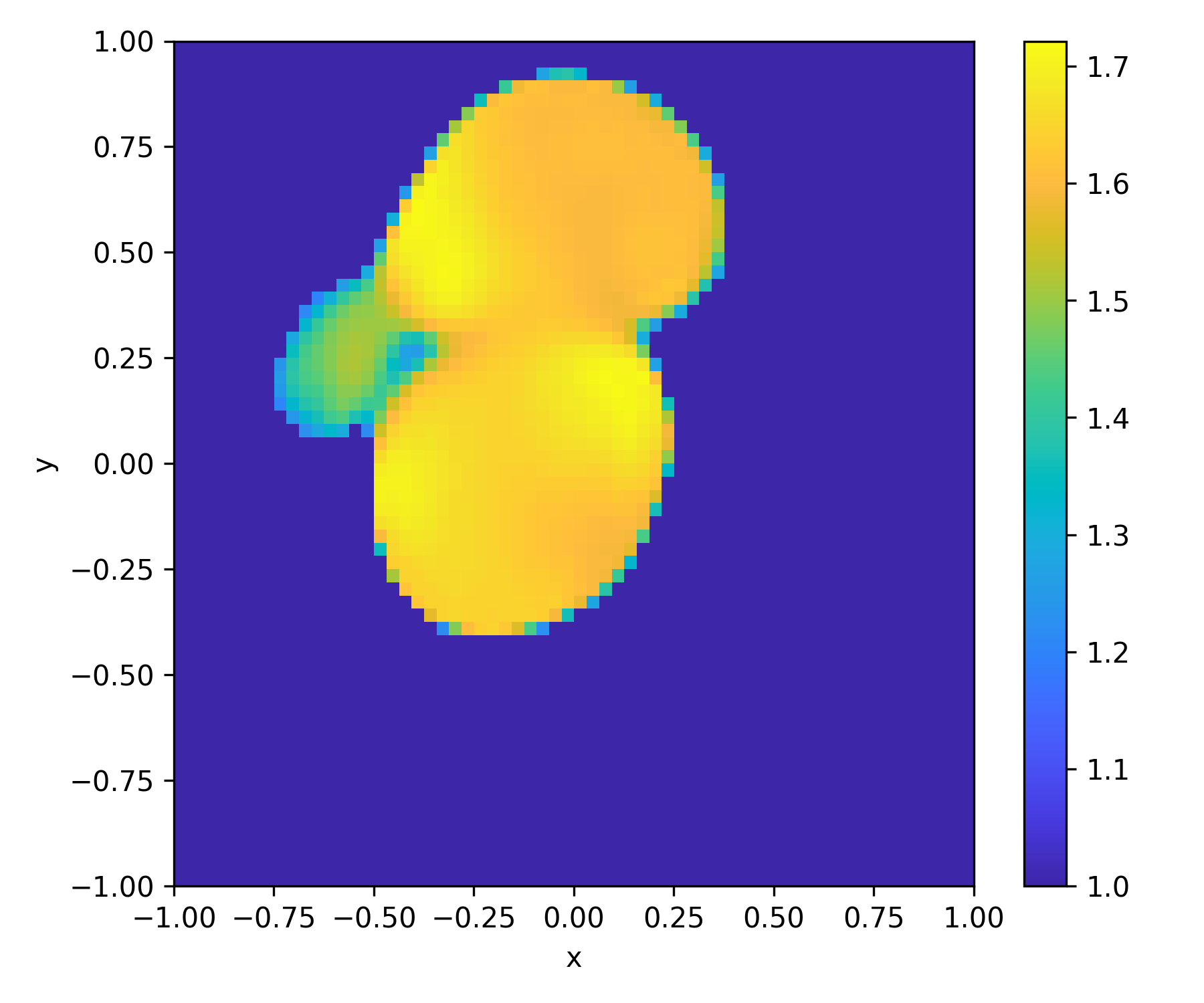}&
				\includegraphics[width=0.15\textwidth]{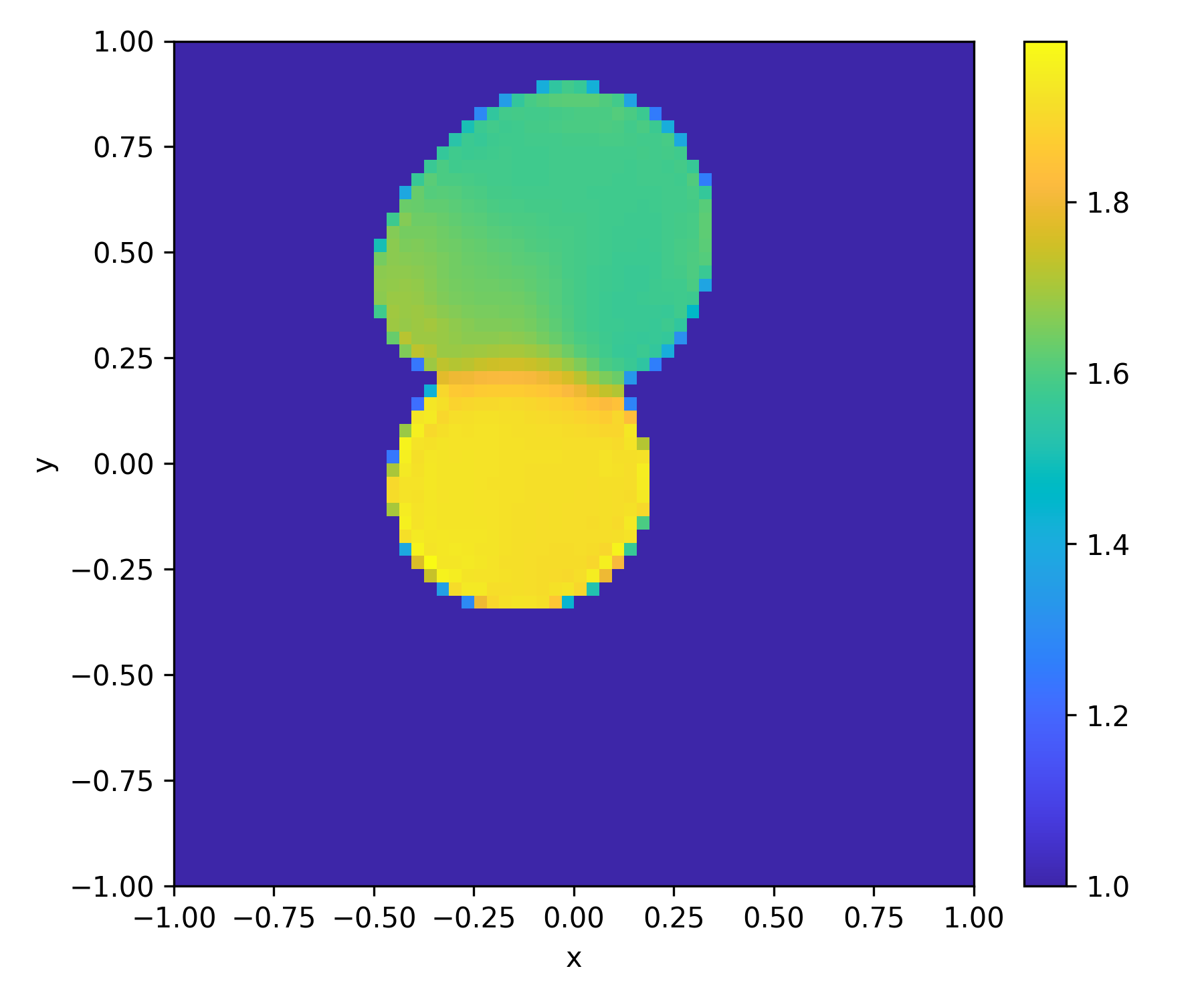}&
				\includegraphics[width=0.15\textwidth]{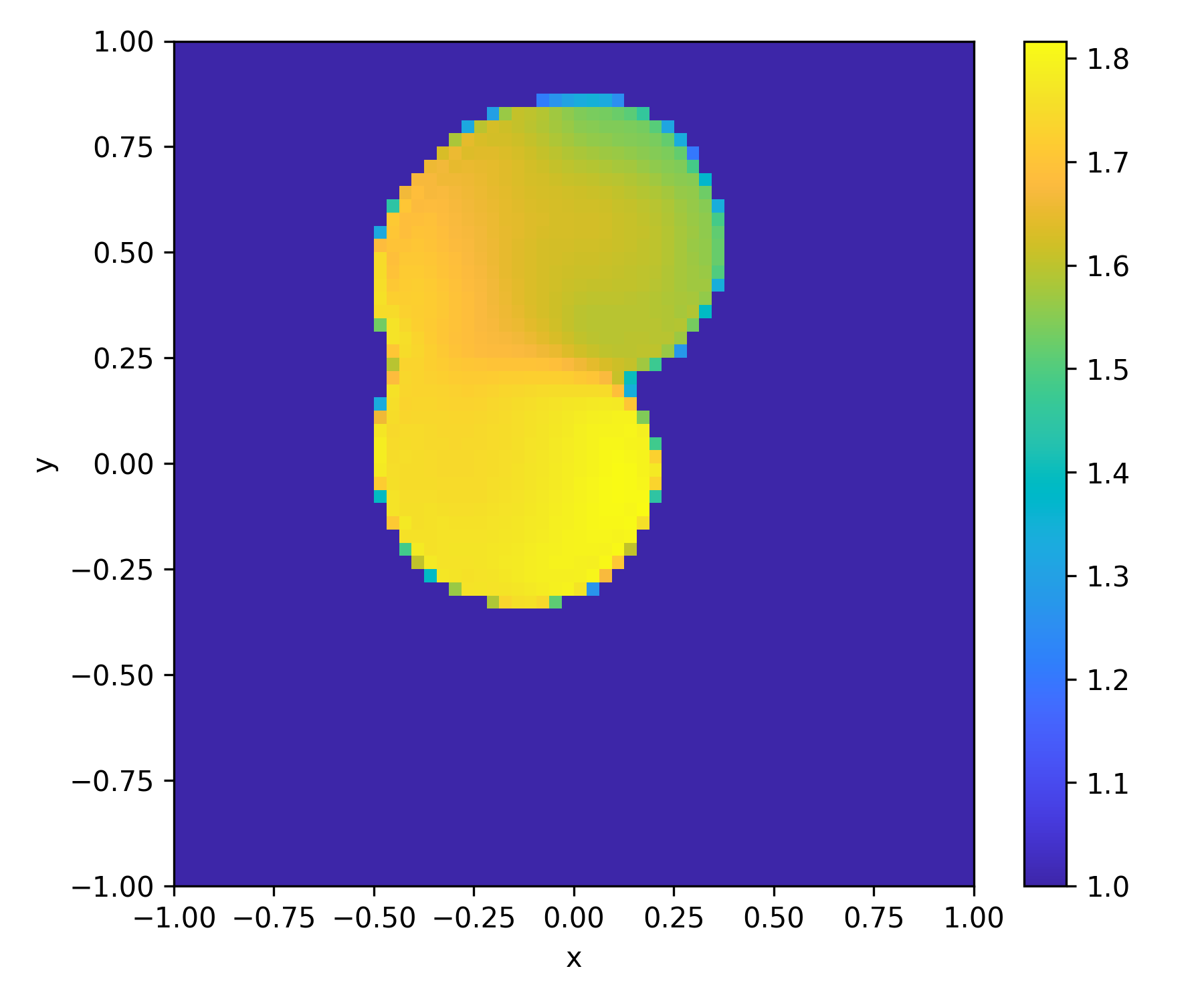}&
				\includegraphics[width=0.15\textwidth]{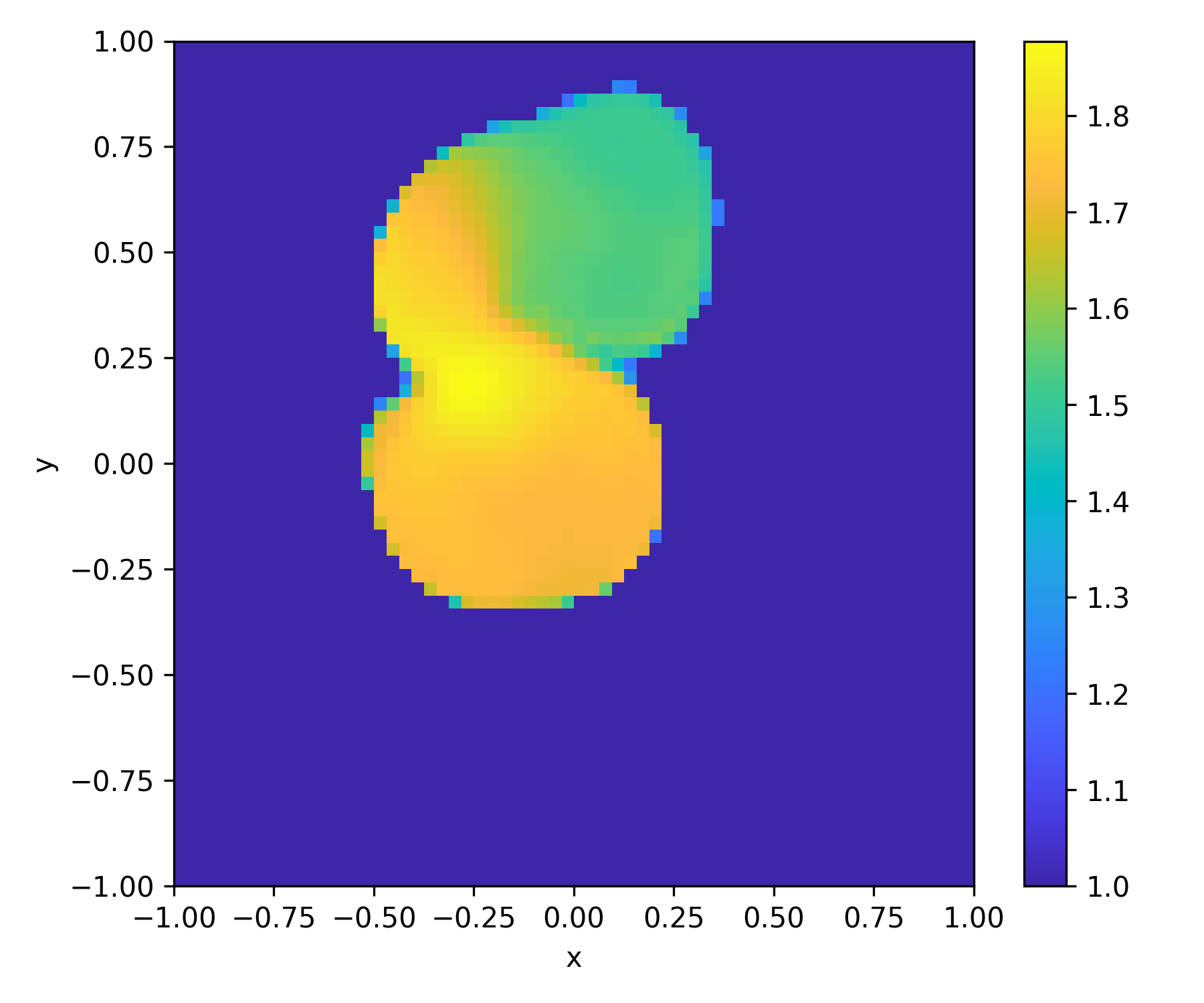} 
			\end{tblr}
			\caption{Image reconstructions from measured scattered fields with $15\%$ and $40\%$ Gaussian noises by using the networks trained by the circle dataset, where the relative permittivity is between 1.5 and 2.0. From left to right: the ground-truth images, the reconstruction with 1,2,4,8, and 16 incident fields.}
			\label{tab:fig-Circle}
		\end{center}
	\end{figure}
 \begin{figure}[htp]
	\centering
	\includegraphics[width=1.0\linewidth]{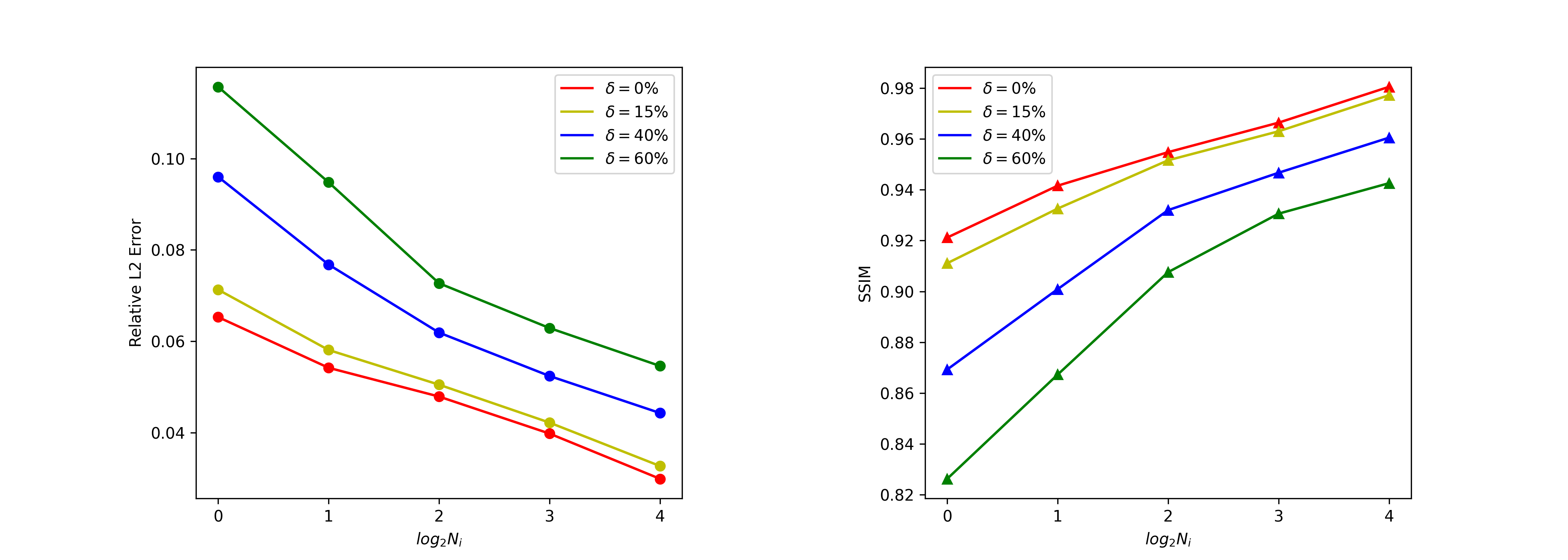}
	\caption{The relative L2 testing error and SSIM for the networks trained by the circle dataset with different noise levels and number of incidences.}
	\label{fig:Error_Circle}
\end{figure}
	
	\subsubsection{Tests with a four-circles example}
	To examine the generality of the trained networks, we consider the reconstruction of scatters of four circles, which is beyond the distribution of the training data. Based on the reconstructions shown in Fig.\,\ref{tab:fig-Four_Circle}, an additional scatterer is mispredicted in the reconstructed profile with $\delta=15\%$ for $N_{i}=1$. It also shows that the noises have a significant impact on the reconstruction with small $N_{i}$, while better reconstruction and generality capability can be observed by using more incident fields. 
		\begin{figure}[htp]\small
		\begin{center}
			\begin{tblr}
				{colspec = {X[-1]X[c]X[c,h]X[c,h]X[c,h]X[c,h]X[c,h]},
					stretch = 0,
					rowsep = 0pt,}
				Noise Level& Ground truth& $N_{i}$=1& $N_{i}$=2 &$N_{i}$=4&$N_{i}$=8& $N_{i}$=16\\
				15\%&\SetCell[r=2]{c}\includegraphics[width=0.15\textwidth]{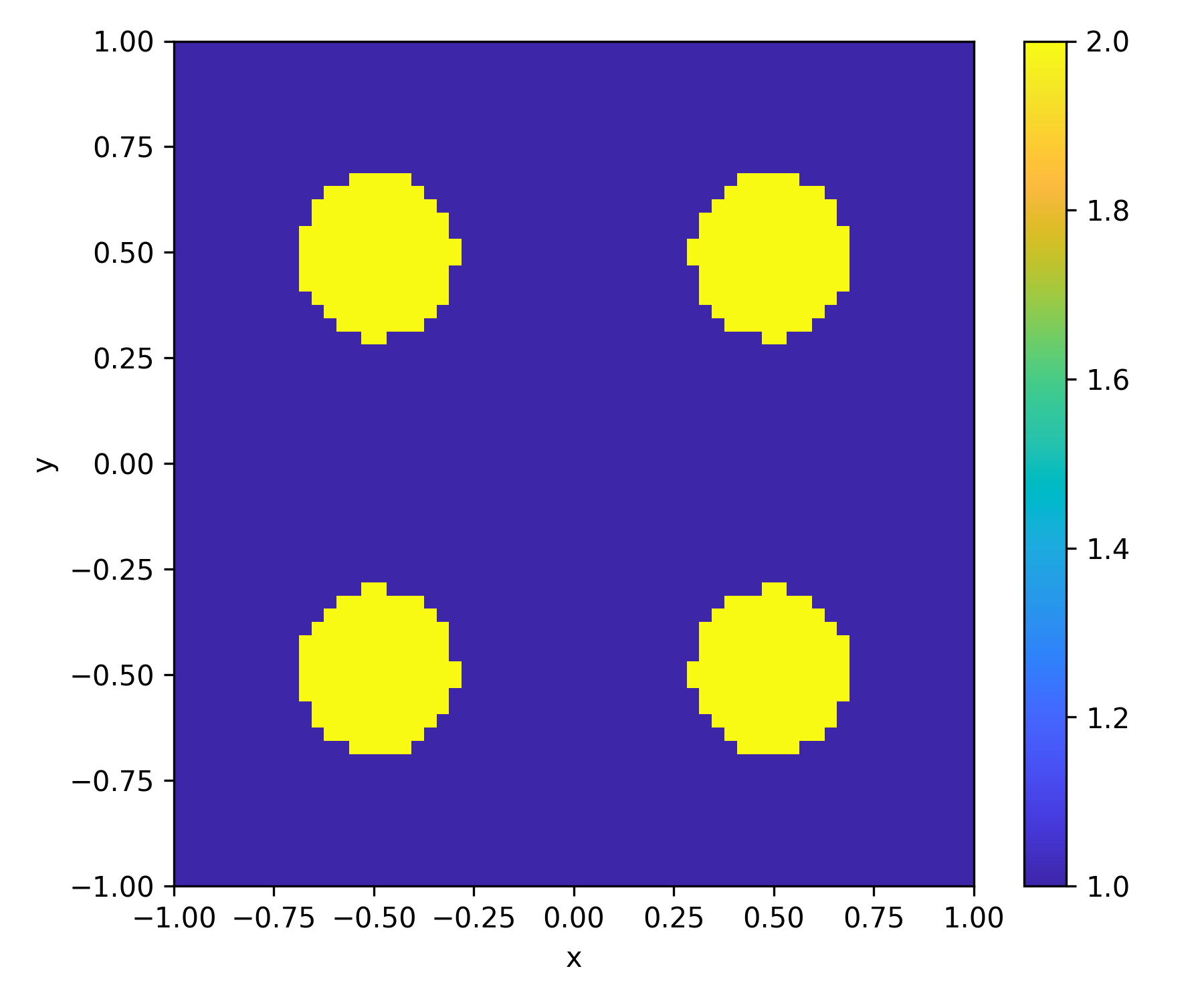}&
				\includegraphics[width=0.15\textwidth]{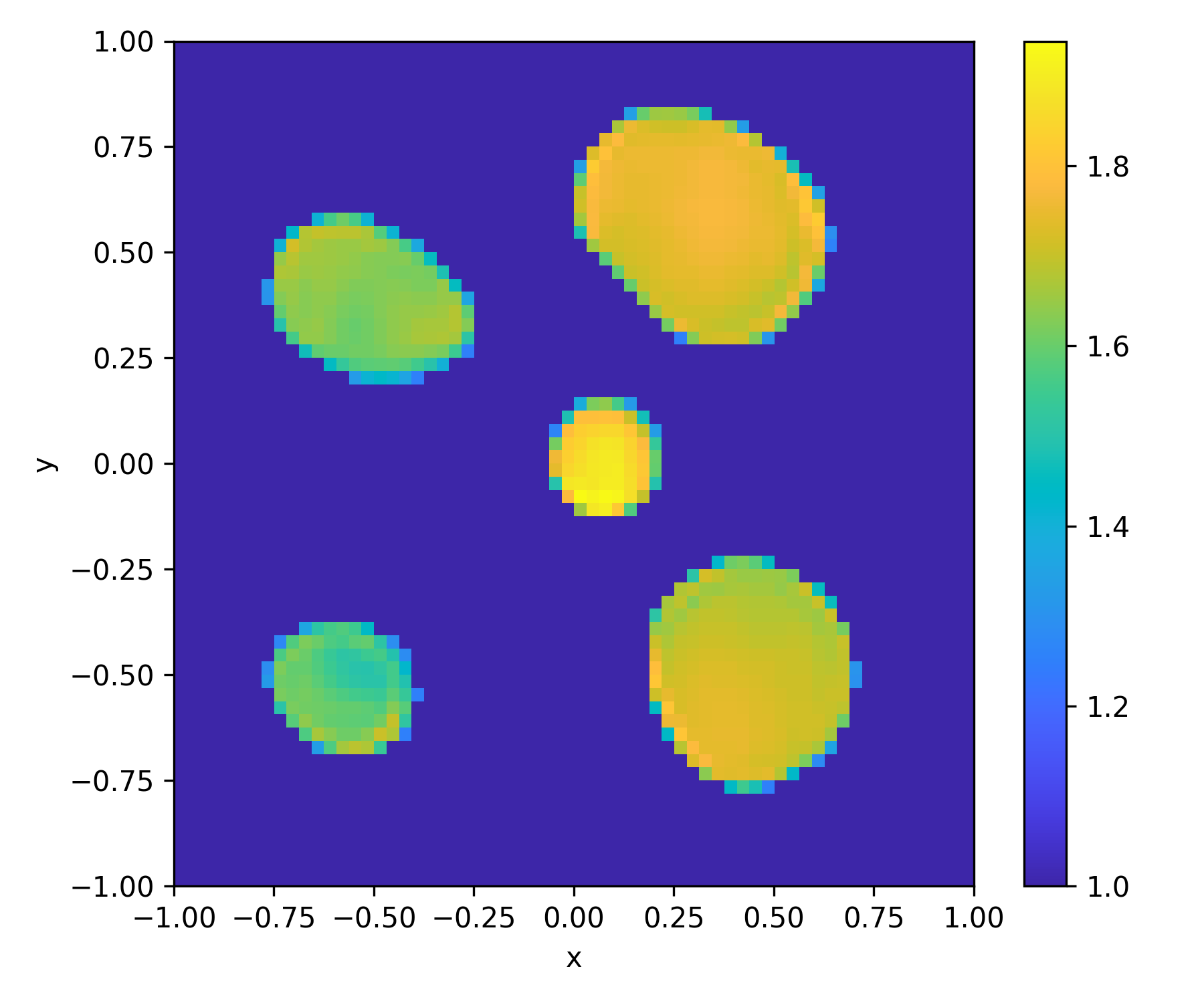}&
				\includegraphics[width=0.15\textwidth]{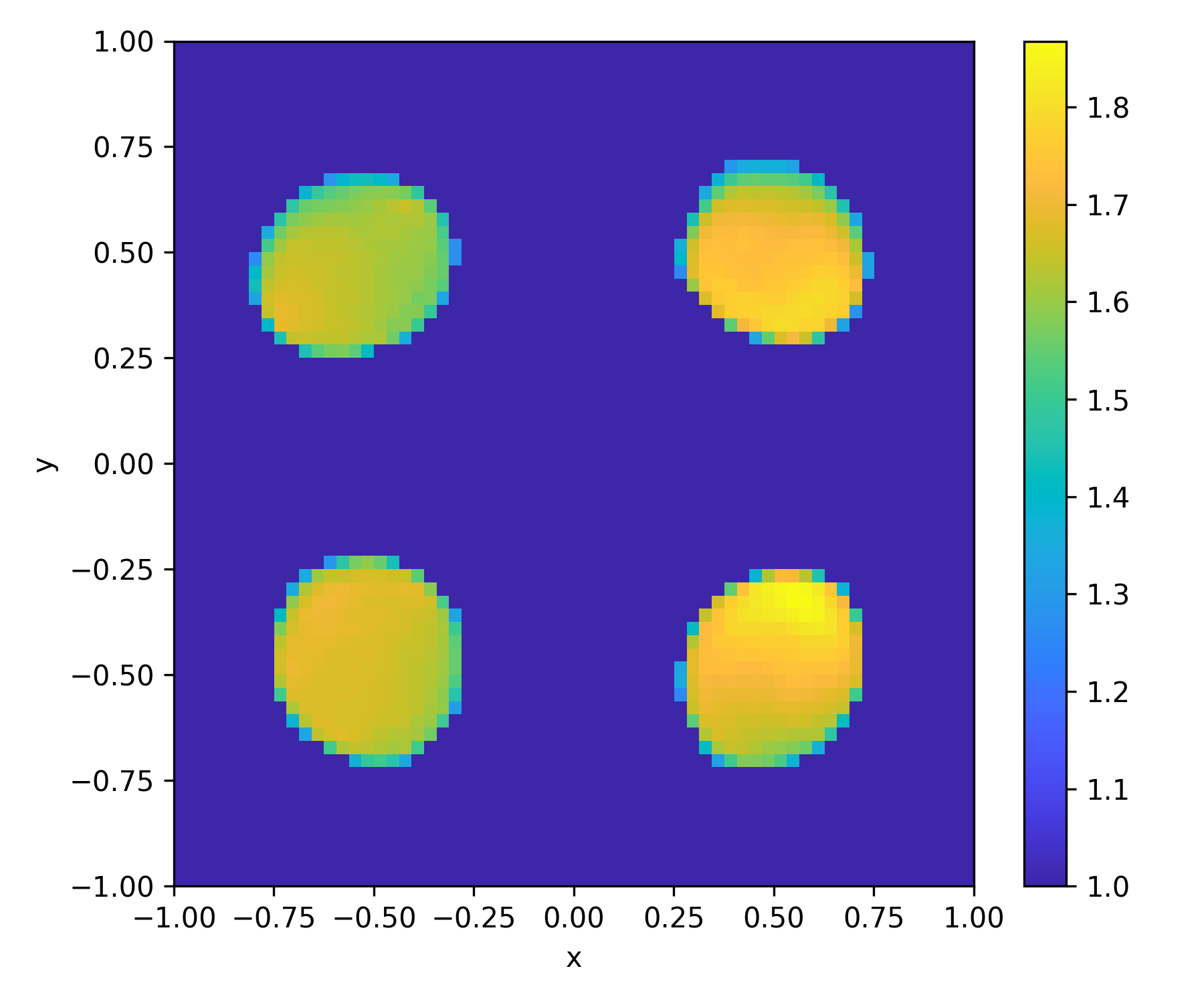}&
				\includegraphics[width=0.15\textwidth]{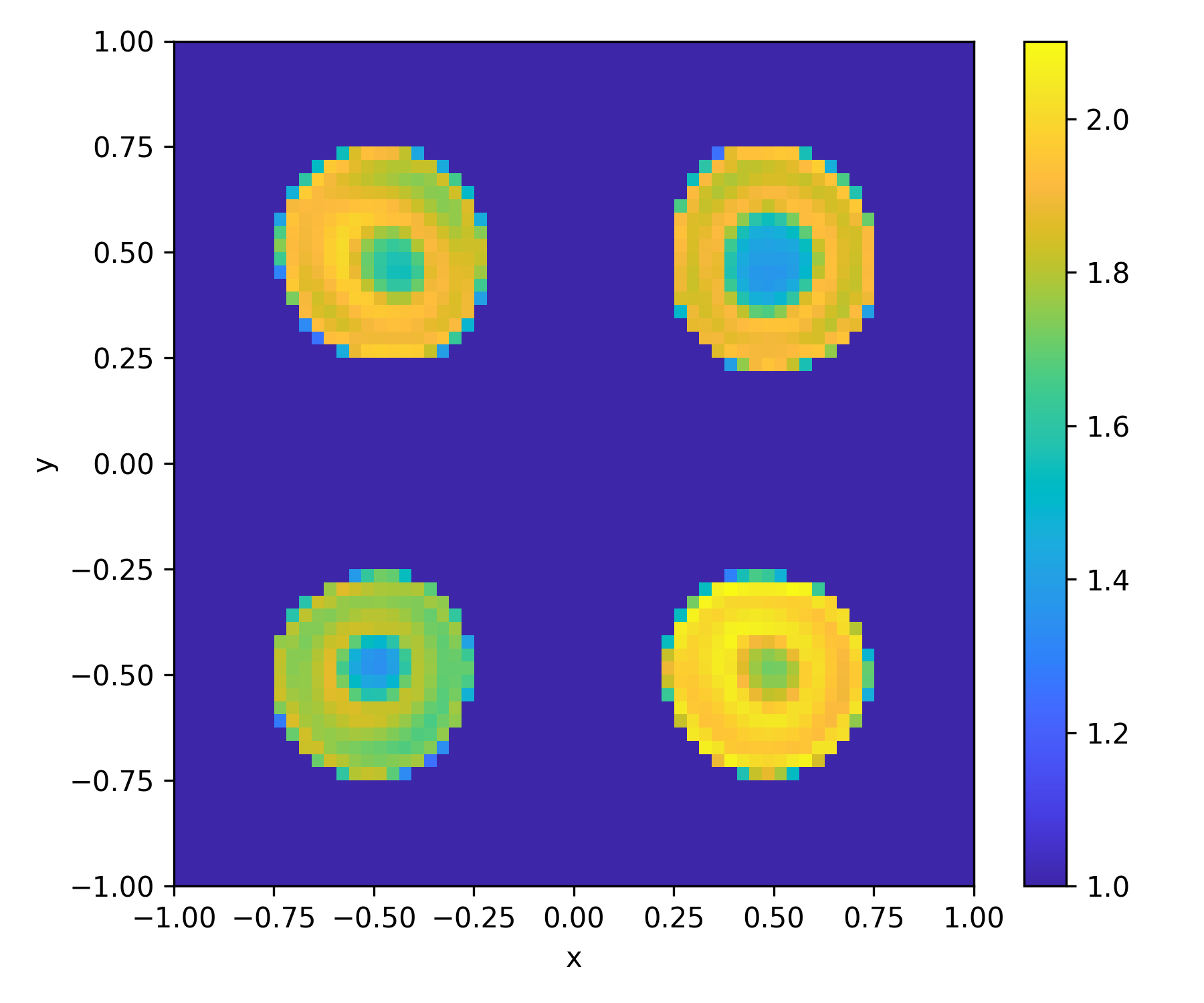}&
				\includegraphics[width=0.15\textwidth]{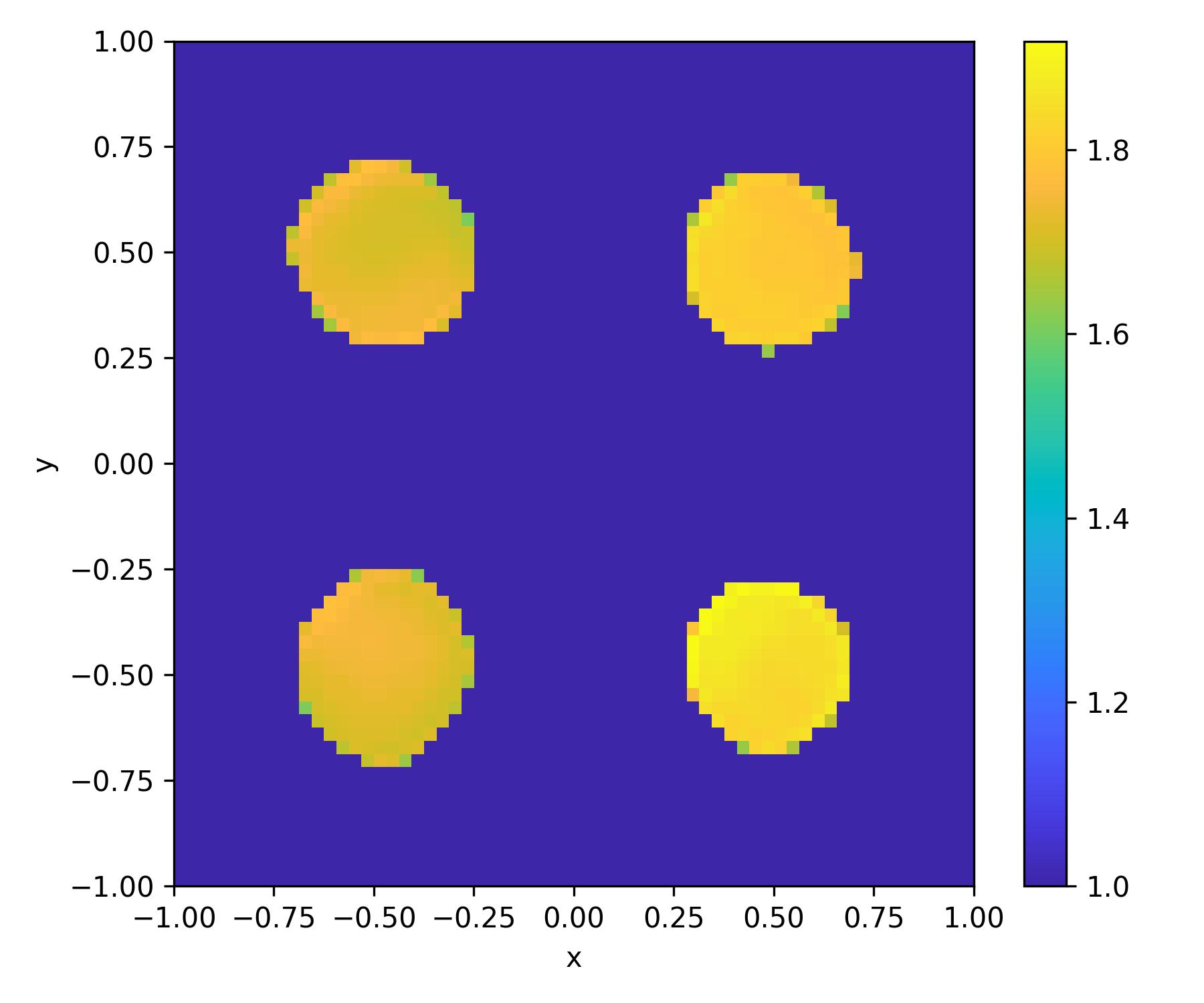}&
				\includegraphics[width=0.15\textwidth]{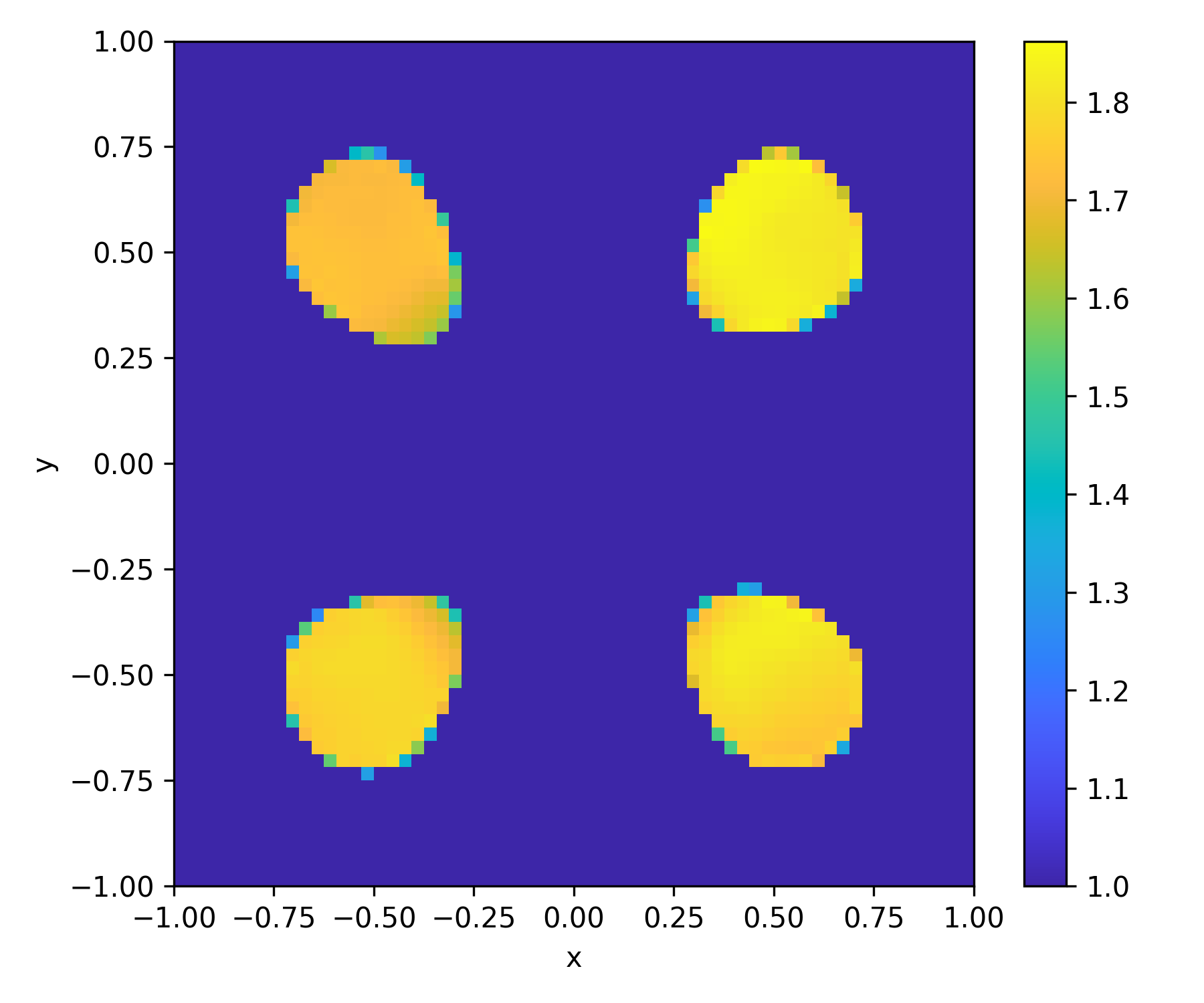}
				\\
				40\%&  &
				\includegraphics[width=0.15\textwidth]{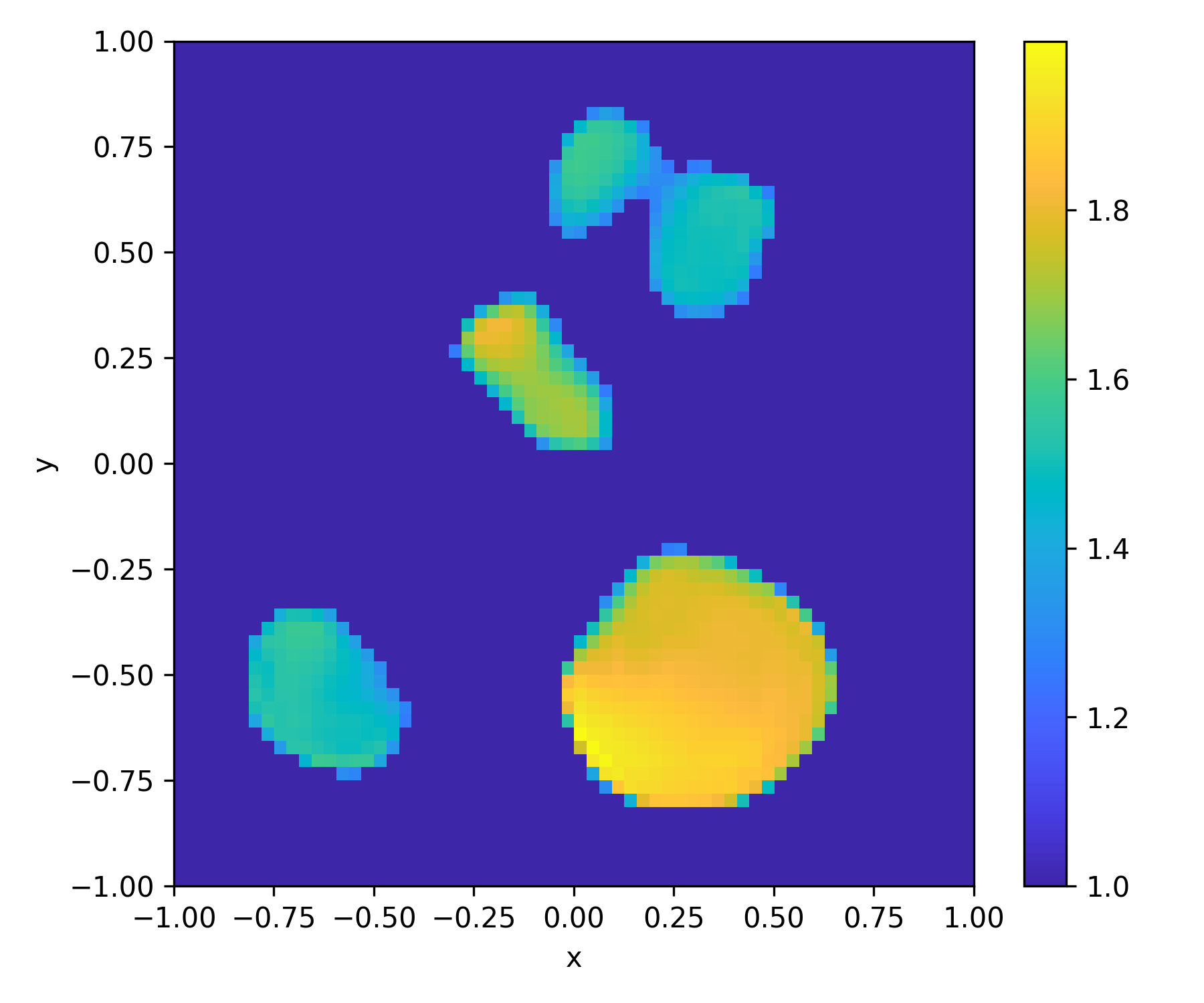}&
				\includegraphics[width=0.15\textwidth]{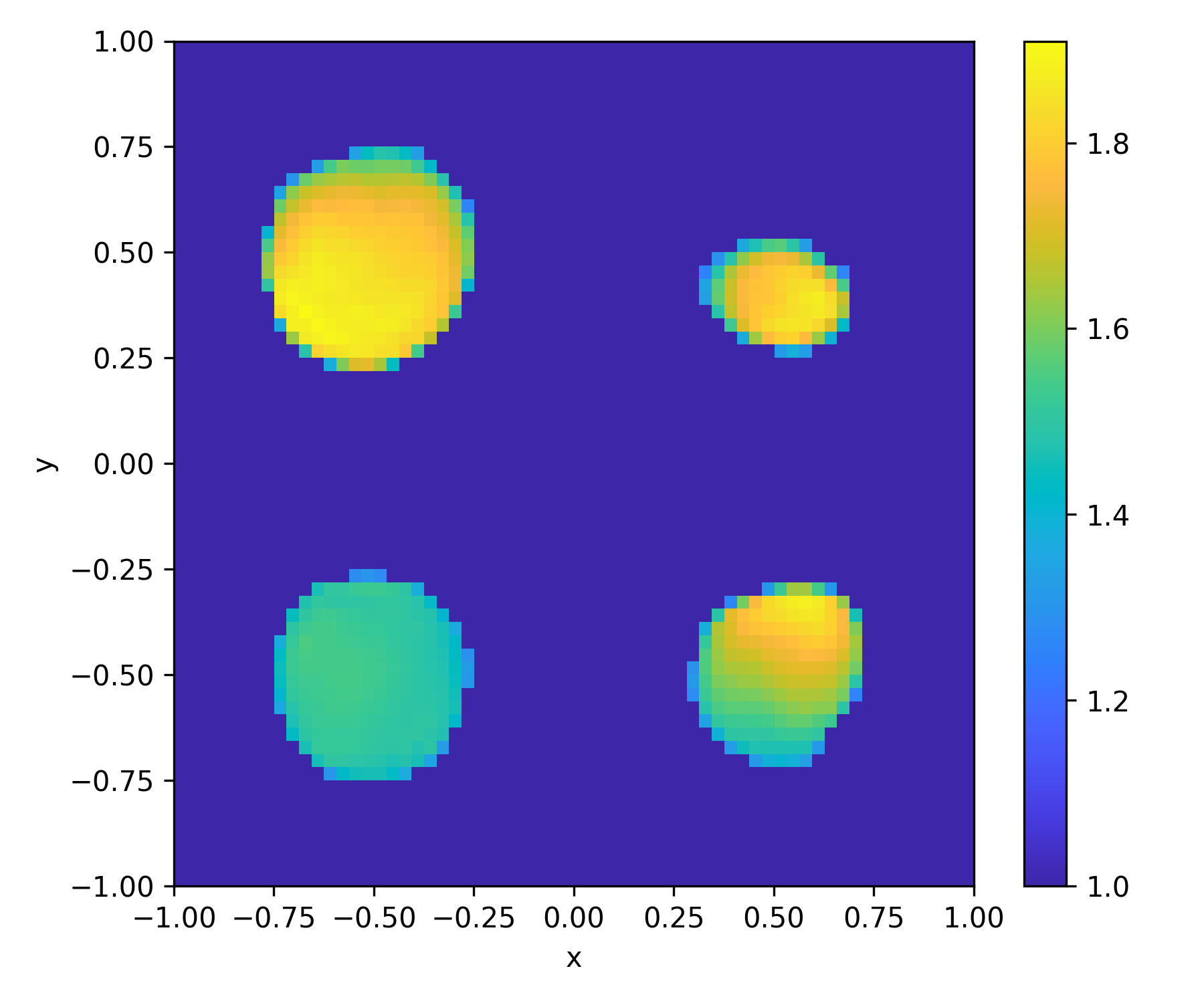}&
				\includegraphics[width=0.15\textwidth]{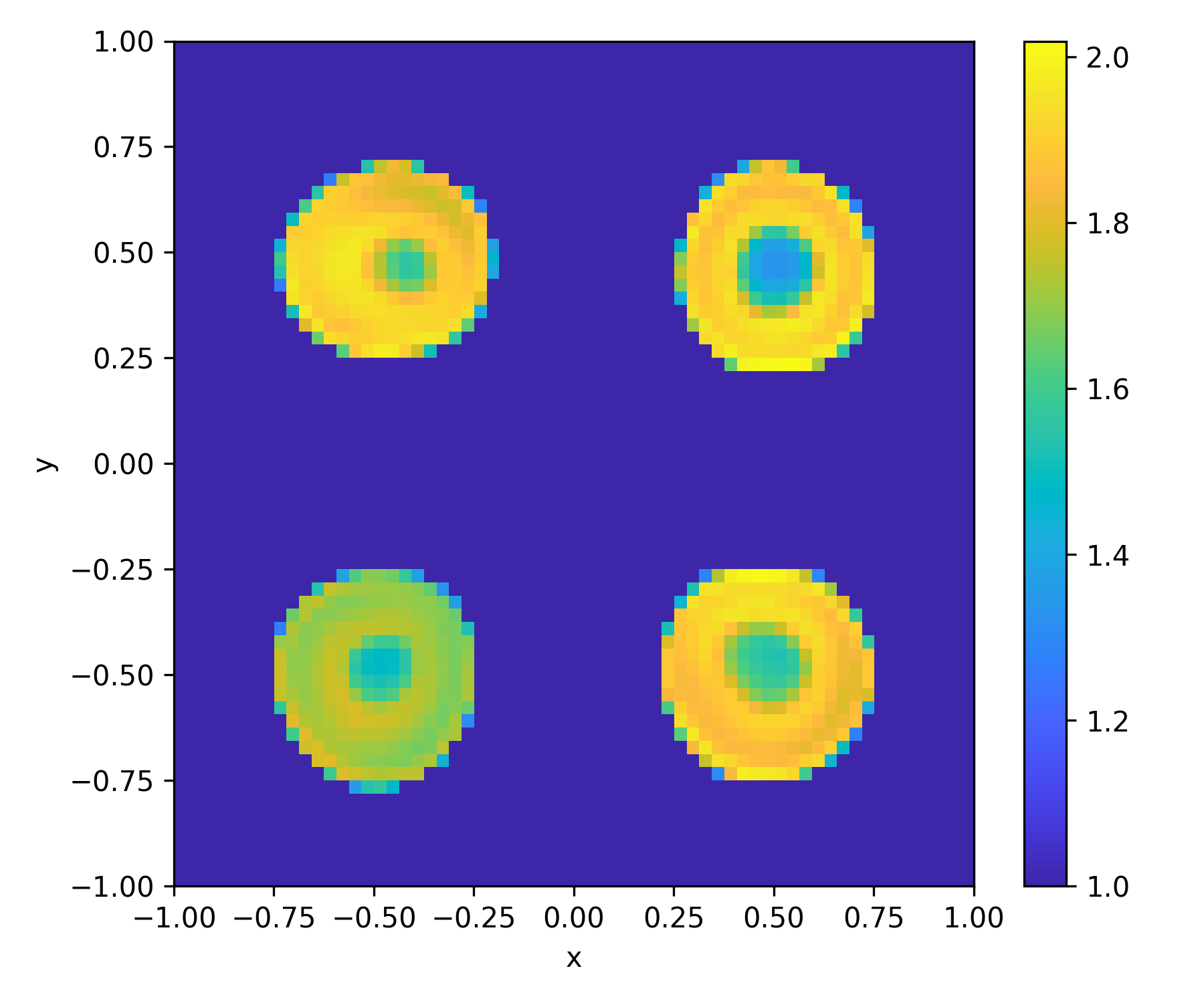}&
				\includegraphics[width=0.15\textwidth]{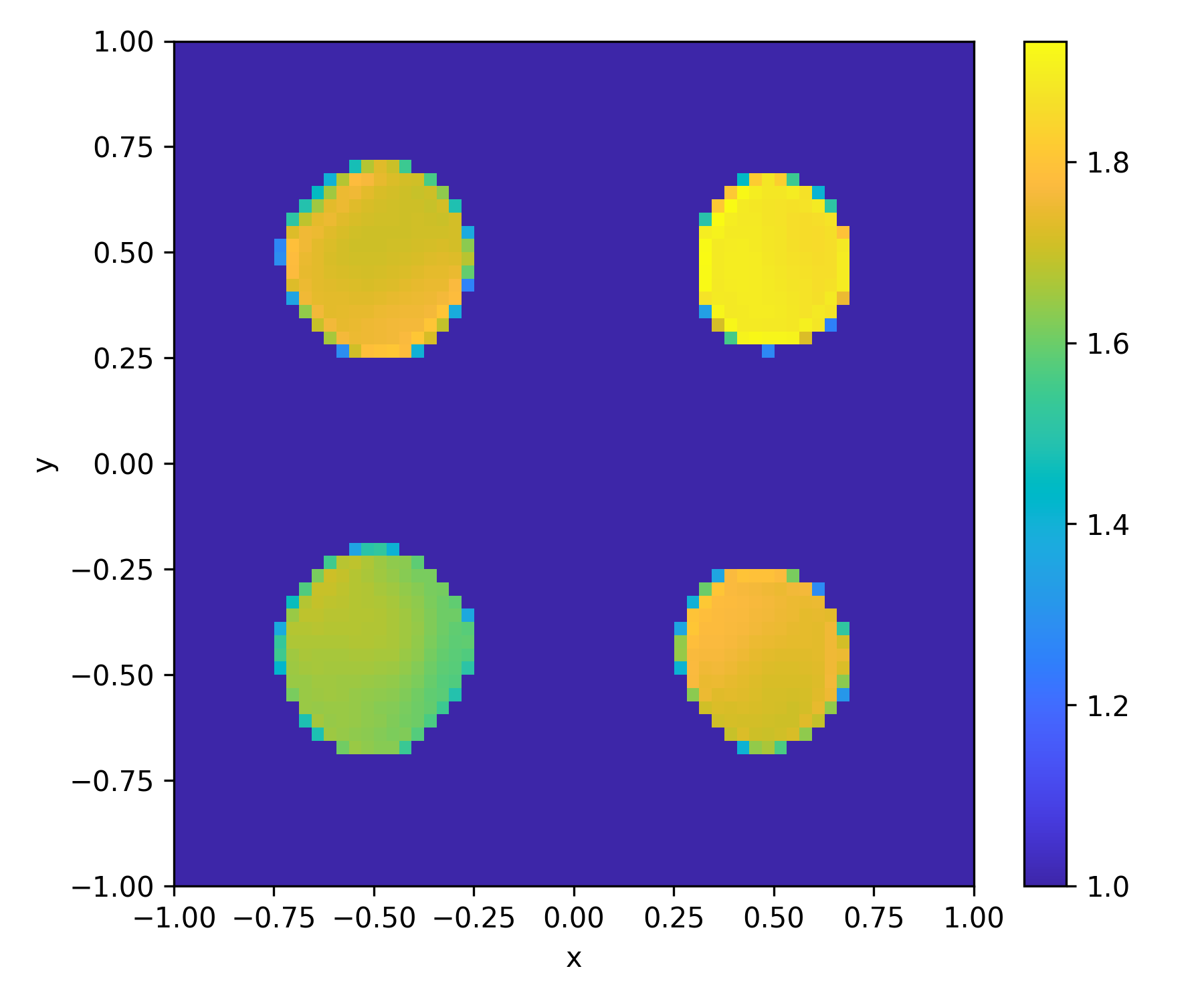}&
				\includegraphics[width=0.15\textwidth]{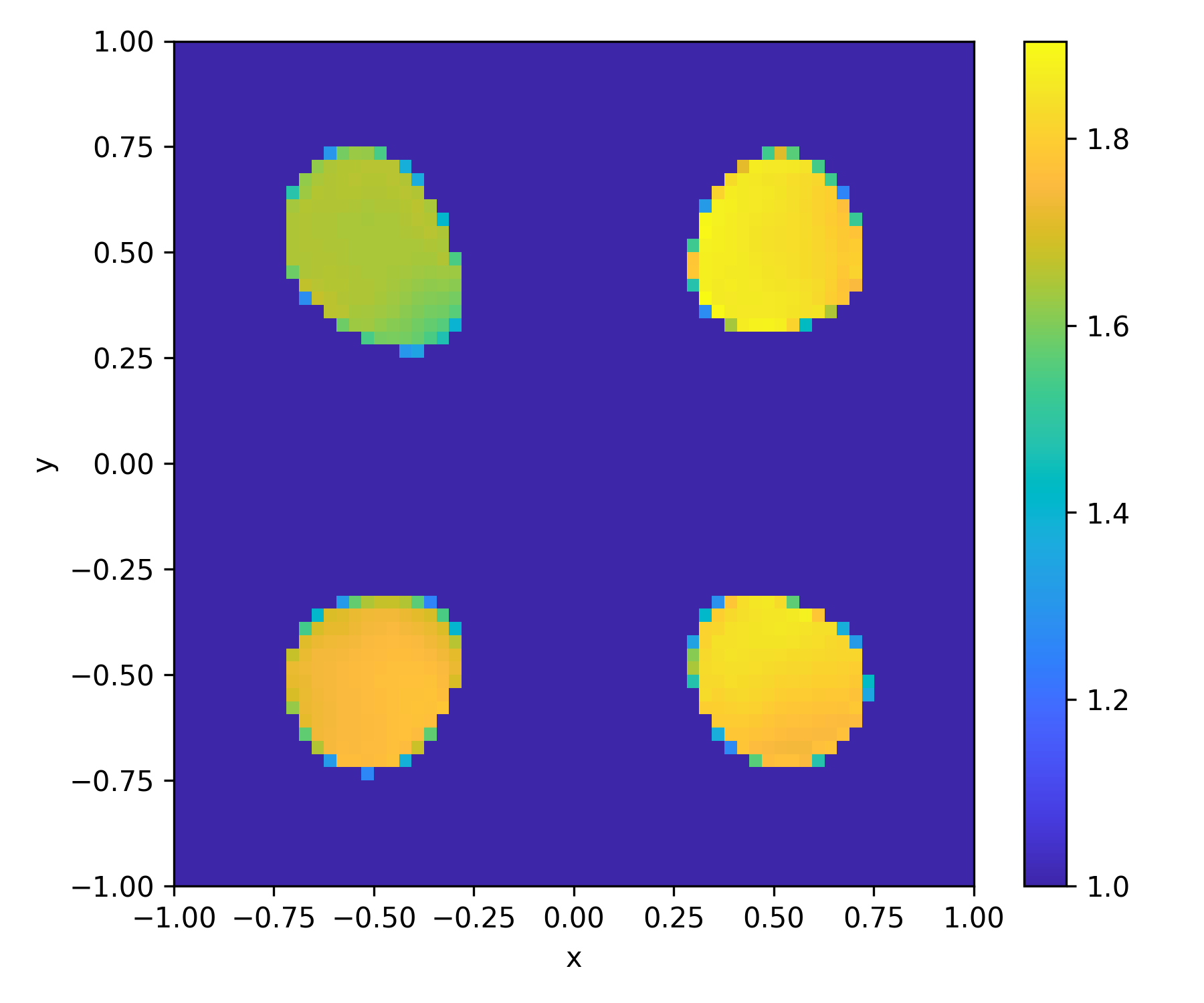}
			\end{tblr}
			\caption{Image reconstructions of a test example consisting of four circles with $15\%$ and $40\%$  Gaussian noises in the scattered fields by the networks trained by the circle dataset. From left to right: the ground-truth images, the reconstruction with 1,2,4,8, and 16 incident fields.}
			\label{tab:fig-Four_Circle}
		\end{center}
	\end{figure}
	
	\subsubsection{Tests with the “Austria ring”}
	Furthermore, we consider a rather special test called “Austria ring”, where the profile is also out of the distribution of the training data as shown in Fig.\,\ref{tab:fig-Austria_Circle} which is a well-known challenging profile in the community of inverse medium scattering problems. From Fig.\,\ref{tab:fig-Austria_Circle},  with very limited incident waves, it is reasonable to see that the DSM-DL can hardly recover the “Austria ring”. Moreover, as we have more measurement data, the accuracy and robustness of the reconstruction are improved gradually. This shows that using multiple data is very important for recovering complicated scatterers.

	\begin{figure}[htp]\small
		\begin{center}
			\begin{tblr}
				{colspec = {X[-1]X[c]X[c,h]X[c,h]X[c,h]X[c,h]X[c,h]},
					stretch = 0,
					rowsep = 0pt,}
				Noise Level& Ground truth& $N_{i}$=1& $N_{i}$=2 &$N_{i}$=4&$N_{i}$=8& $N_{i}$=16\\
				15\%&\SetCell[r=2]{c}\includegraphics[width=0.15\textwidth]{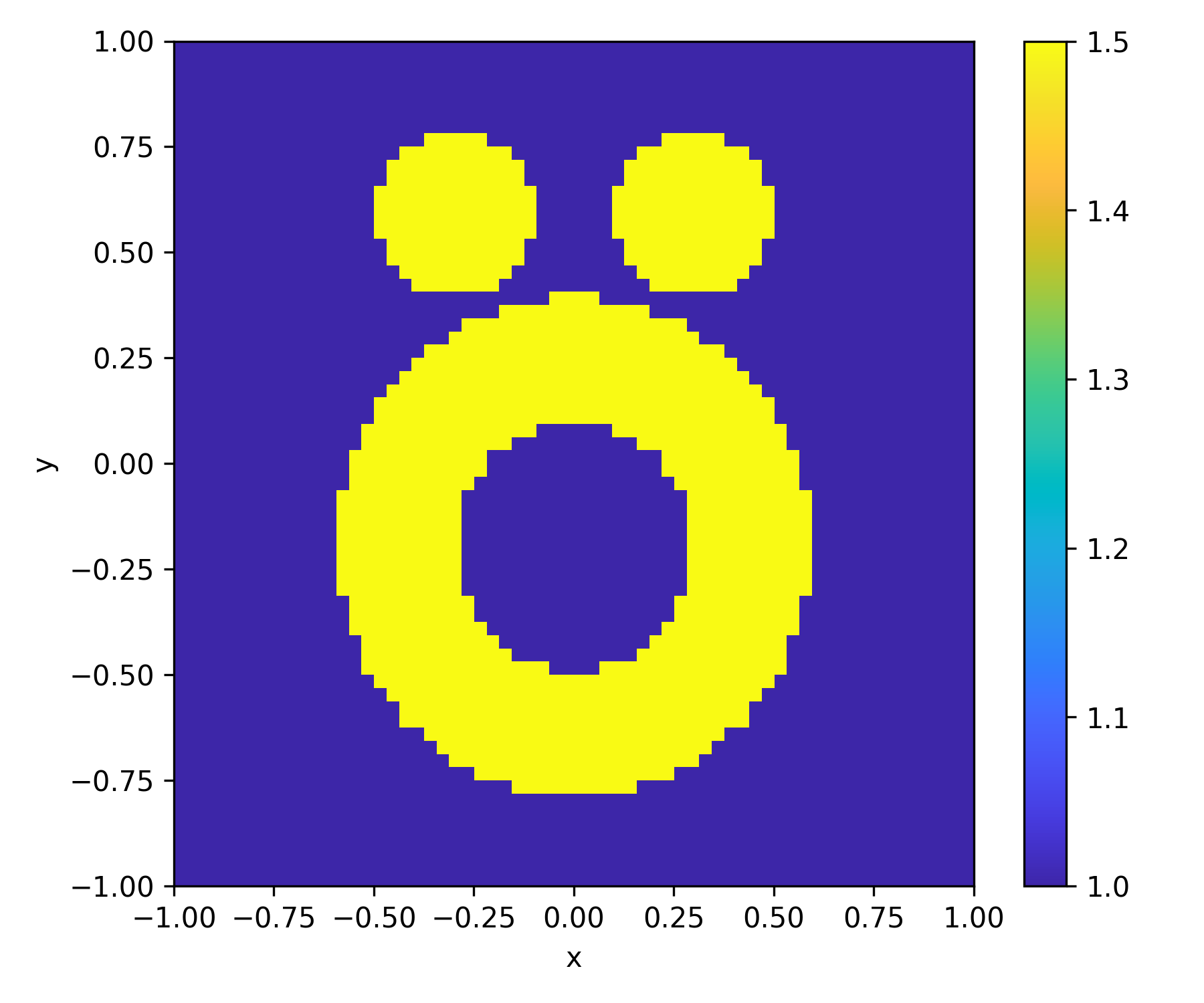}& 
				\includegraphics[width=0.15\textwidth]{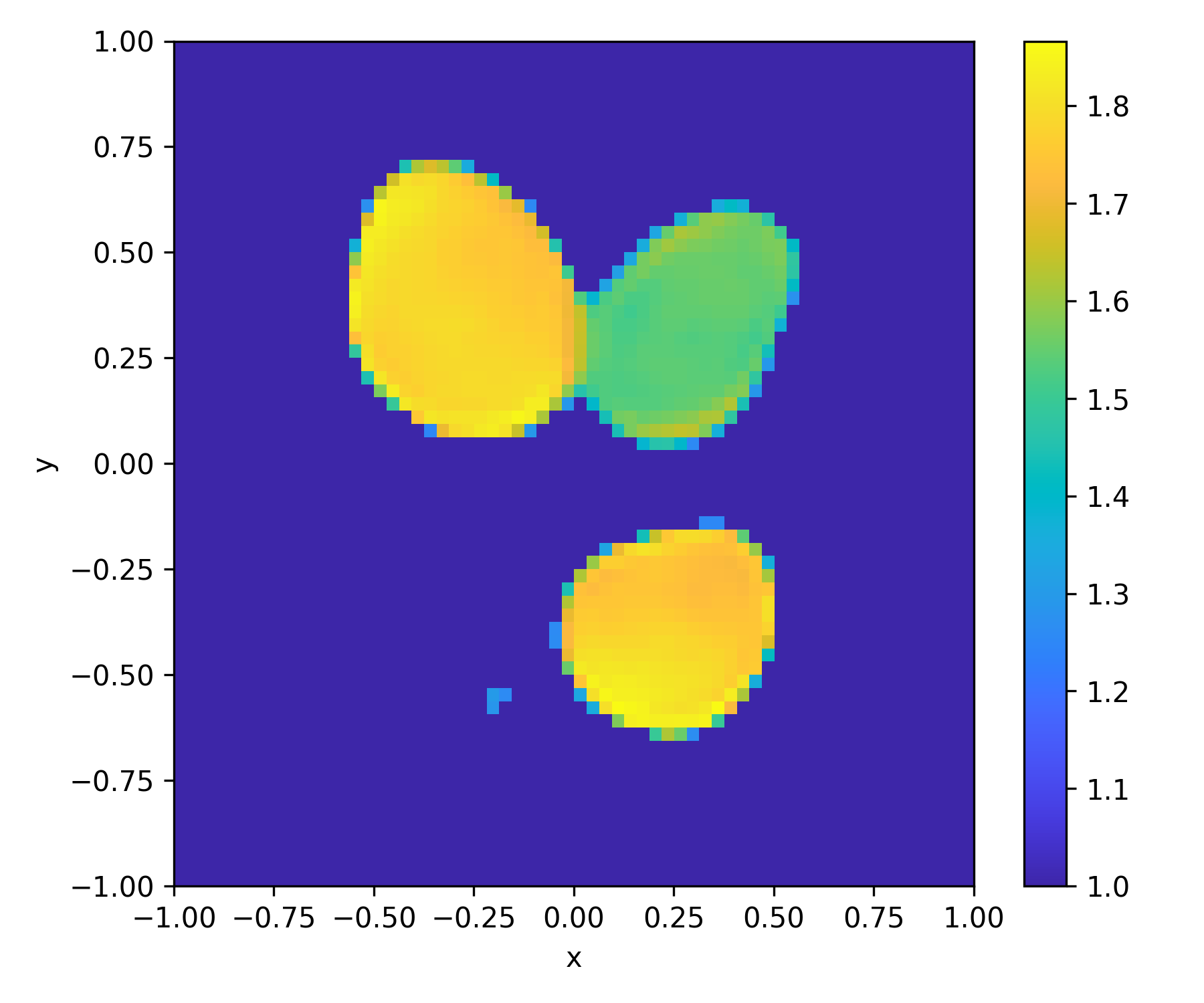}&
				\includegraphics[width=0.15\textwidth]{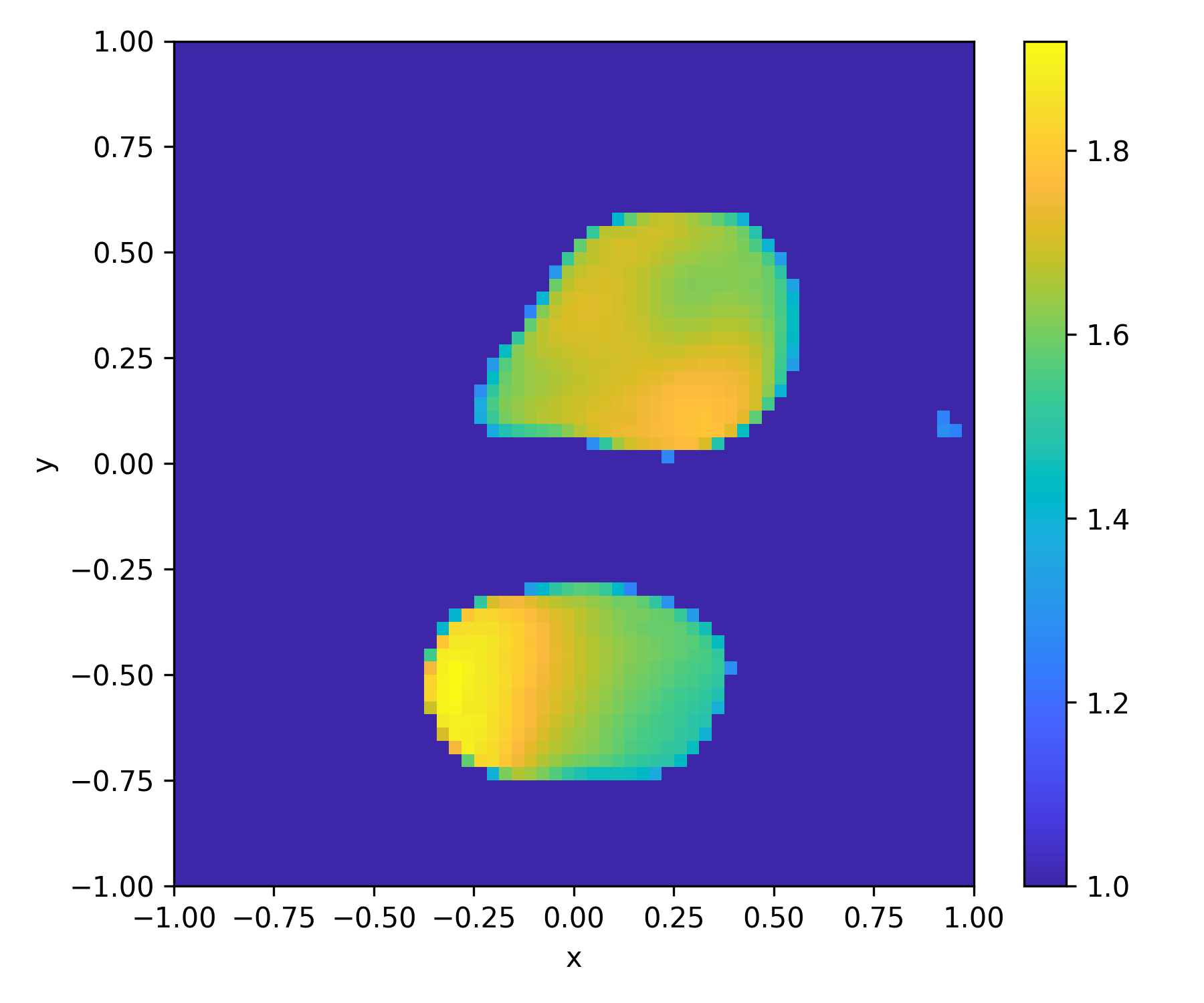}&
				\includegraphics[width=0.15\textwidth]{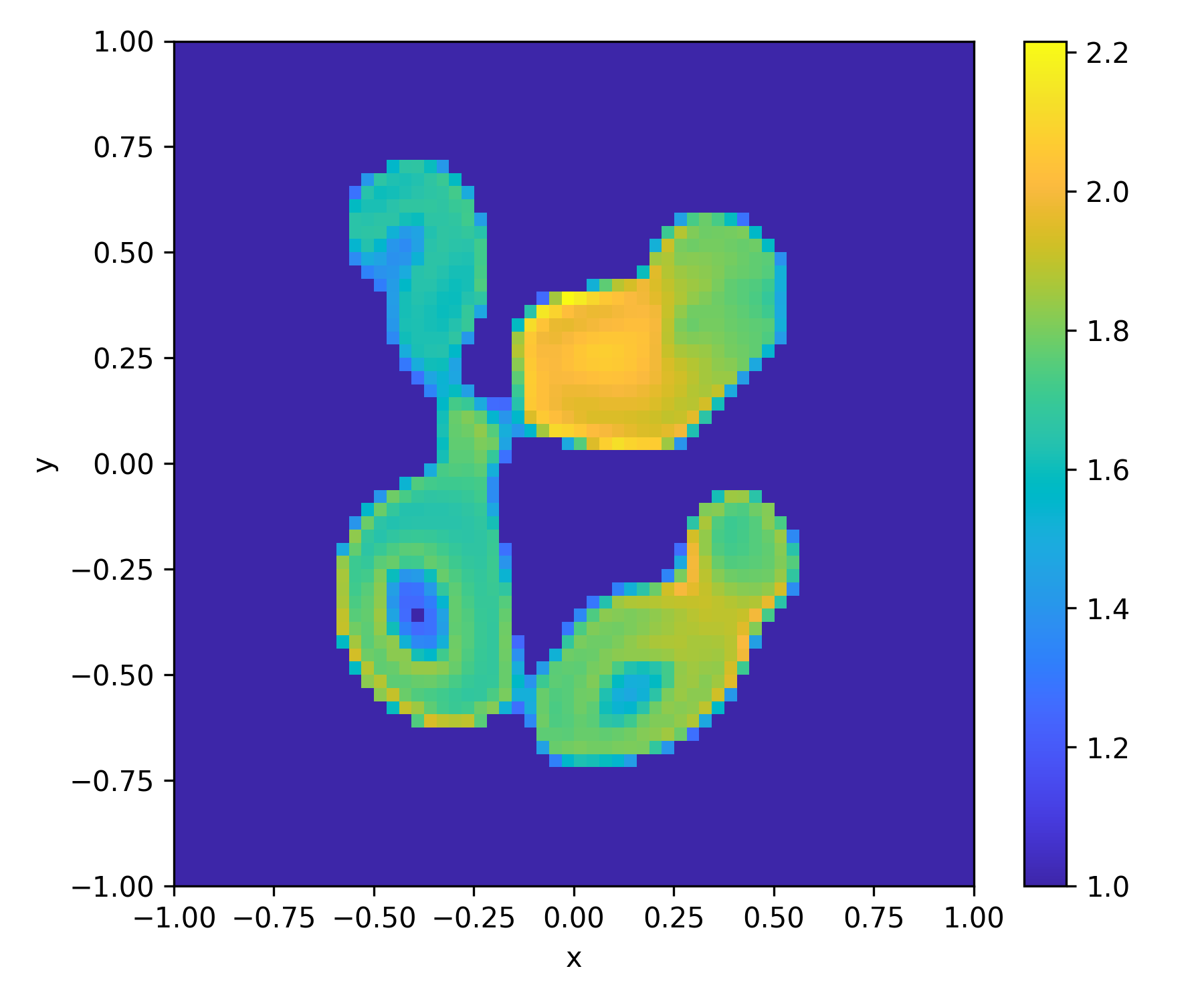}&
				\includegraphics[width=0.15\textwidth]{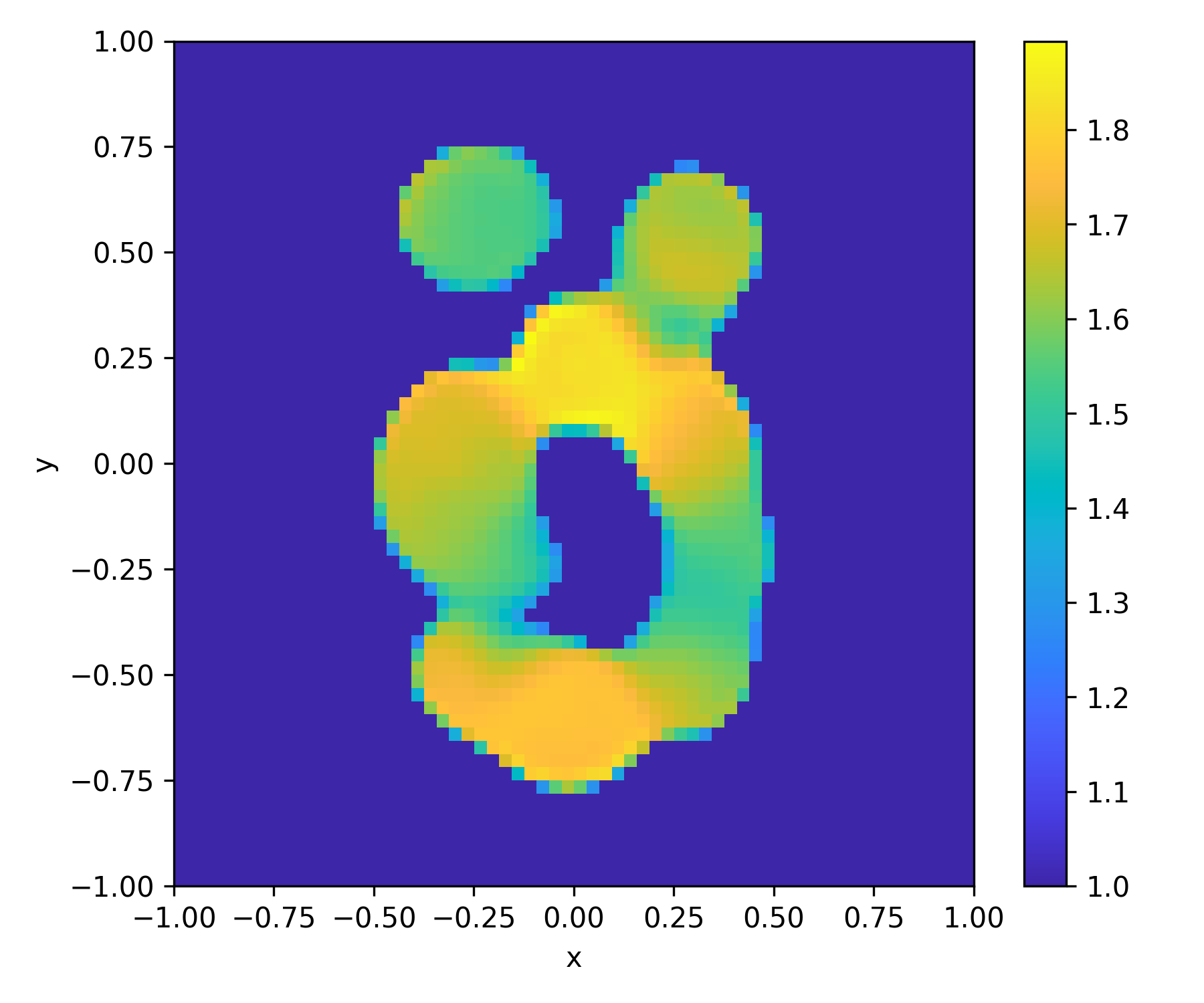}&
				\includegraphics[width=0.15\textwidth]{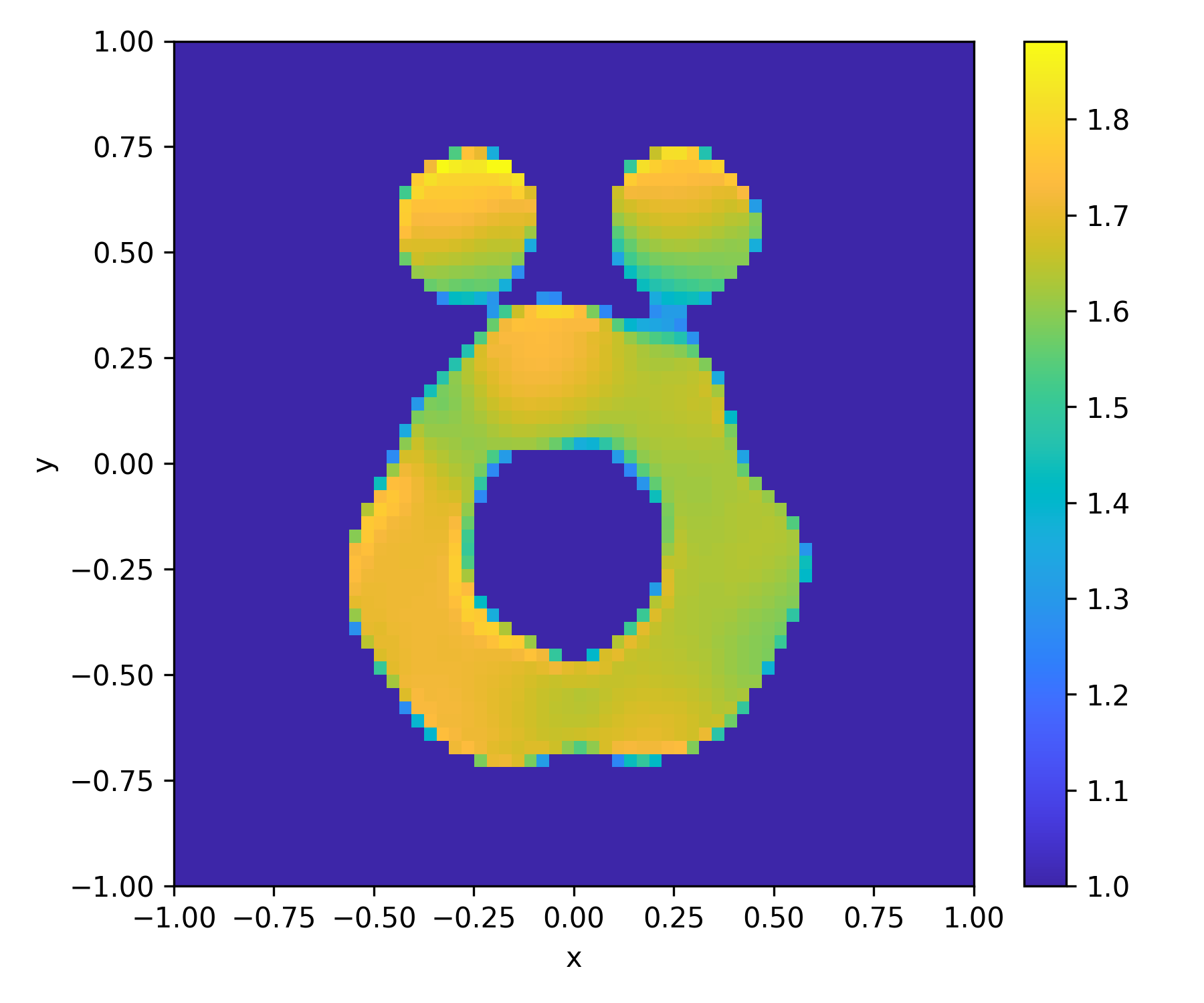}
				\\
				40\%& & 
				\includegraphics[width=0.15\textwidth]{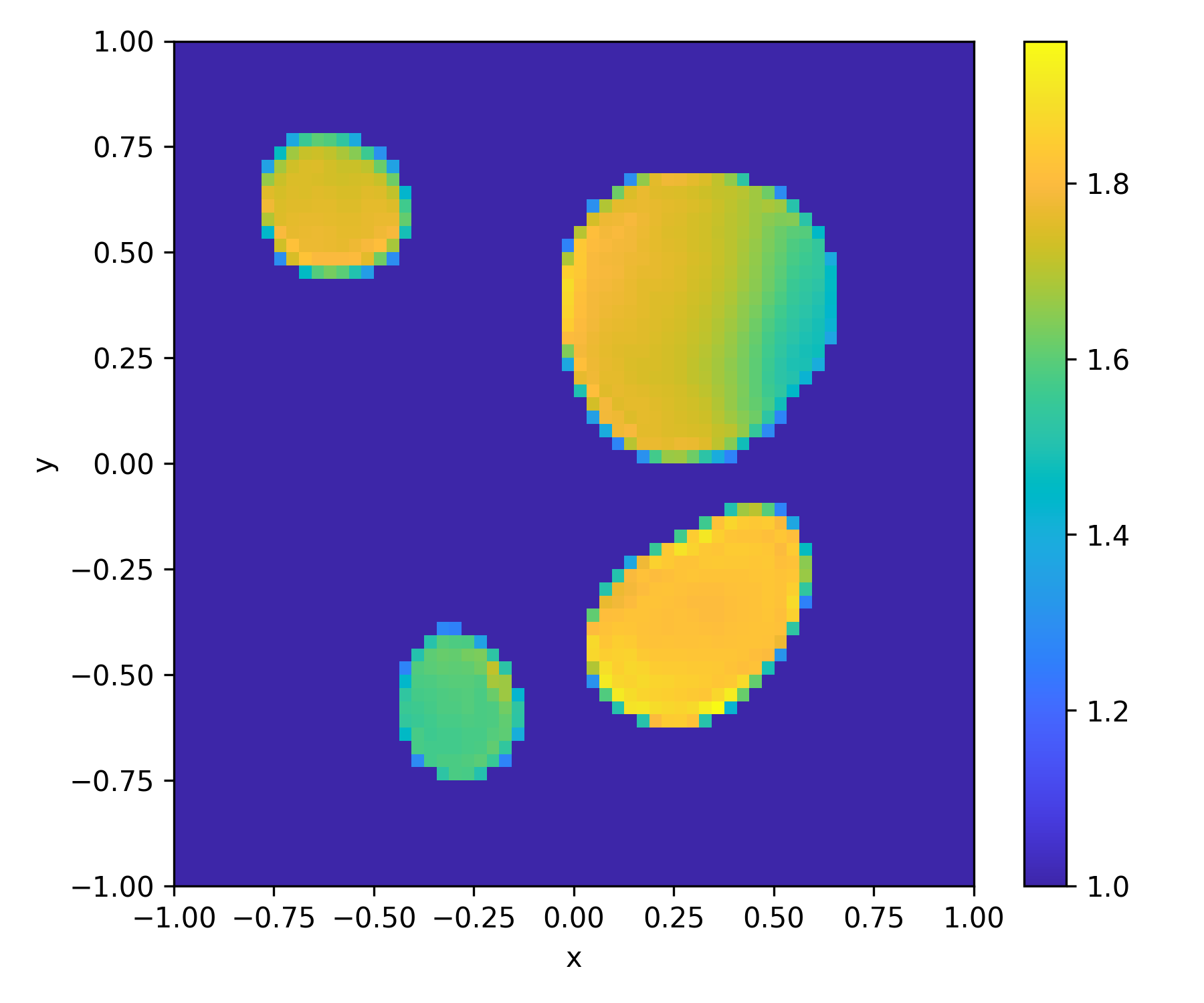}&
				\includegraphics[width=0.15\textwidth]{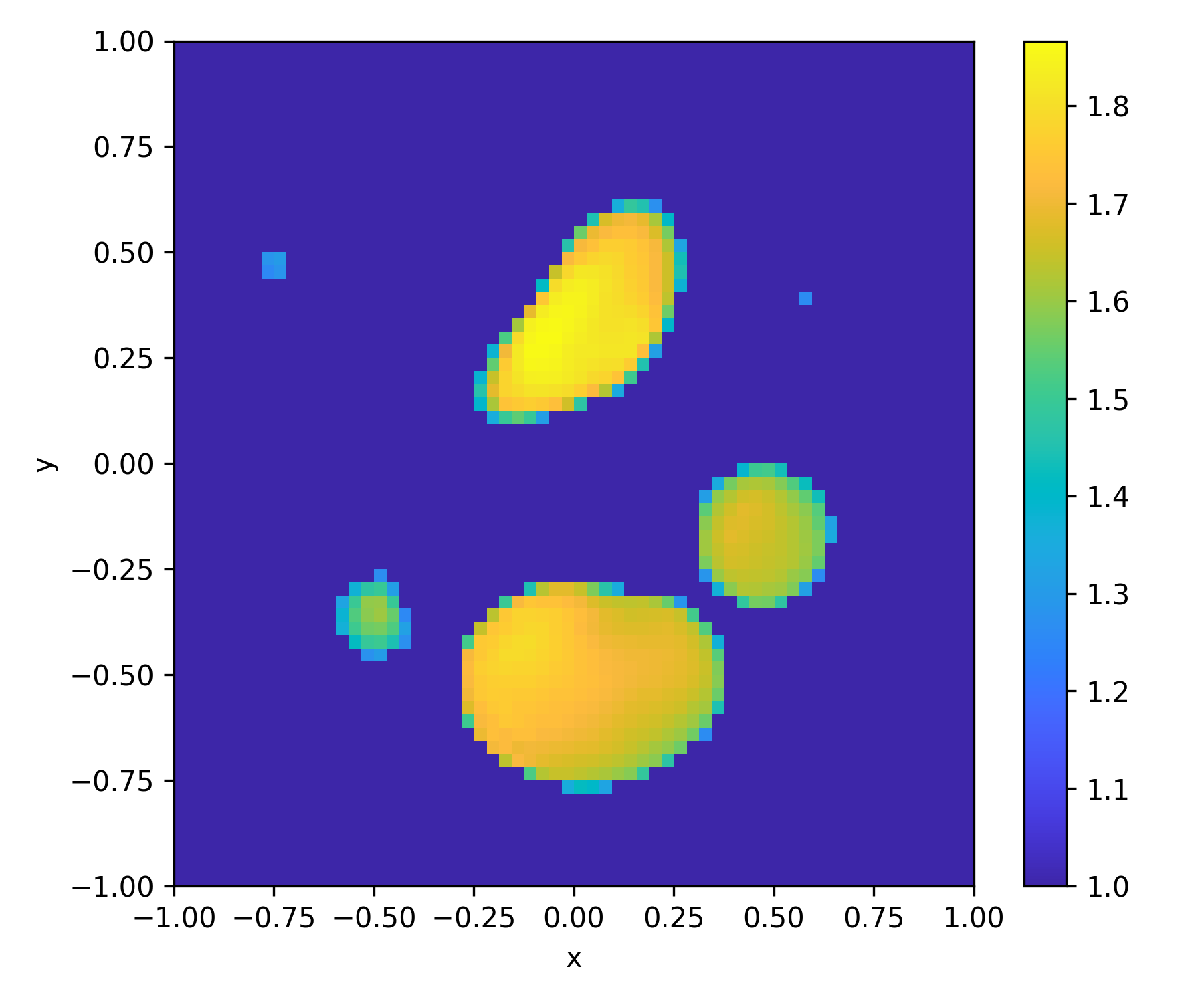}&
				\includegraphics[width=0.15\textwidth]{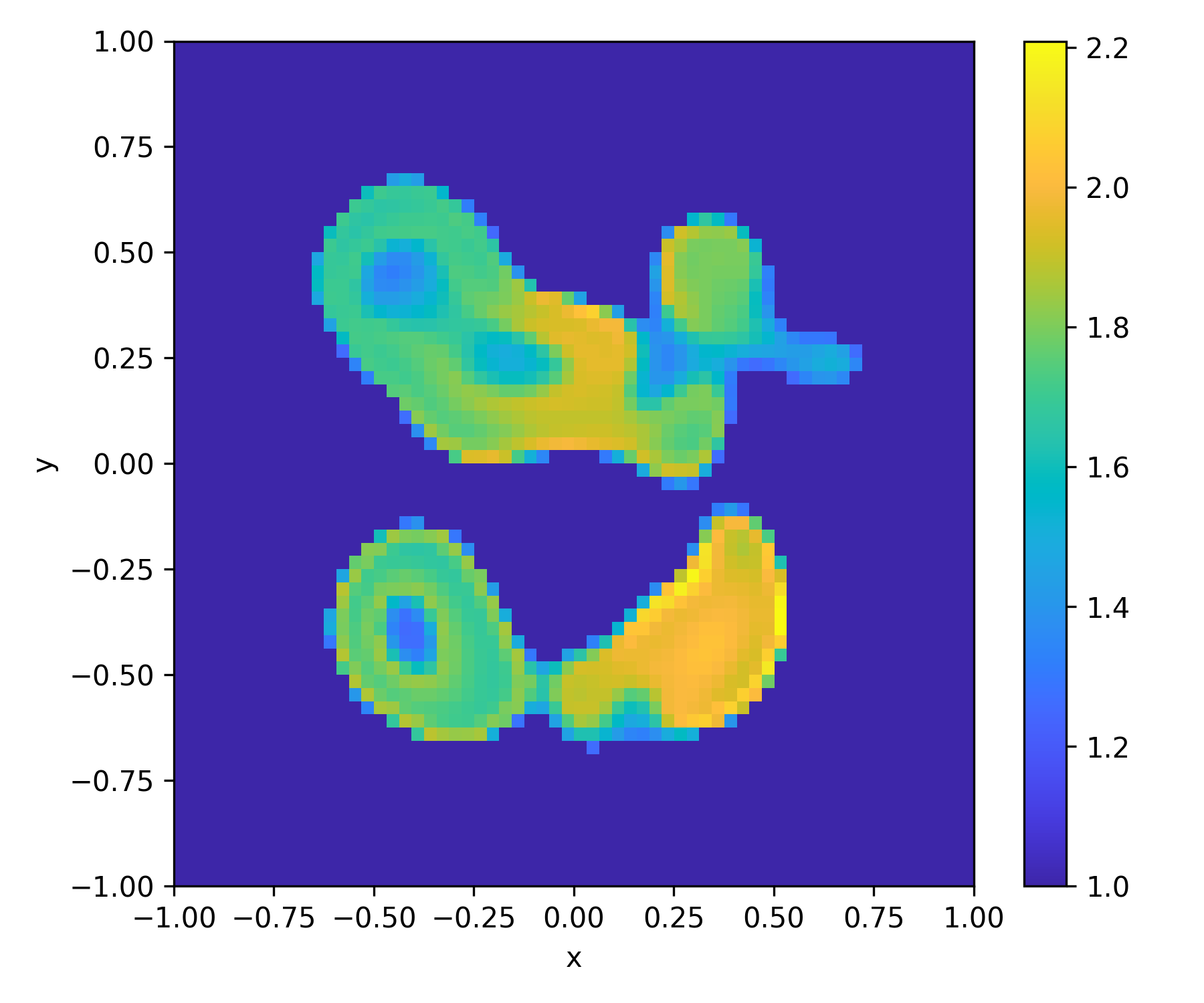}&
				\includegraphics[width=0.15\textwidth]{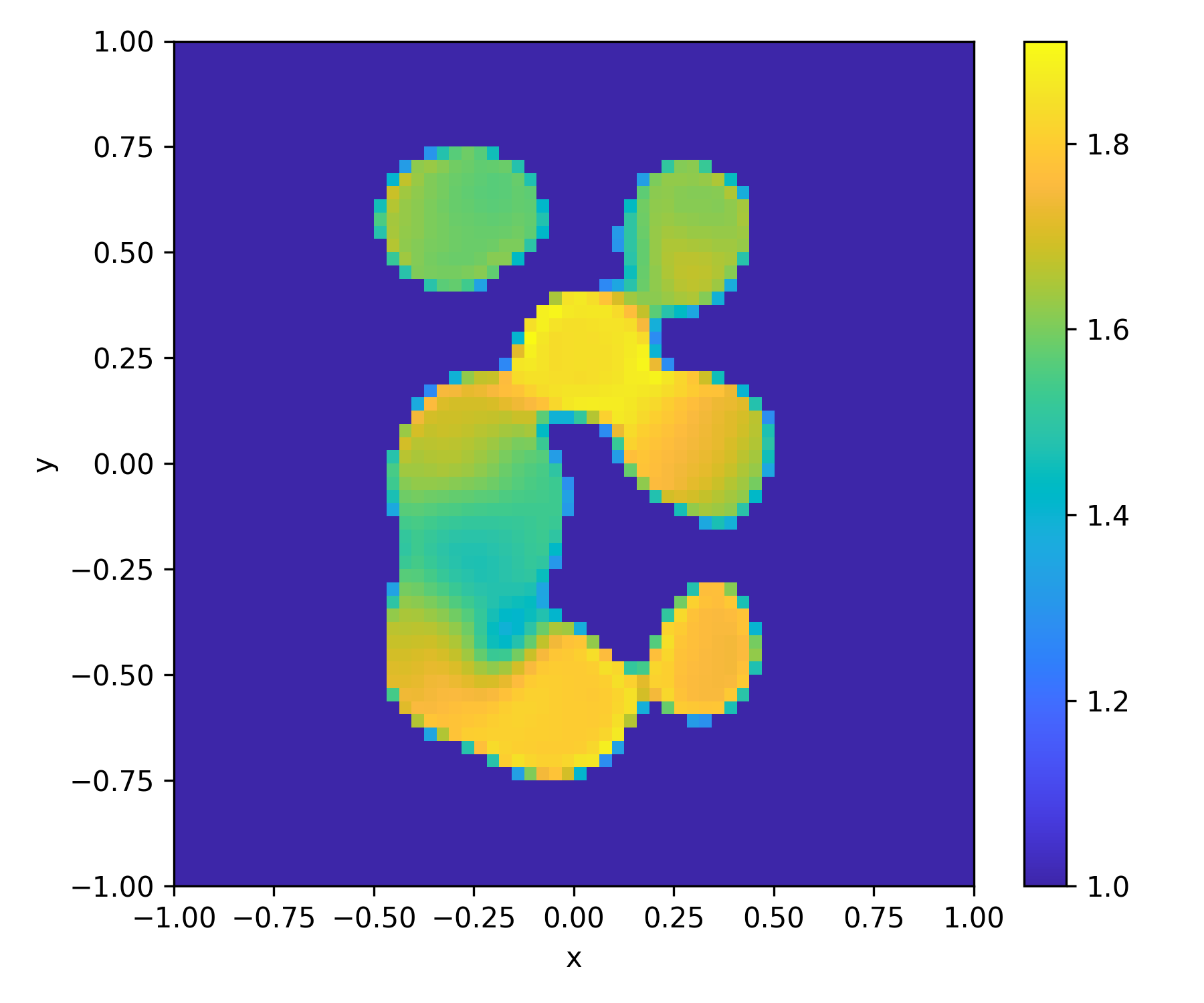}&
				\includegraphics[width=0.15\textwidth]{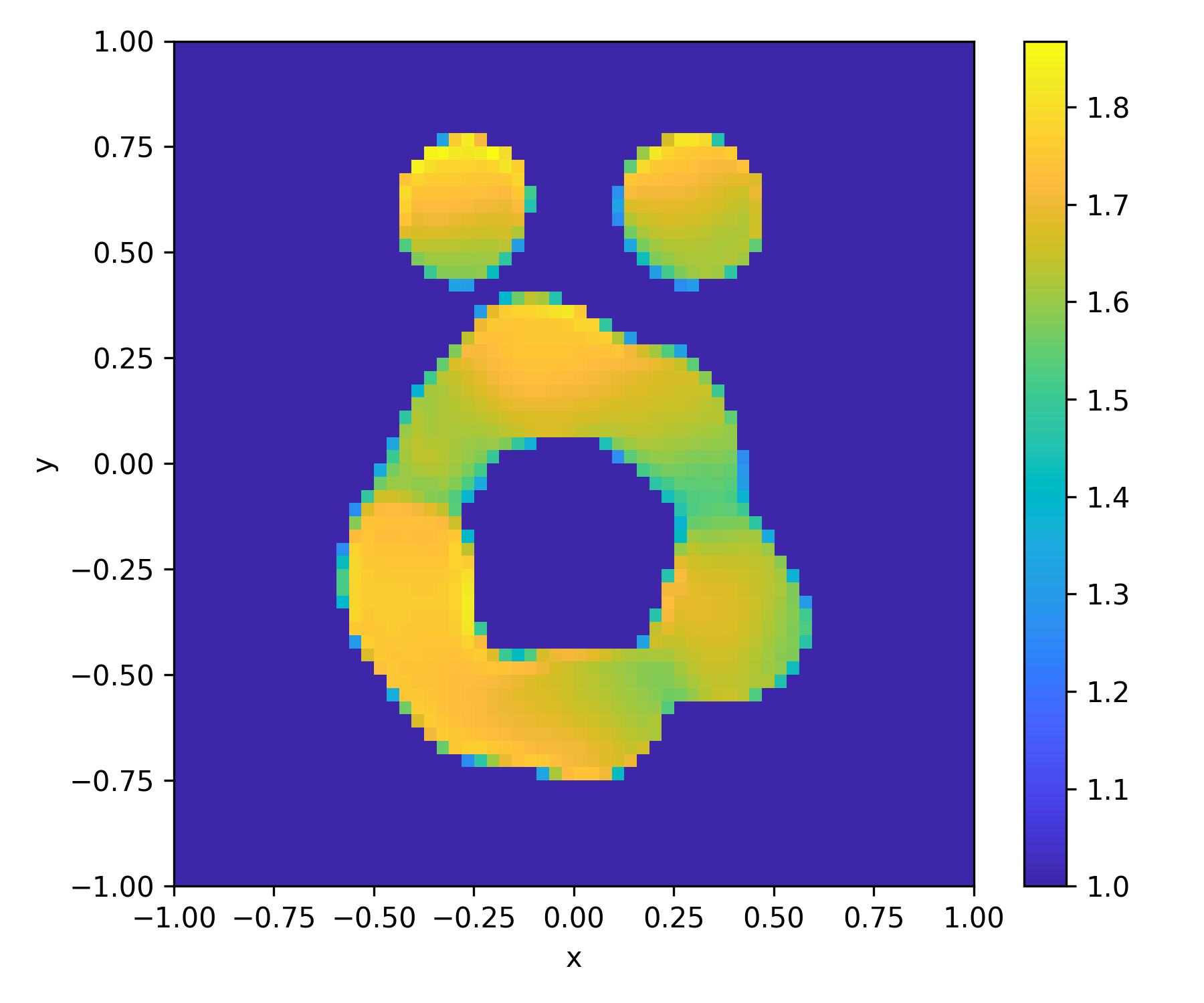}
			\end{tblr}
			\caption{Image reconstructions of the “Austria profile” with $15\%$ and $40\%$ Gaussian noises in the scattered fields by the networks trained by the circle dataset. From left to right: the ground-truth images, the reconstruction with 1,2,4,8, and 16 incident fields.}
			\label{tab:fig-Austria_Circle}
		\end{center}
	\end{figure}

  \subsubsection{Tests with the high-contrast scatterers}
  In this example, a high-contrast circle dataset is used to train and test the proposed neural network. The relative permittivity of the scatterers is sampled from $U(3.5,4.0)$, while other settings are the same as the previous circle dataset example. Thus, this inverse problem is much more ill-posed and nonlinear compared to the previous ones. Some reconstructions from the testing data are shown in Fig.\,\ref{tab:fig-Circle_Strong}, and the relative L2 testing error and SSIM are listed in Table\,\ref{tab:error_circle}. For small $N_i$, especially for the case of three scatterers in the domain, the reconstructed scatterers are distorted and there are some unexpected artifacts. The reconstruction quality can be significantly improved as we have more data and is quite satisfactory with at least $4$ incidences fields which may be because the classical DSM can be applied to the high-contrast case.

	\begin{figure}[htp]\small
	\begin{center}
		\begin{tblr}
			{colspec = {X[-1]X[c]X[c,h]X[c,h]X[c,h]X[c,h]X[c,h]},
				stretch = 0,
				rowsep = 0pt,}
			Noise Level& Ground truth& $N_{i}$=1& $N_{i}$=2 &$N_{i}$=4&$N_{i}$=8& $N_{i}$=16\\
			15\%&\SetCell[r=2]{c}\includegraphics[width=0.15\textwidth]{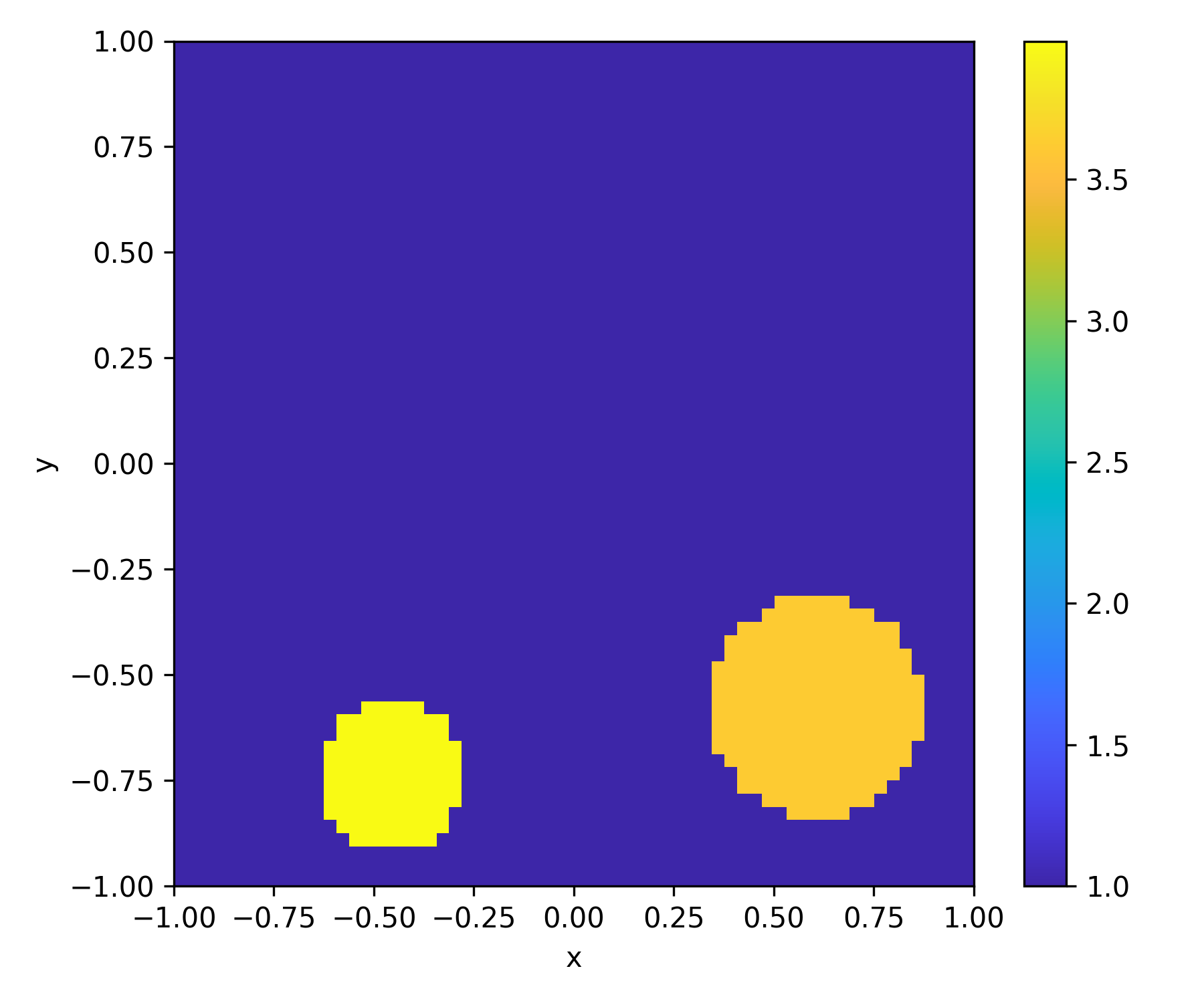}&
			\includegraphics[width=0.15\textwidth]{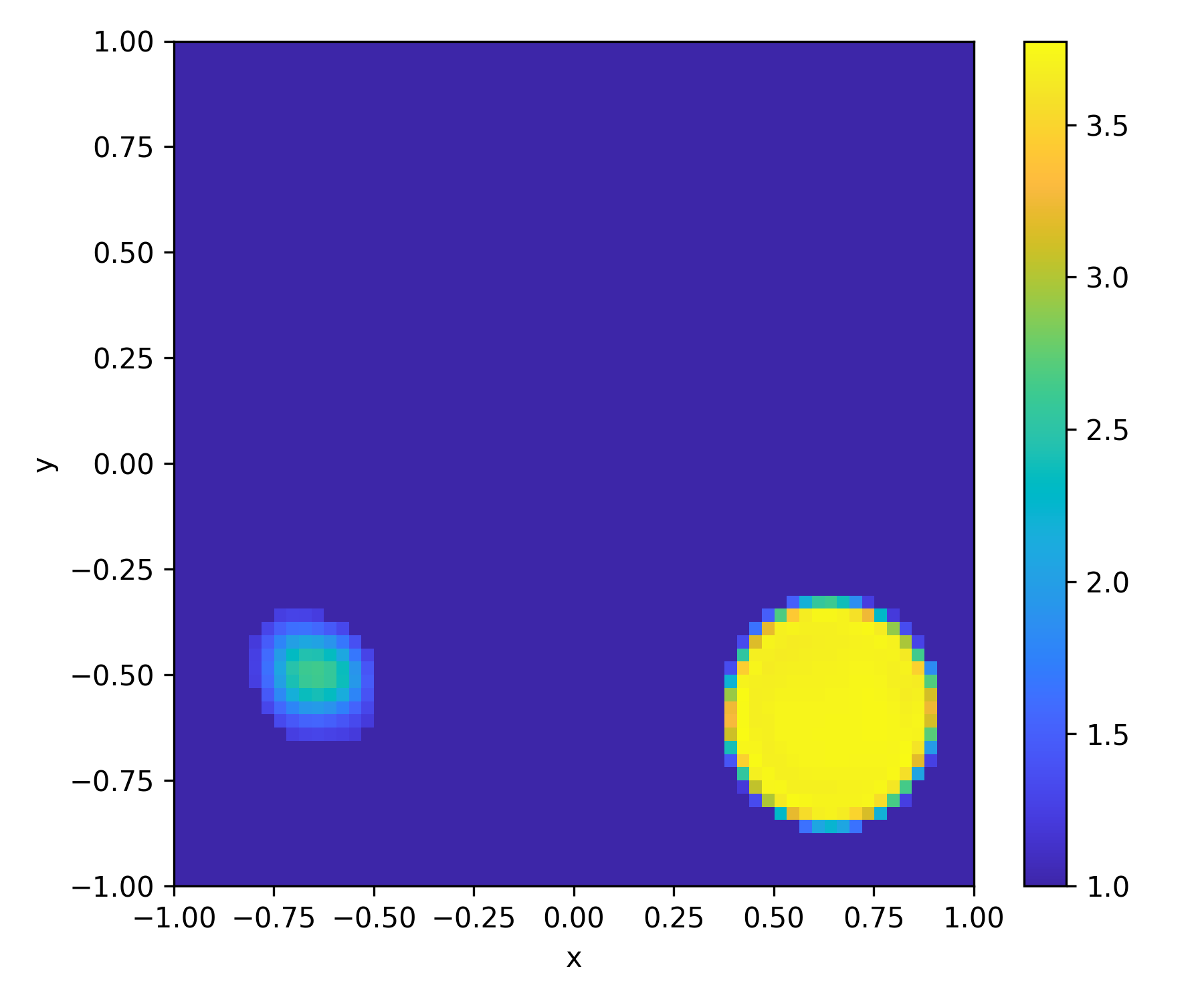}&
			\includegraphics[width=0.15\textwidth]{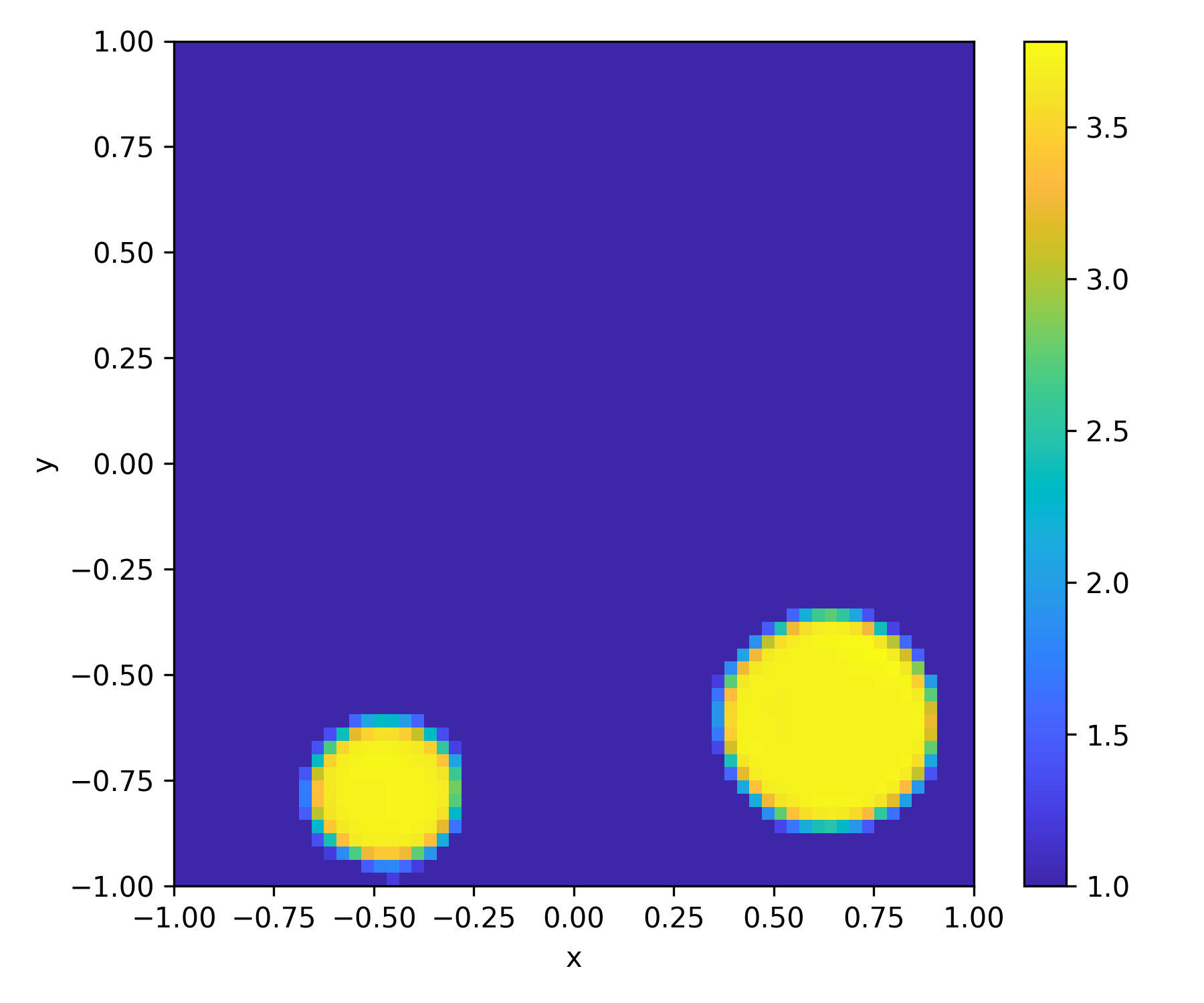}&
			\includegraphics[width=0.15\textwidth]{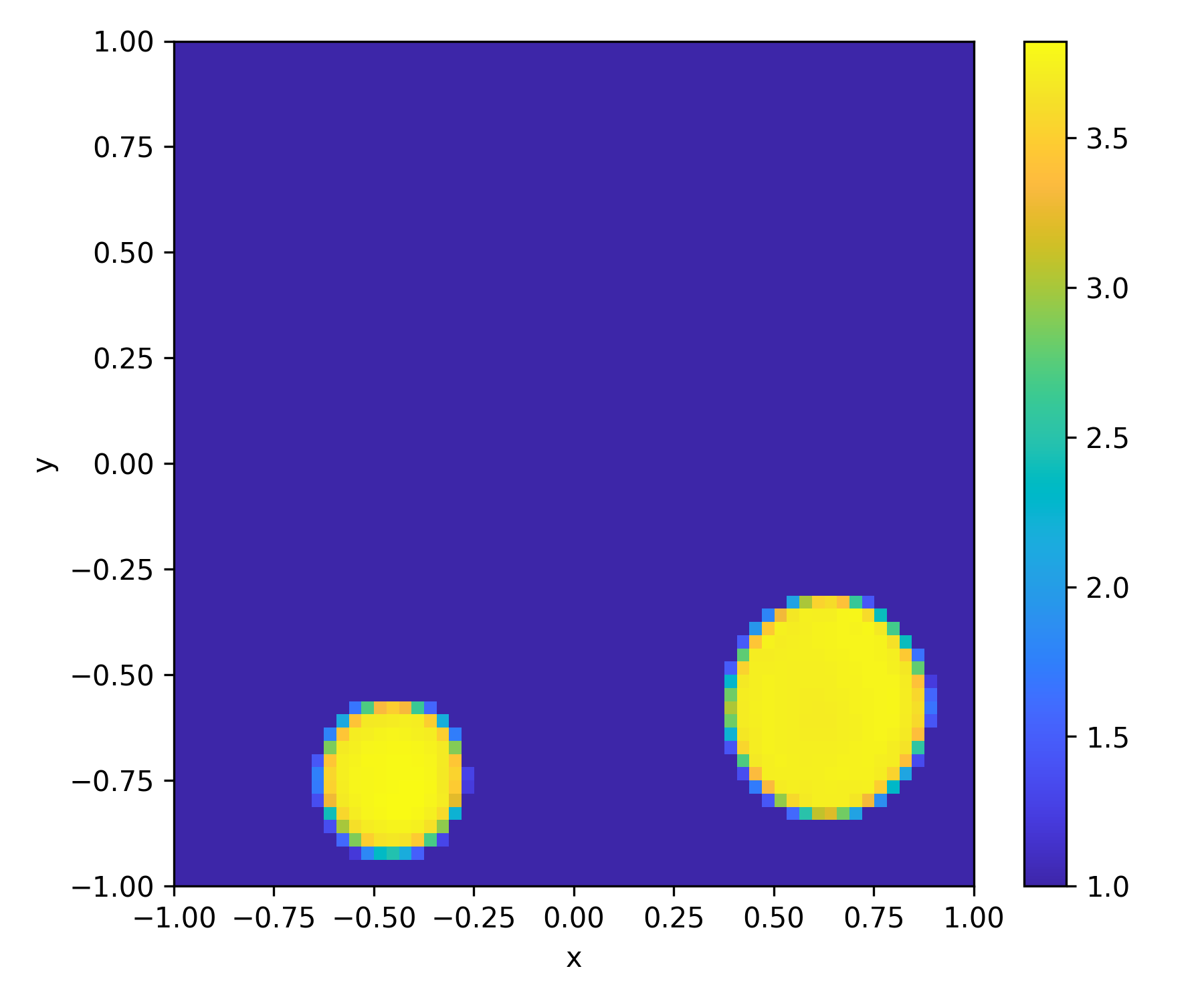}&
			\includegraphics[width=0.15\textwidth]{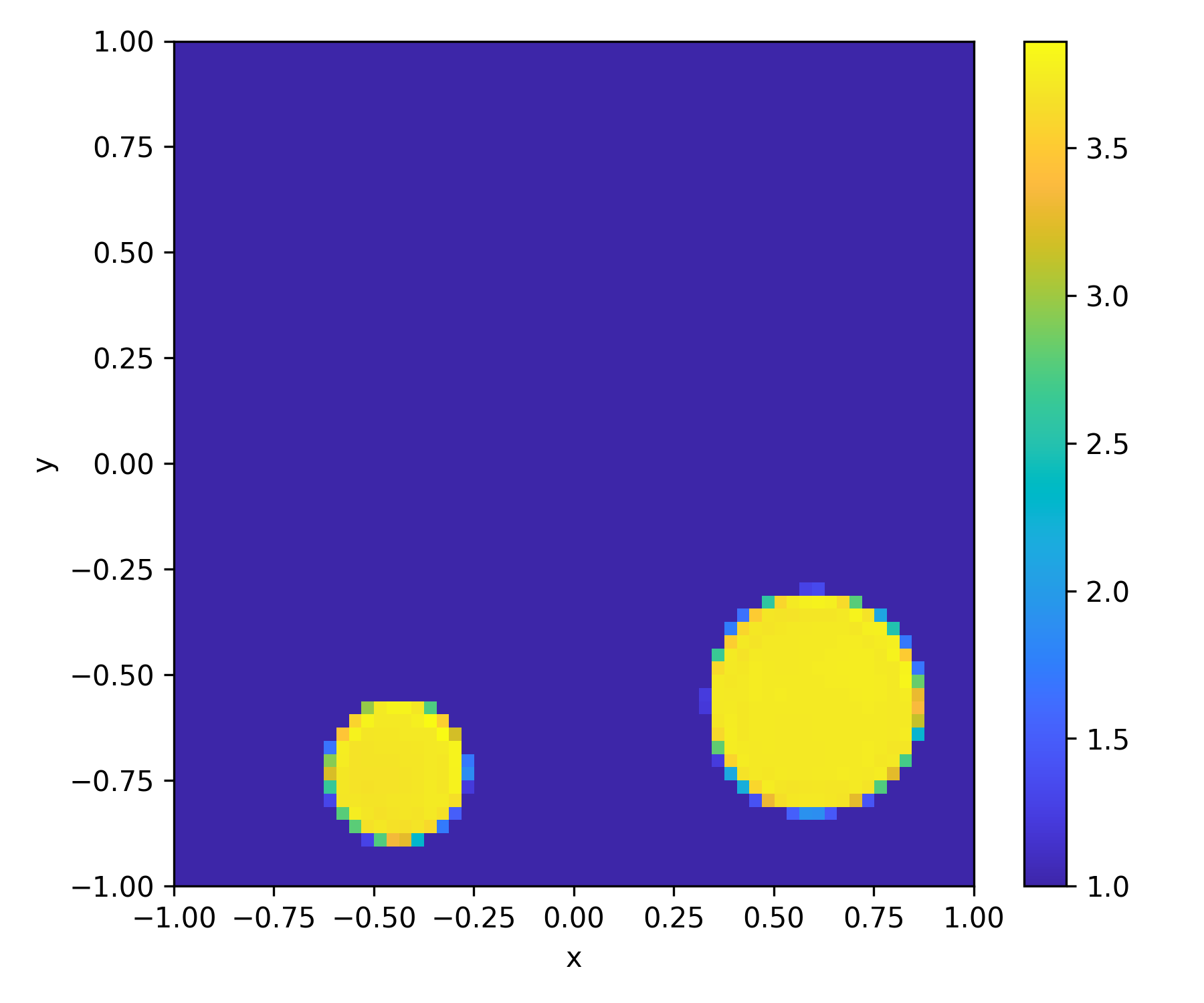}&
			\includegraphics[width=0.15\textwidth]{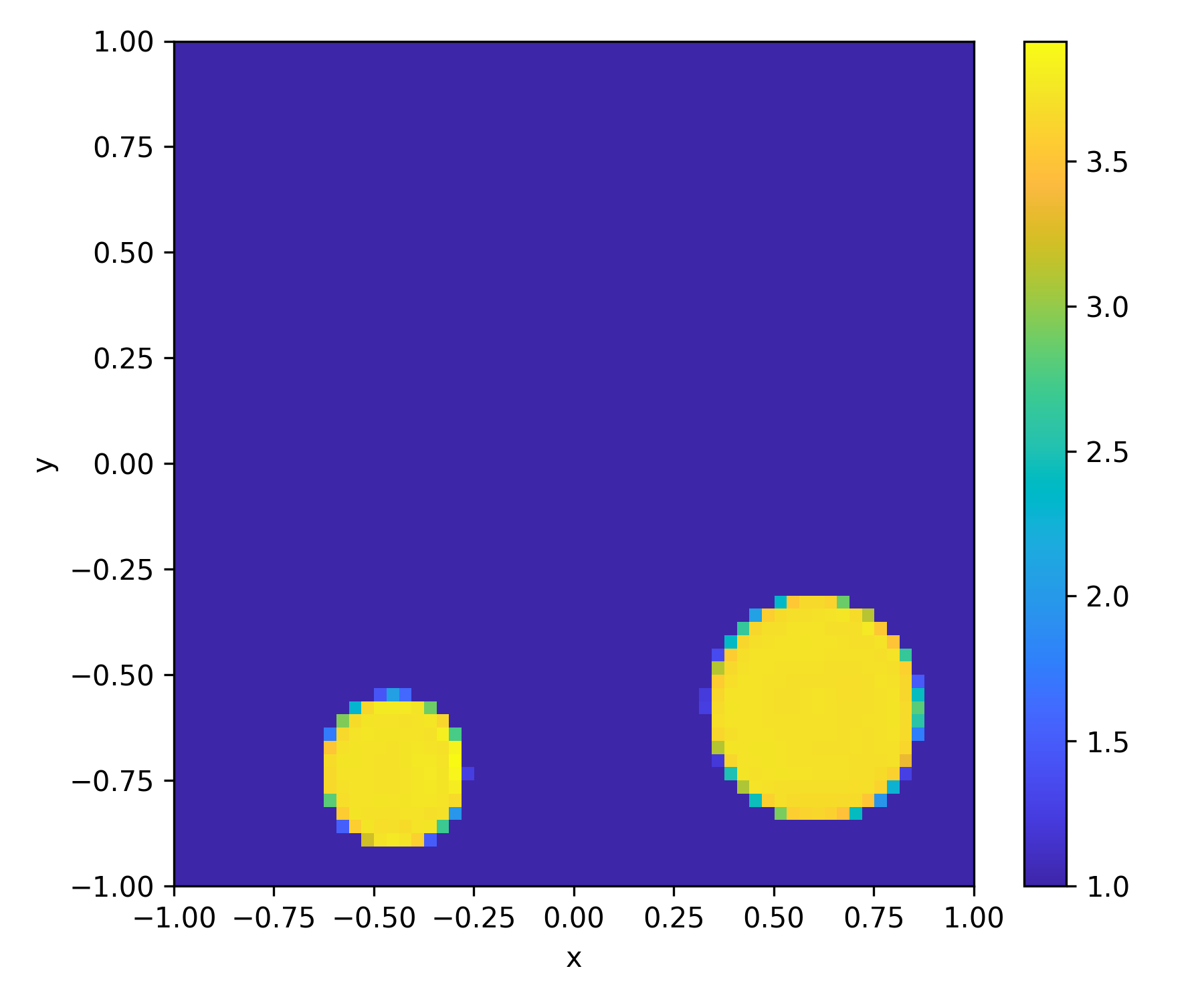}
			\\
			40\%& &
			\includegraphics[width=0.15\textwidth]{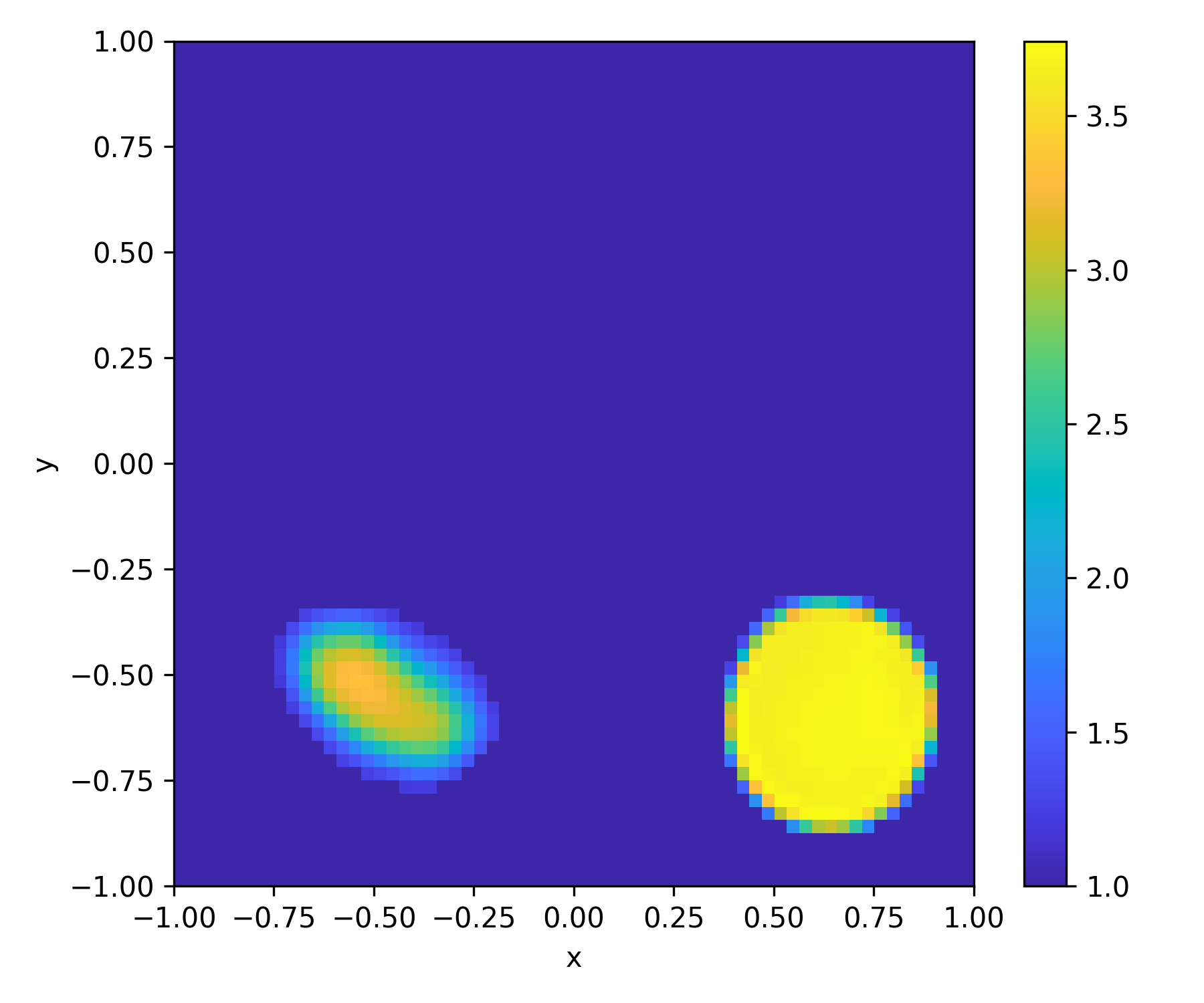}&
	    	\includegraphics[width=0.15\textwidth]{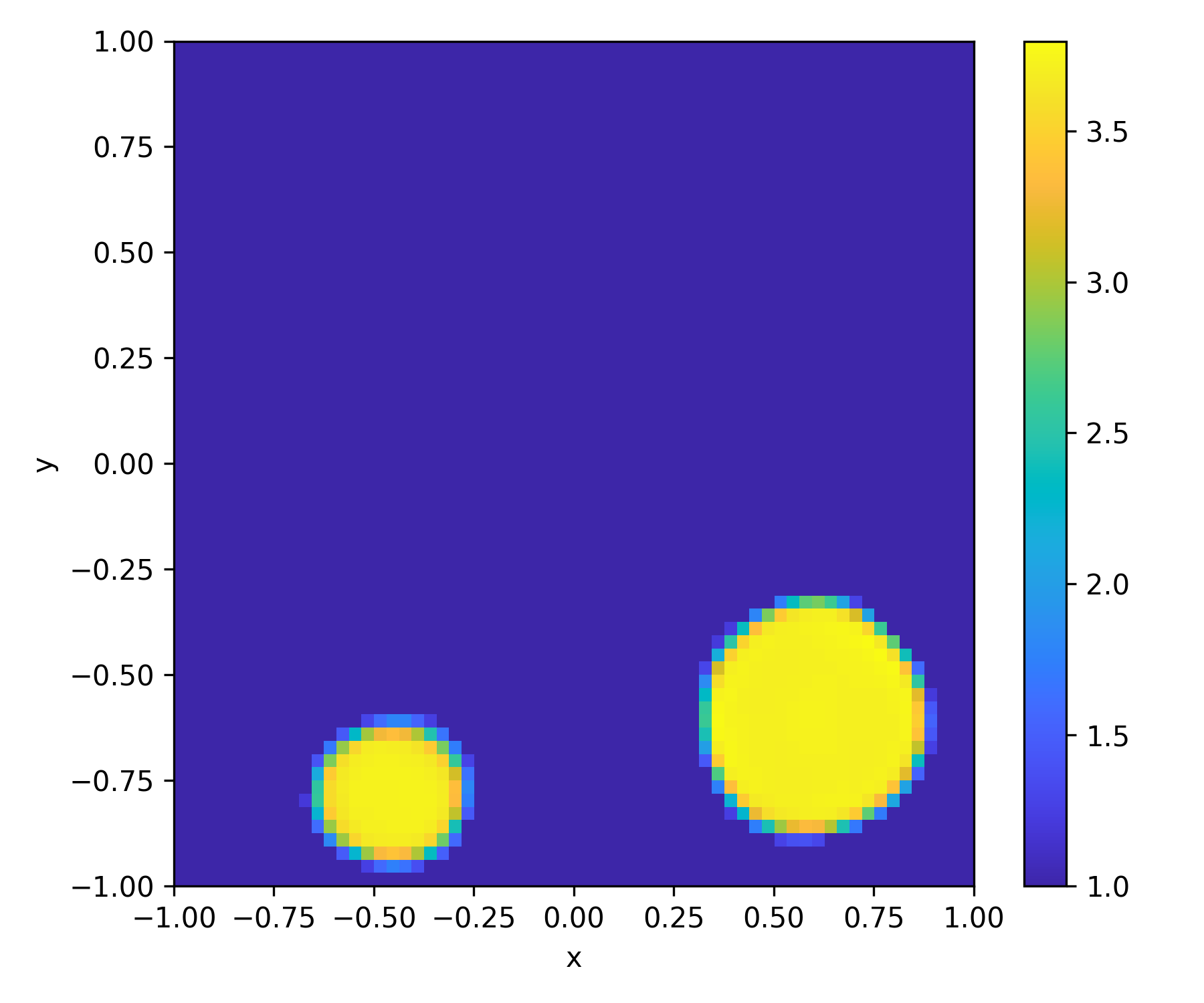}&
		   \includegraphics[width=0.15\textwidth]{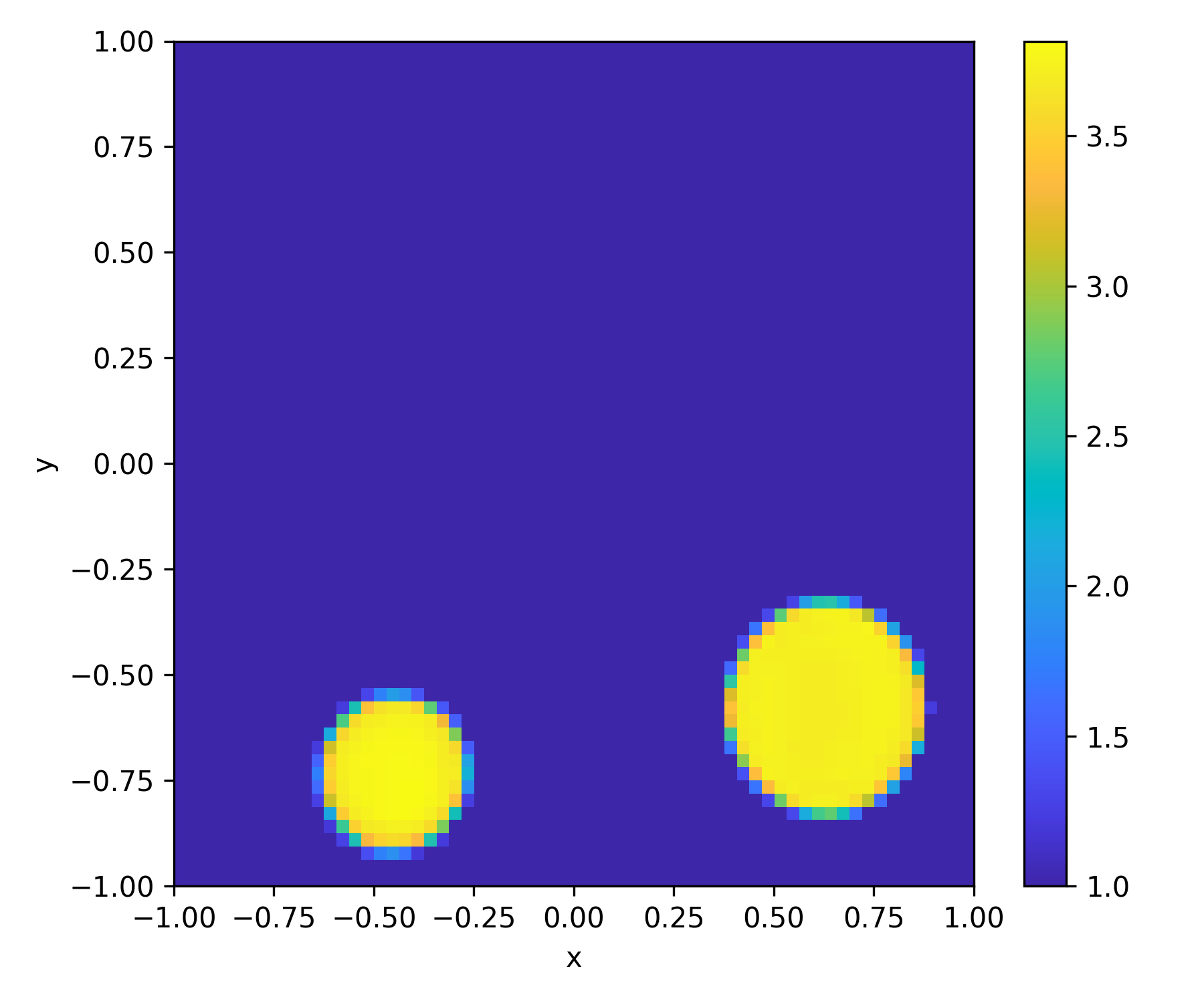}&
		  \includegraphics[width=0.15\textwidth]{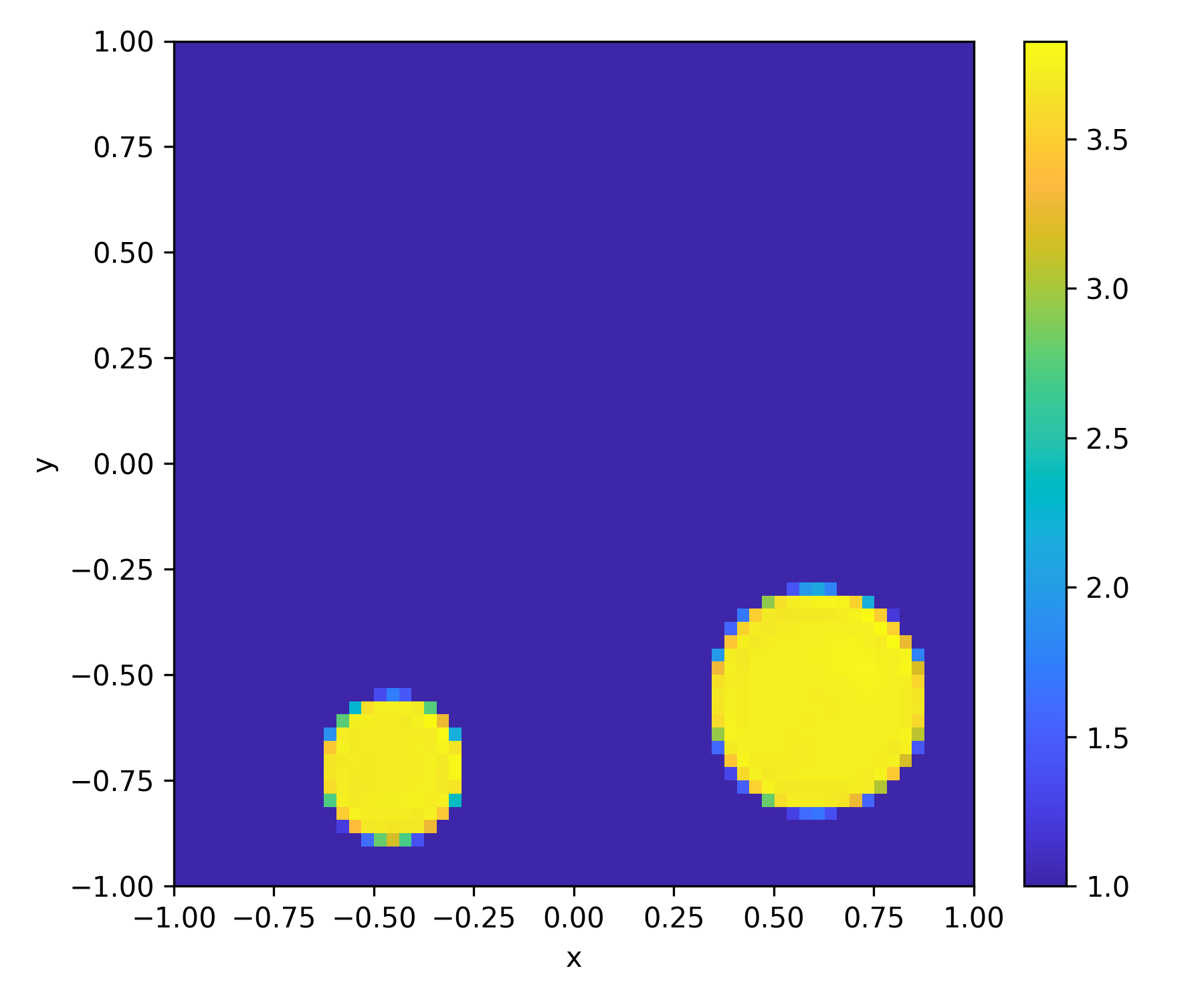}&
		  \includegraphics[width=0.15\textwidth]{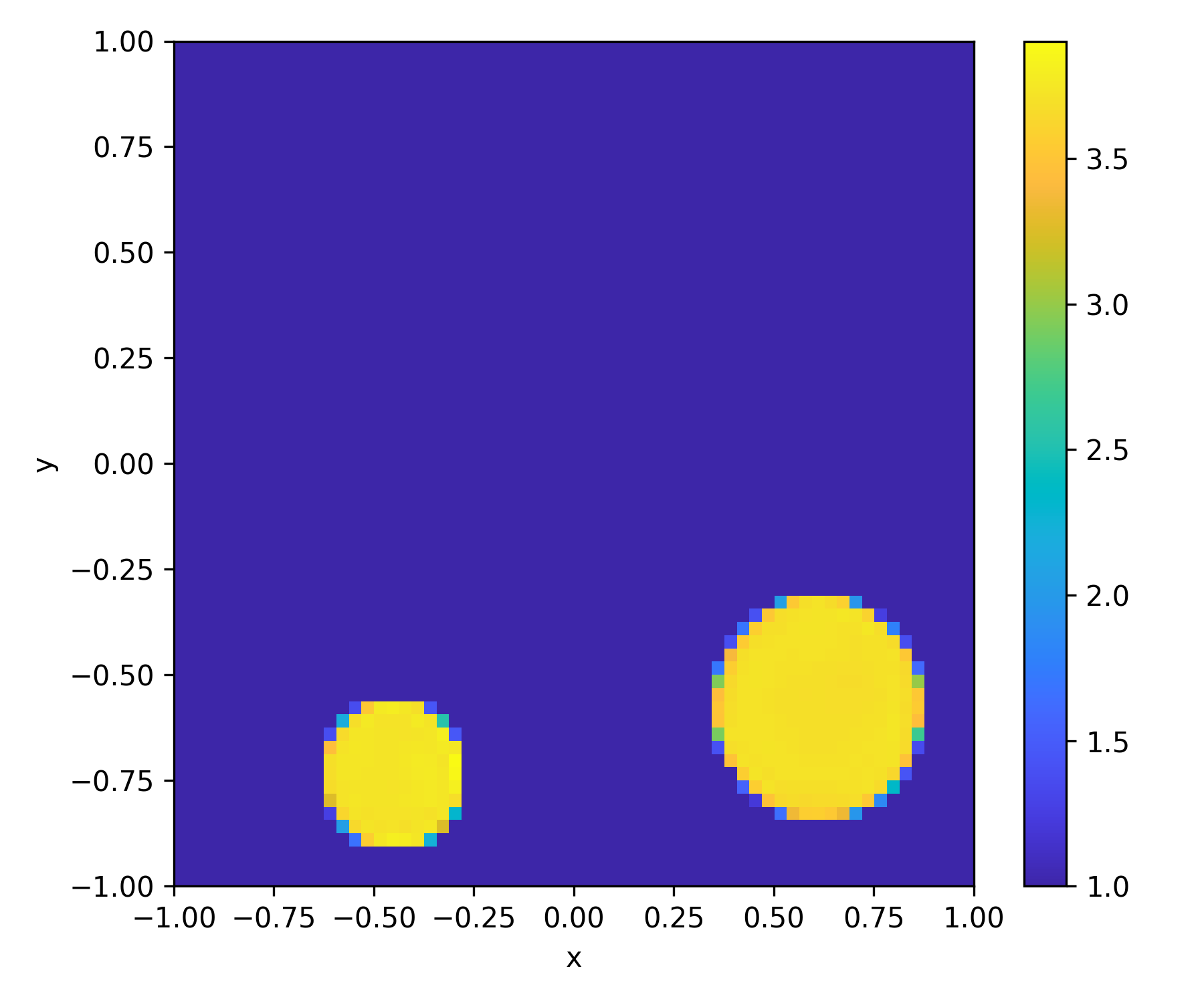}
			\\
			
			15\%&\SetCell[r=2]{c}\includegraphics[width=0.15\textwidth]{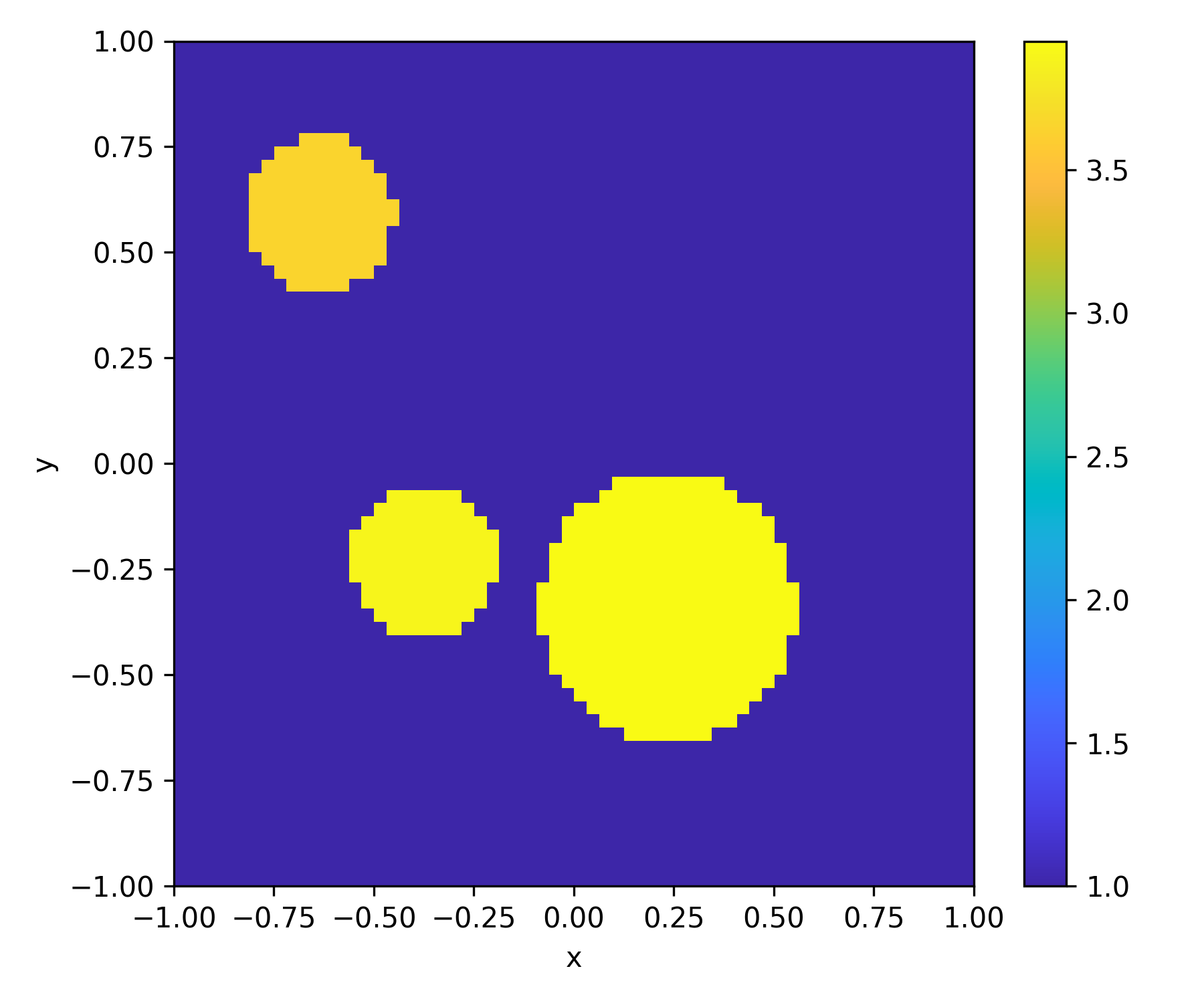}&
			\includegraphics[width=0.15\textwidth]{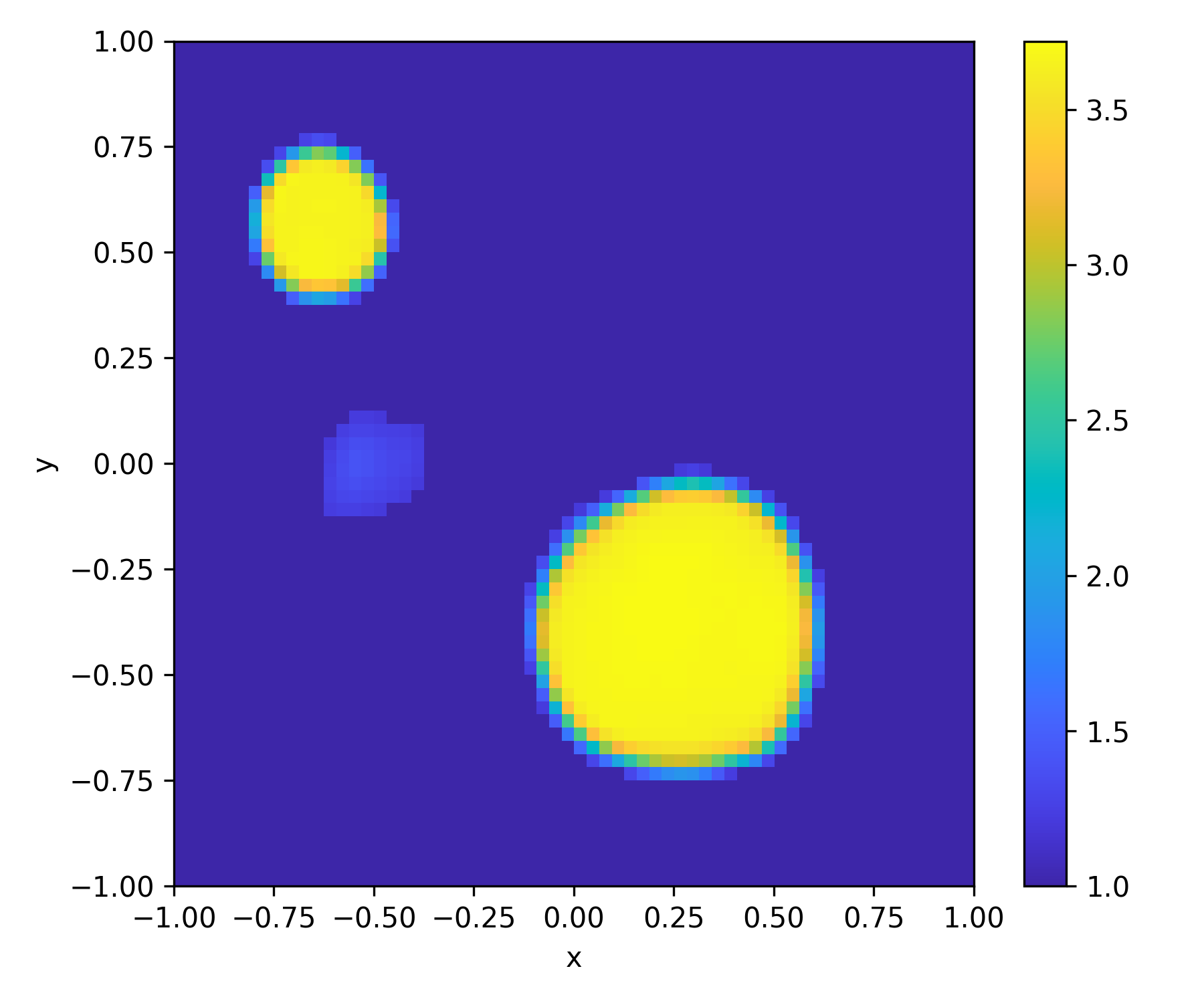}&
			\includegraphics[width=0.15\textwidth]{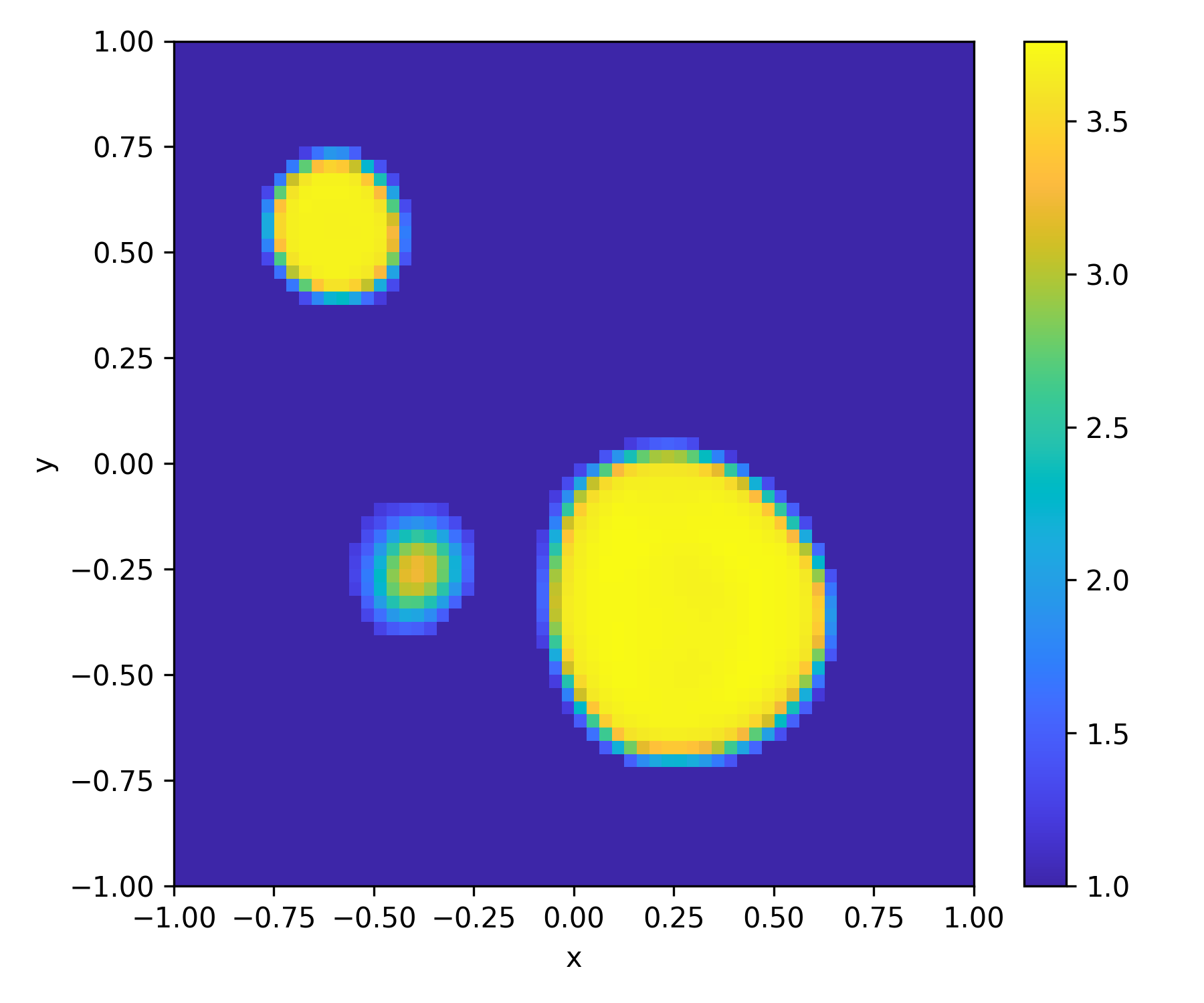}&
			\includegraphics[width=0.15\textwidth]{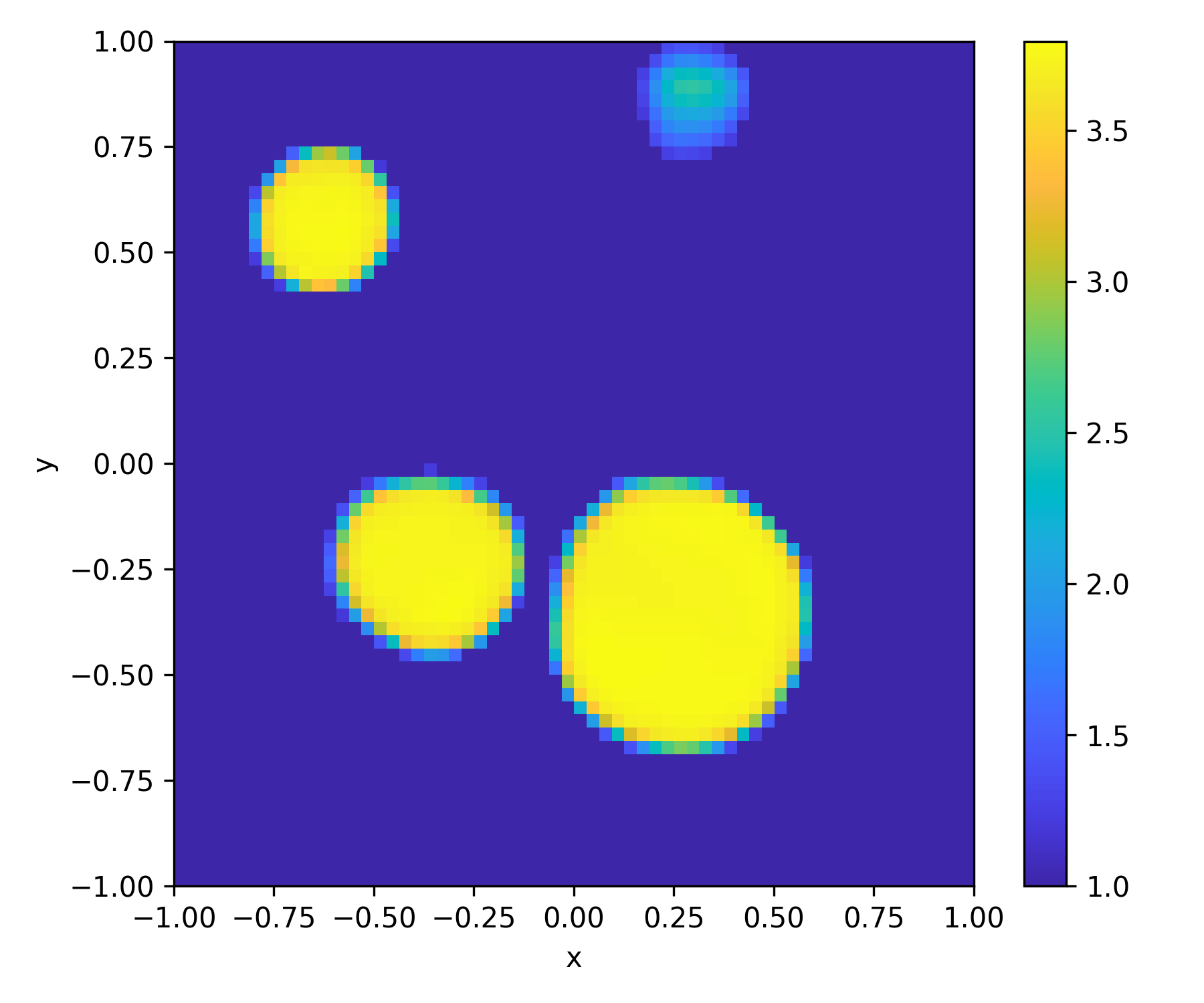}&
			\includegraphics[width=0.15\textwidth]{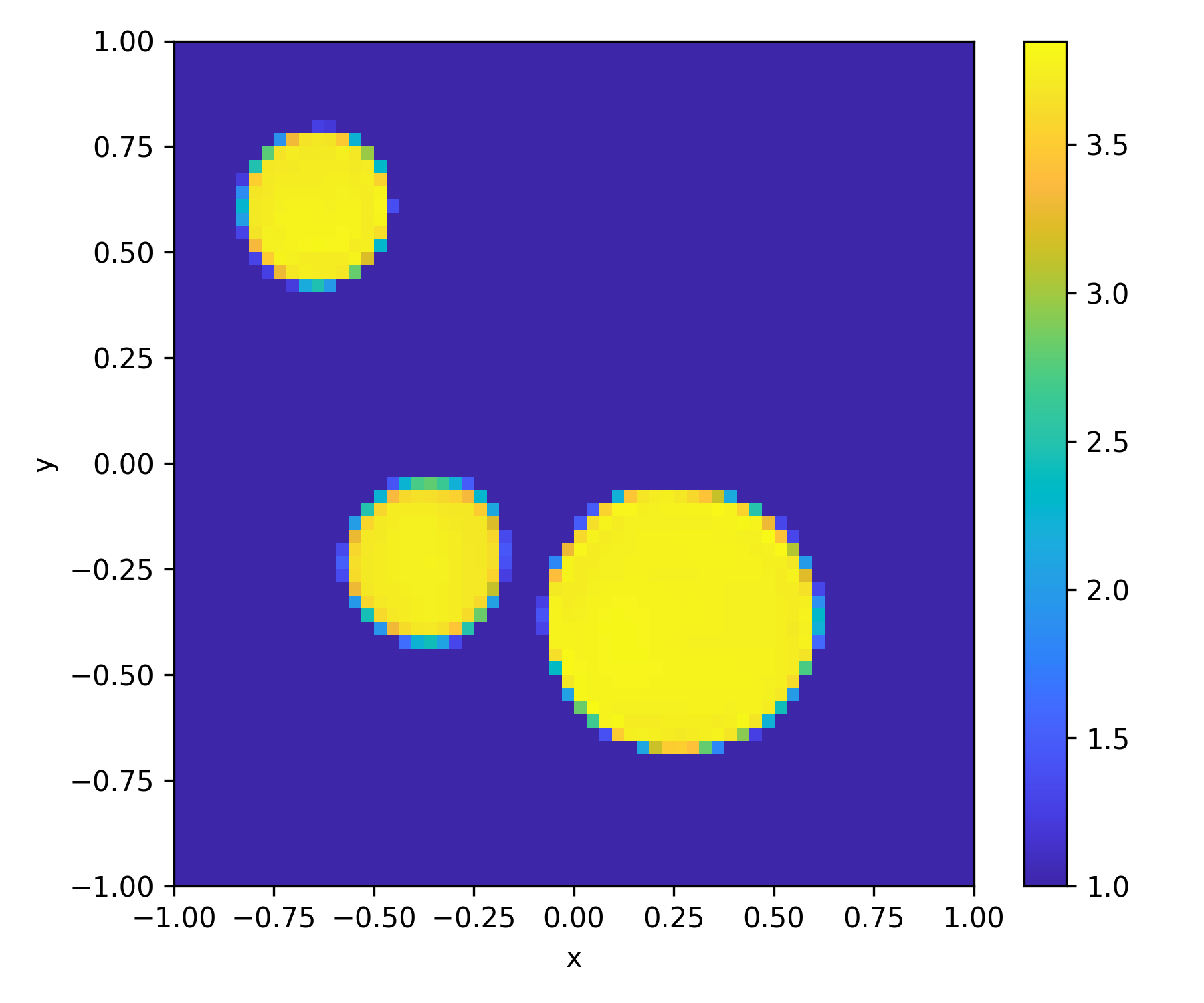}&
			\includegraphics[width=0.15\textwidth]{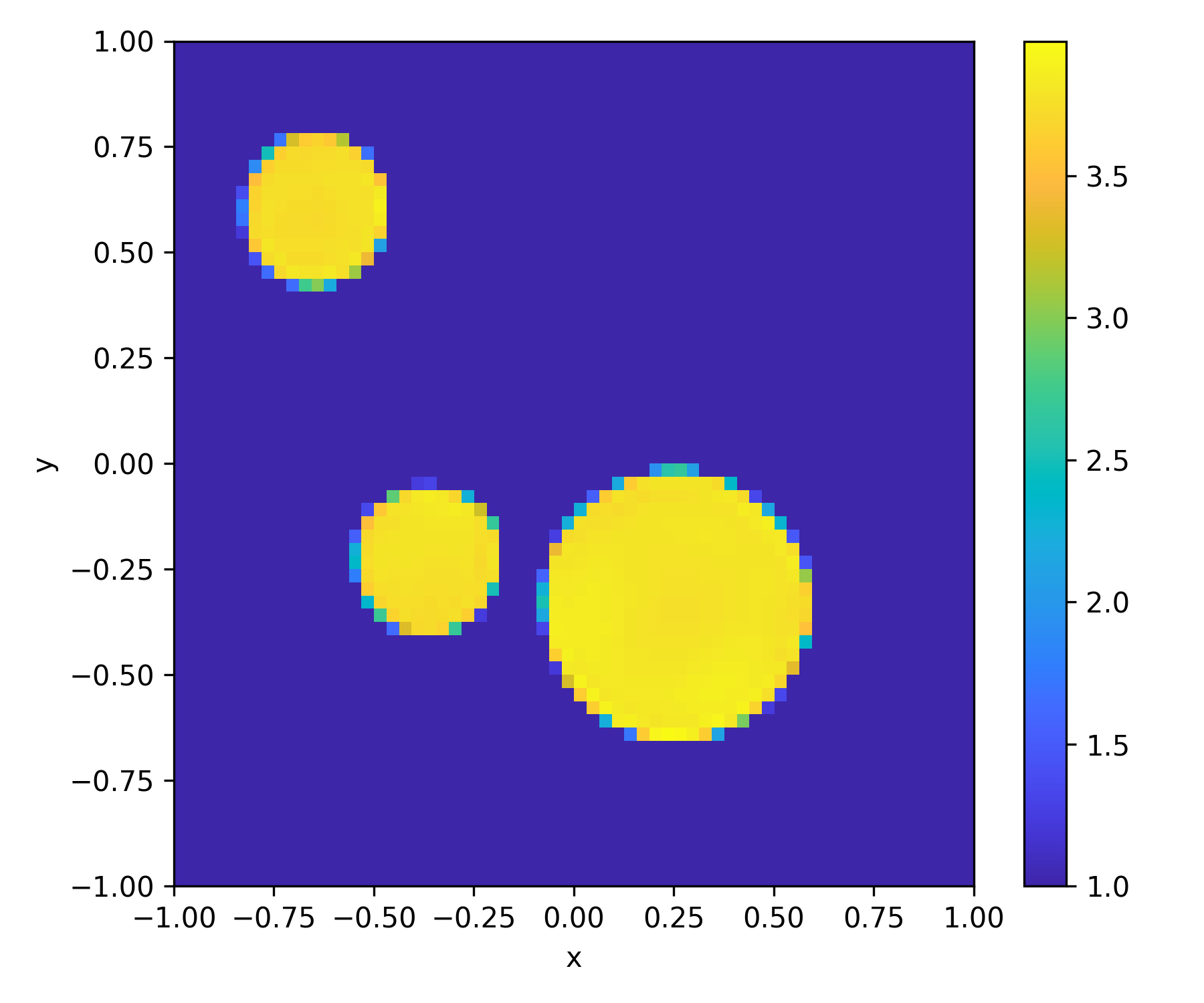}
			\\
			40\%& &
			\includegraphics[width=0.15\textwidth]{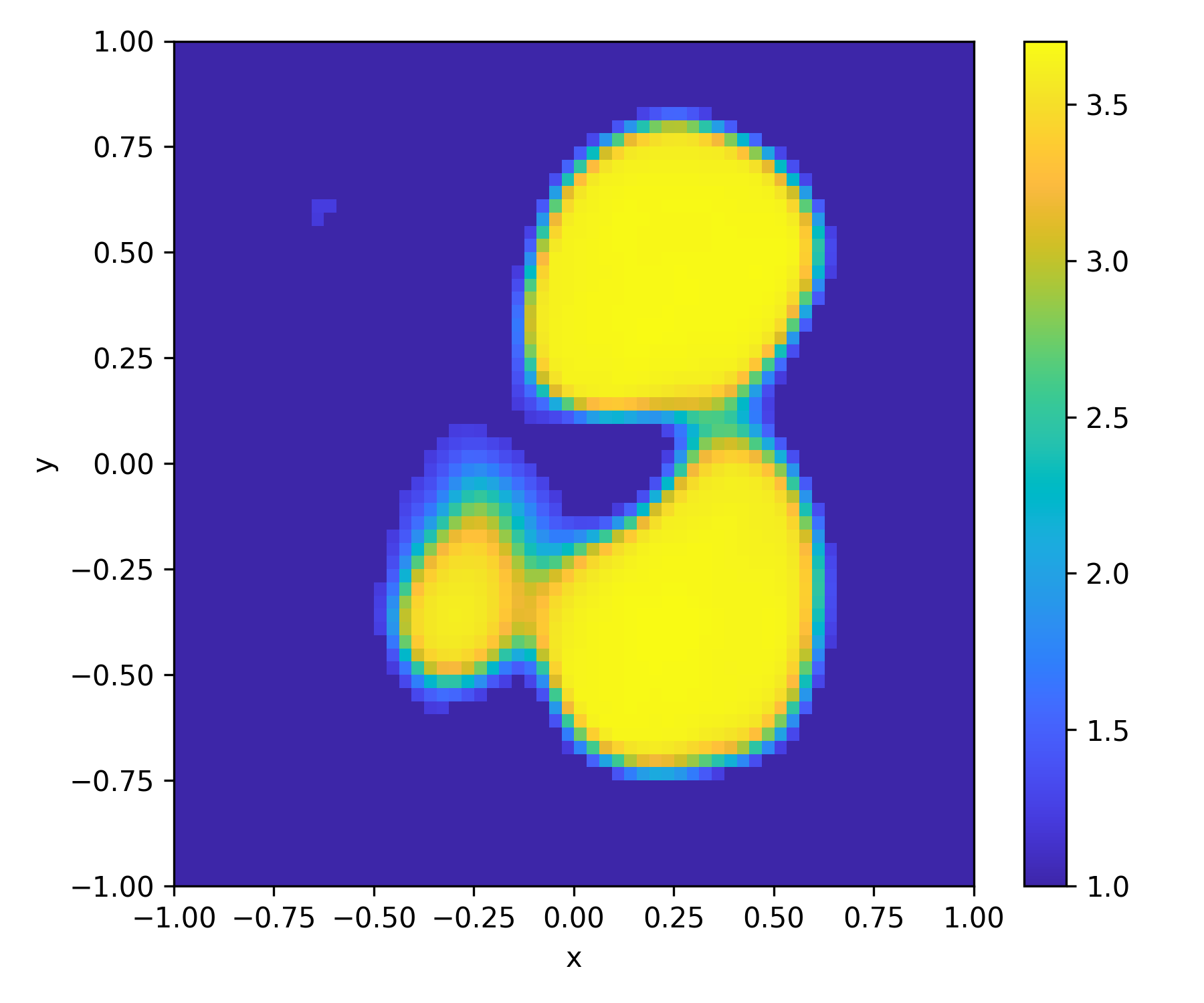}&
			\includegraphics[width=0.15\textwidth]{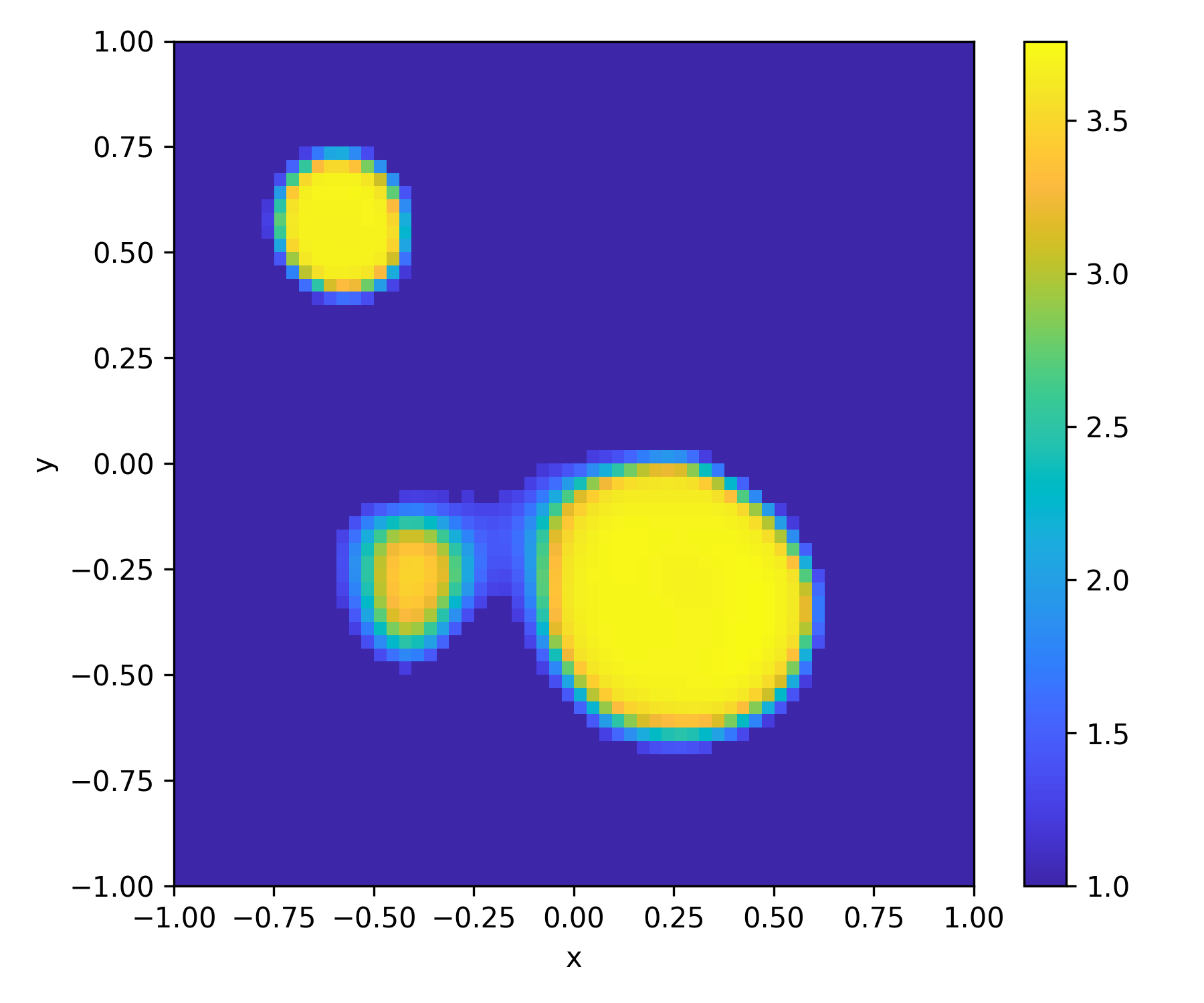}&
			\includegraphics[width=0.15\textwidth]{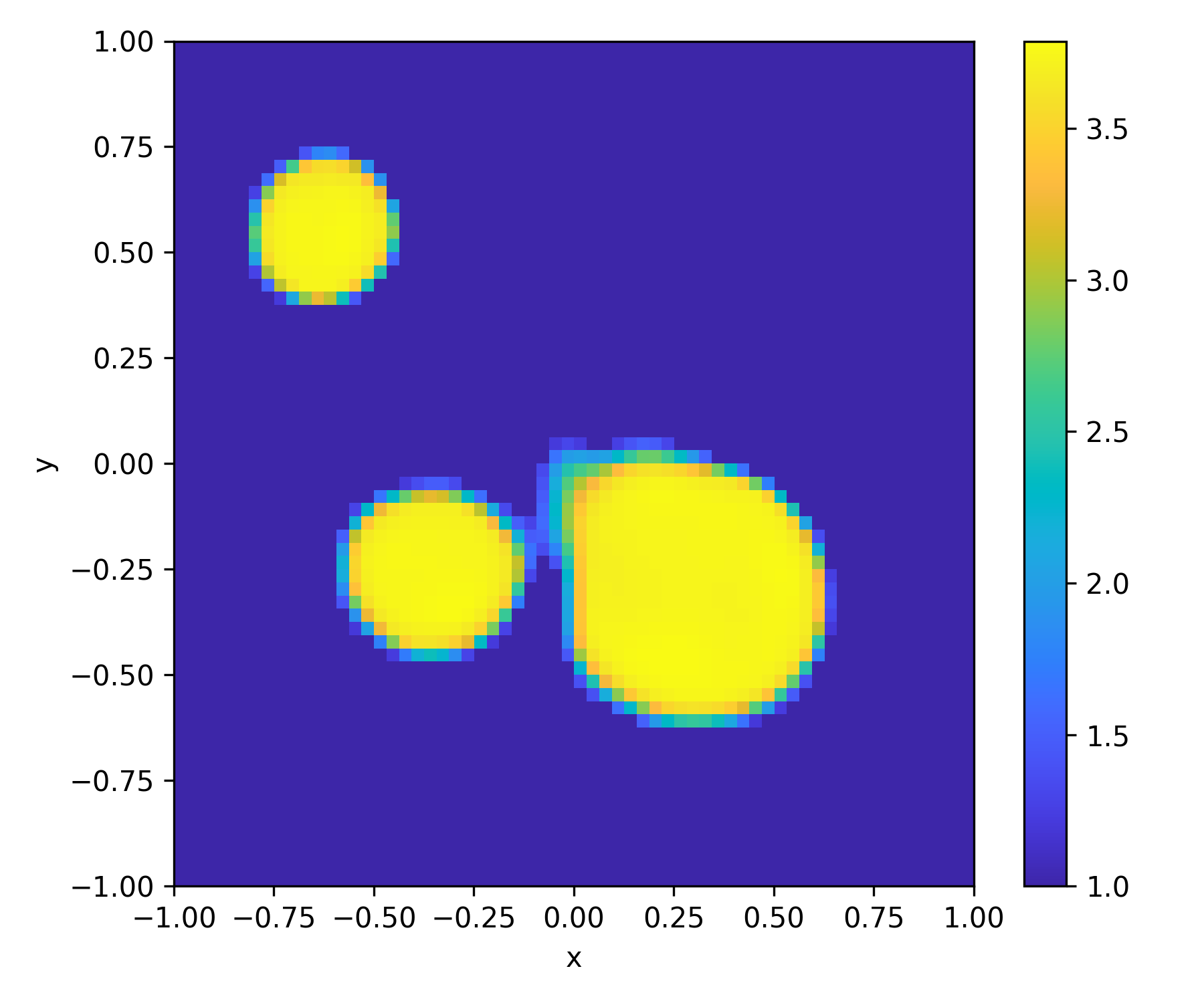}&
			\includegraphics[width=0.15\textwidth]{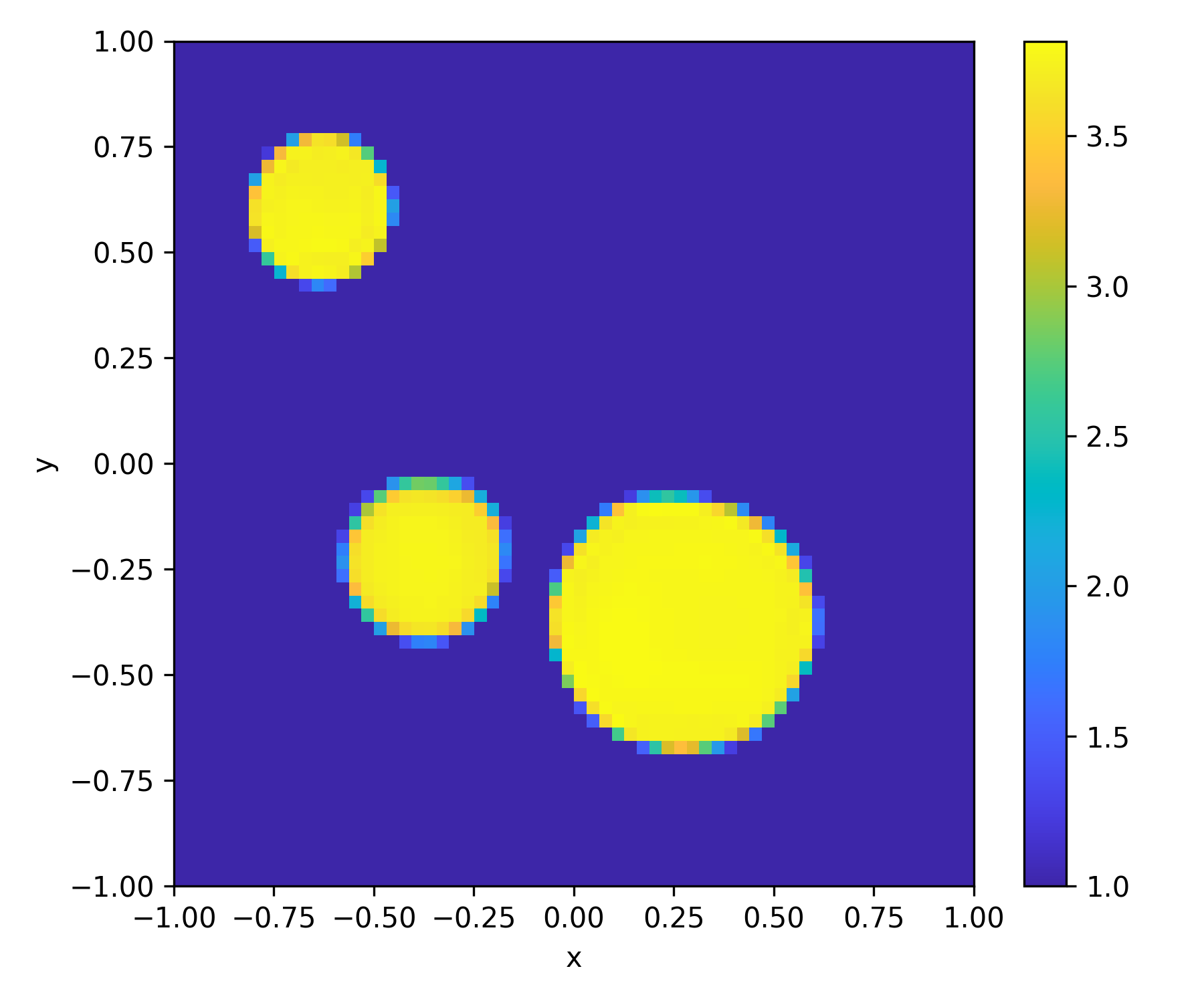}&
			\includegraphics[width=0.15\textwidth]{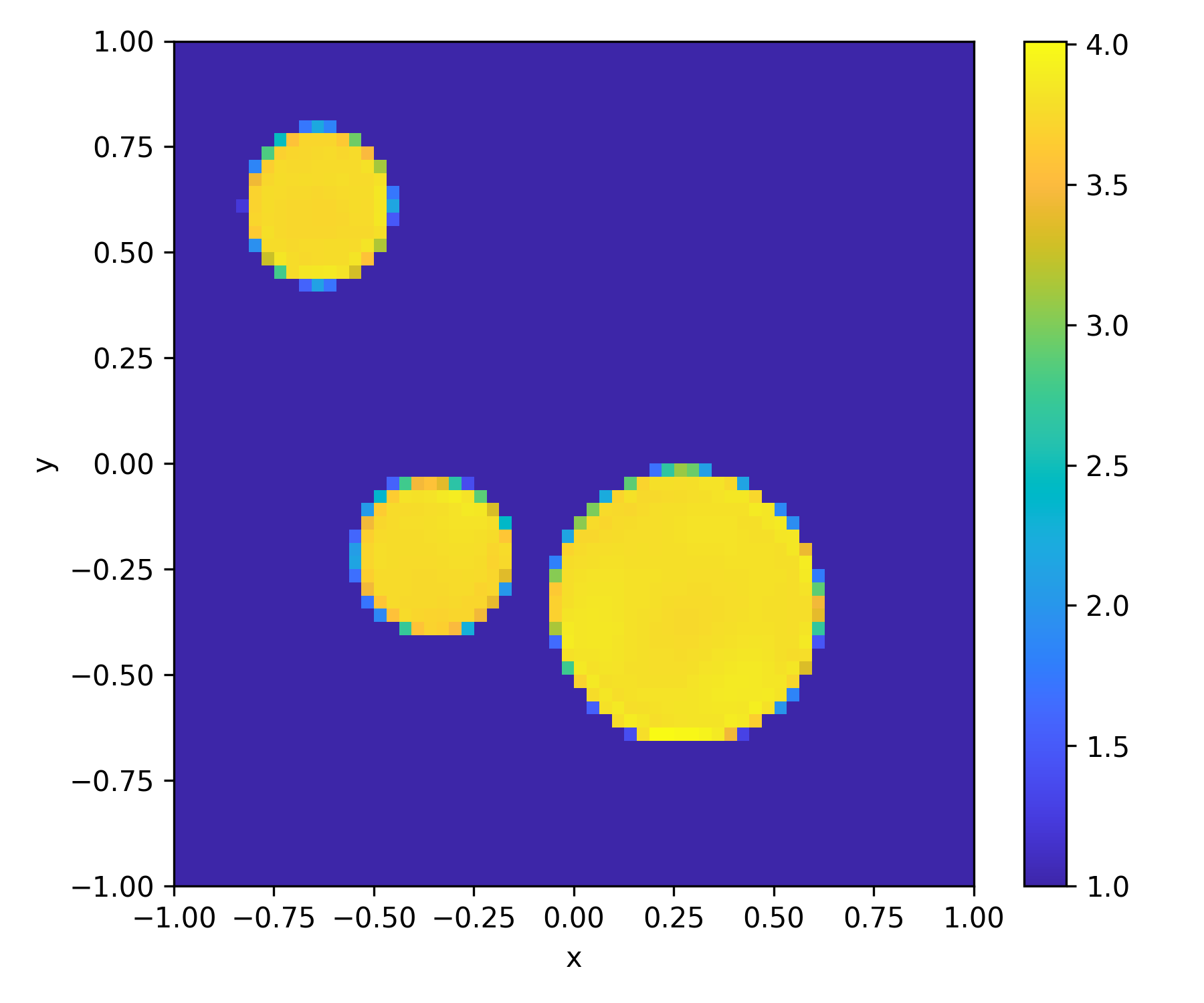}
			\\
			
			15\%&\SetCell[r=2]{c}\includegraphics[width=0.15\textwidth]{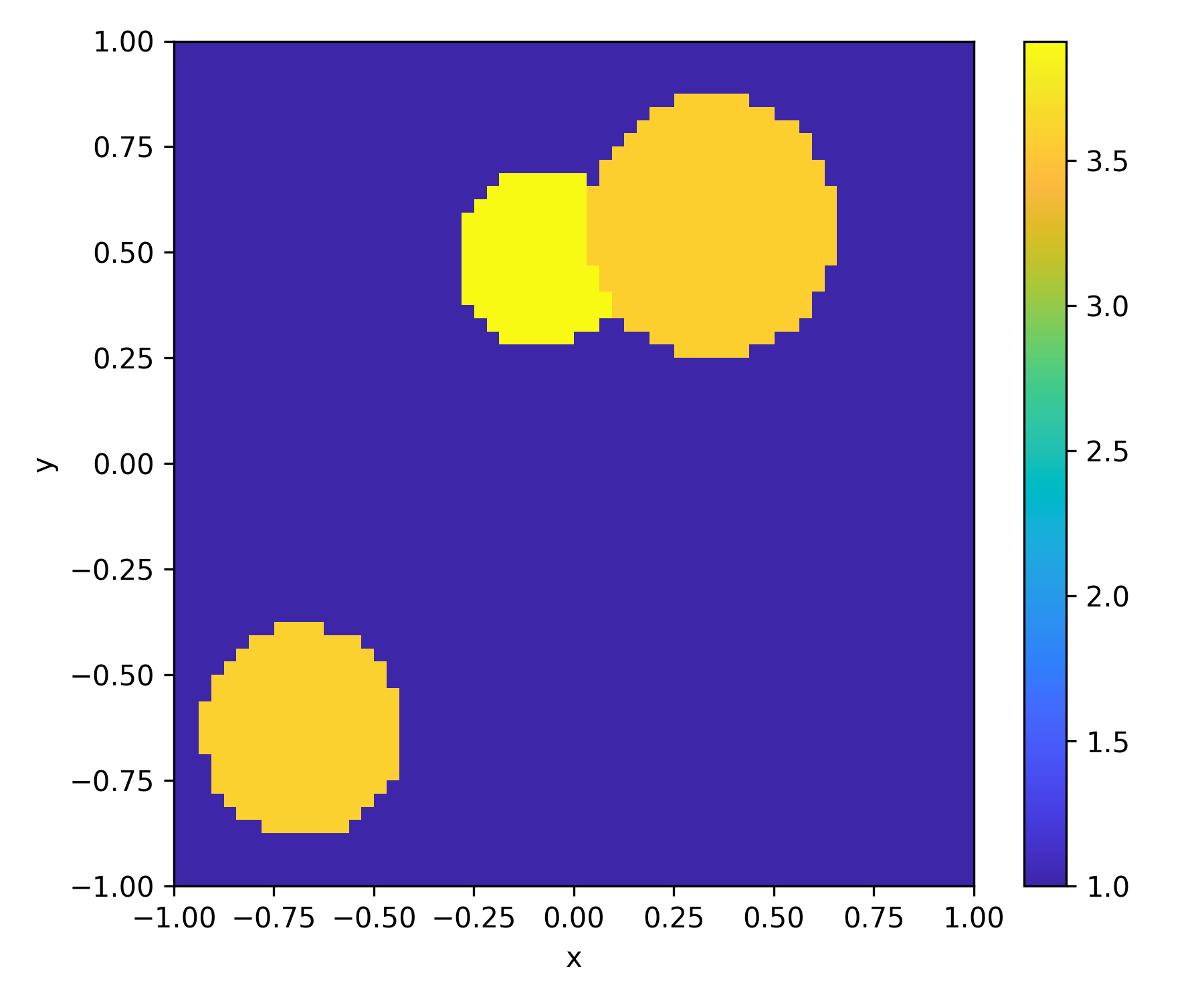}&
			\includegraphics[width=0.15\textwidth]{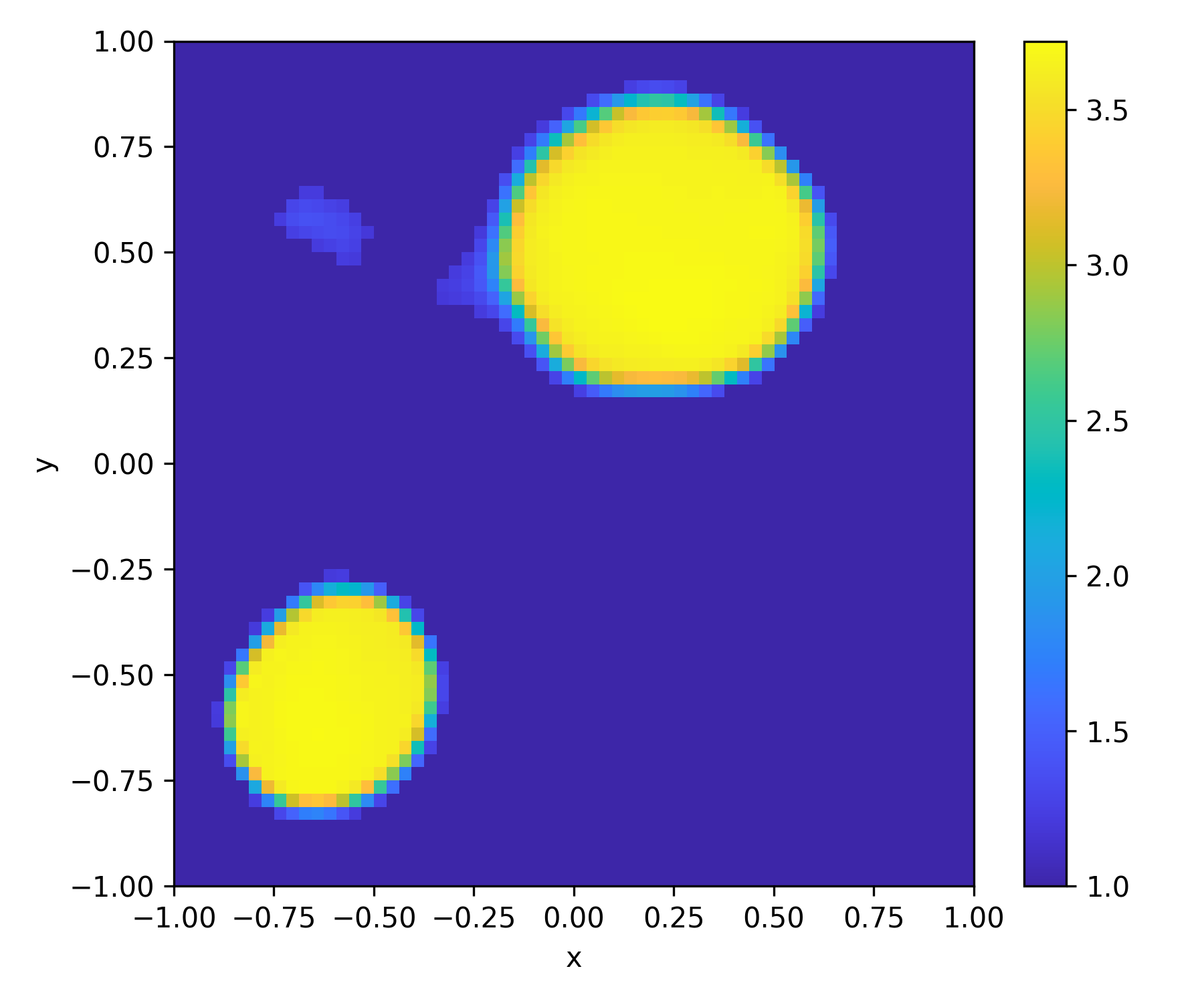}&
			\includegraphics[width=0.15\textwidth]{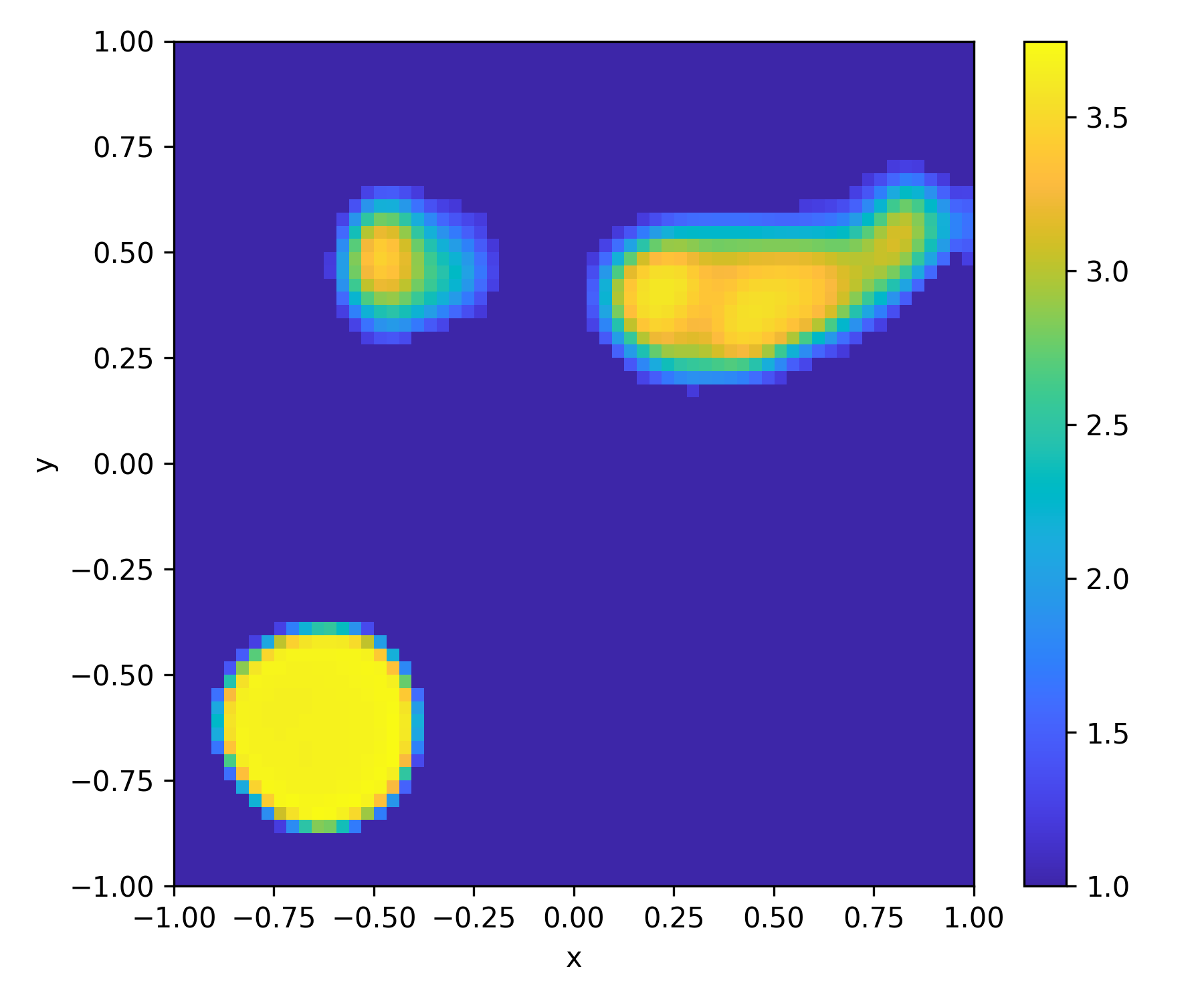}&
			\includegraphics[width=0.15\textwidth]{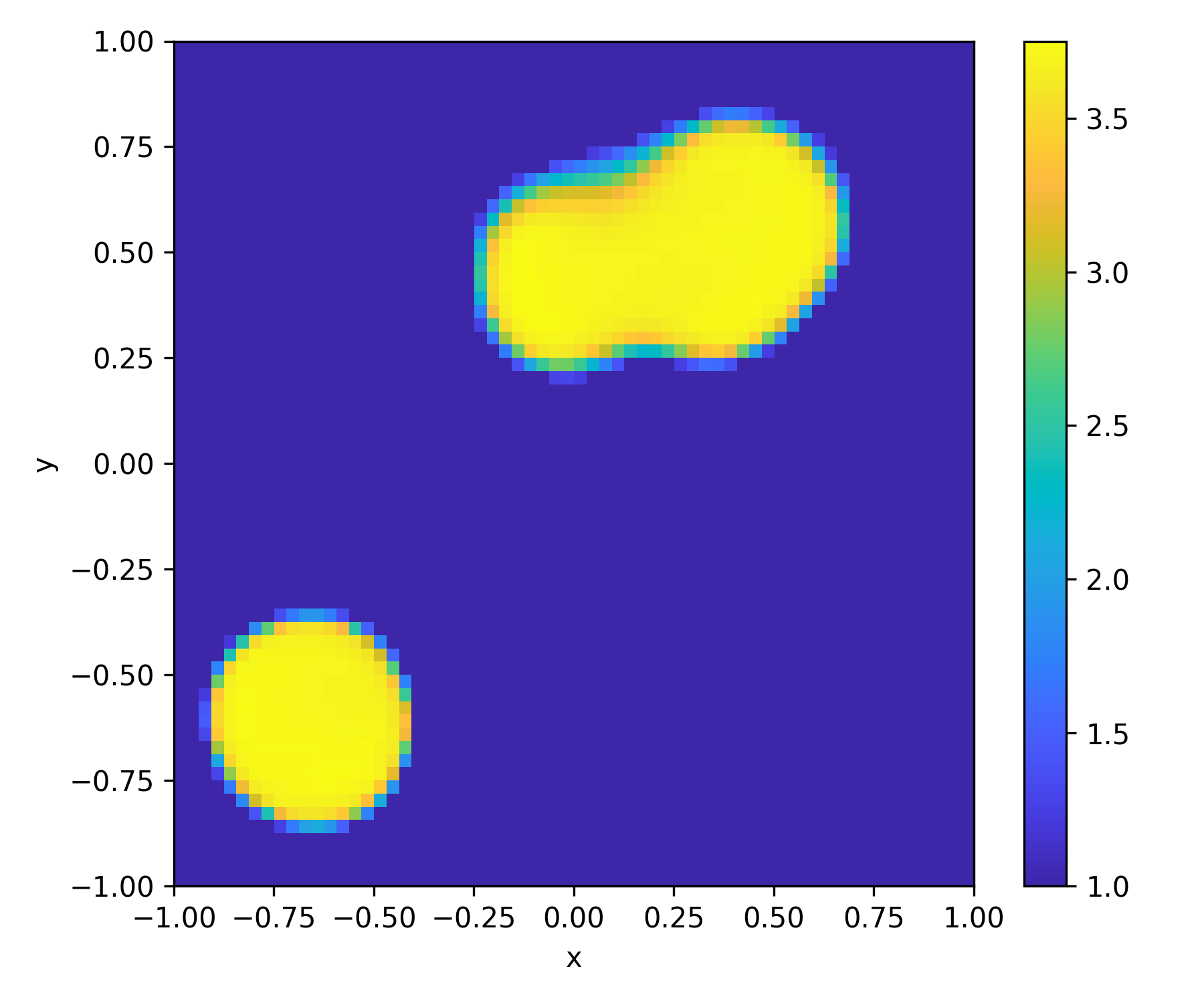}&
			\includegraphics[width=0.15\textwidth]{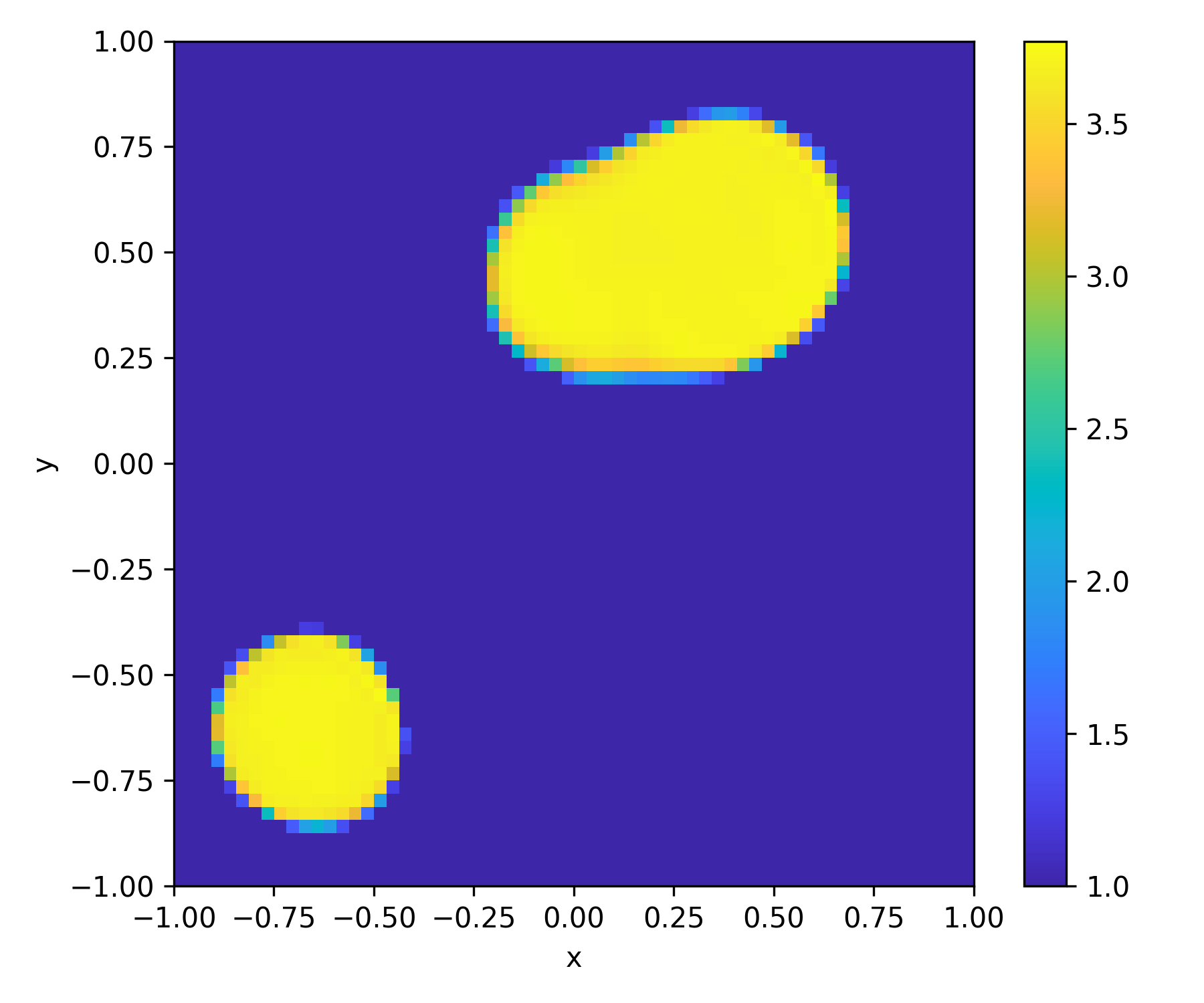}&
			\includegraphics[width=0.15\textwidth]{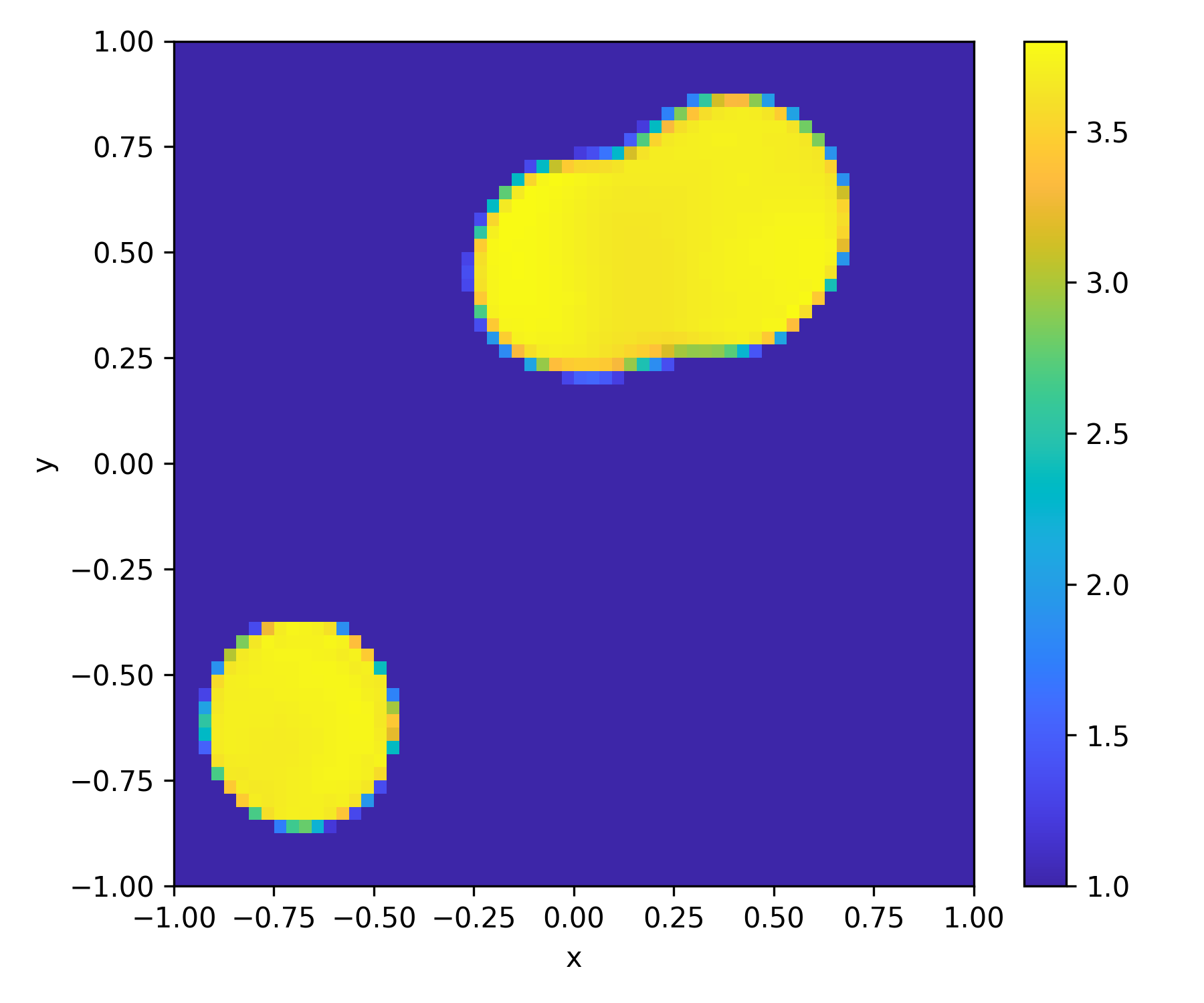}
			\\
			40\%& &
			\includegraphics[width=0.15\textwidth]{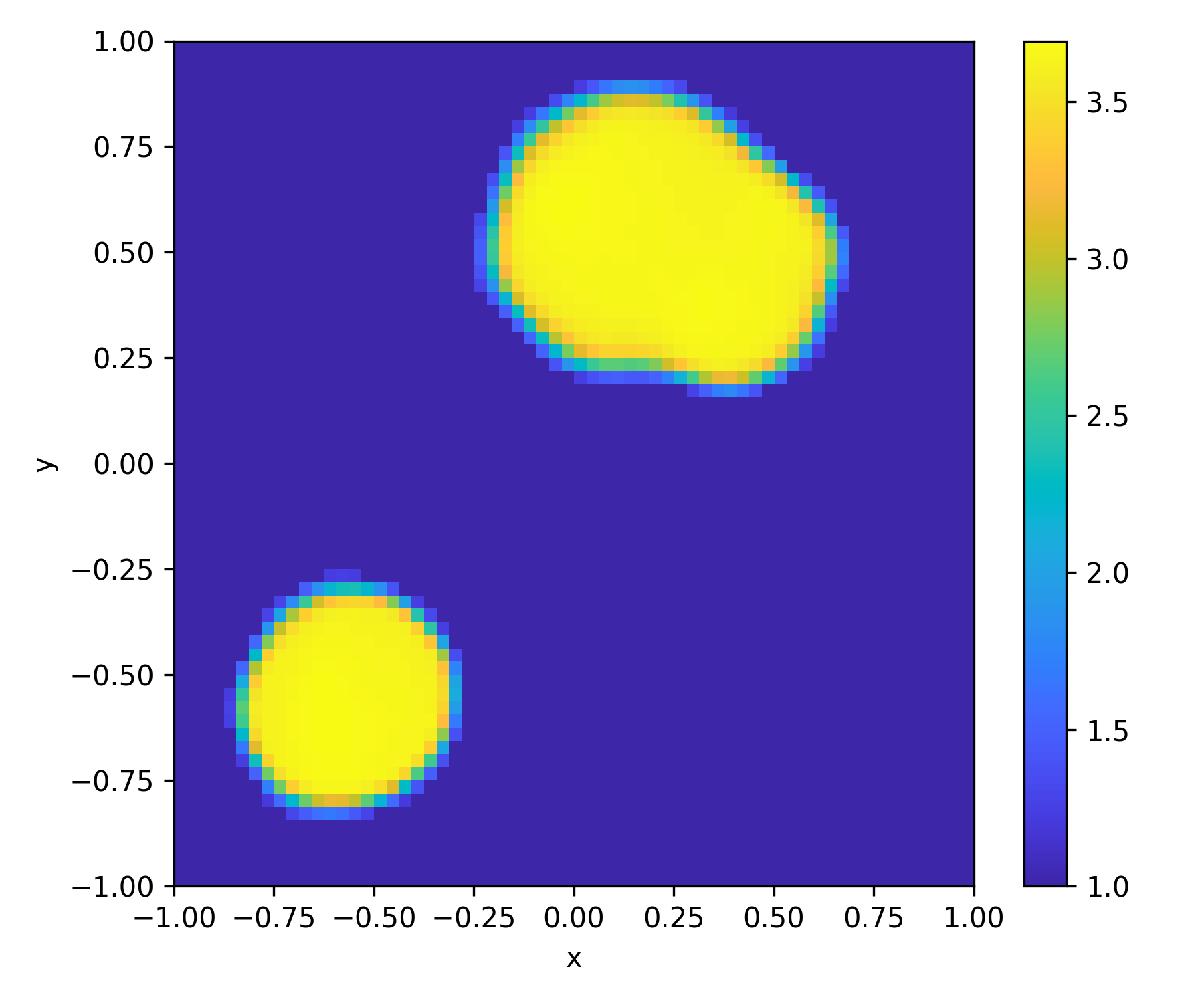}&
			\includegraphics[width=0.15\textwidth]{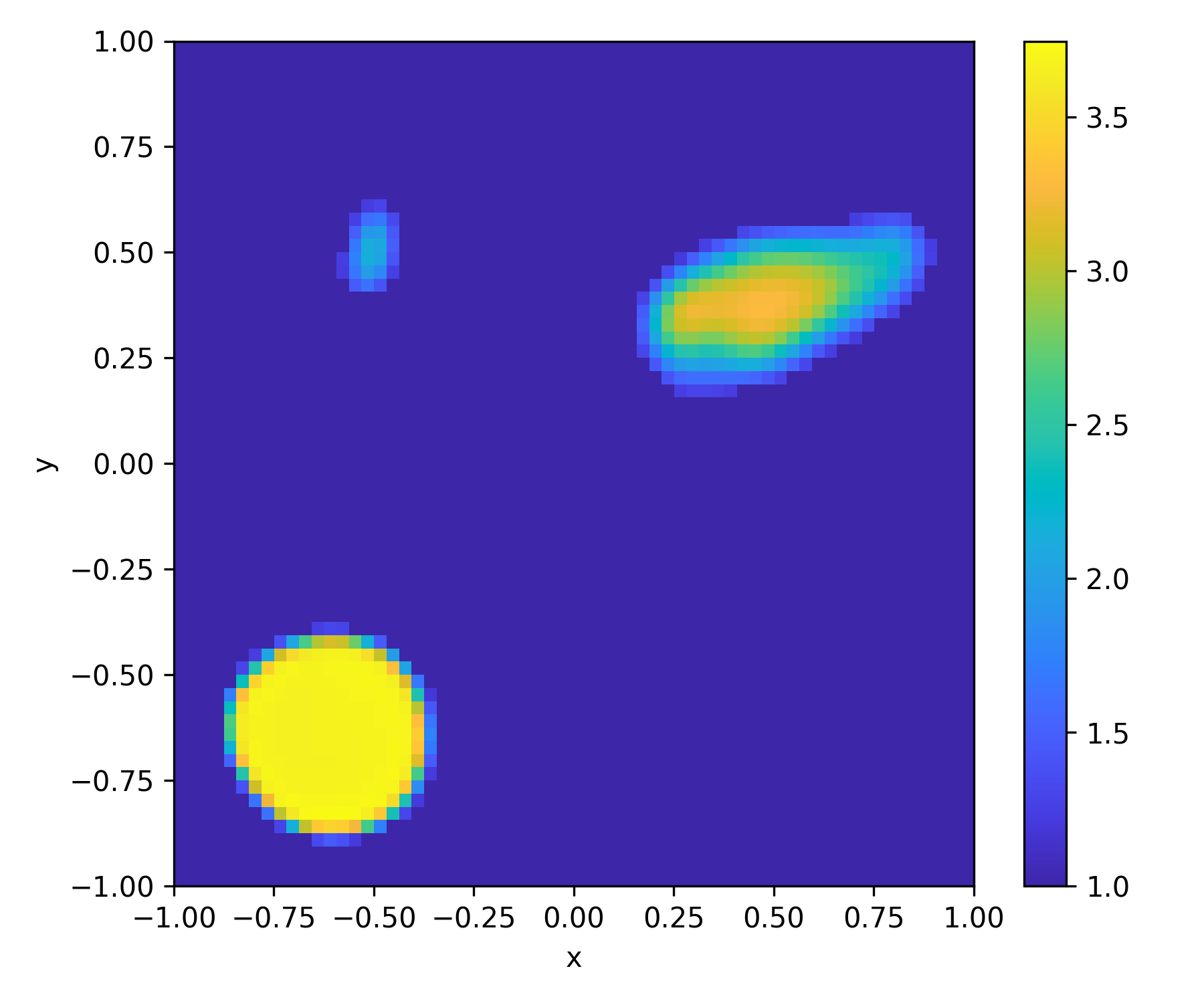}&
			\includegraphics[width=0.15\textwidth]{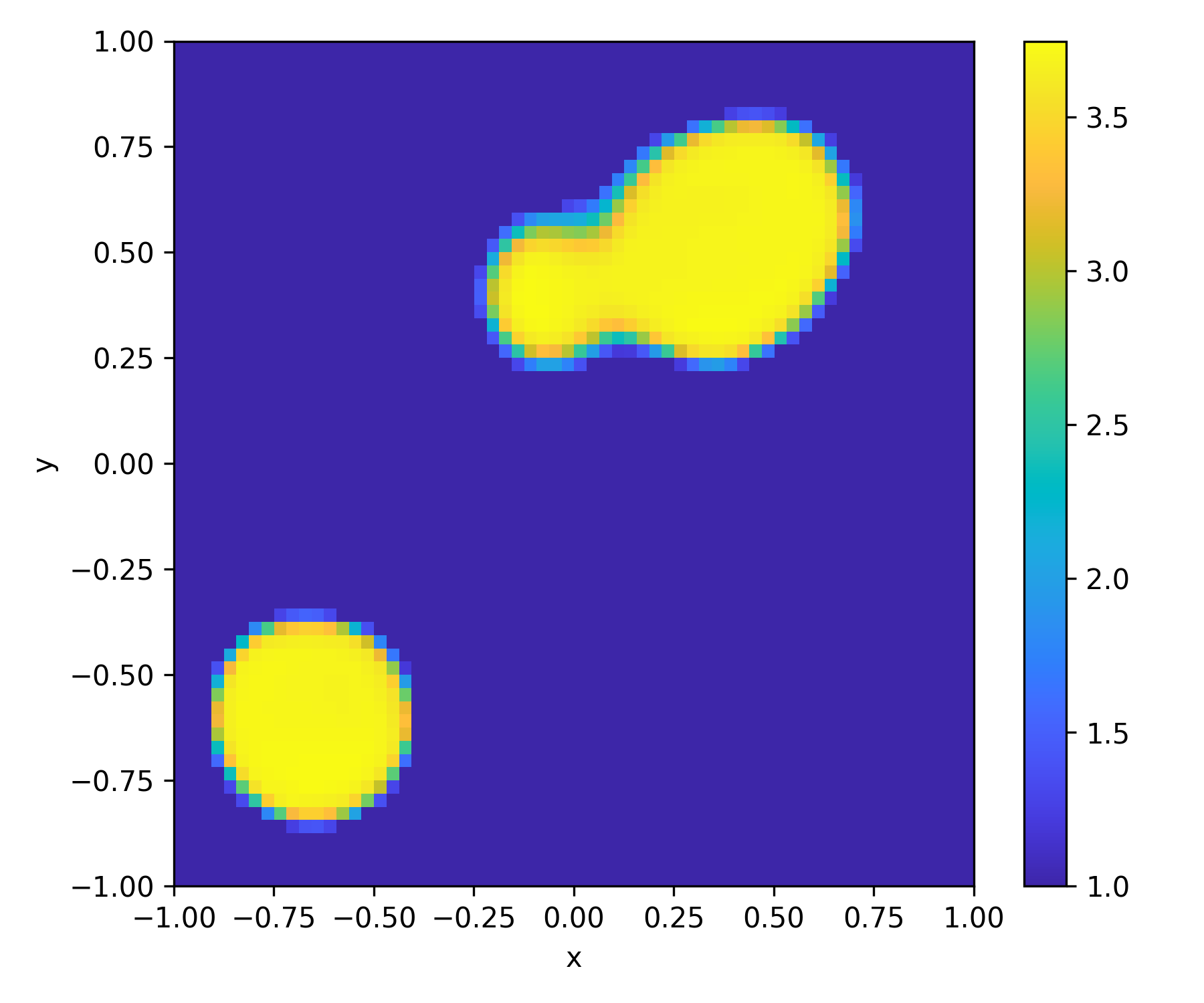}&
			\includegraphics[width=0.15\textwidth]{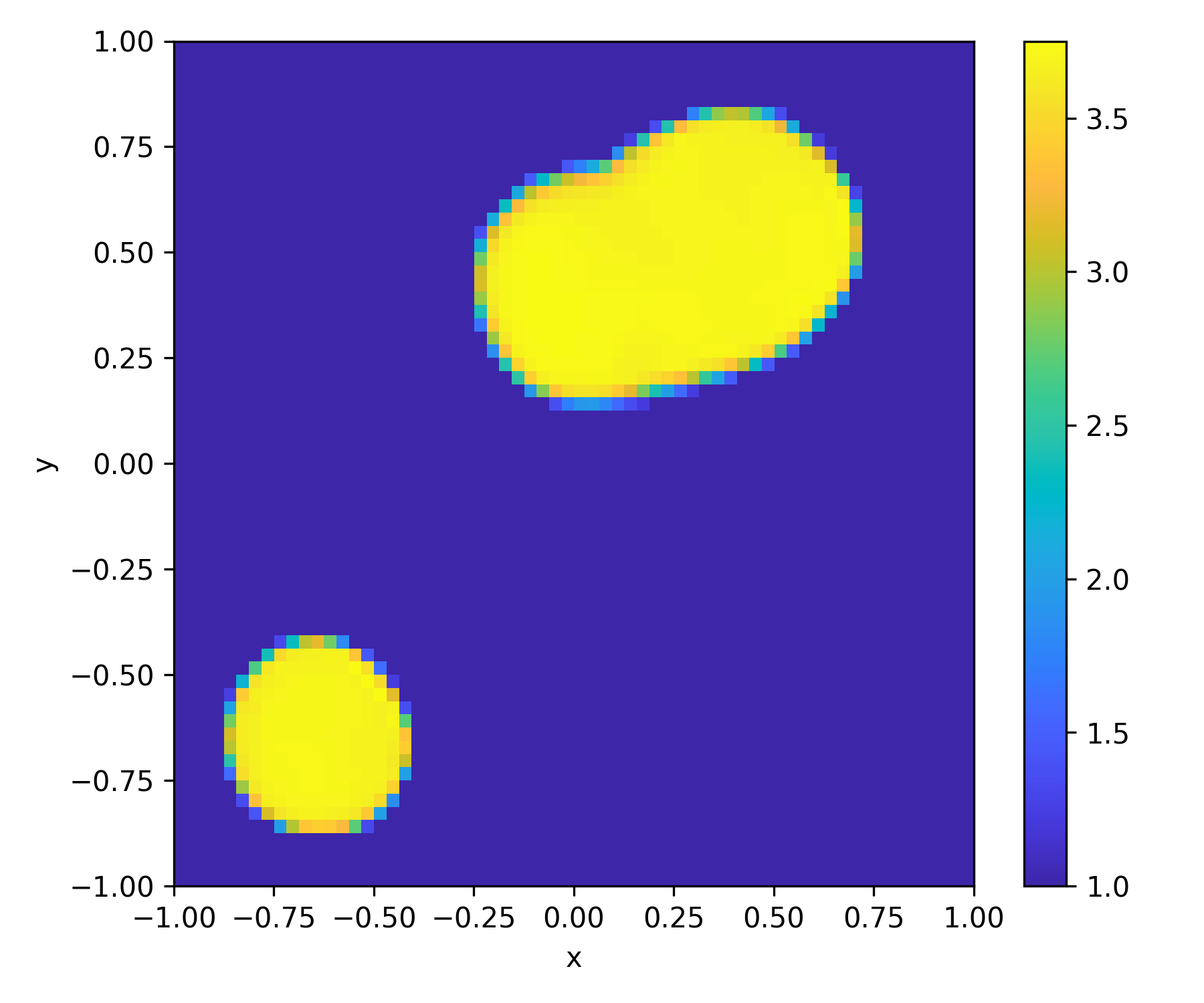}&
			\includegraphics[width=0.15\textwidth]{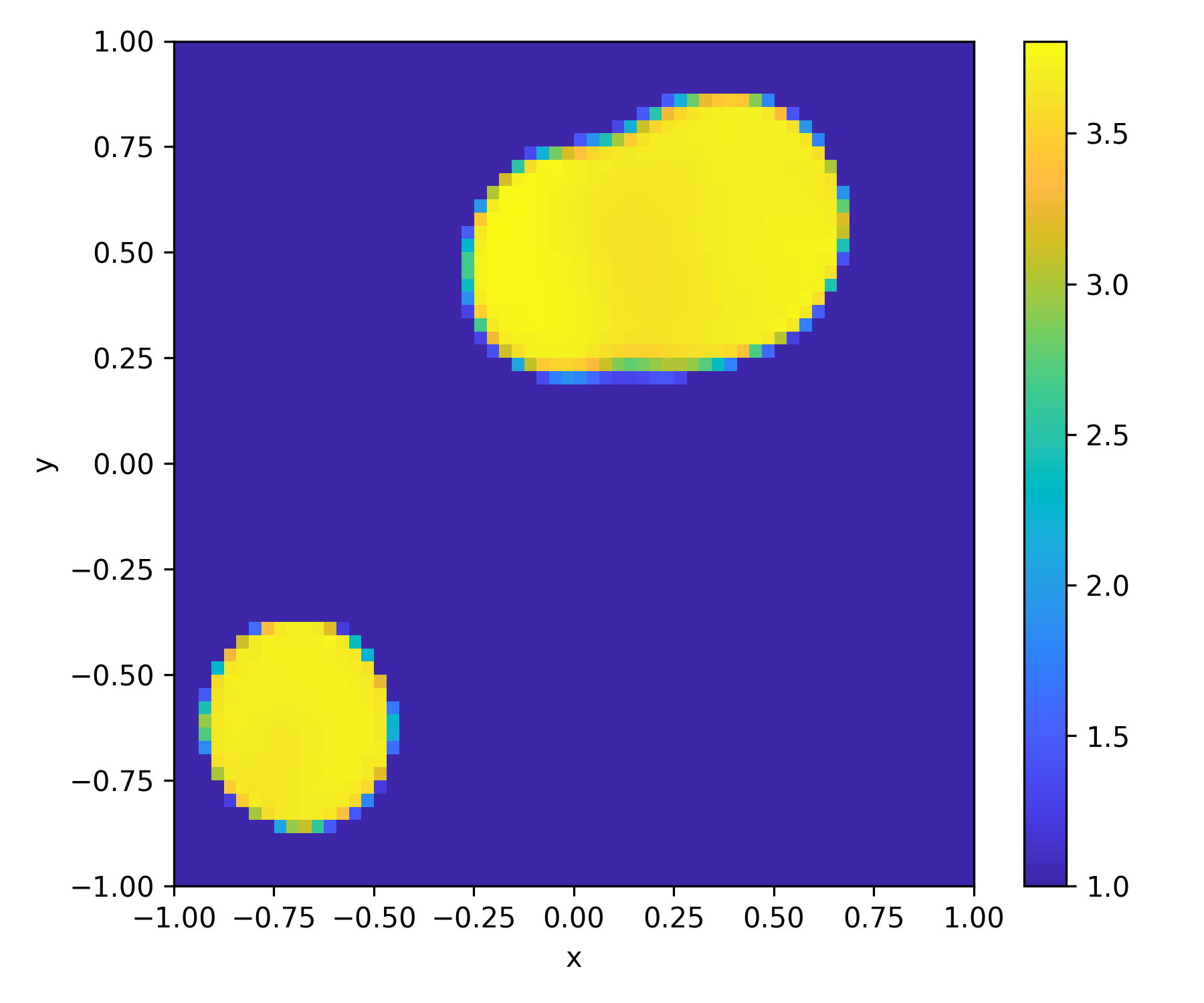}
		\end{tblr}
		\caption{Image reconstructions from measured scattered fields with $15\%$ and $40\%$ Gaussian noises by using the networks trained by the high contrast circle dataset, where the relative permittivity is between 3.5 and 4.0. From left to right: the ground-truth images, the reconstruction with 1,2,4,8, and 16 incident fields.}
		\label{tab:fig-Circle_Strong}
	\end{center}
\end{figure}
	\begin{table}[htp]\small
	\begin{center}
		\begin{tabular}{ |c|c|c|c|c|c|c|c|}
			\hline
			Example&Metric&Noise Level& $N_{i}=1$& $N_{i}=2$ & $N_{i}=4$ & $N_{i}=8$ & $N_{i}=16$ \\
			\hline
			\multirow{4}*{Circle}&\multirow{2}*{Relative L2 error}&15\% &0.0713& 0.0581 & 0.0505 & 0.0422&0.0327 \\
			\cline{3-8}
			& &40\% &0.0960& 0.0768 & 0.0619 & 0.0524&0.0443 \\
			\cline{2-8}
			&\multirow{2}*{SSIM}&15\% &0.9123& 0.9352 & 0.9496 & 0.9632&0.9771 \\
			\cline{3-8}
			& &40\% &0.8645& 0.9009 & 0.9305 & 0.9474&0.9608\\
			\hline
			
			\multirow{4}*{High Contrast Circle}&\multirow{2}*{Relative L2 error}&15\% &0.2462& 0.2192 & 0.1842 & 0.1575&0.1322 \\
			\cline{3-8}
			& &40\% &0.2941& 0.2309 & 0.1981 & 0.1701&0.1439 \\
			\cline{2-8}
			&\multirow{2}*{SSIM}&15\% &0.8388& 0.8702 & 0.8951 & 0.9161&0.9366 \\
			\cline{3-8}
			& &40\% &0.8024& 0.8566 & 0.8846 & 0.9075&0.9260\\
			\hline
		\end{tabular}
		\caption{The relative L2 testing error and SSIM for circle and high contrast circle examples with different noise levels and number of incidences.}
		\label{tab:error_circle}
	\end{center}
\end{table}
	
	\subsection{MNIST database example}
	In this example, the unknown inhomogeneous medium is modeled by a modification of the well-known MNIST dataset. In the MNIST dataset, the resolution of each image is $28\times 28$ and the pixel values range from $0$ to $1$. We rescale the size to be $64\times 64$ and apply a threshold with value $0.5$ so that the pixel values only take 0 and 1. To represent multiple scatterers and enhance the diversity of the data, we randomly rotate the digits and also randomly add a circle with a radius from $U(0.1,0.3)$ to the domain of interest $\Omega$. The relative permittivities of the digits and the circles are randomly chosen from $U(1.5,2.5)$. In the examples of the MNIST database, which is more complicated than the circle dataset, we use $10000$ images as the training data, $200$ images as the validation data, and $200$ images as the testing data. The batch size is set as $10$ and we use a total of $20$ epochs to train the models. As the loss function (\ref*{loss_func}) is applied, we actually have $20000$ training data and the batch size is $20$. The learning rate starts at $0.001$ and decreases by a factor of $0.5$ every $2$ epoch.

	\subsubsection{Tests with testing data}
	In Fig.\,\ref{tab:fig-Mnist}, reconstructions of four images from the MNIST dataset are presented. When $N_{i}$ is small, it is difficult for the trained networks to recover the physical properties, and noises have a large impact on the reconstruction, but the DSM-DL can still provide some reasonable reconstruction for the shapes and locations of the scatterers.  The error and SSIM for different noises and incidences are presented in Table\,\ref{tab:error} and Fig.\,\ref{fig:Error_Mnist}, we can also see significant improvement as $N_{i}$ becomes larger.
		\begin{figure}[htp]\small
		\begin{center}
			\begin{tblr}
				{colspec = {X[-1]X[c]X[c,h]X[c,h]X[c,h]X[c,h]X[c,h]},
					stretch = 0,
					rowsep = 0pt,}
				Noise Level& Ground truth& $N_{i}$=1& $N_{i}$=2 &$N_{i}$=4&$N_{i}$=8& $N_{i}$=16\\
				15\%&\SetCell[r=2]{c} \includegraphics[width=0.15\textwidth]{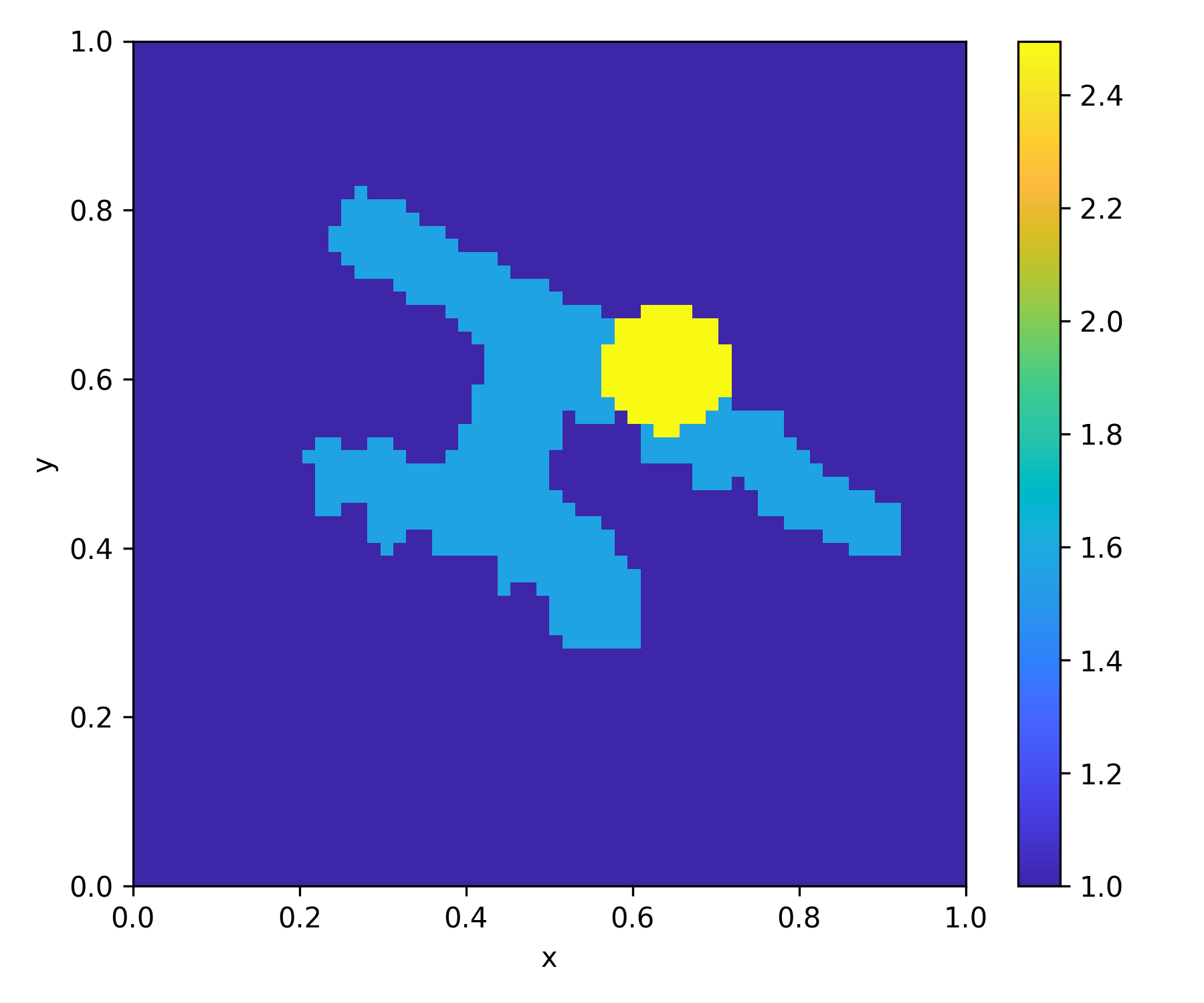} &
				\includegraphics[width=0.15\textwidth]{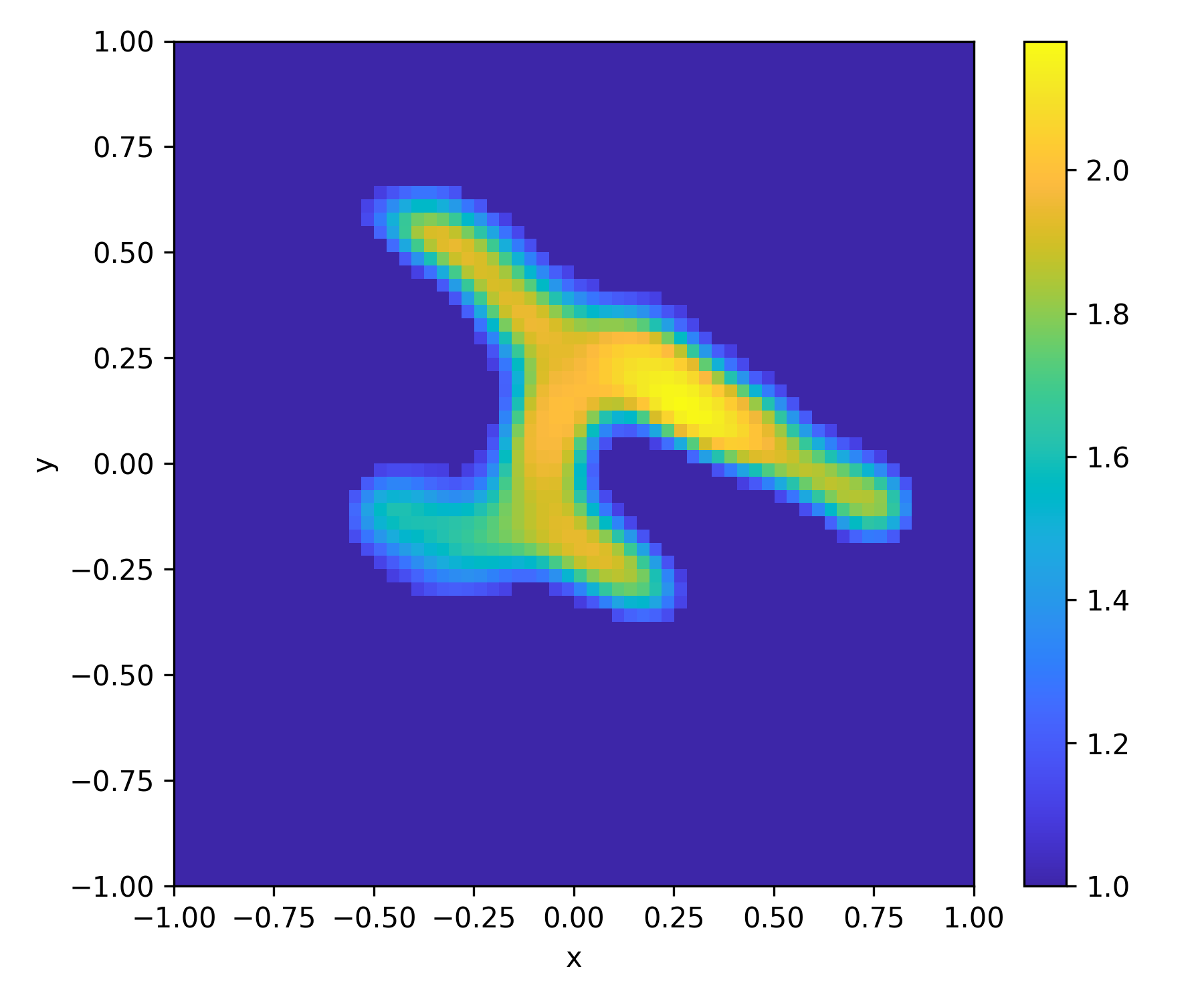}&
				\includegraphics[width=0.15\textwidth]{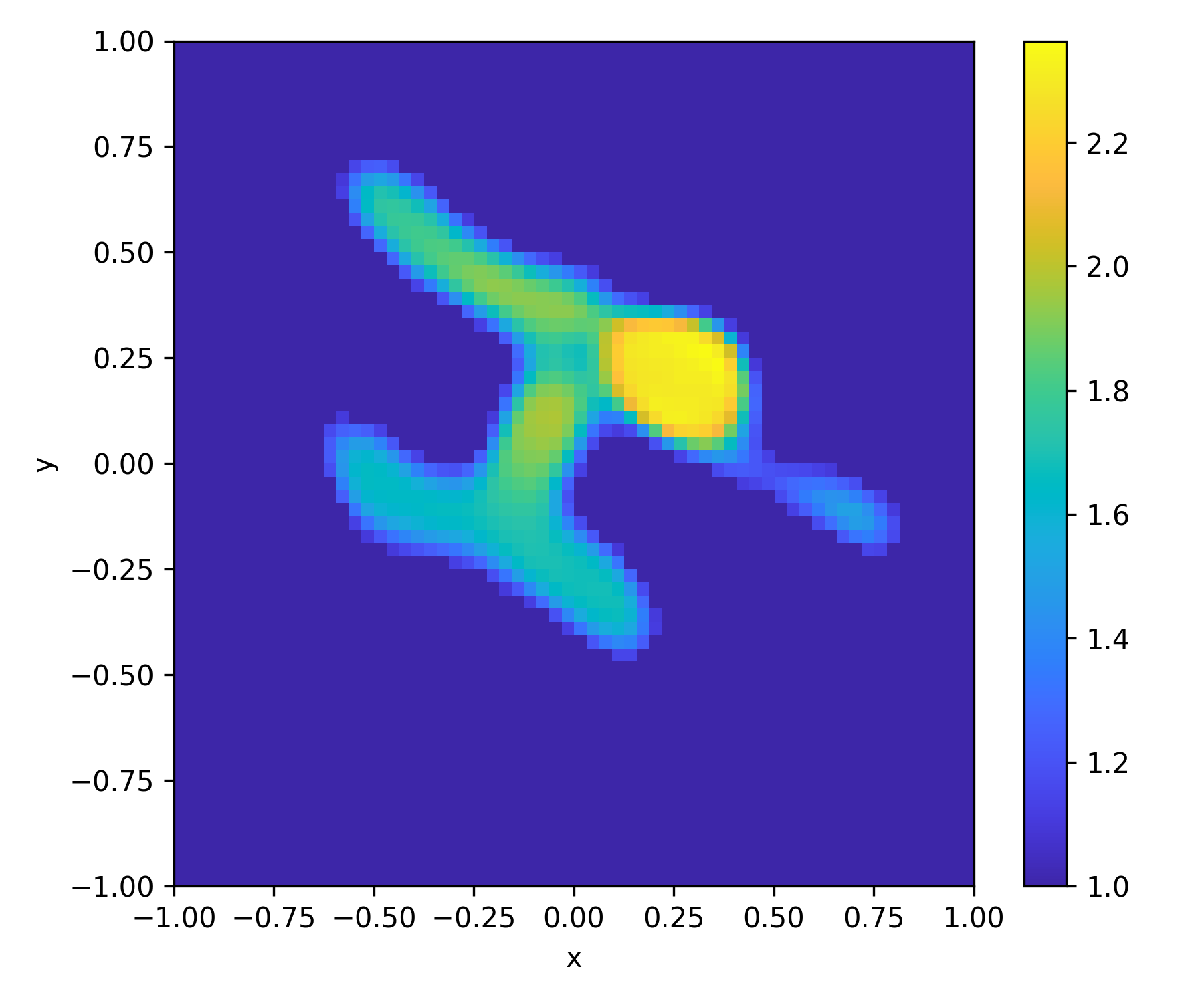}&
				\includegraphics[width=0.15\textwidth]{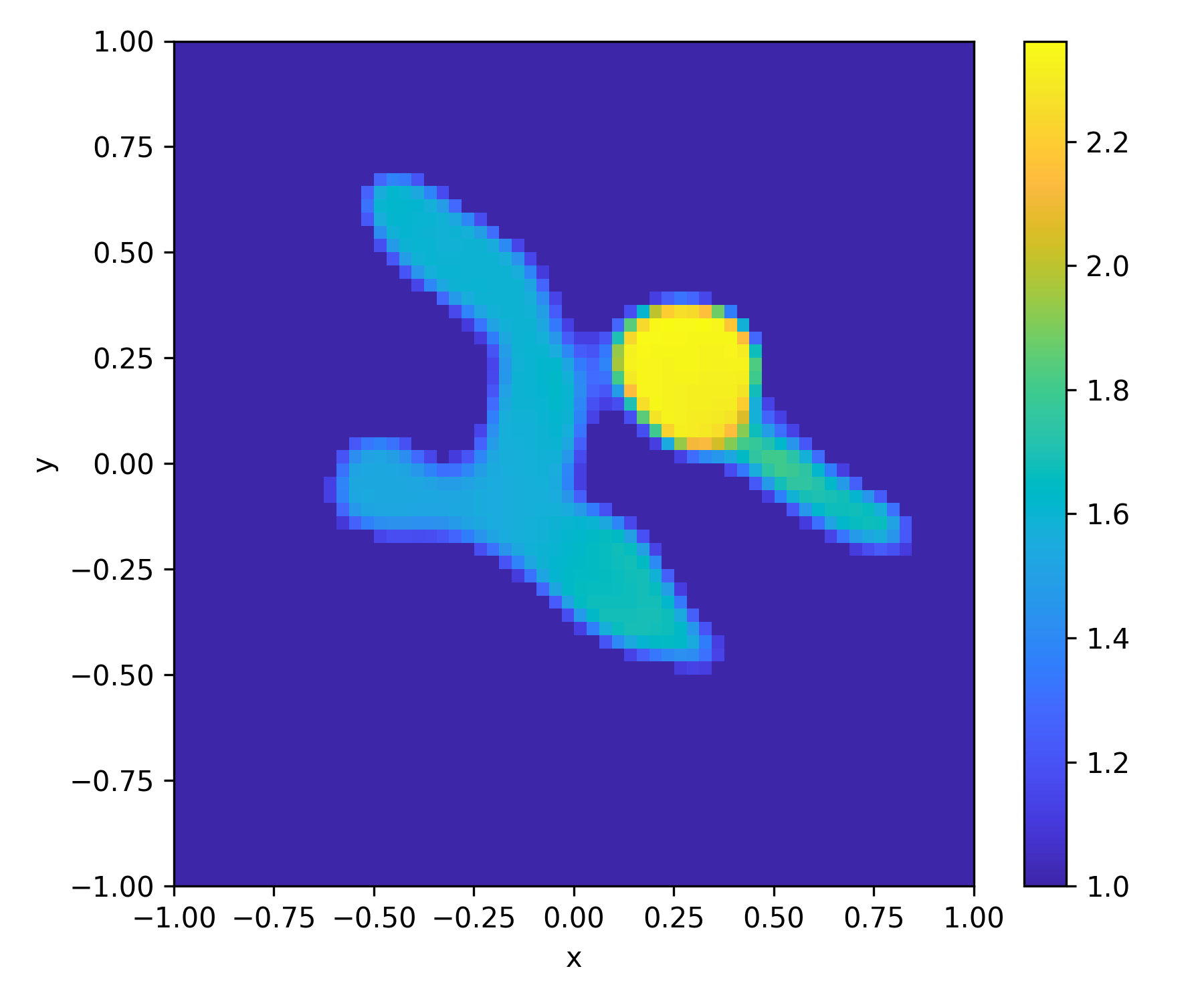}&
				\includegraphics[width=0.15\textwidth]{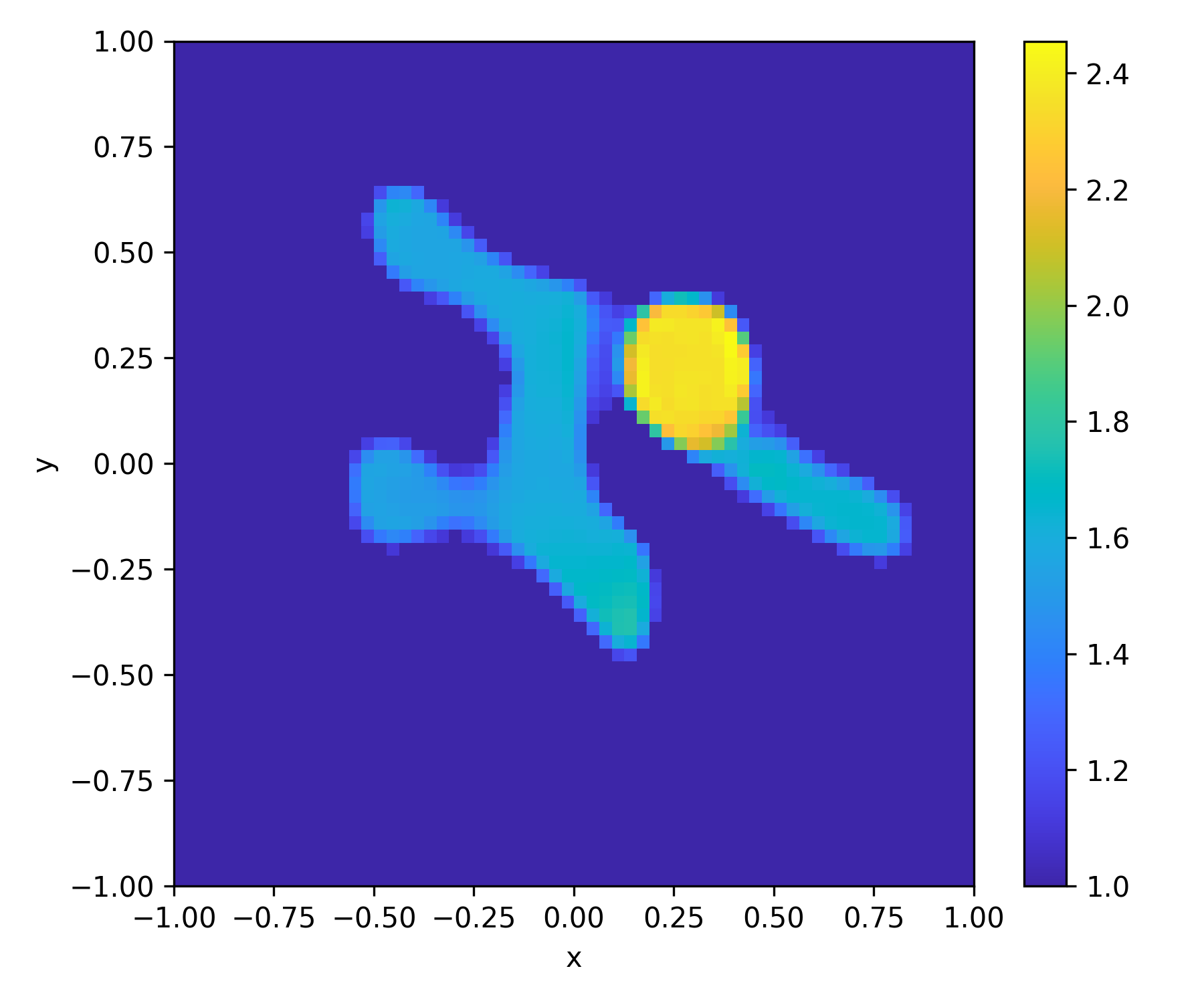}&
				\includegraphics[width=0.15\textwidth]{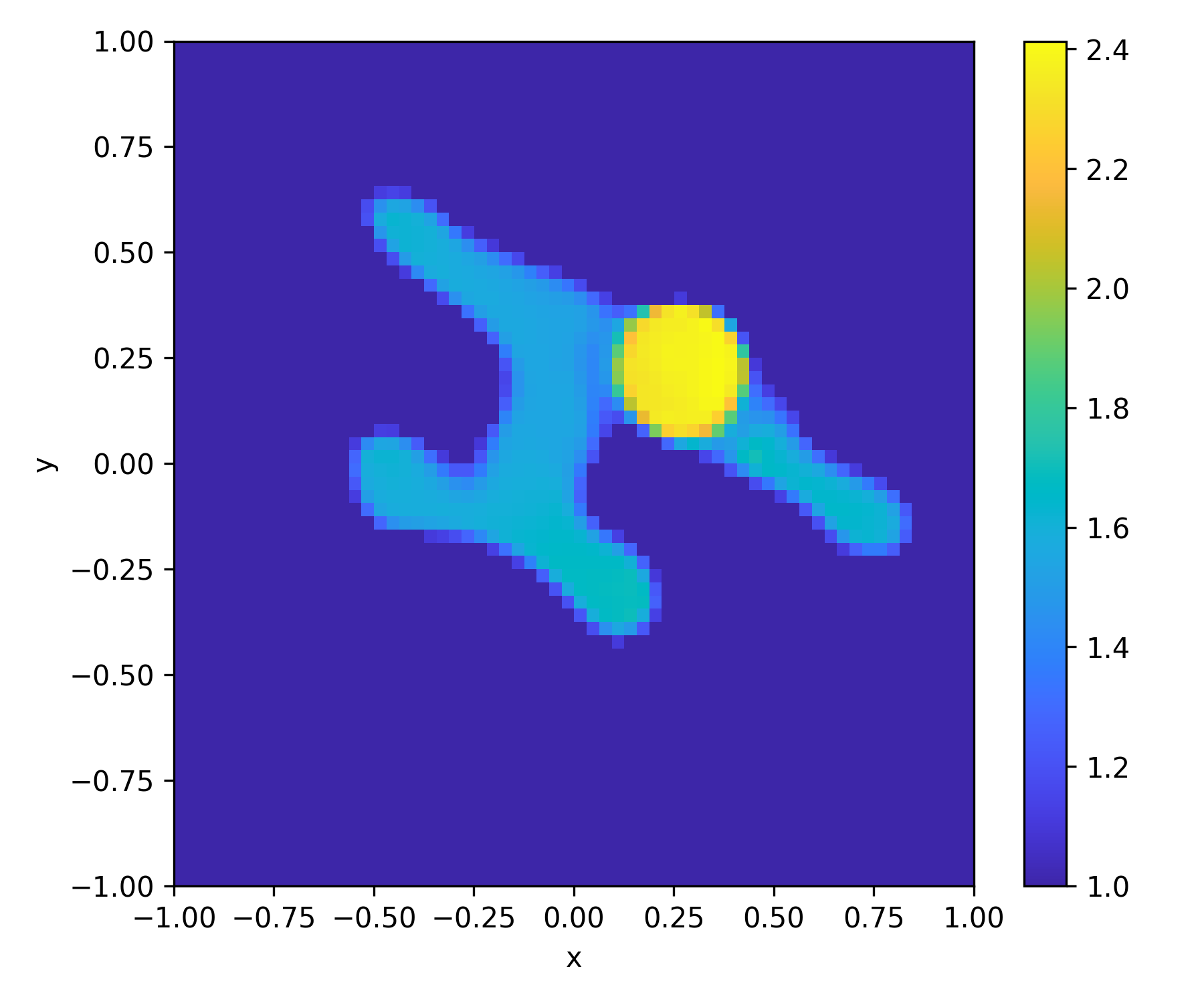}
				\\
				40\%& &
				\includegraphics[width=0.15\textwidth]{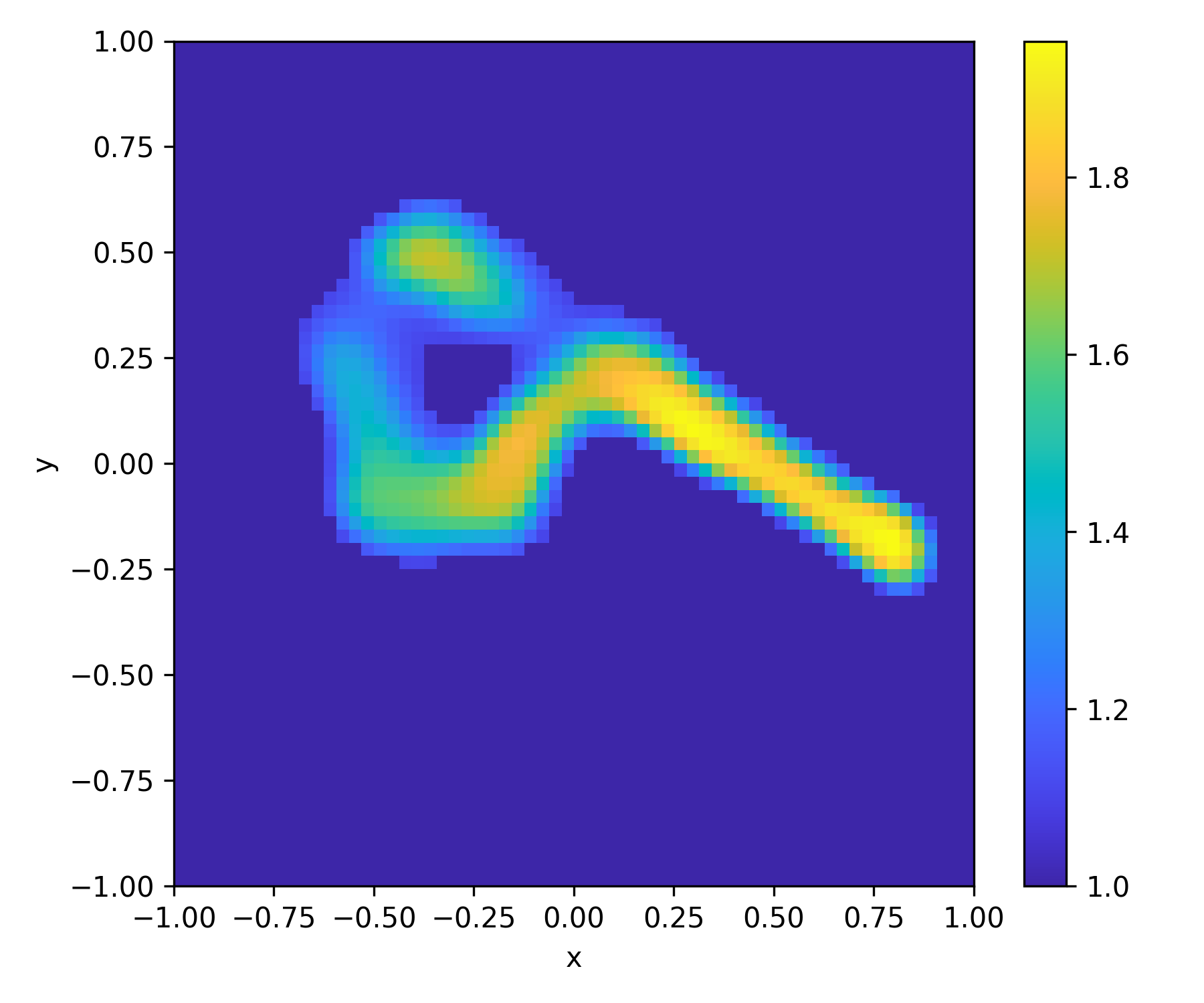}&
				\includegraphics[width=0.15\textwidth]{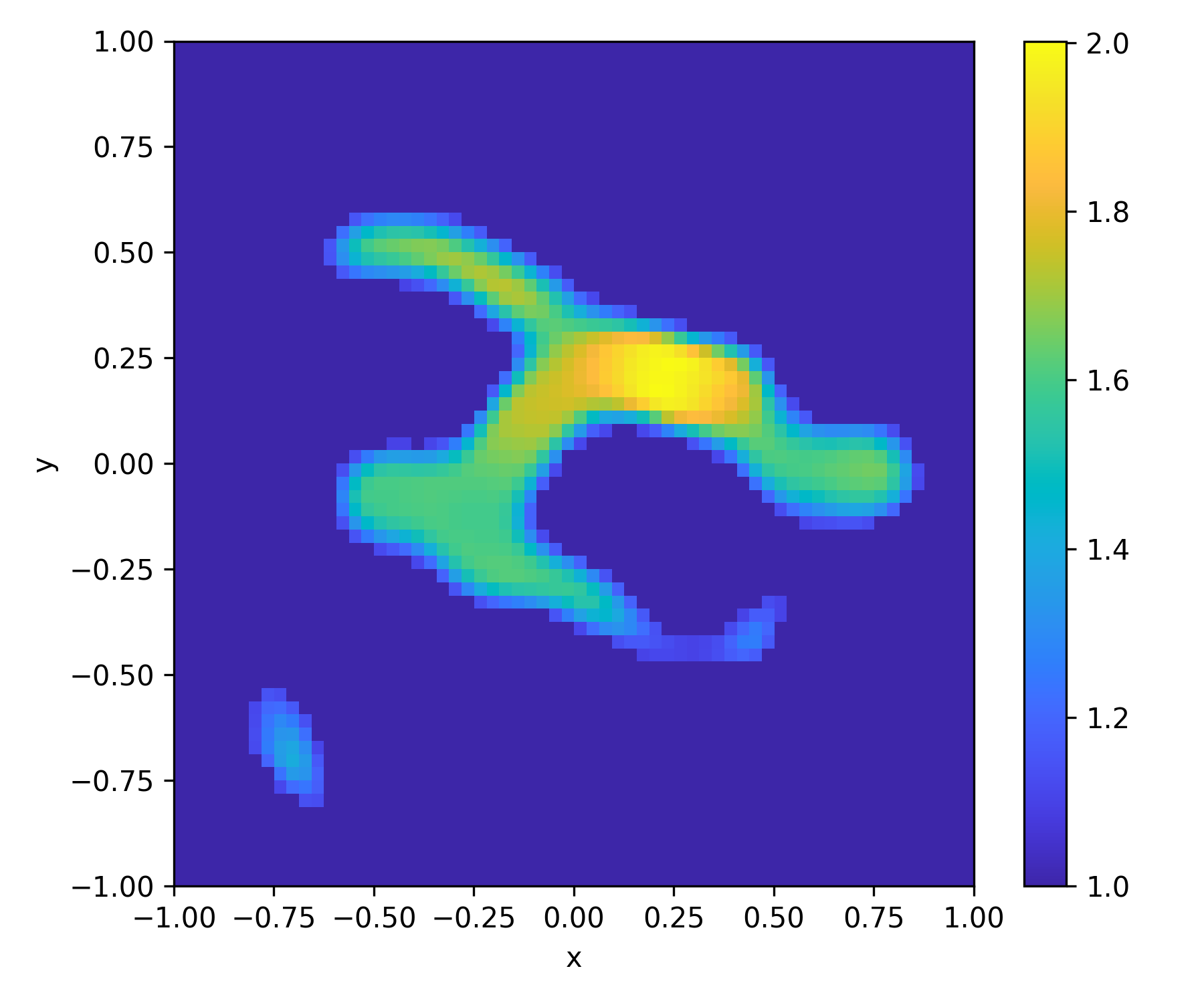}&
				\includegraphics[width=0.15\textwidth]{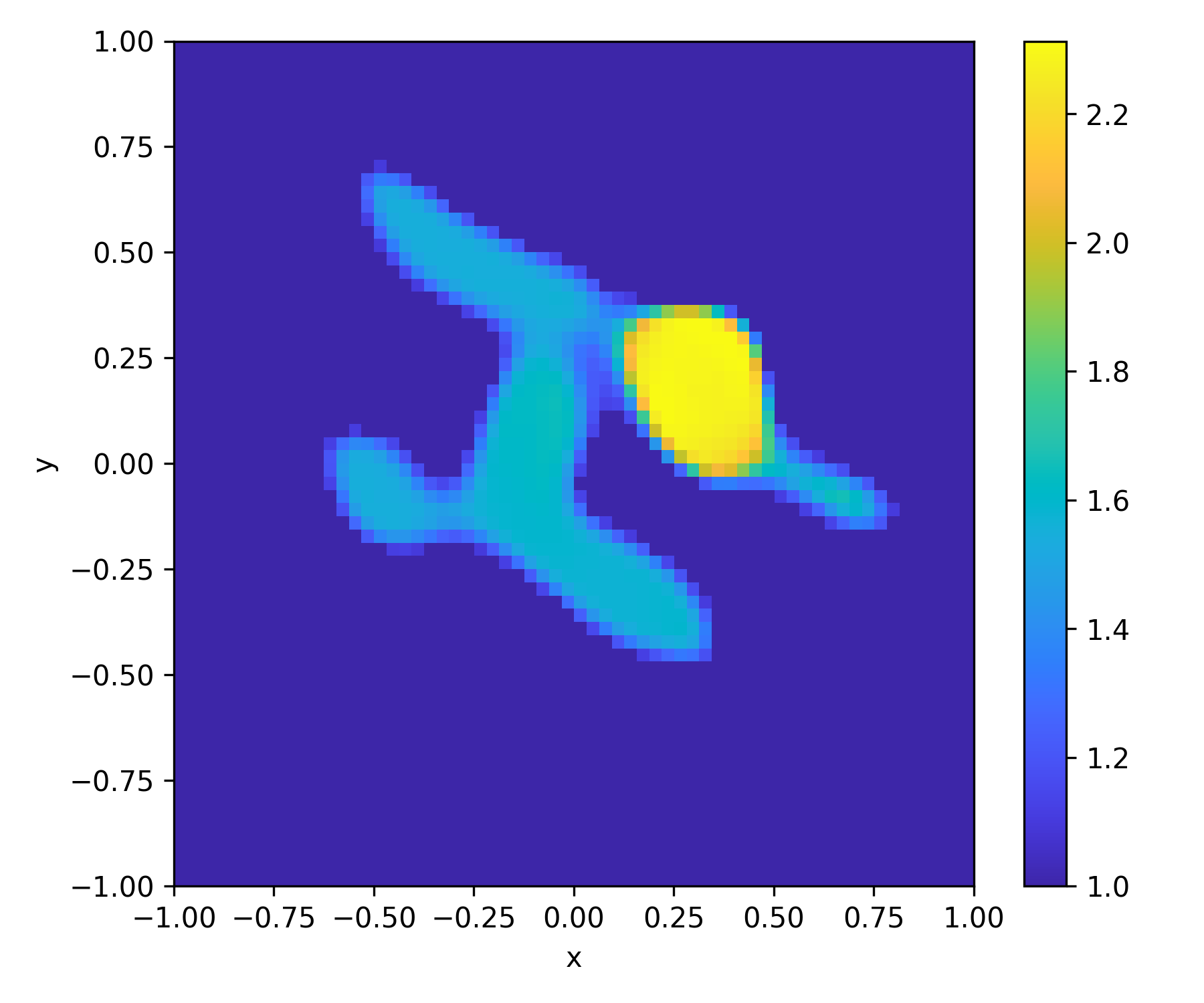}&
				\includegraphics[width=0.15\textwidth]{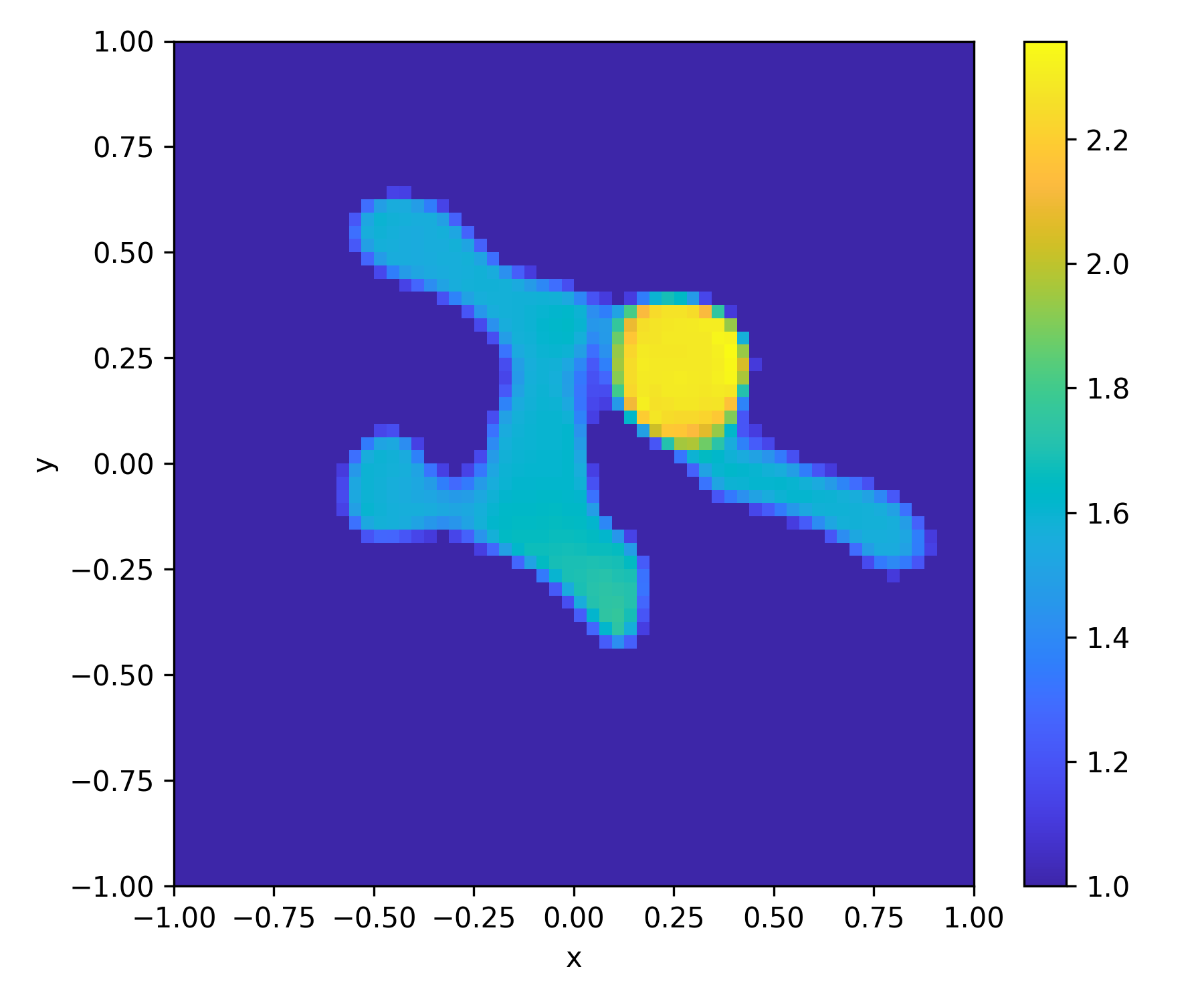}&
				\includegraphics[width=0.15\textwidth]{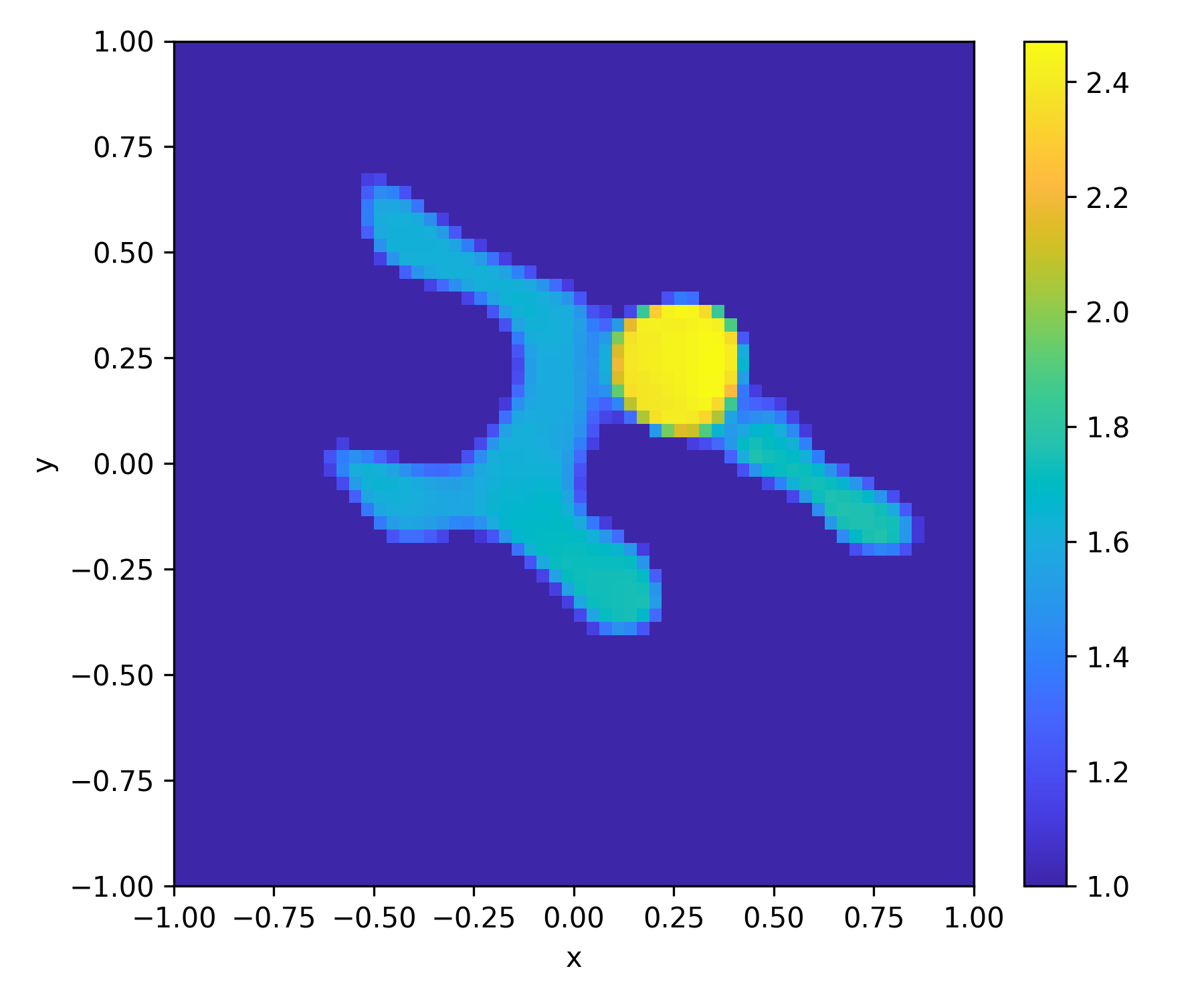}
				\\
				
				15\%&\SetCell[r=2]{c} \includegraphics[width=0.15\textwidth]{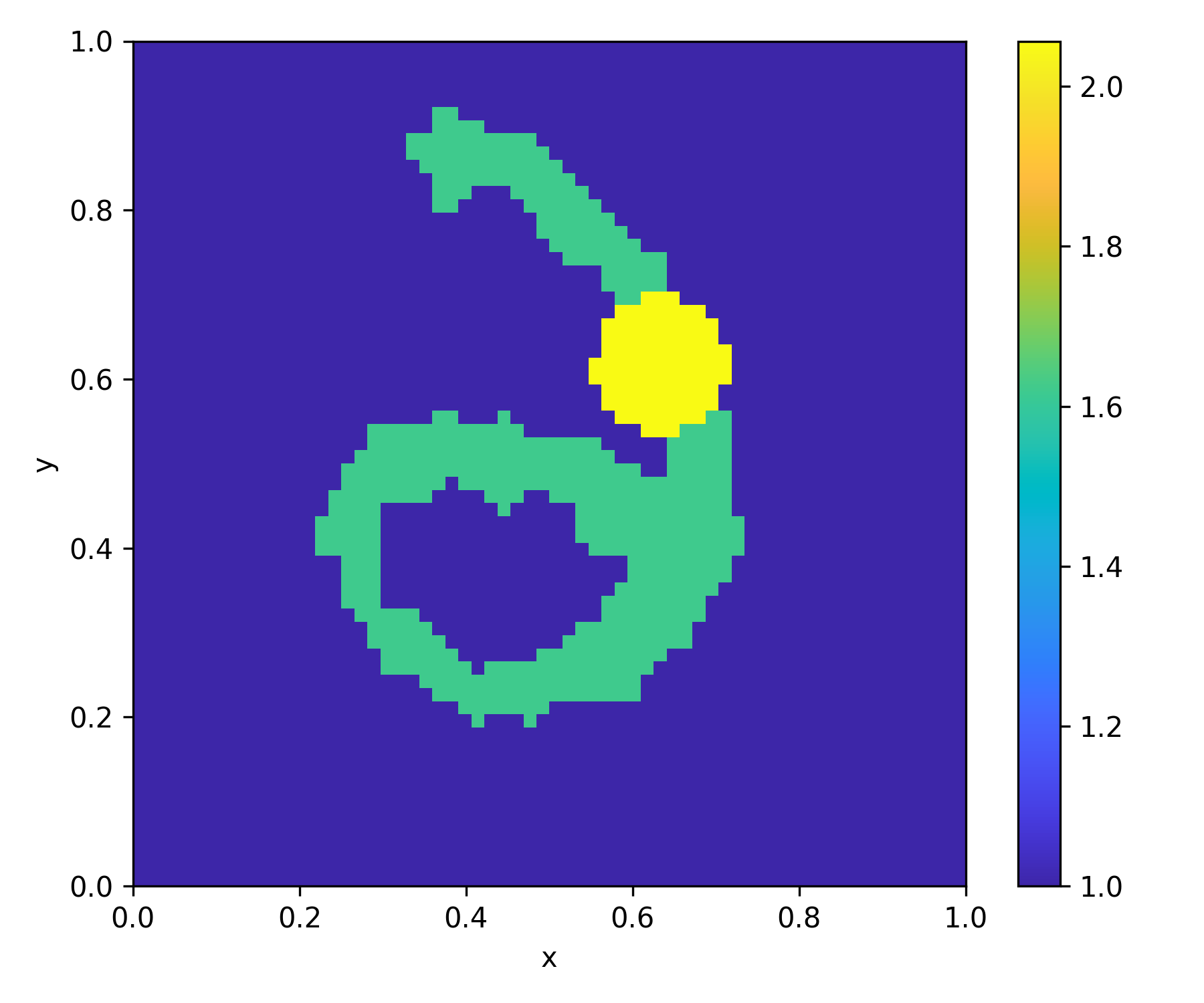} &
				\includegraphics[width=0.15\textwidth]{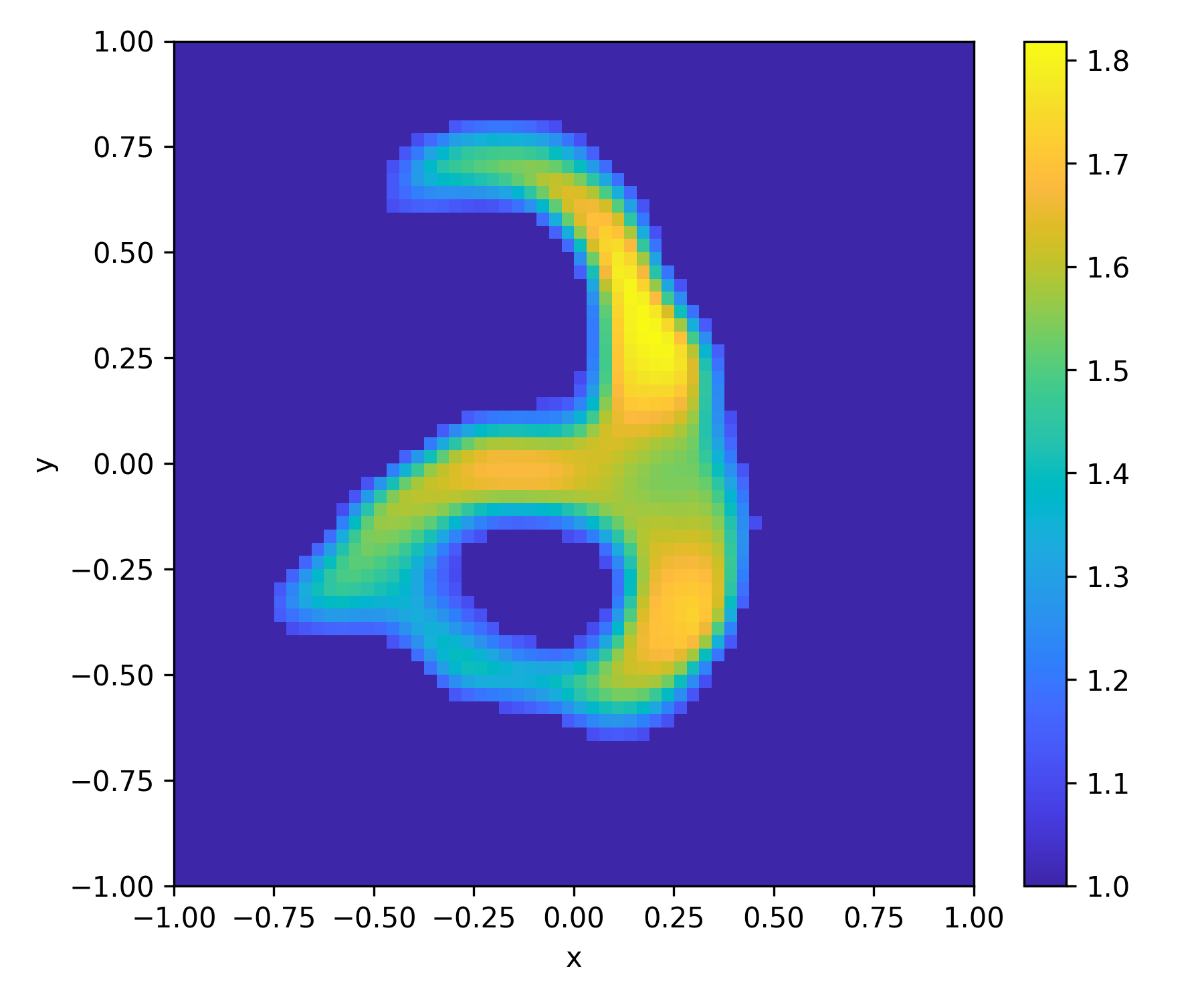}&
				\includegraphics[width=0.15\textwidth]{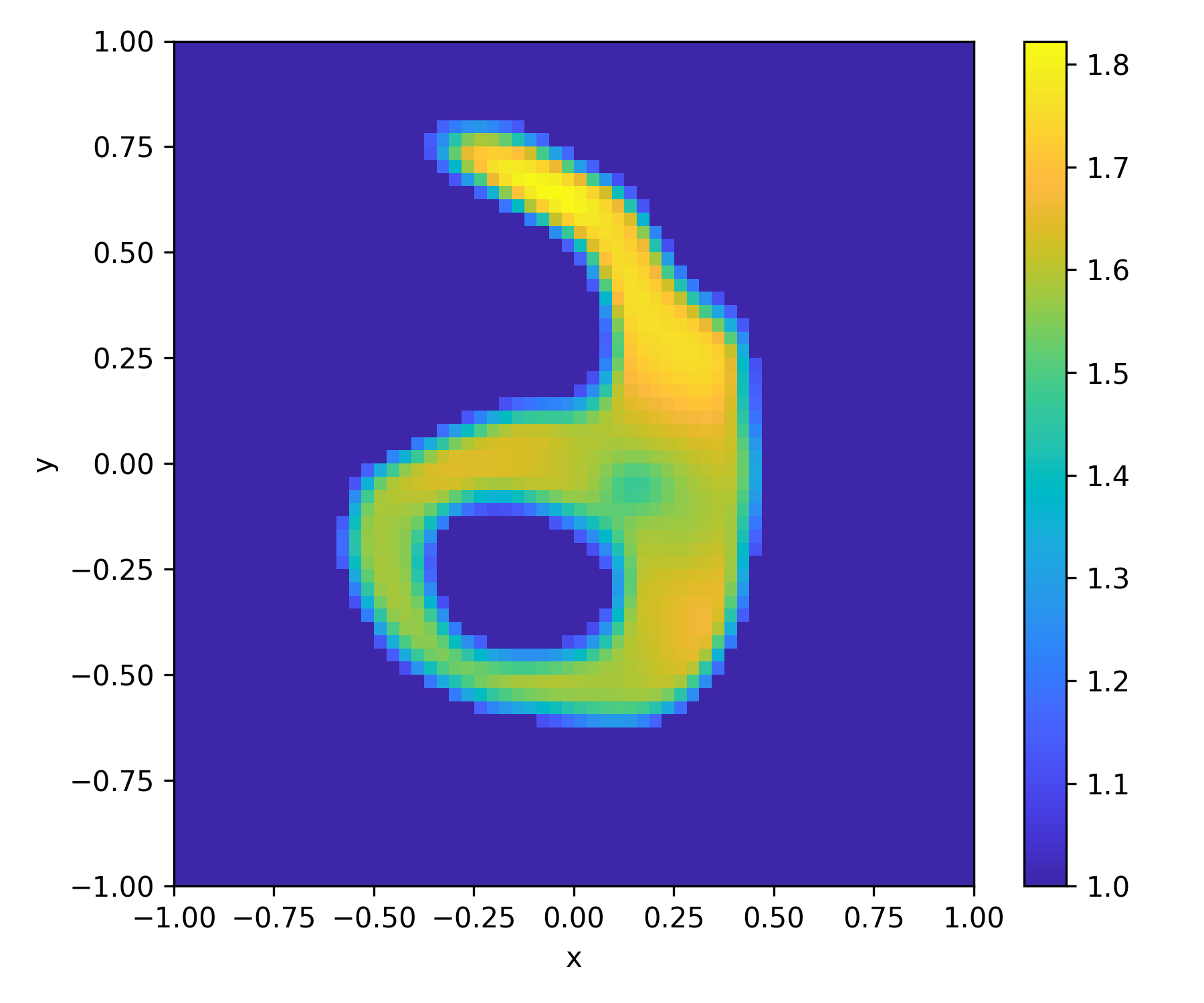}&
				\includegraphics[width=0.15\textwidth]{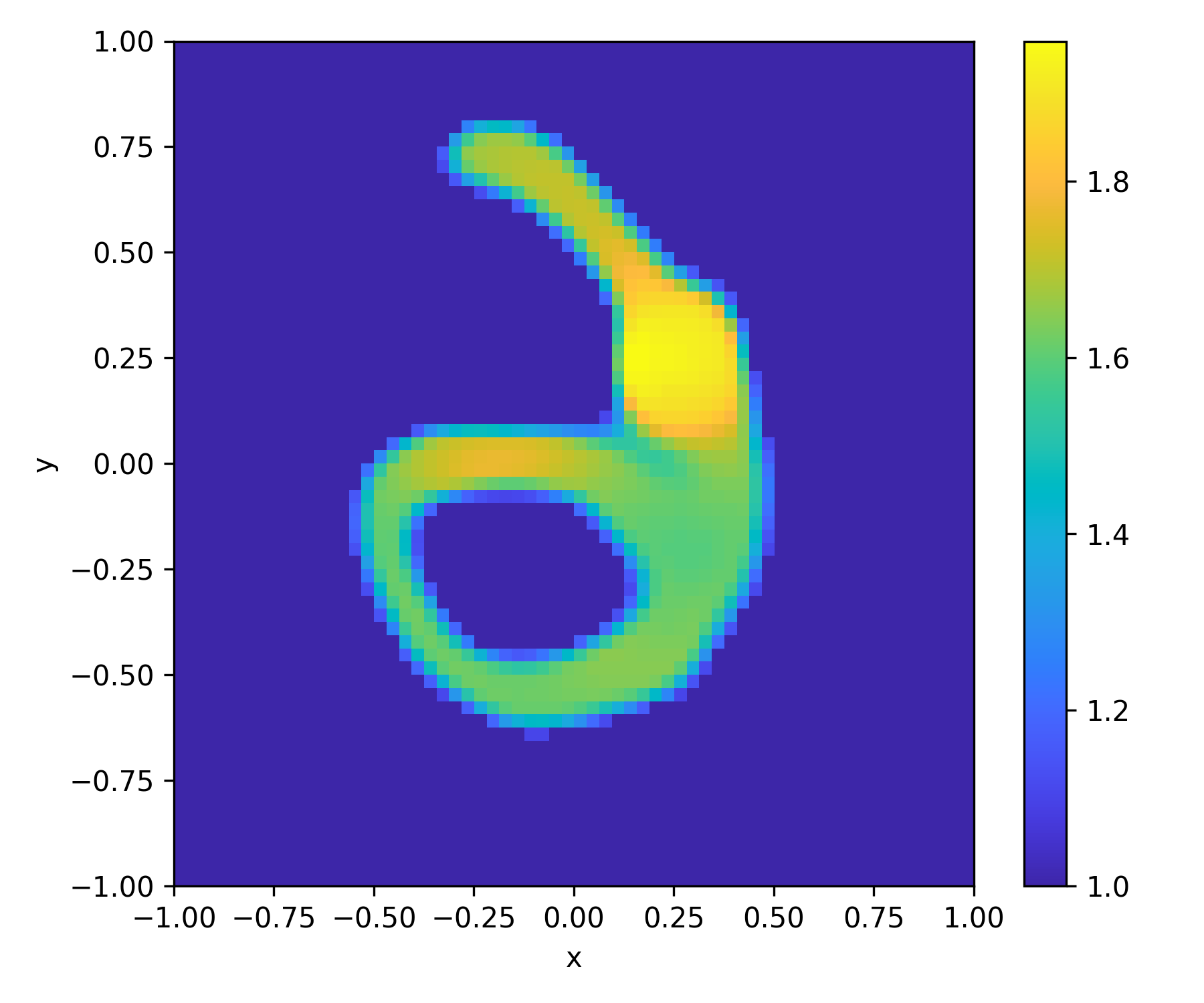}&
				\includegraphics[width=0.15\textwidth]{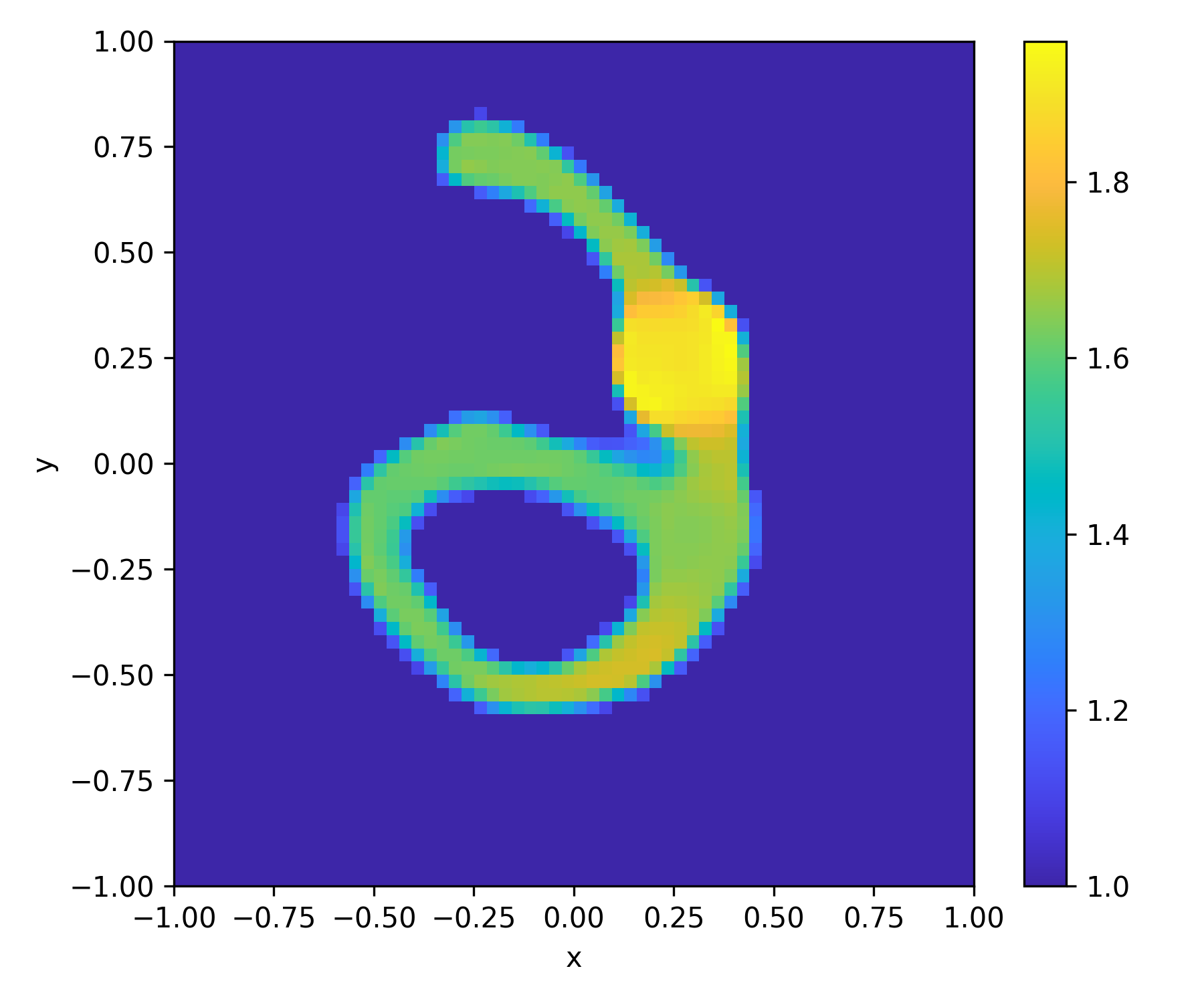}&
				\includegraphics[width=0.15\textwidth]{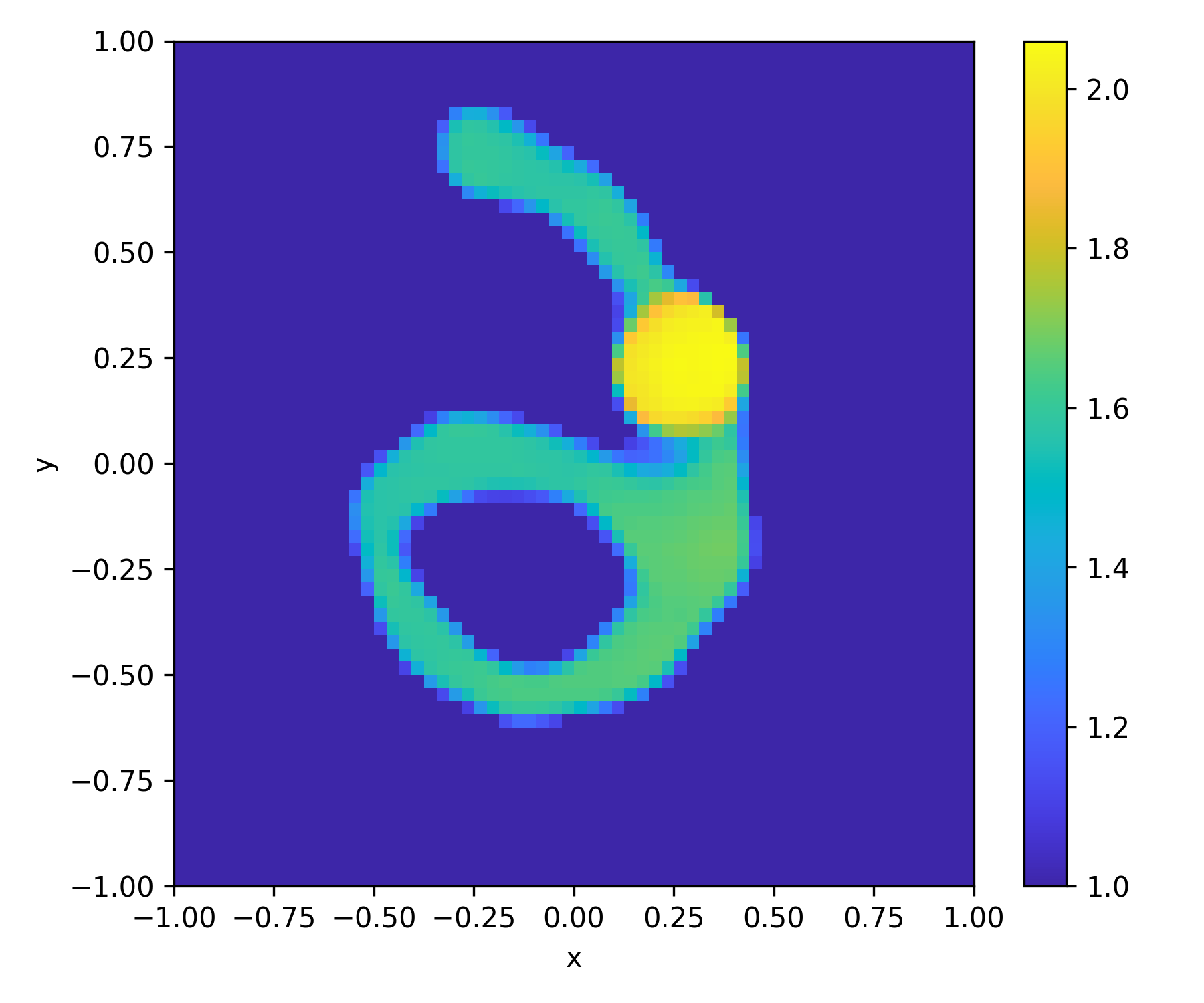}
				\\
				40\%& &
				\includegraphics[width=0.15\textwidth]{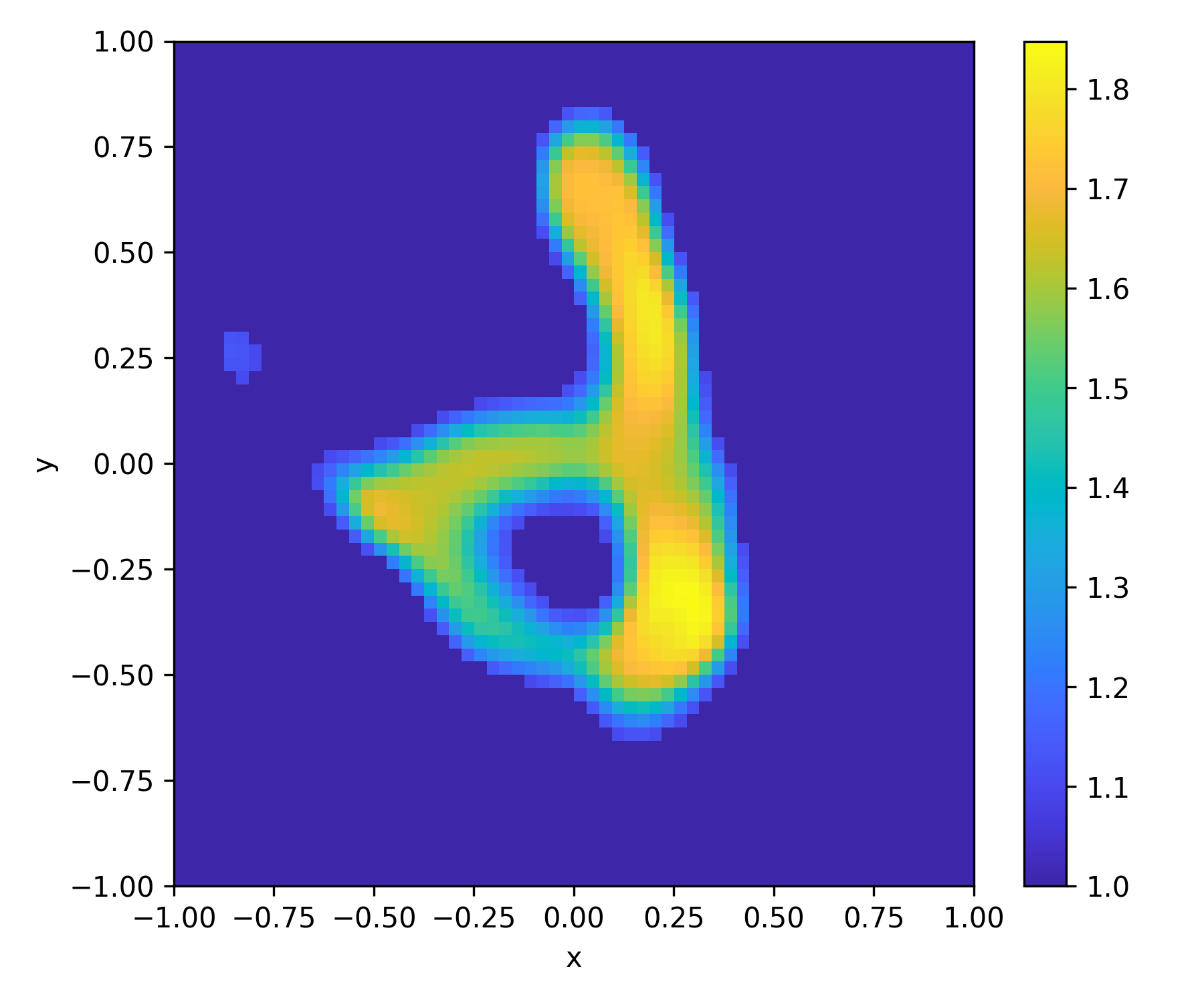}&
				\includegraphics[width=0.15\textwidth]{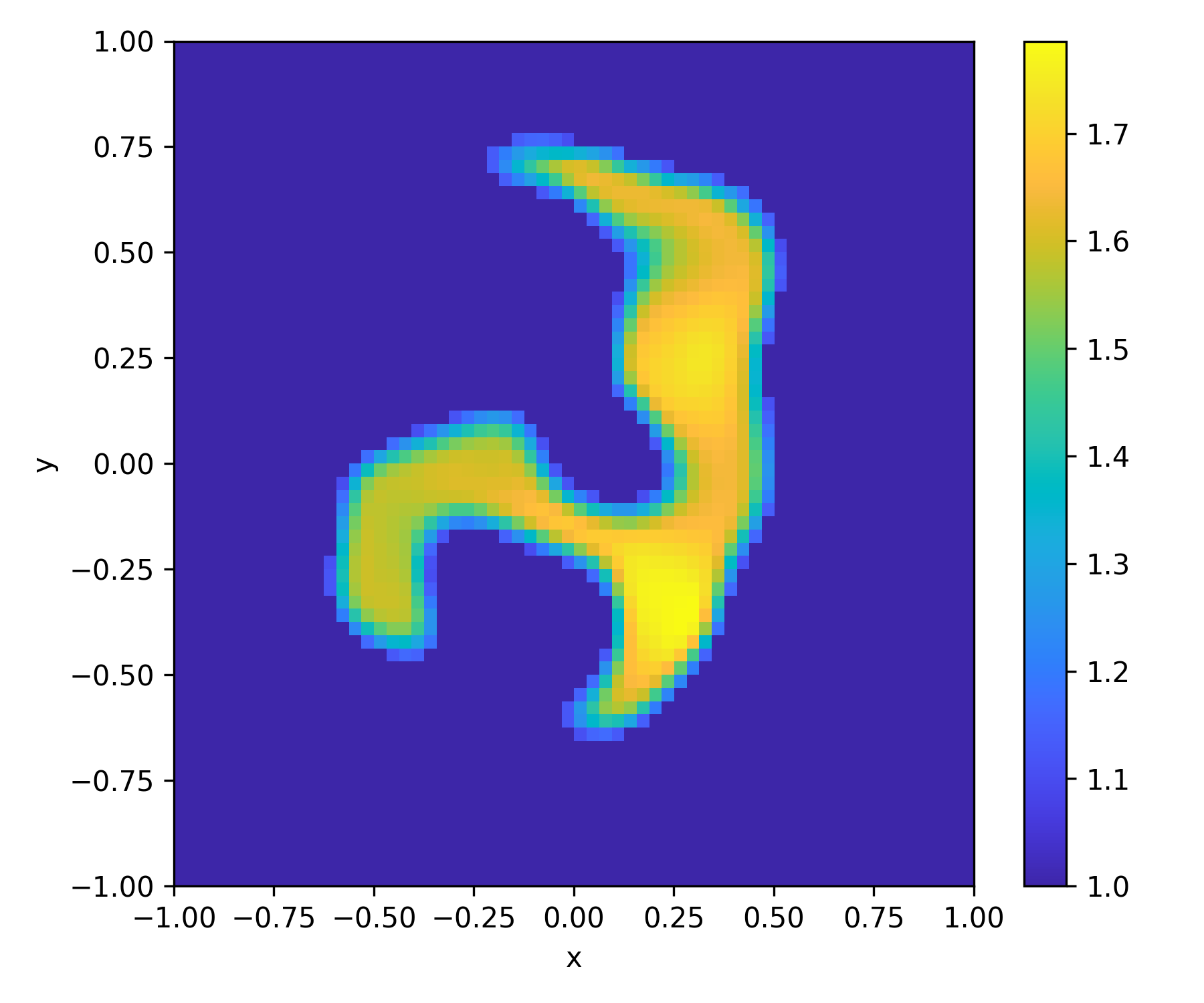}&
				\includegraphics[width=0.15\textwidth]{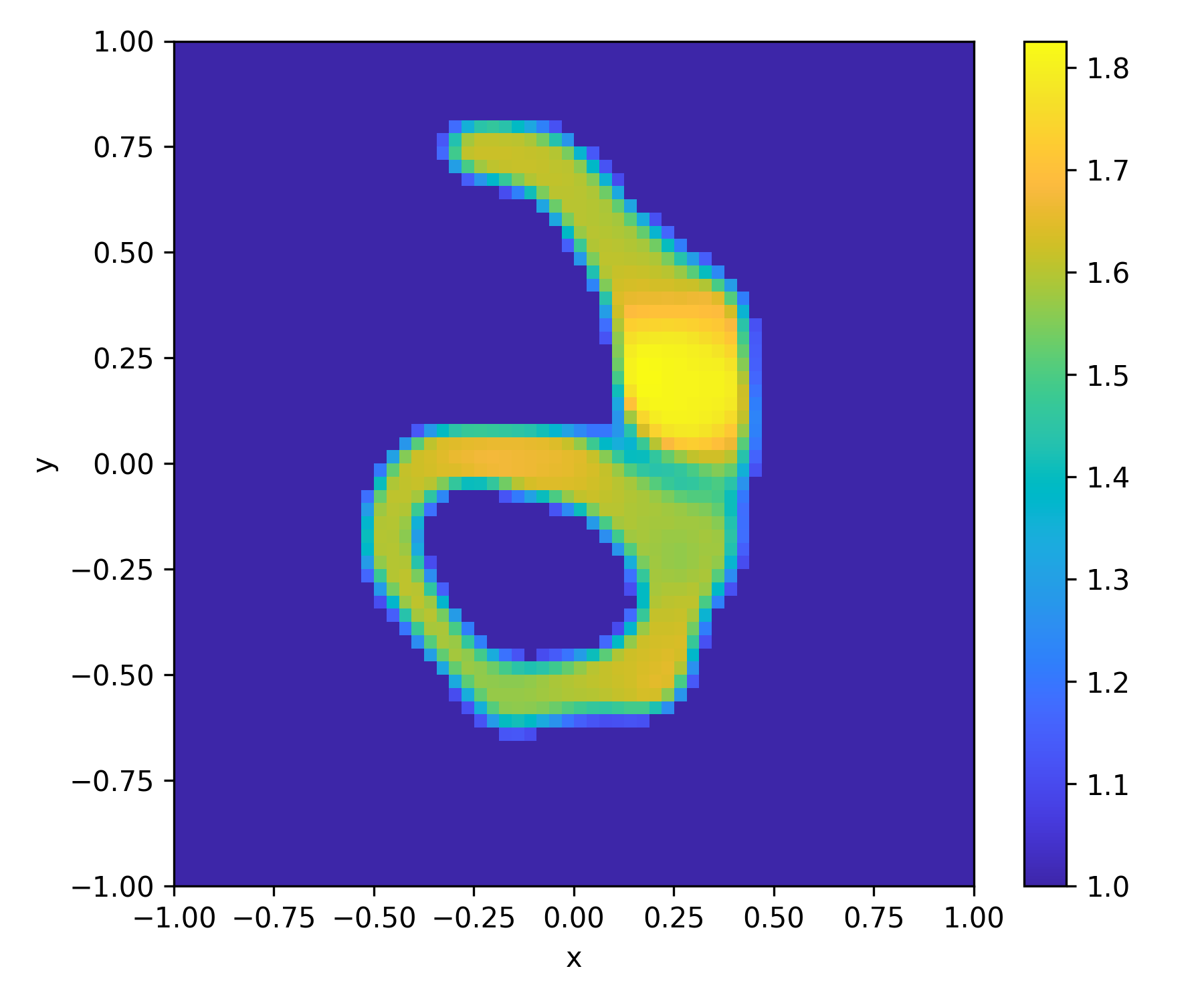}&
				\includegraphics[width=0.15\textwidth]{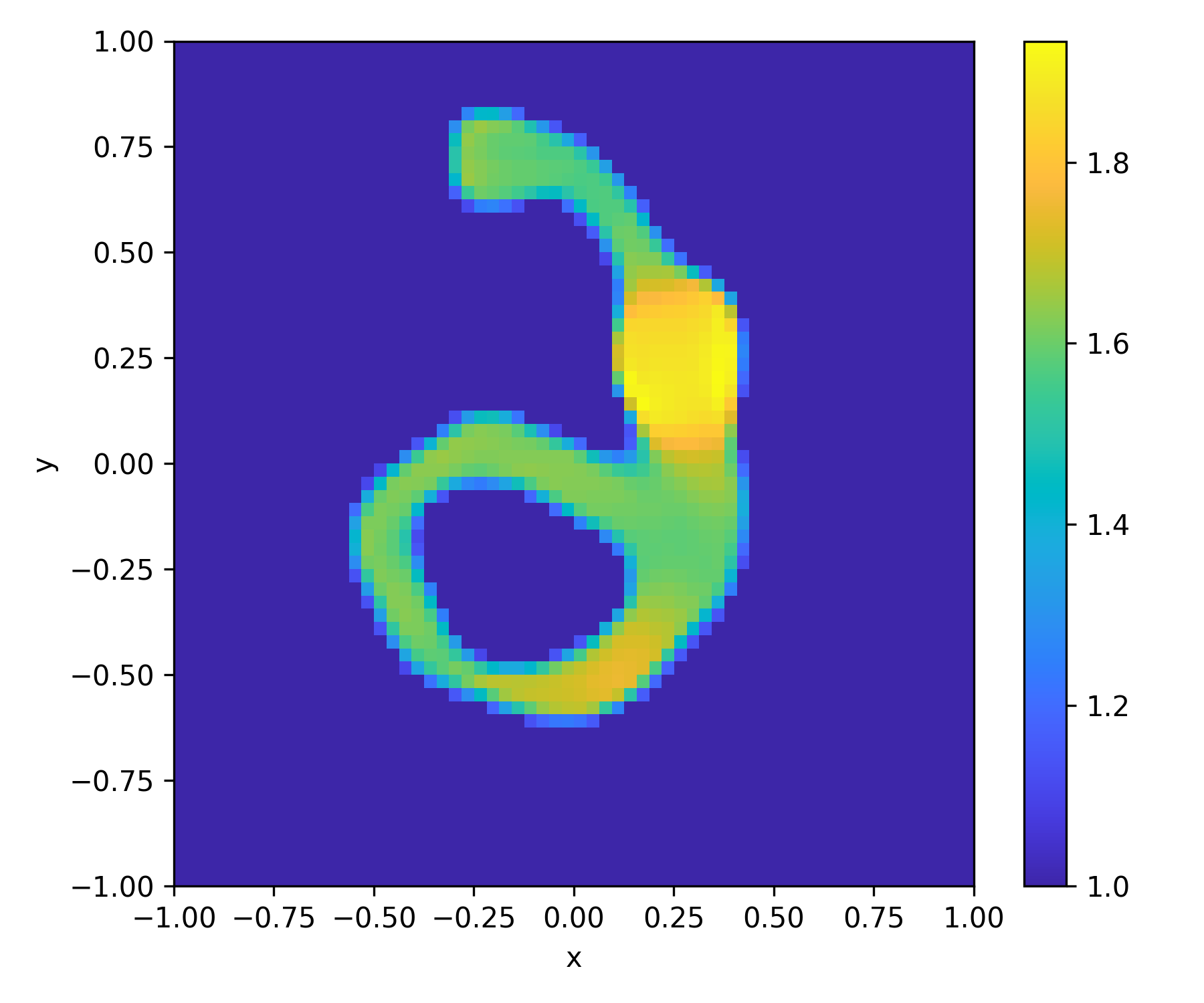}&
				\includegraphics[width=0.15\textwidth]{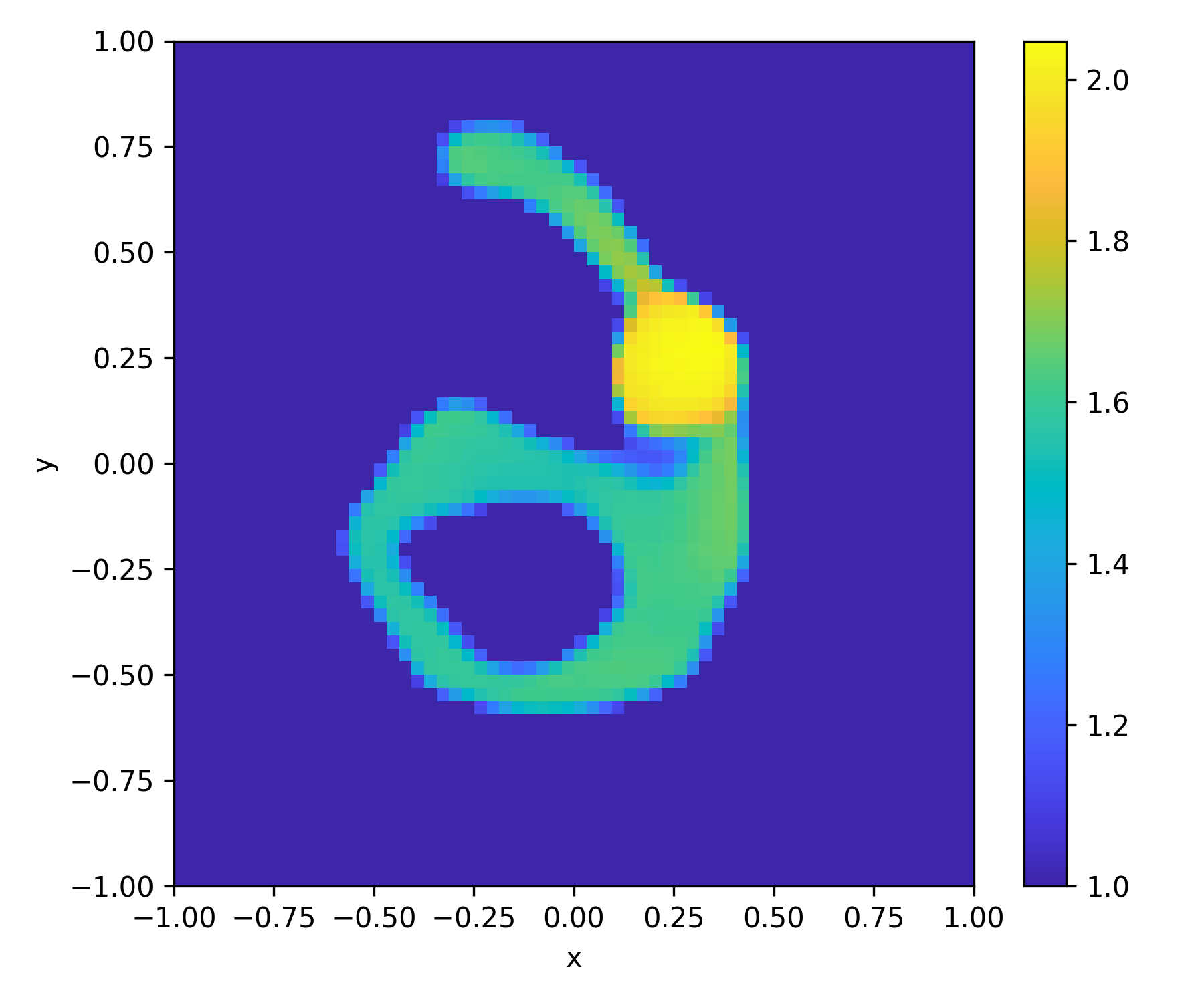}
				\\ 15\%&\SetCell[r=2]{c} \includegraphics[width=0.15\textwidth]{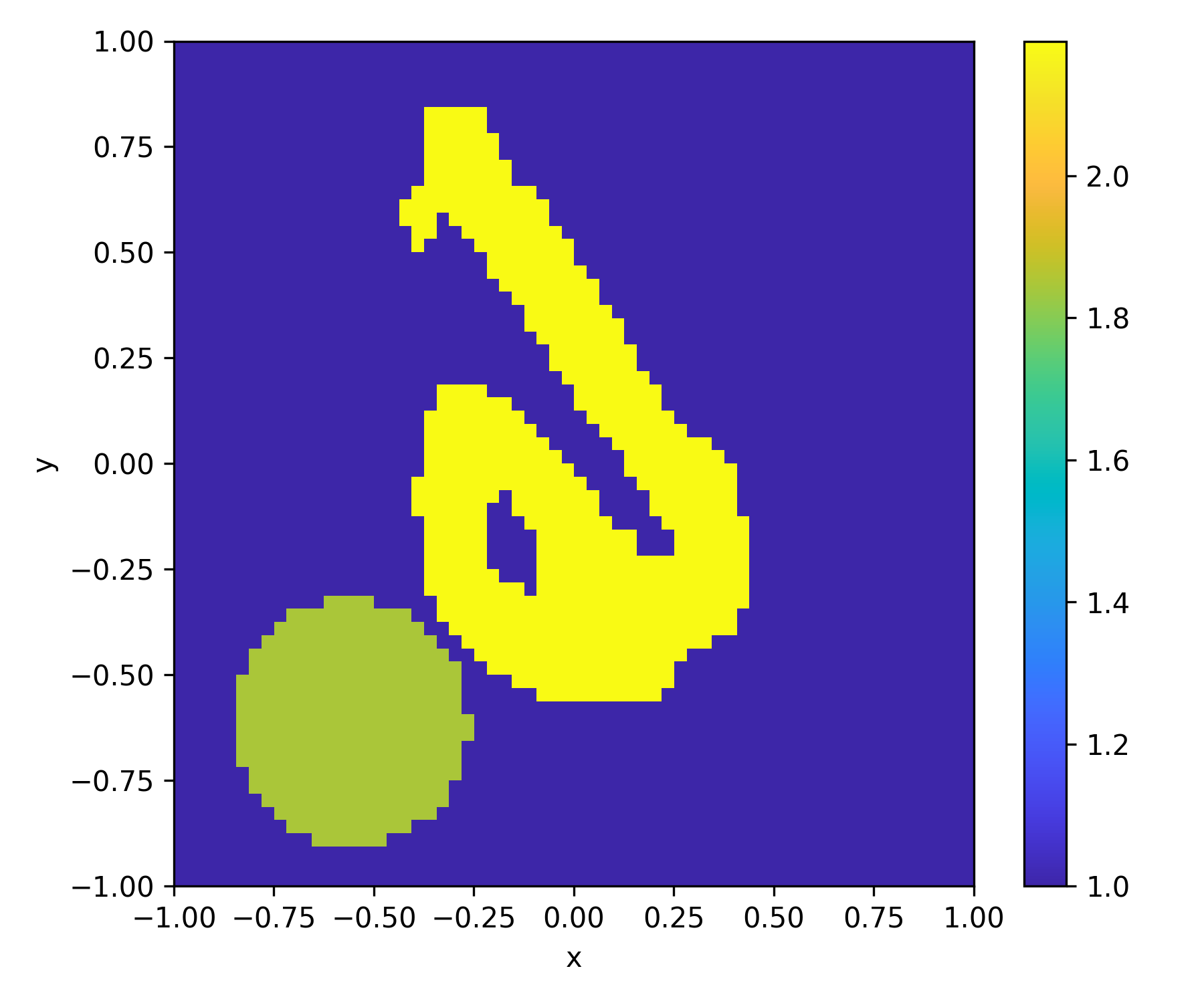} &
				\includegraphics[width=0.15\textwidth]{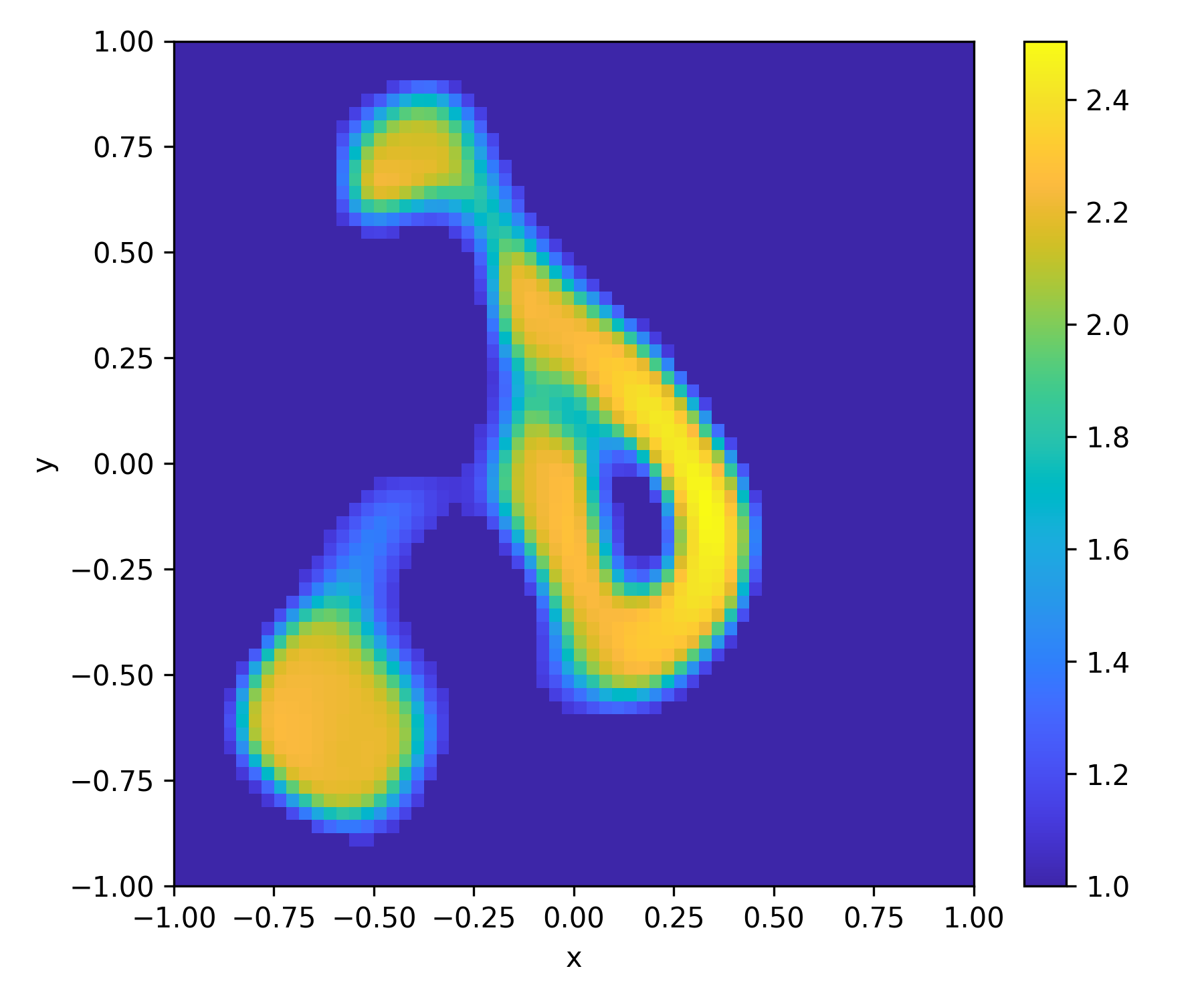}&
				\includegraphics[width=0.15\textwidth]{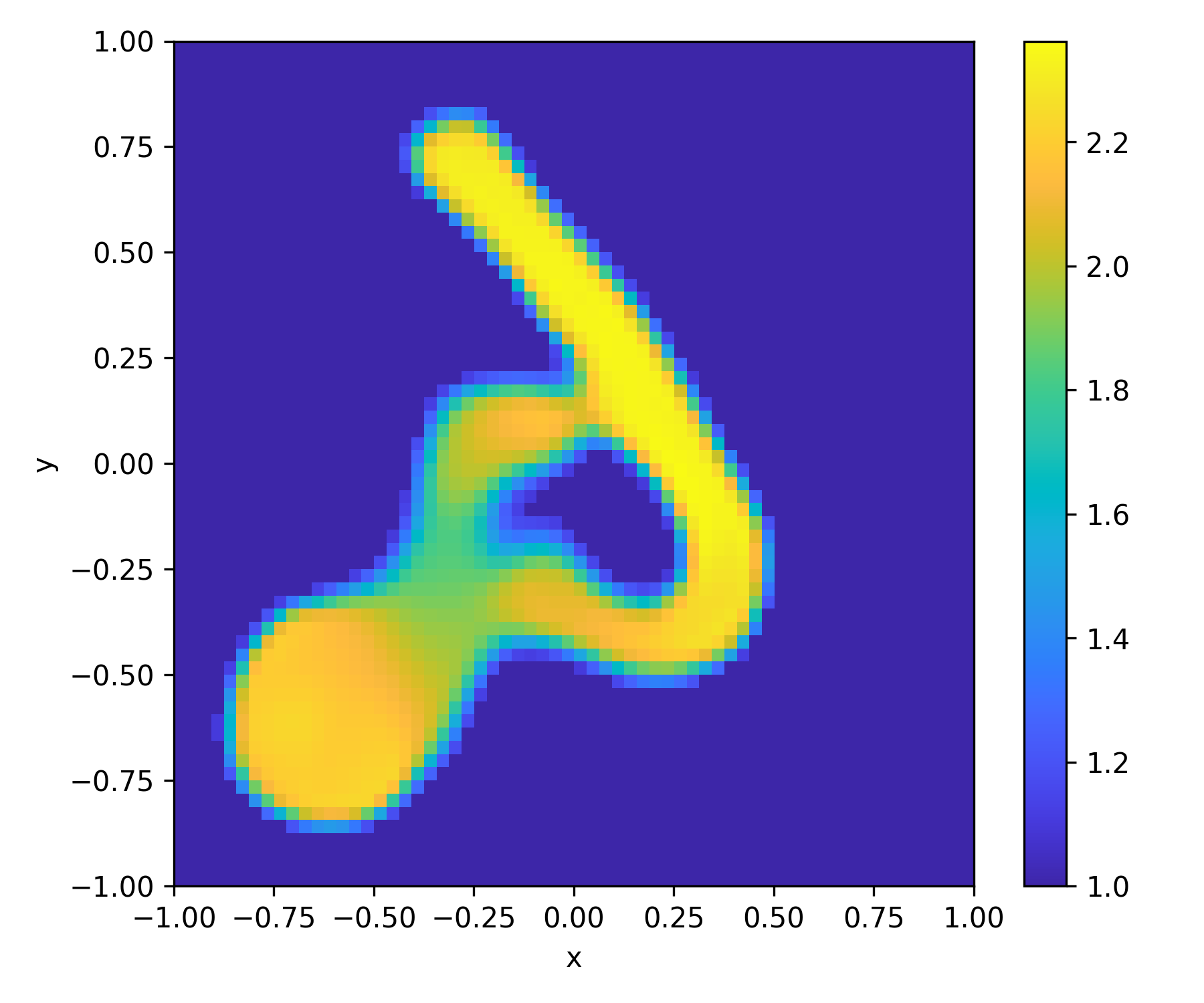}&
				\includegraphics[width=0.15\textwidth]{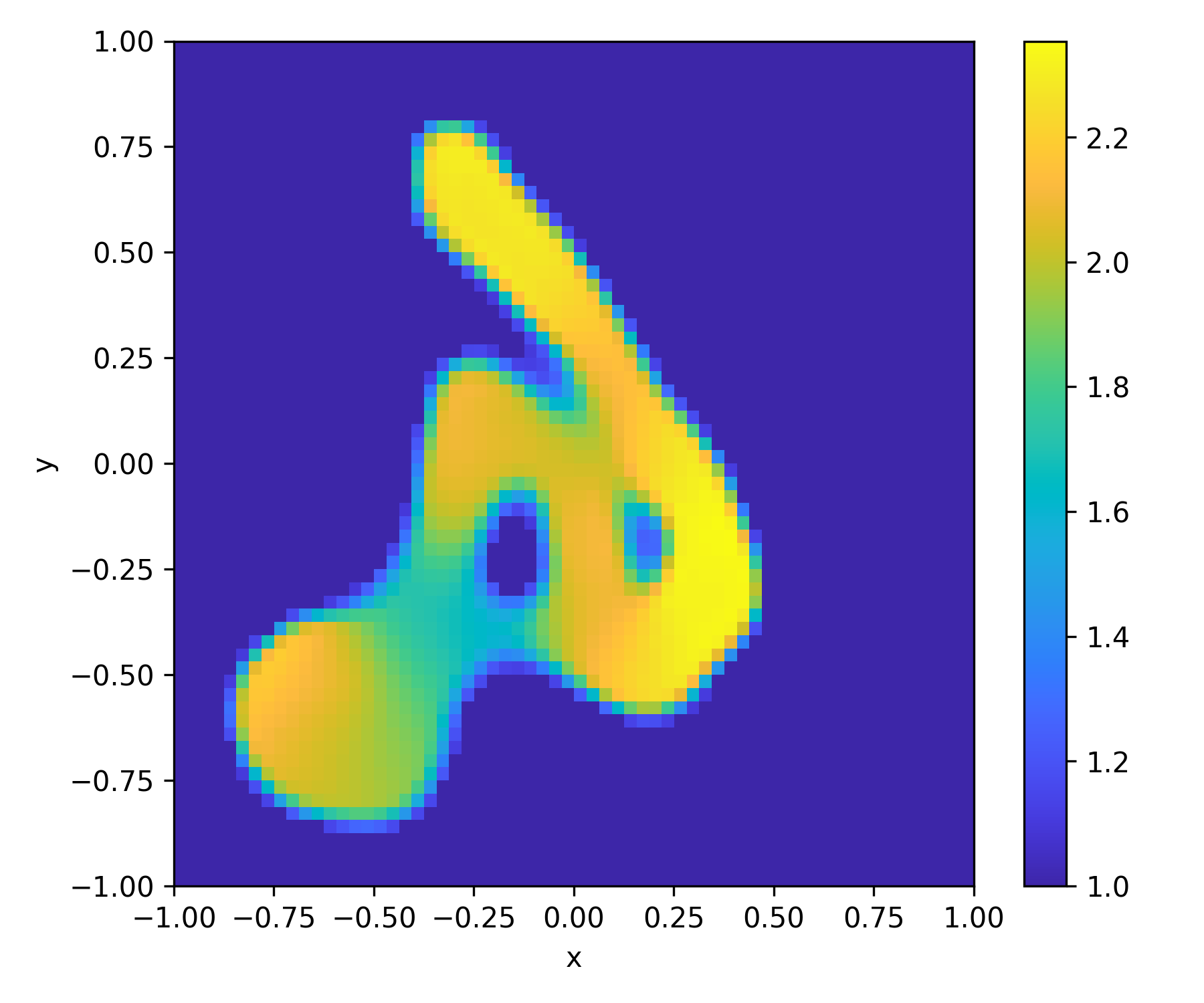}&
				\includegraphics[width=0.15\textwidth]{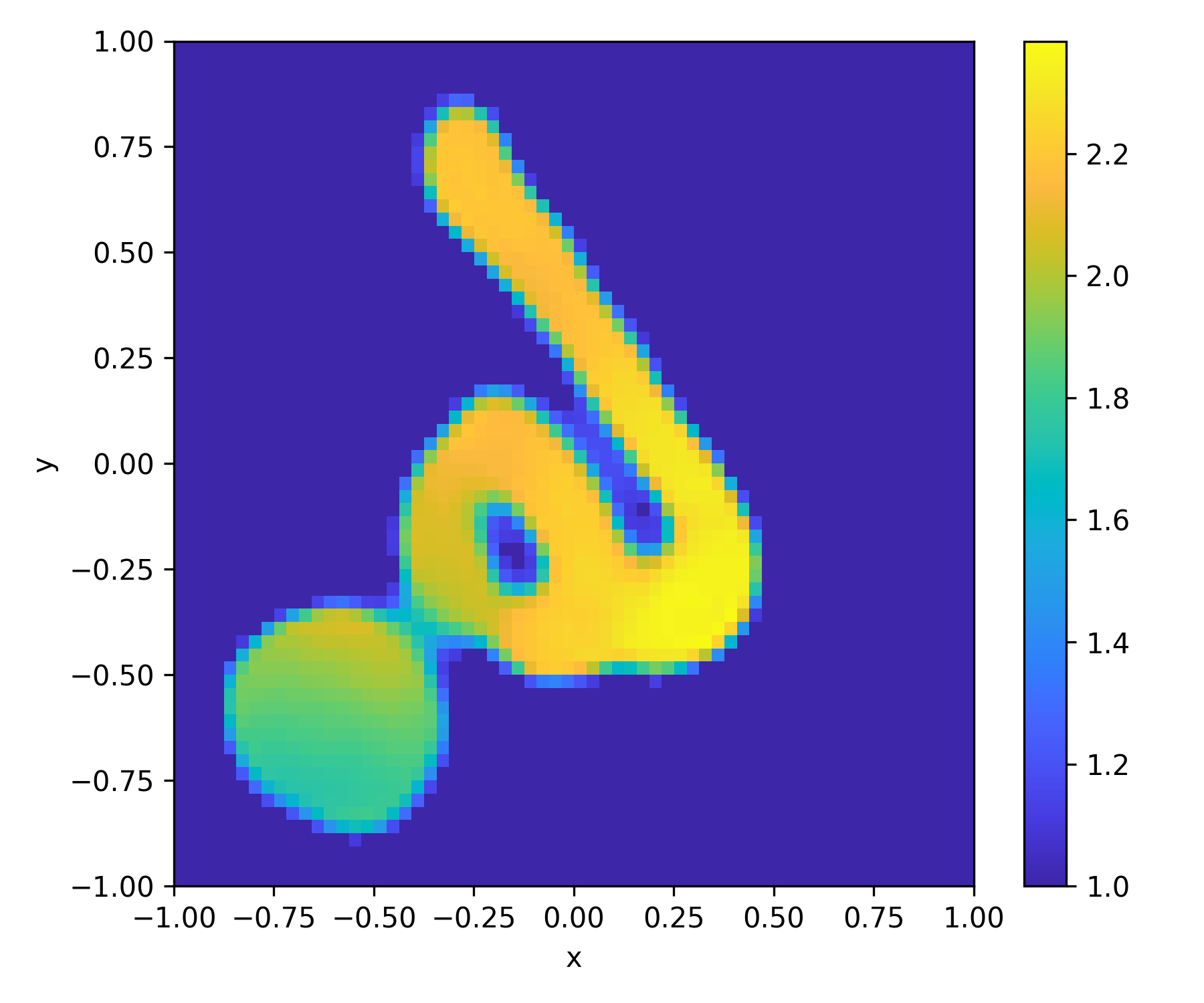}&
				\includegraphics[width=0.15\textwidth]{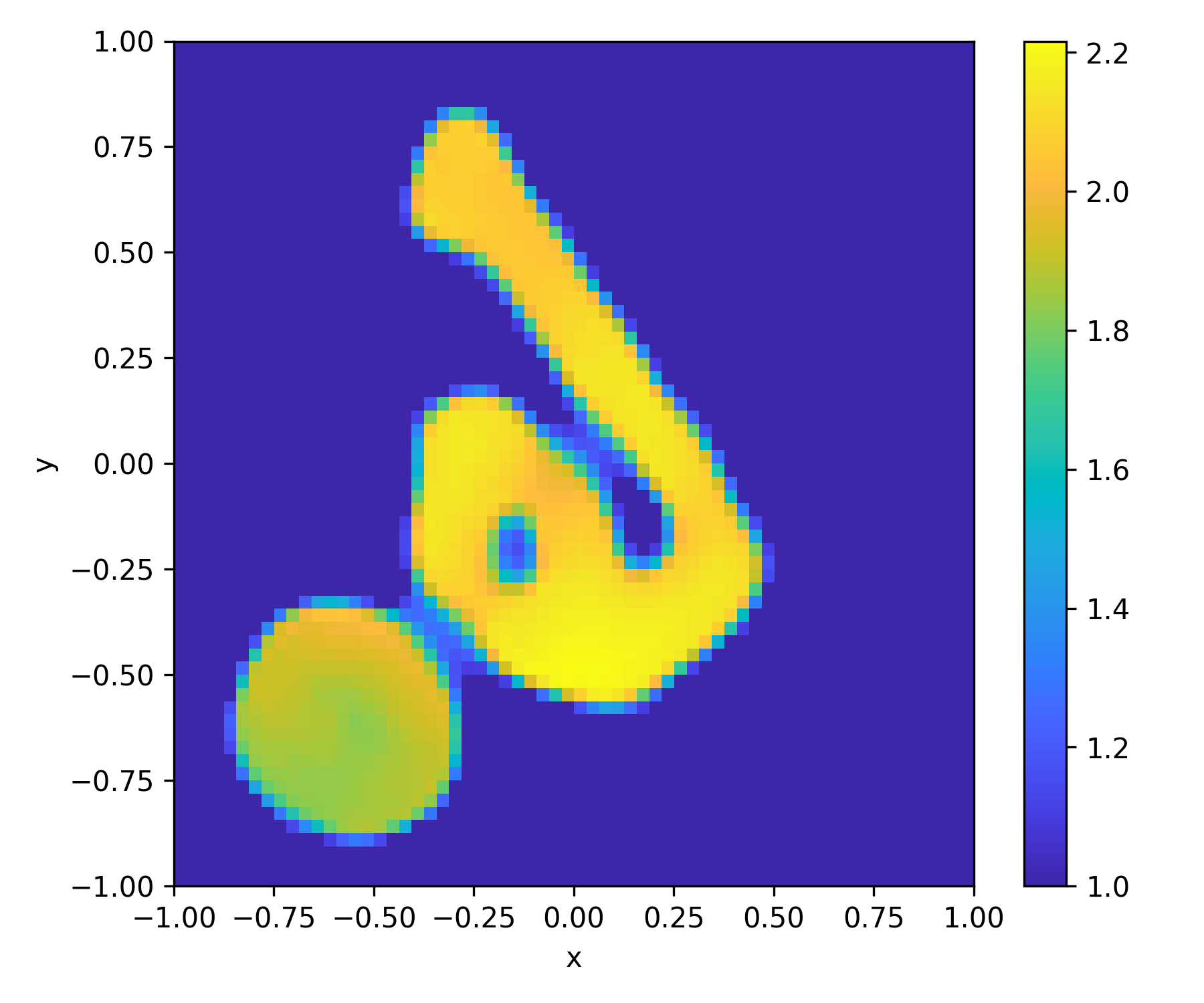}
				\\
				40\%& &
				\includegraphics[width=0.15\textwidth]{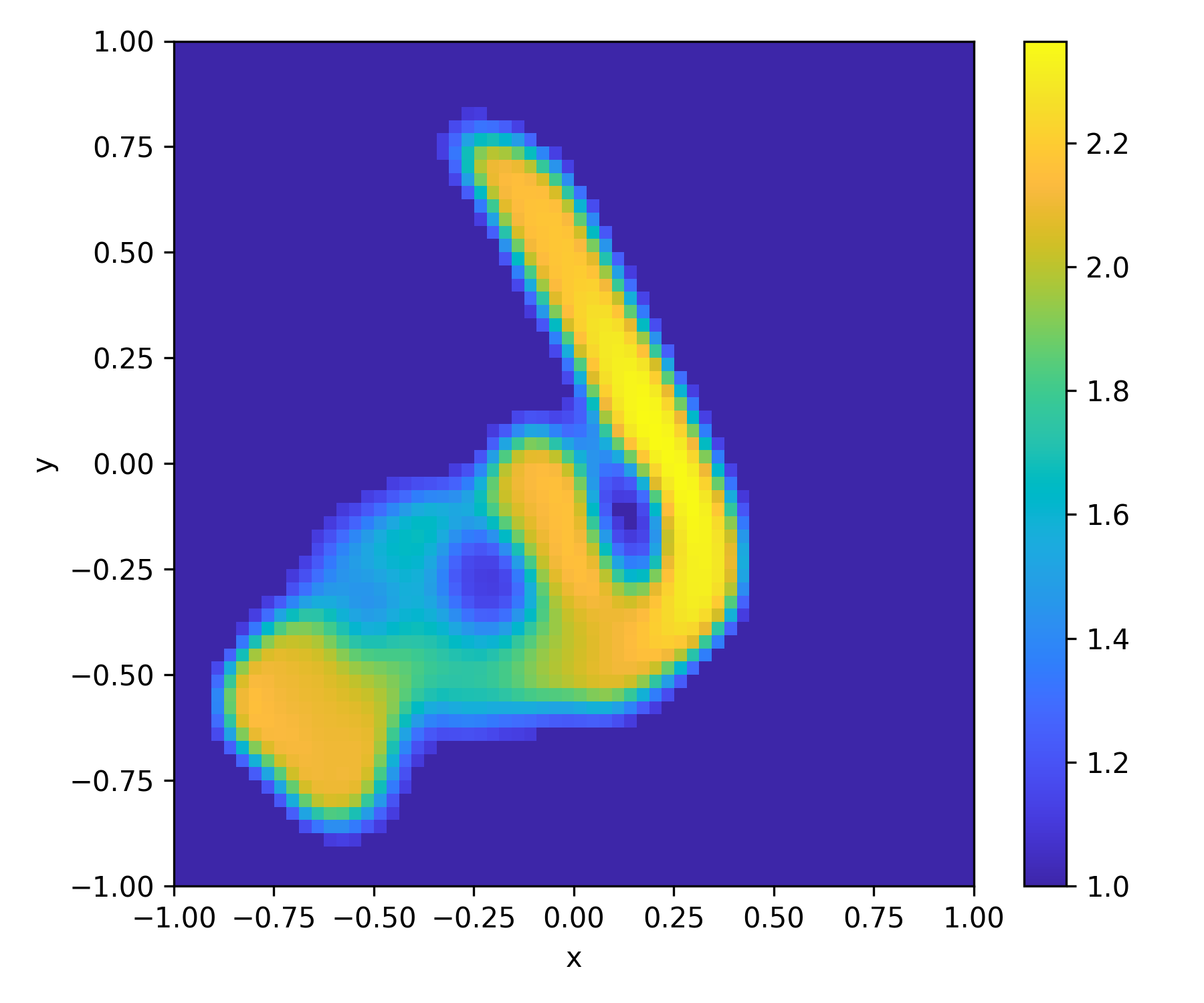}&
				\includegraphics[width=0.15\textwidth]{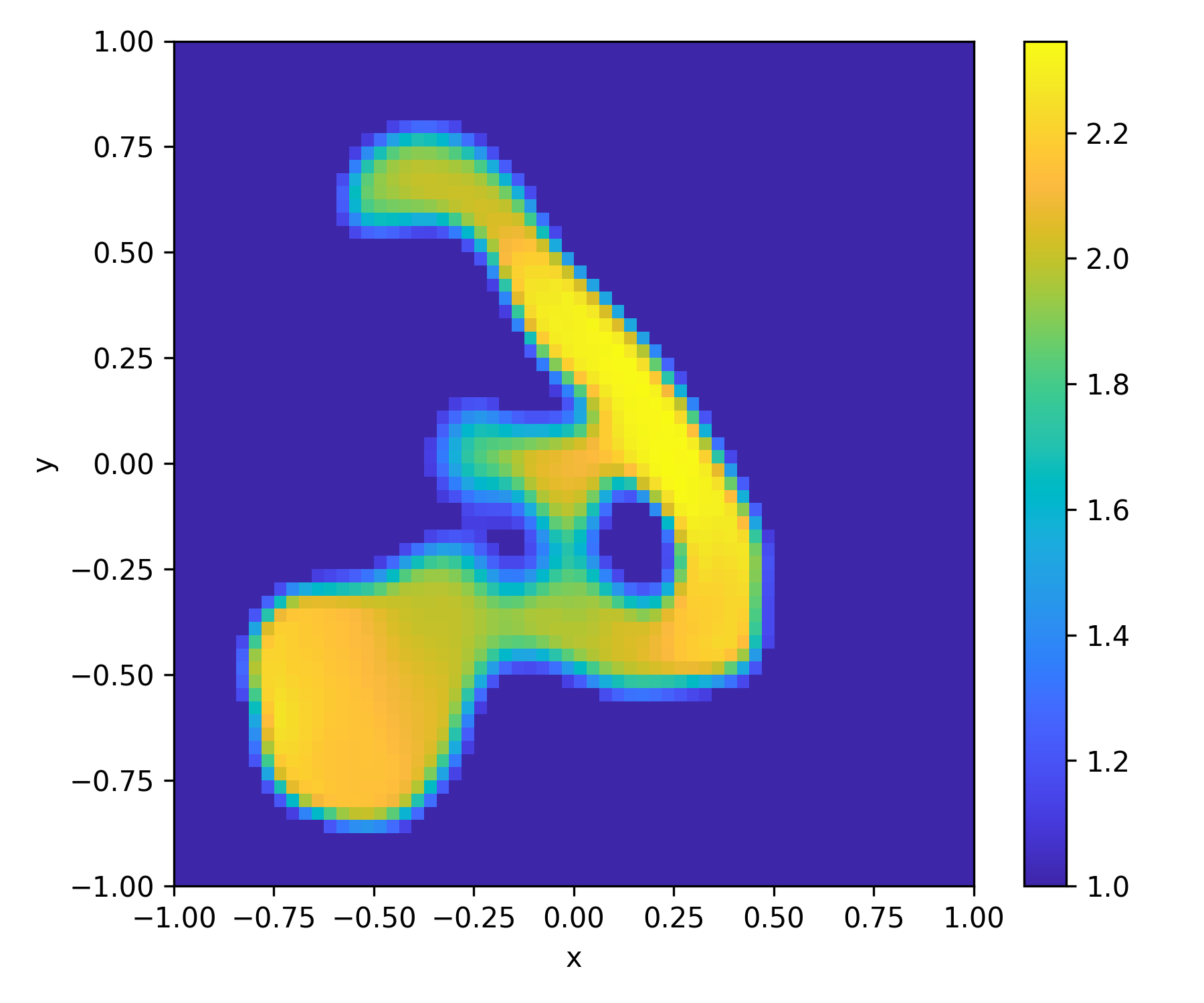}&
				\includegraphics[width=0.15\textwidth]{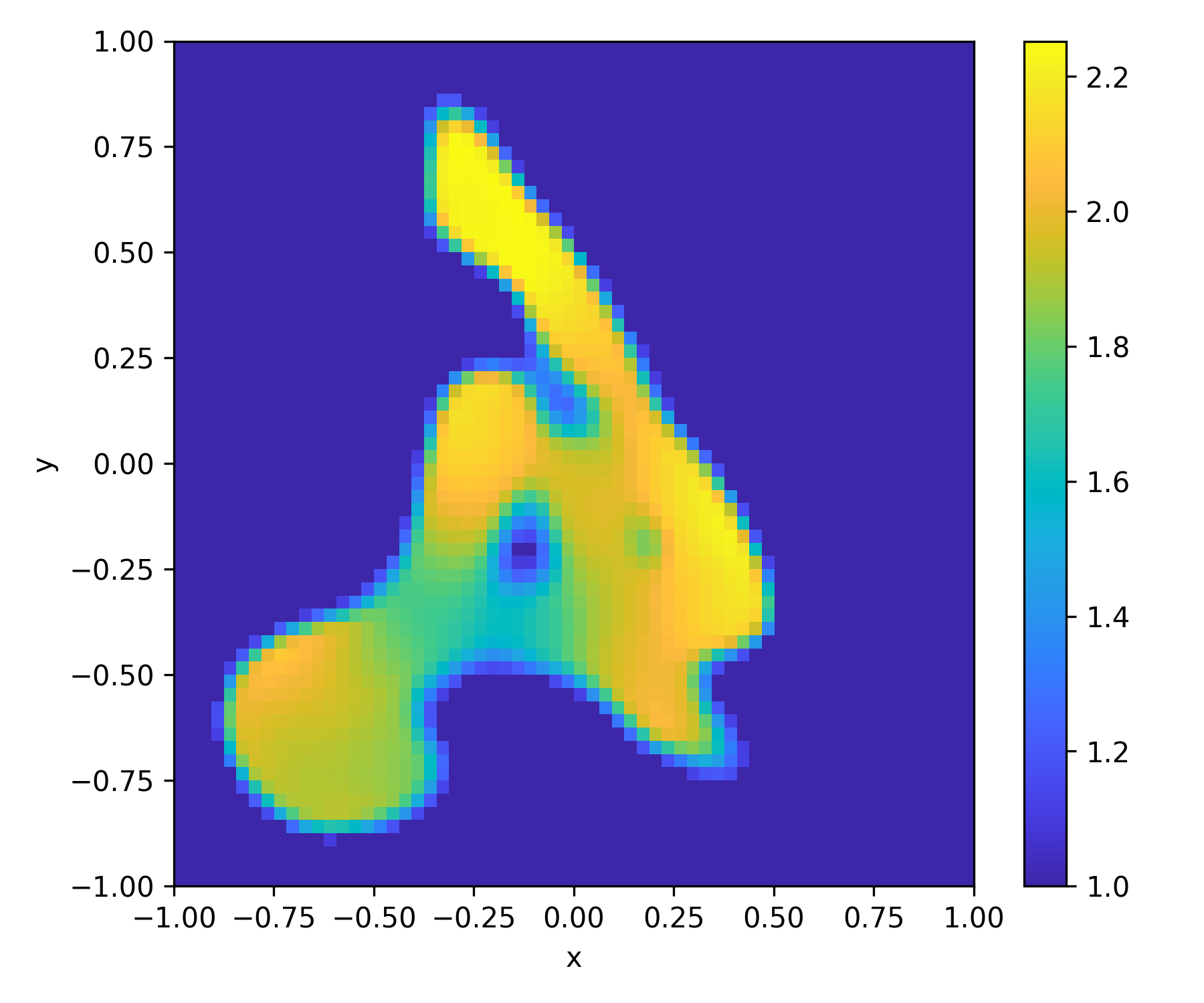}&
				\includegraphics[width=0.15\textwidth]{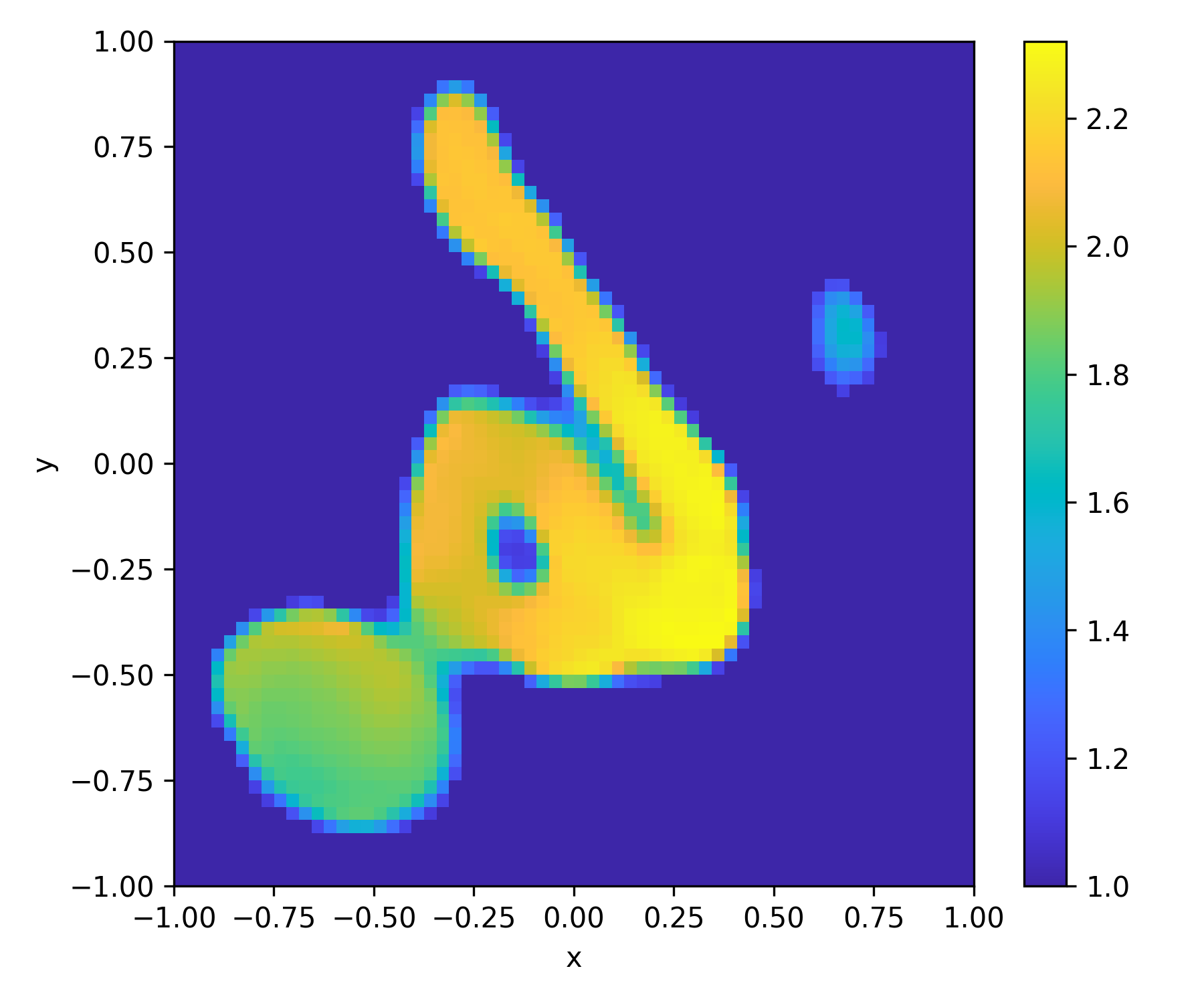}&
				\includegraphics[width=0.15\textwidth]{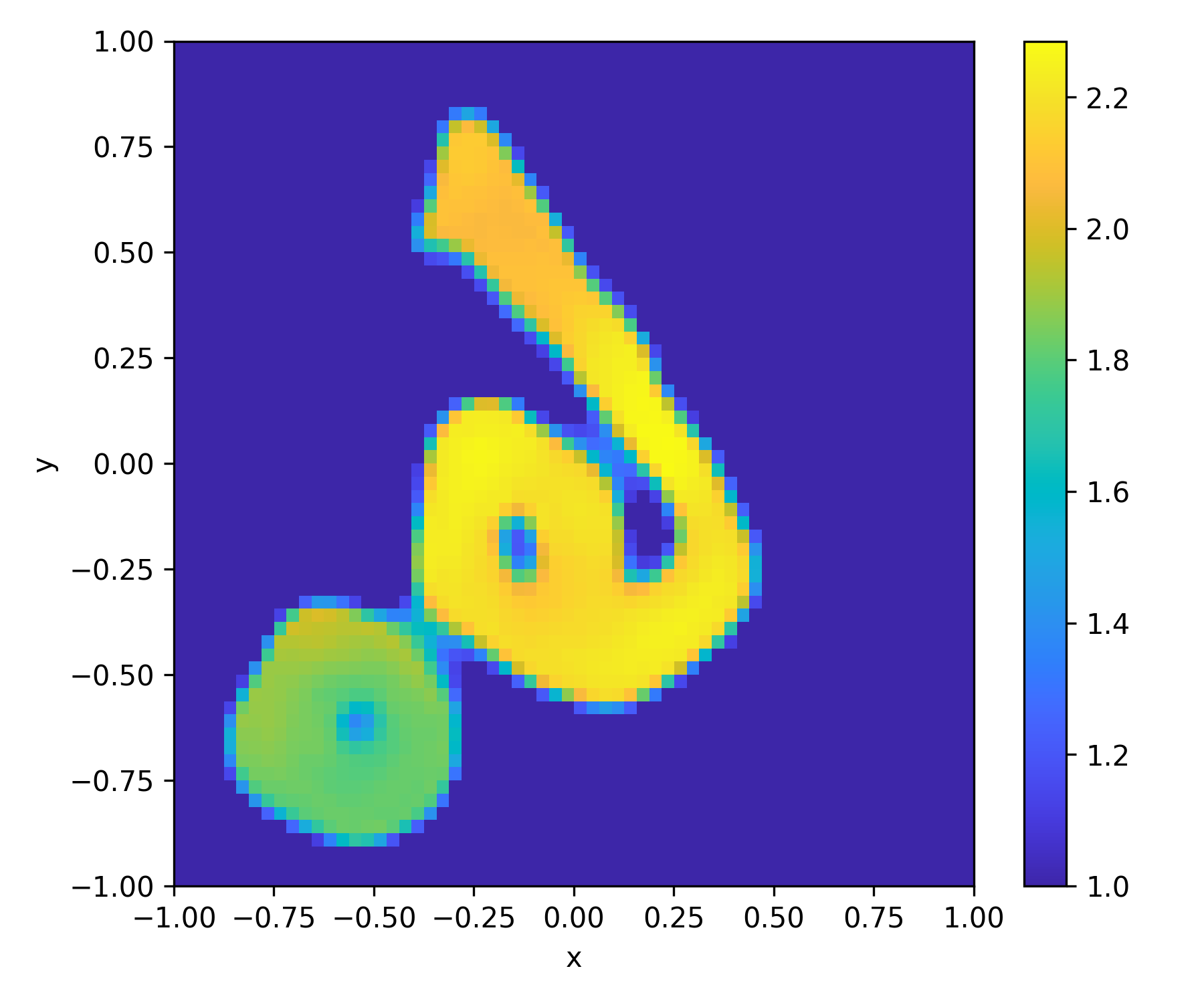}
				\\
				15\%&\SetCell[r=2]{c}
				\includegraphics[width=0.15\textwidth]{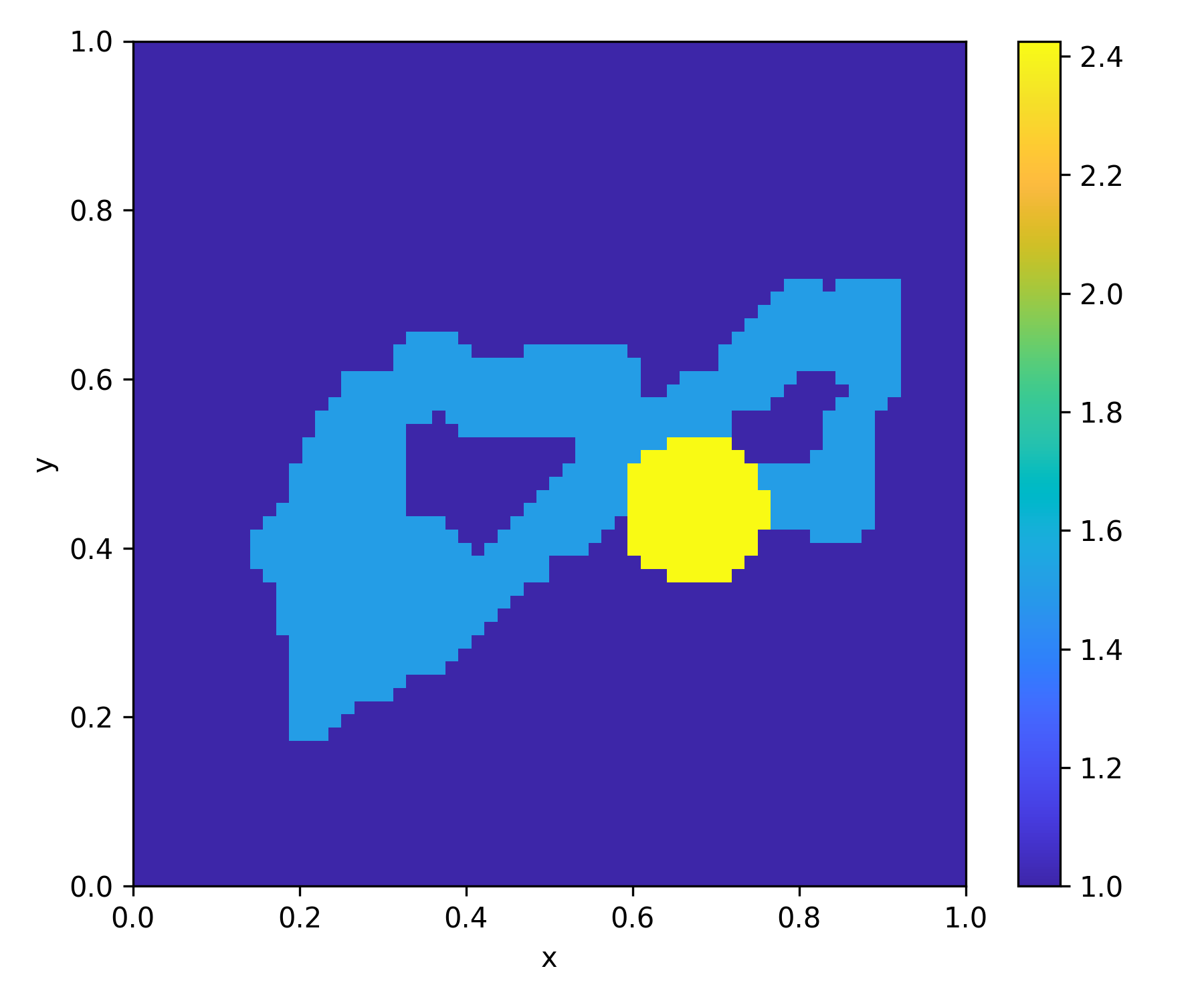} &
				\includegraphics[width=0.15\textwidth]{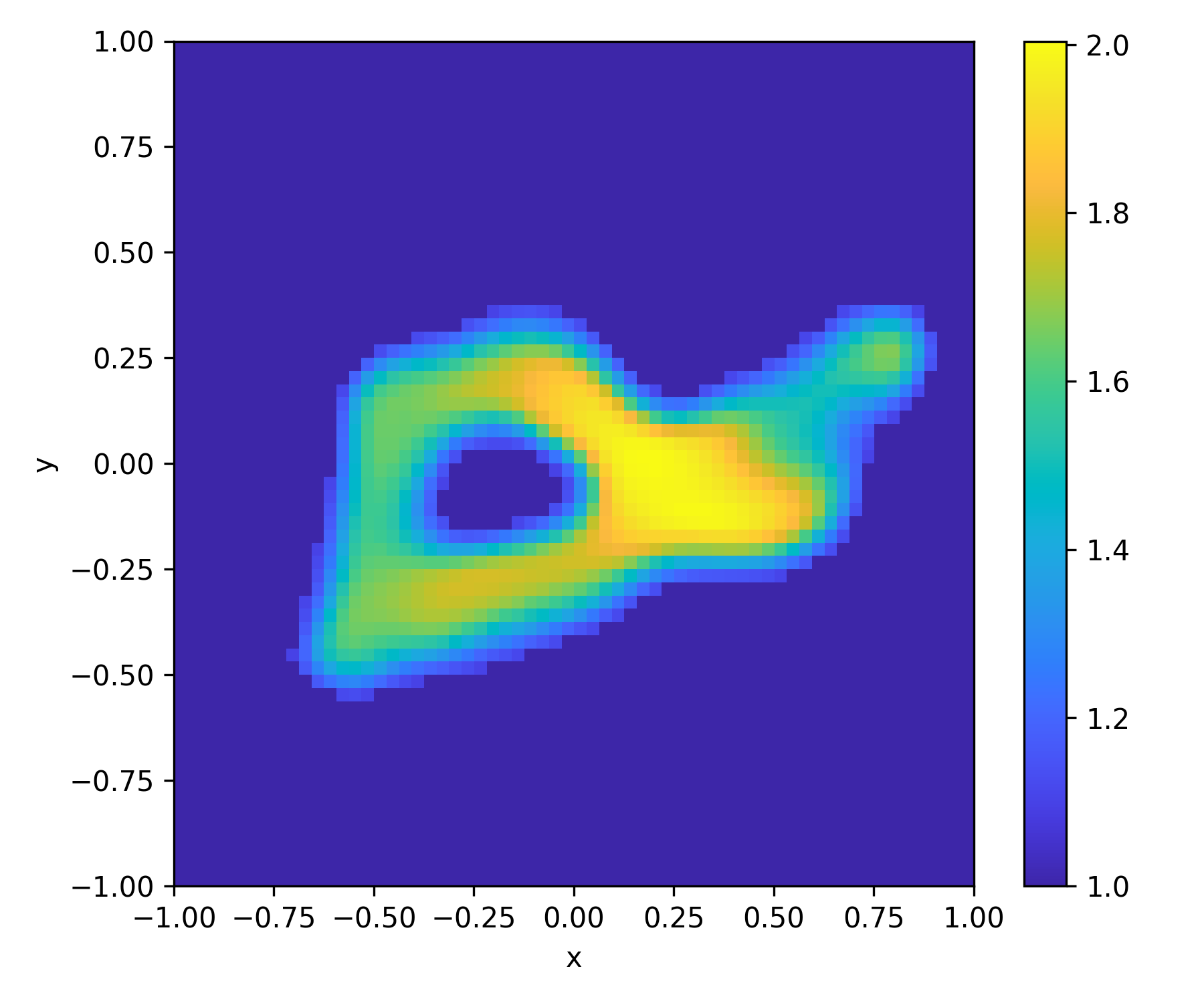}&
				\includegraphics[width=0.15\textwidth]{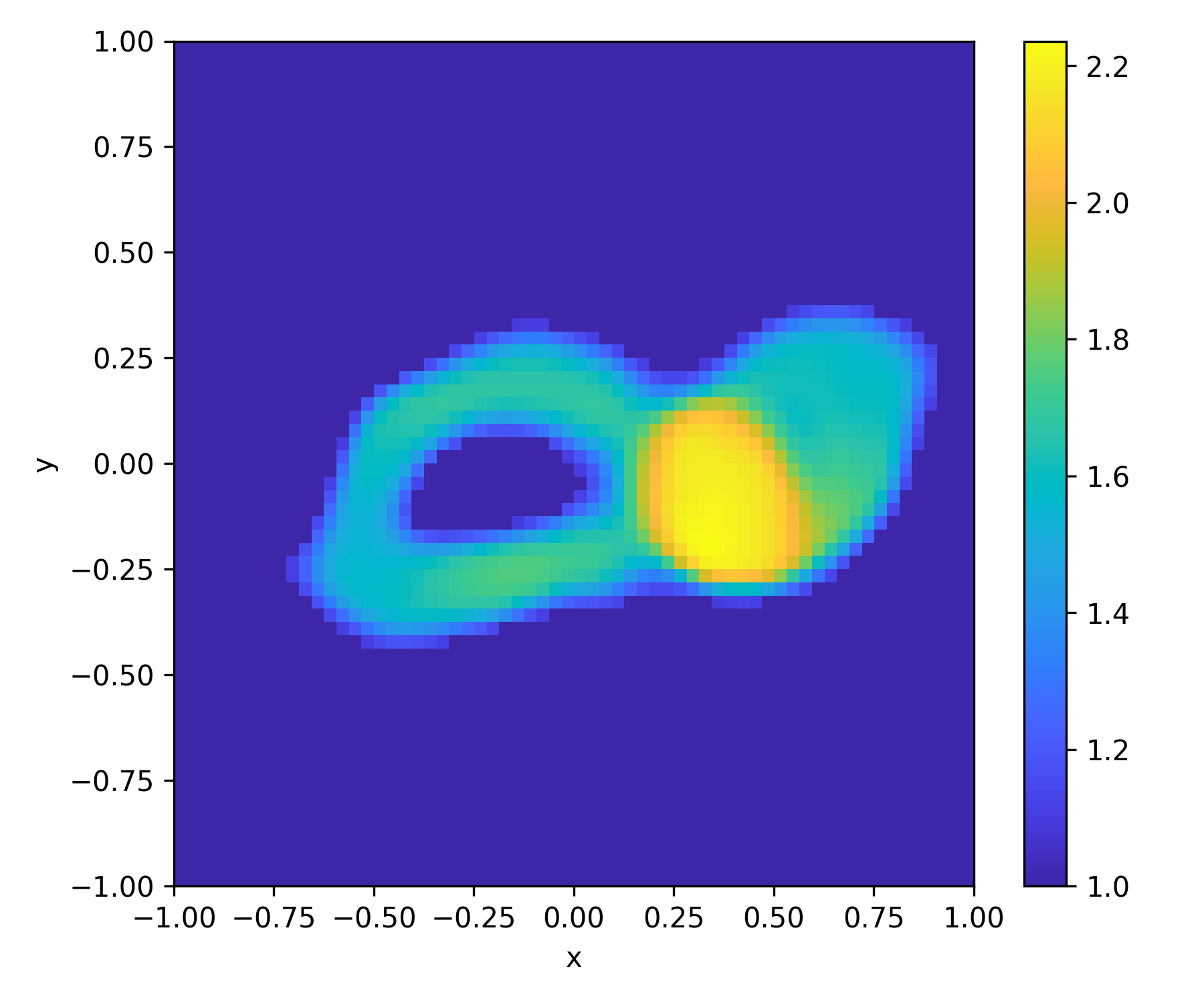}&
				\includegraphics[width=0.15\textwidth]{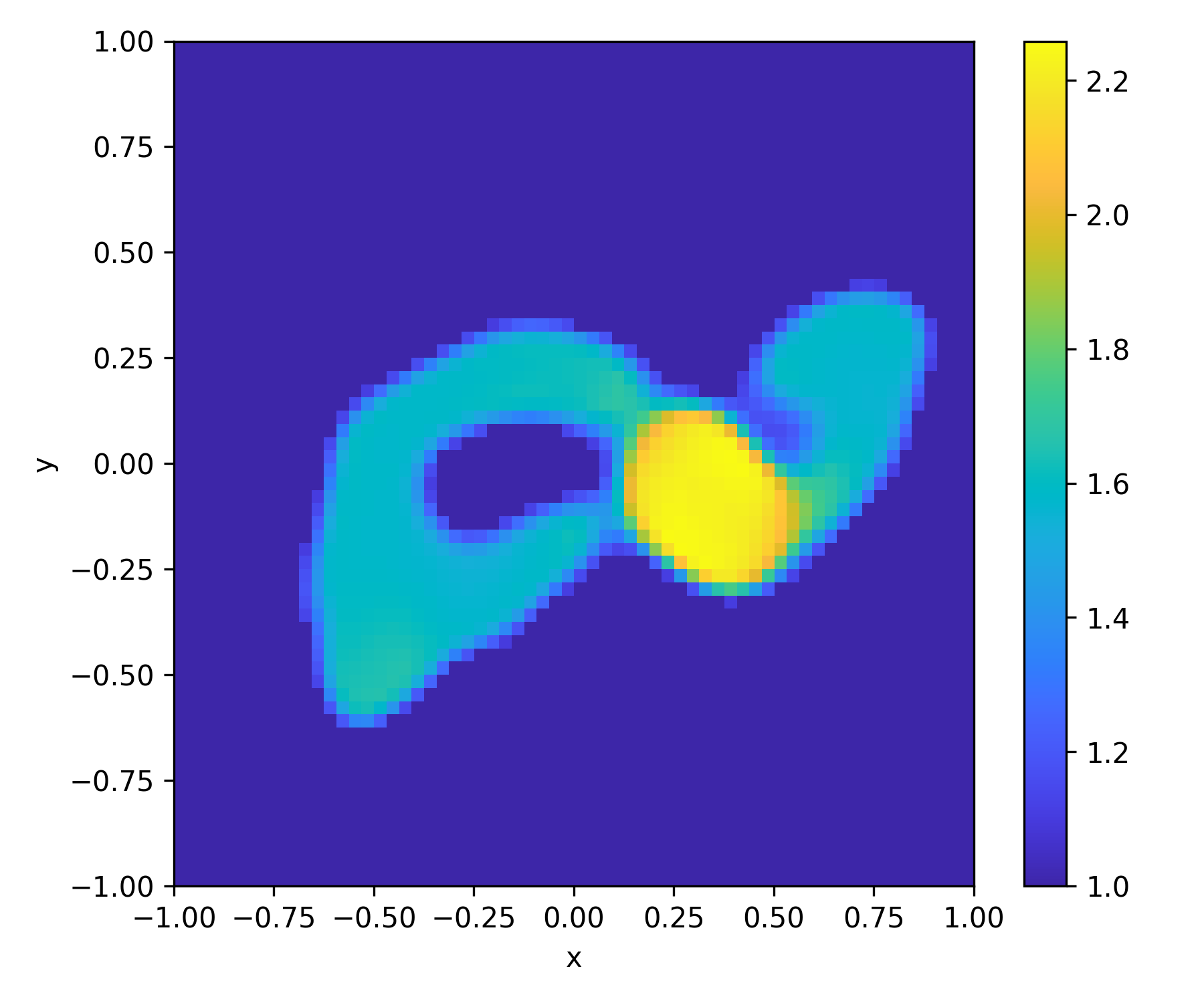}&
				\includegraphics[width=0.15\textwidth]{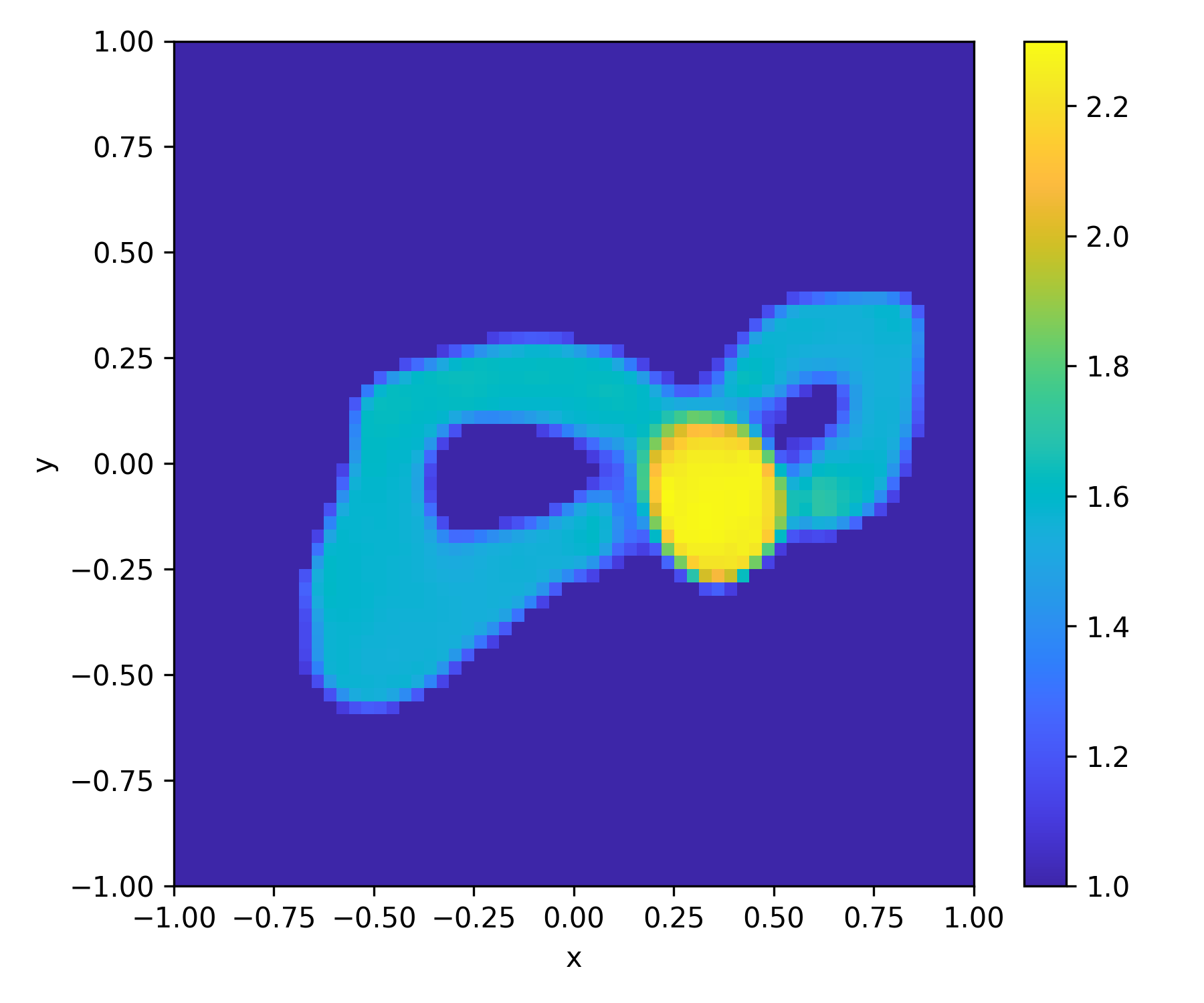}&
				\includegraphics[width=0.15\textwidth]{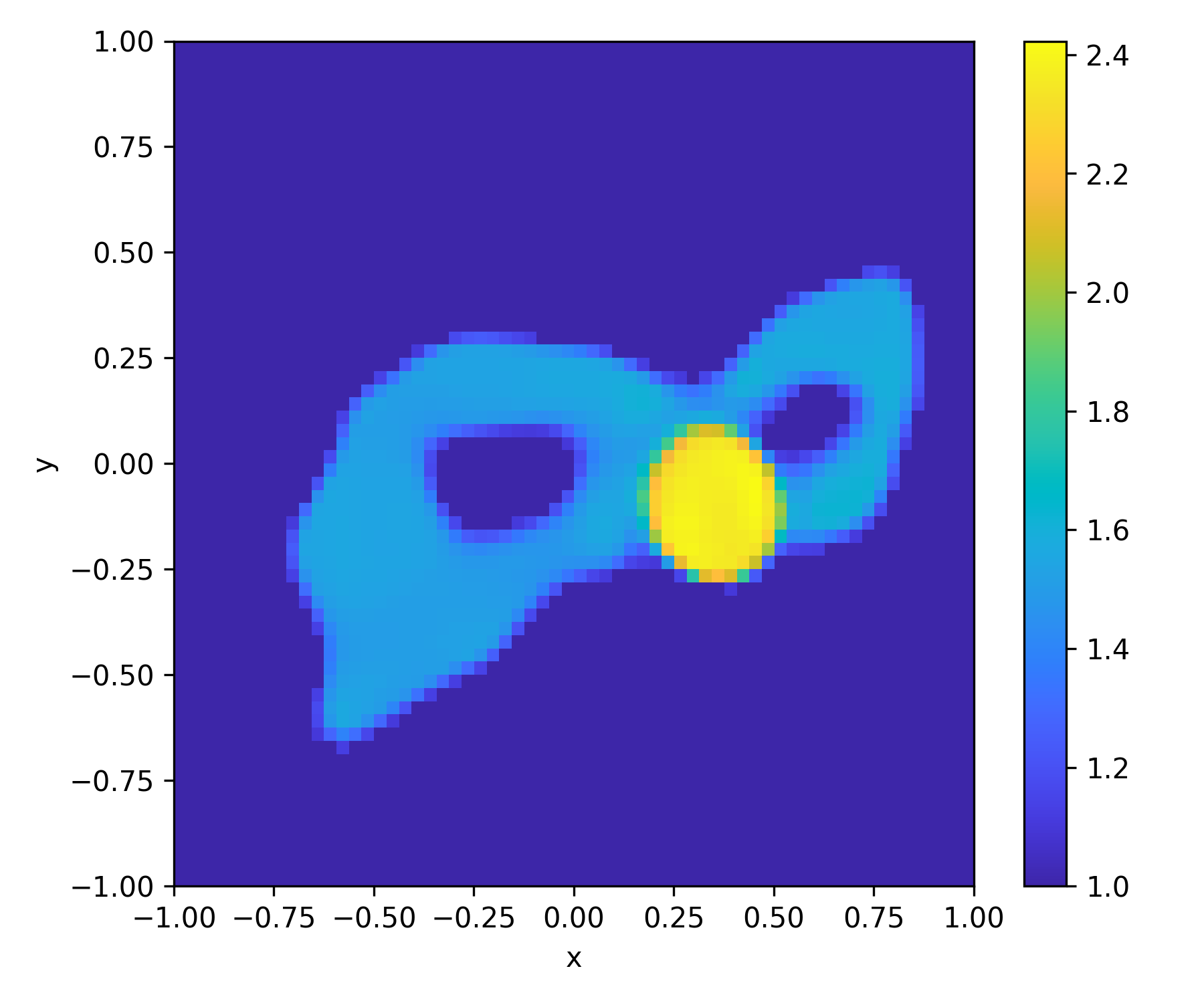}
				\\
				40\%& &
				\includegraphics[width=0.15\textwidth]{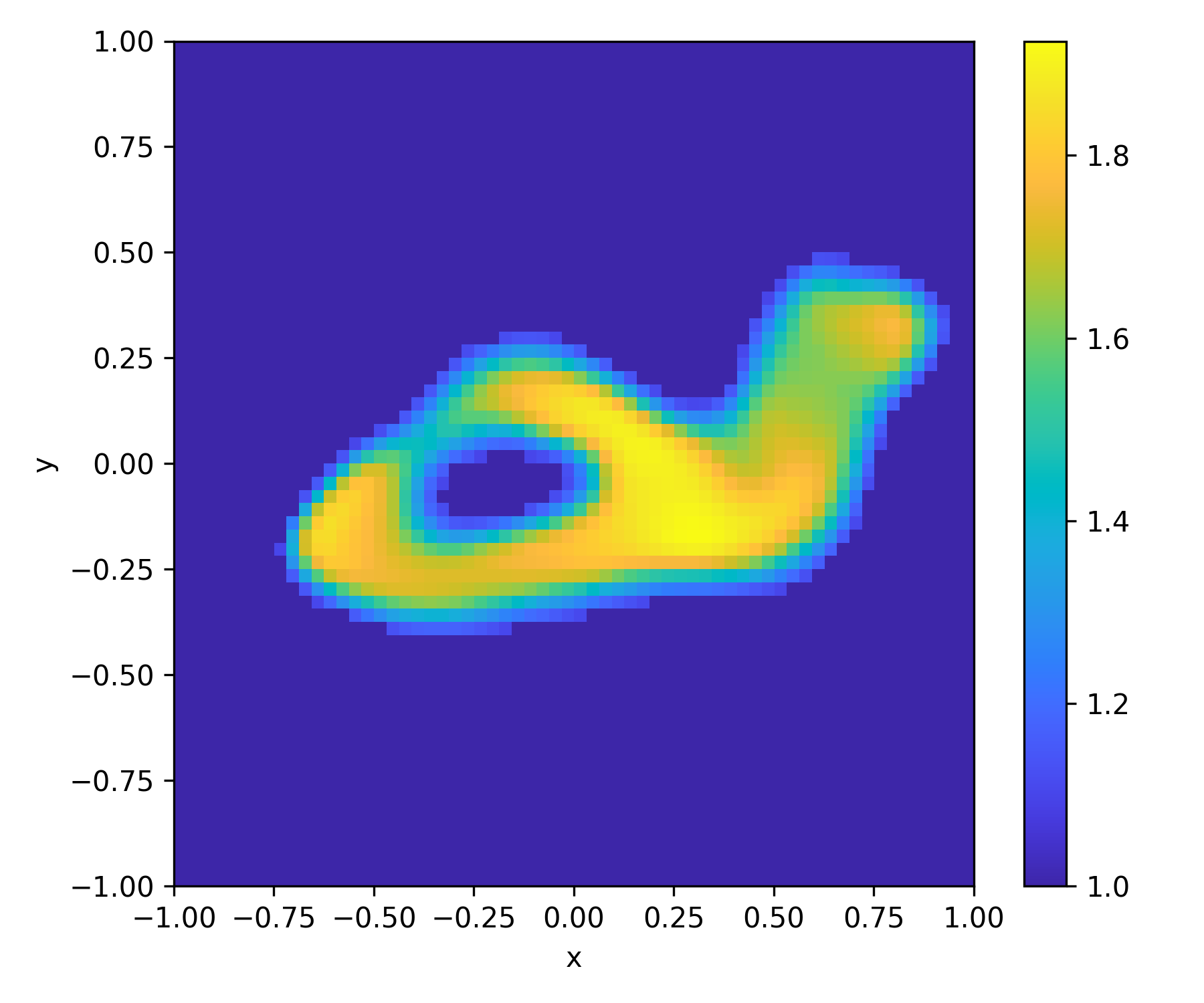}&
				\includegraphics[width=0.15\textwidth]{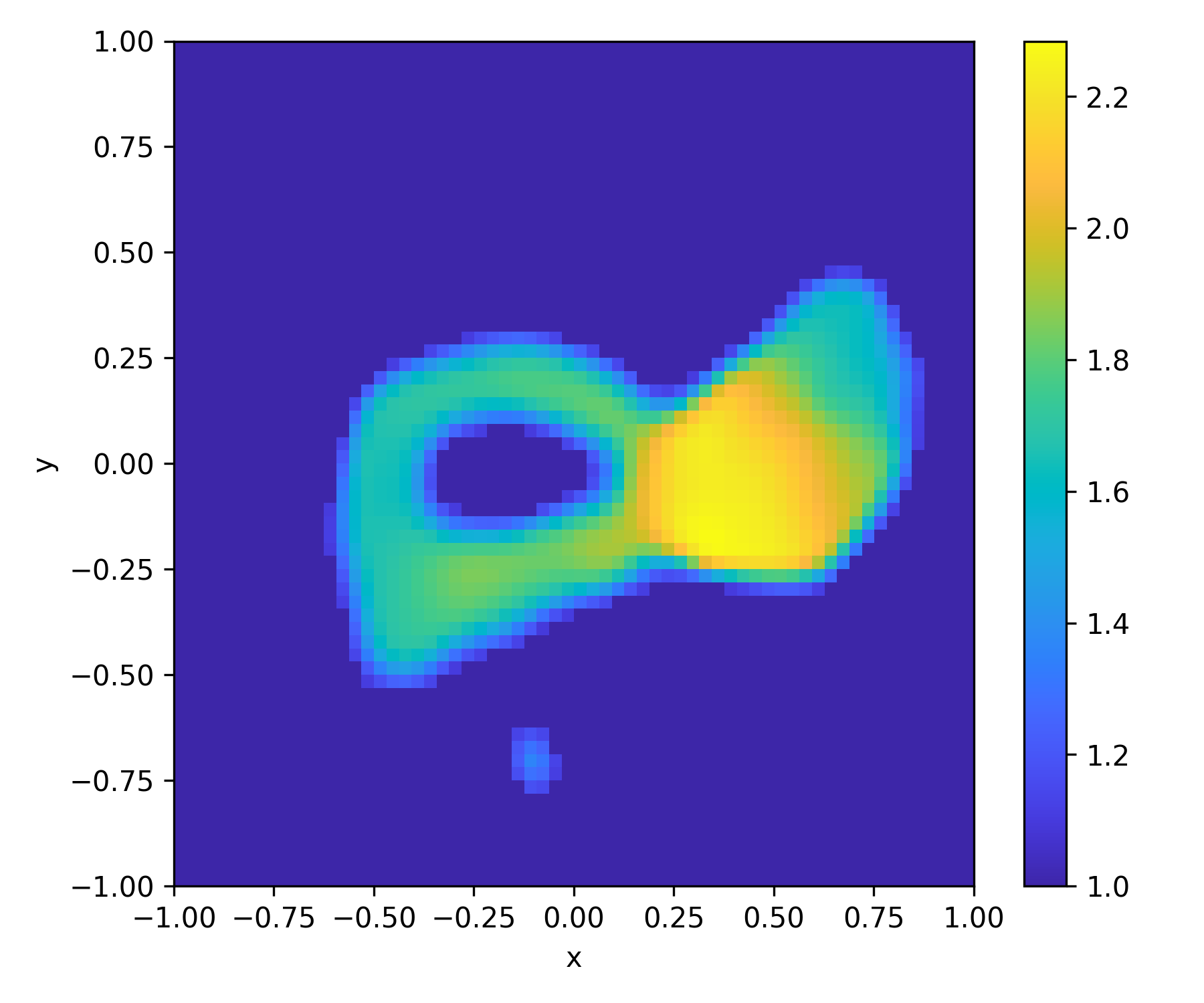}&
				\includegraphics[width=0.15\textwidth]{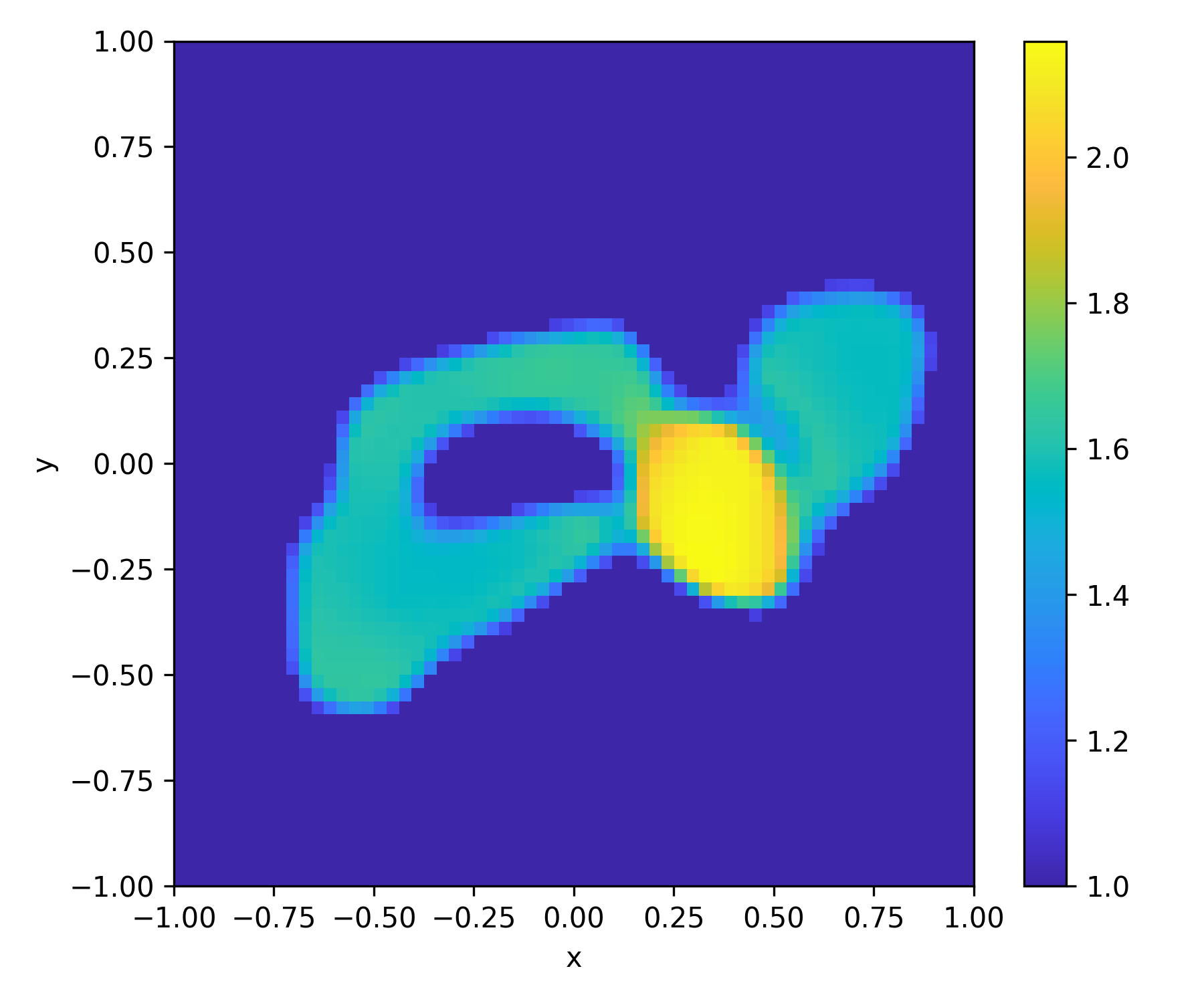}&
				\includegraphics[width=0.15\textwidth]{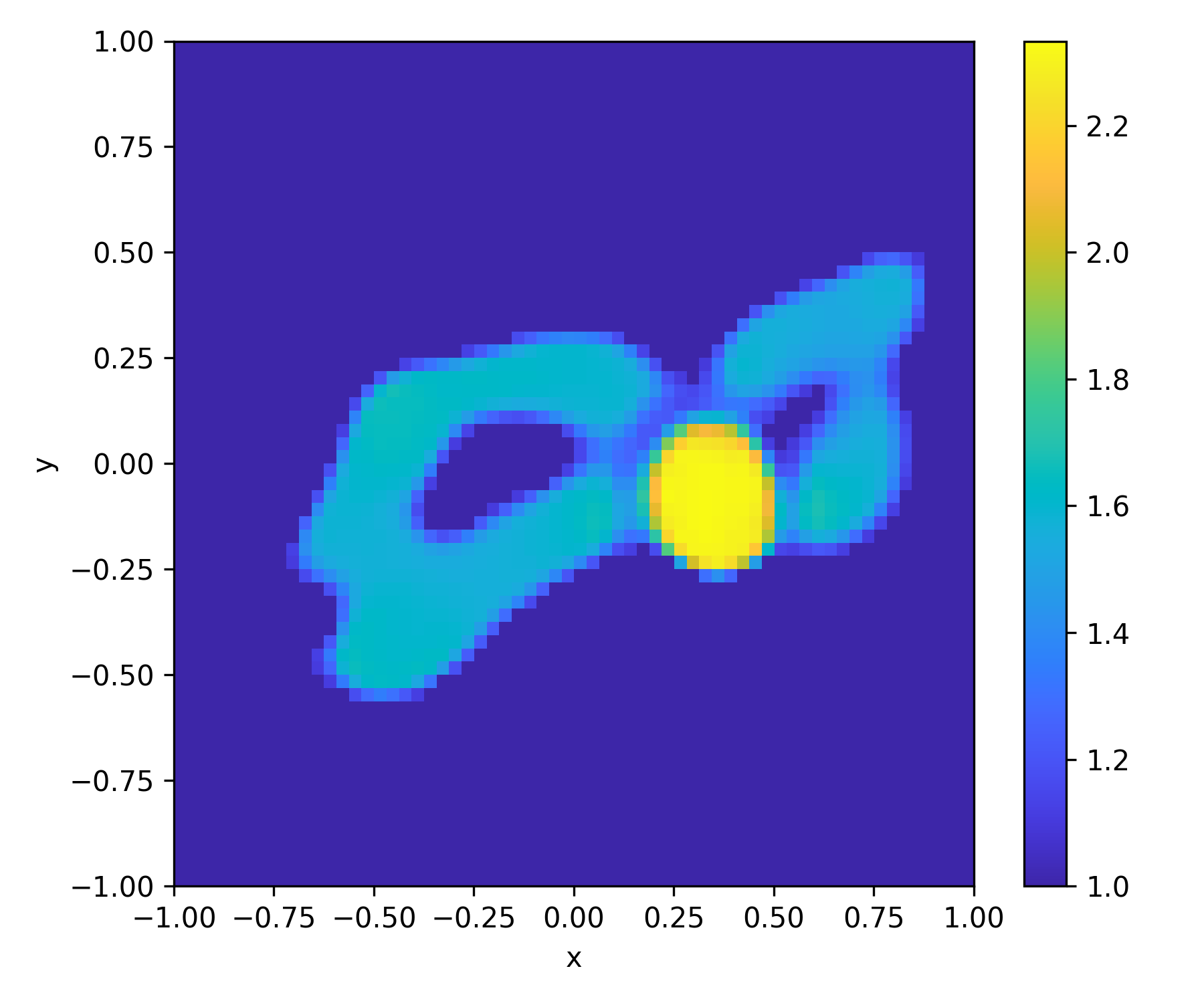}&
				\includegraphics[width=0.15\textwidth]{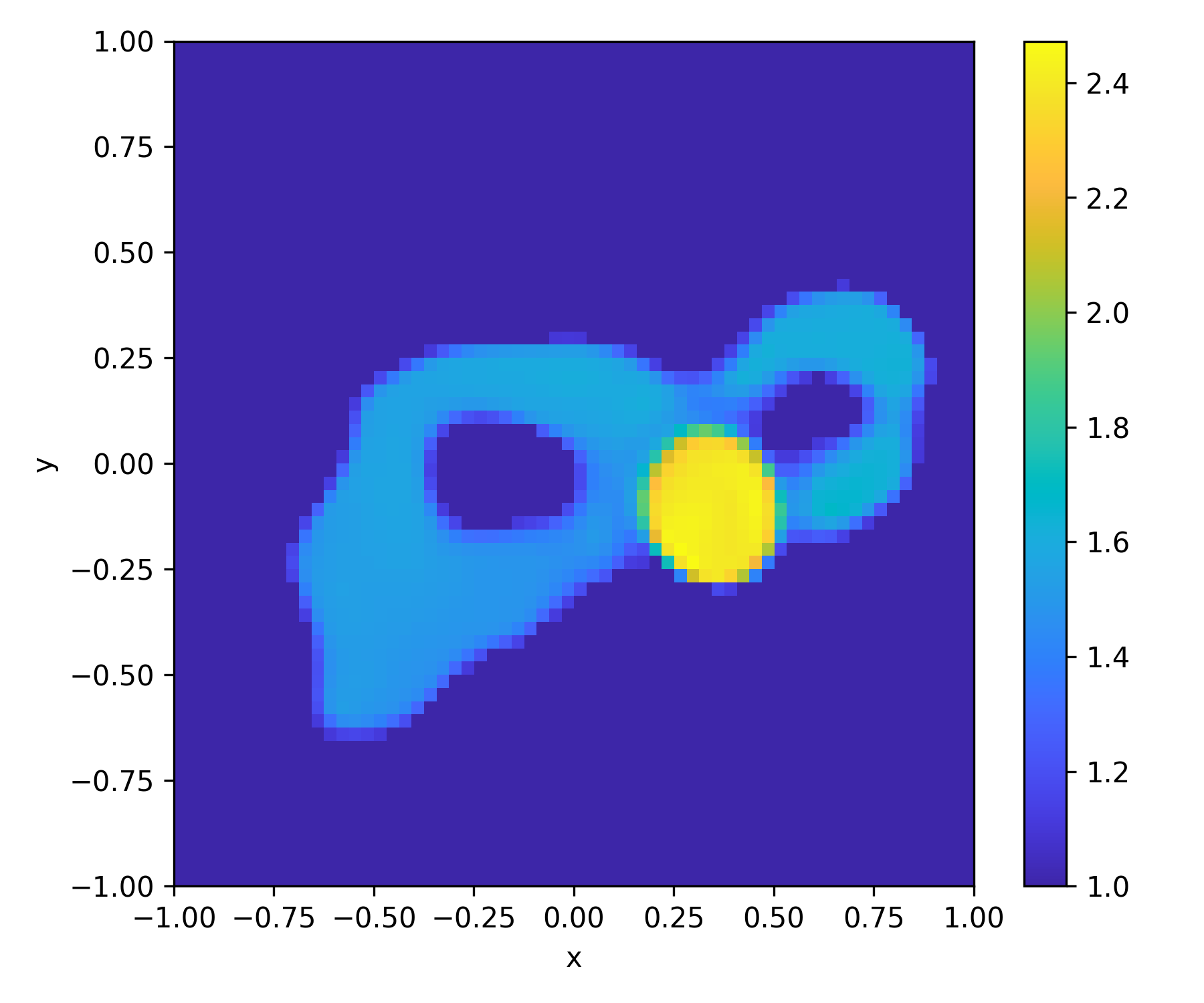}
			\end{tblr}
			\caption{Reconstructed images from scattered fields with $15\%$ and $40\%$  Gaussian noises by using the networks trained by the MNIST dataset, where the relative permittivity is between 1.5 and 2.5. From left to right: the ground-truth images, the reconstruction with 1,2,4,8, and 16 incident fields.}
			\label{tab:fig-Mnist}
		\end{center}
	\end{figure}
    
\begin{figure}[htp]
	\centering
	\includegraphics[width=1.0\linewidth]{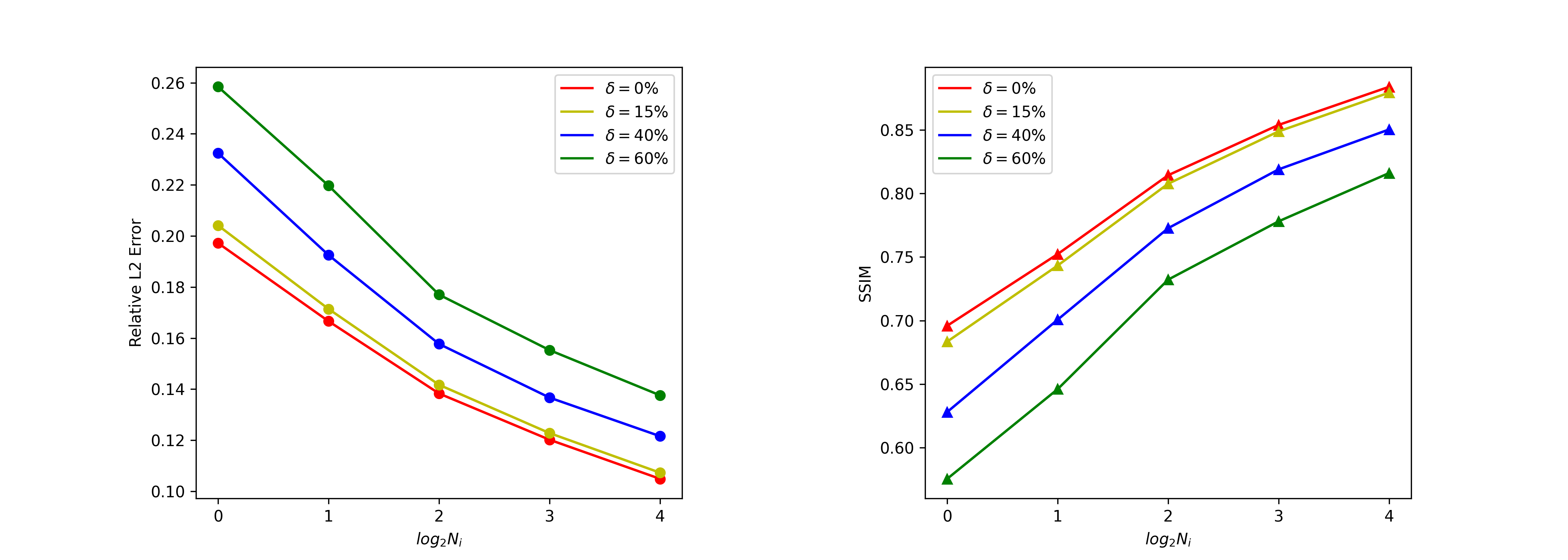}
	\caption{The relative L2 testing error and SSIM for the trained networks from the MNIST dataset with different noise levels and numbers of incidences.}
	\label{fig:Error_Mnist}
\end{figure}
	
	\subsubsection{Tests with “Austria ring”}
	
	We further test the generalization ability of the models using two different “Austria profiles” which are out of the distribution of the training parameters. For the first  “Austria profile”, the relative permittivity of the two small circles is set to be $3.0$ (Note that the range in the training data is from $1.5$ to $2.5$) and the relative permittivity of the ring is set to be $1.5$. For the second “Austria profile”, the relative permittivities of the three scatters are $1.5$, $2.0$ and $2.5$, respectively, while for each example in the training data, the scatterers only take two different relative permittivities. Thus, the two “Austria profiles” are out of the distribution of the training parameters, and are very challenging for the trained networks. The results are shown in Fig.\,\ref{tab:fig-Mnist-Austria}, which also show the benefits of using more incident fields. It can be also observed that when multiple scatterers exist, the DSM-DL can provide reasonable reconstruction for the strong scatterers (the yellow small circles), while it is distorted for weak scatterers with small $N_i$. This is mainly because the strong scatterers contribute to the dominant part of the scattered field. In addition, by comparing Fig.\,\ref{tab:fig-Austria_Circle} and Fig.\,\ref{tab:fig-Mnist-Austria}, we can see that the trained models from the MNIST dataset have better generalization than the trained models from the circle dataset; this is reasonable since the MNIST dataset is more diverse and we use more training data for the MNIST example.  
	
	\begin{figure}[htp]\small
		\begin{center}
			\begin{tblr}
				{colspec = {X[-1]X[c]X[c,h]X[c,h]X[c,h]X[c,h]X[c,h]},
					stretch = 0,
					rowsep = 0pt,}
				Noise Level& Ground truth& $N_{i}$=1& $N_{i}$=2 &$N_{i}$=4&$N_{i}$=8& $N_{i}$=16\\
				15\%&\SetCell[r=2]{c}
				\includegraphics[width=0.15\textwidth]{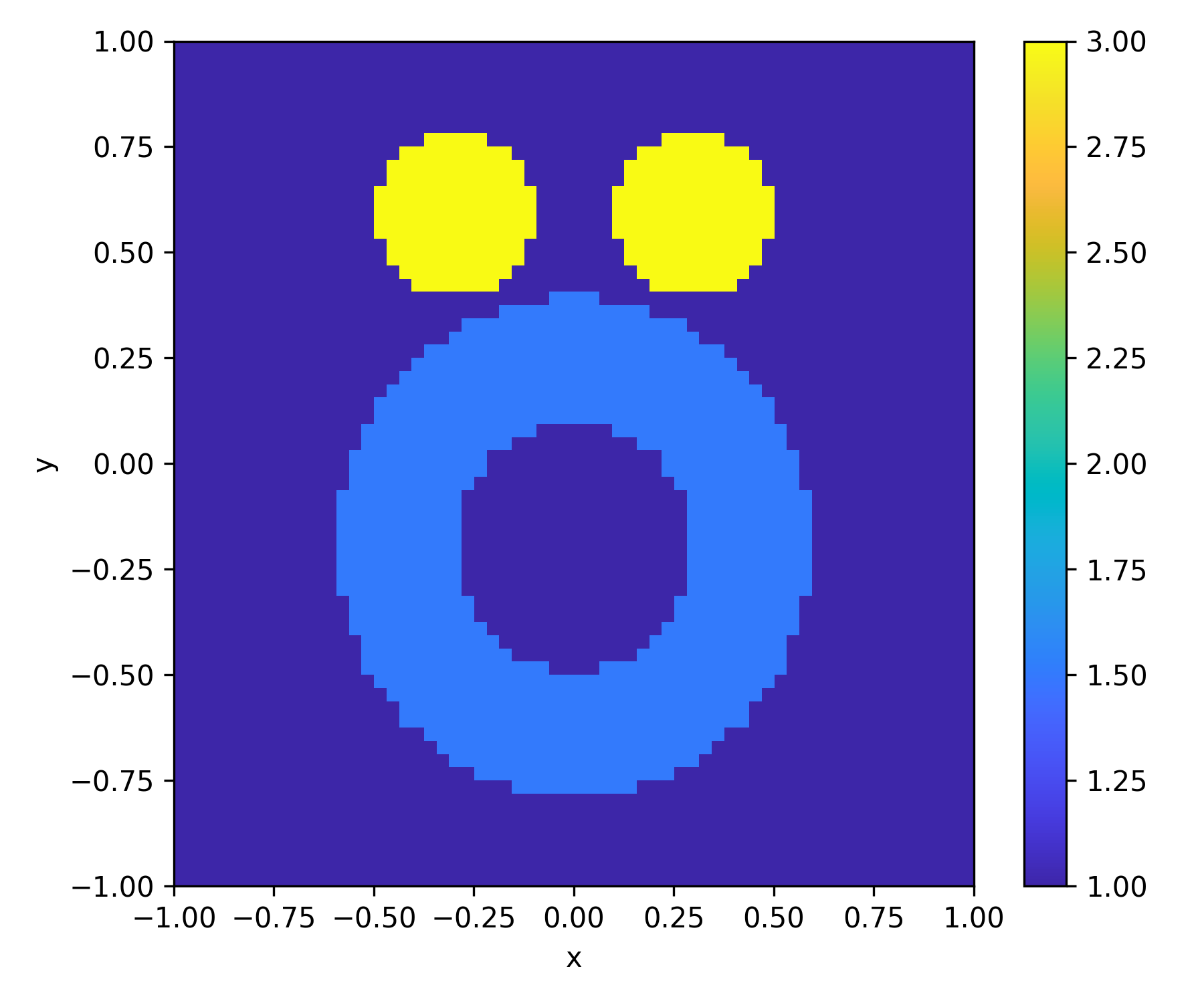}& 
				\includegraphics[width=0.15\textwidth]{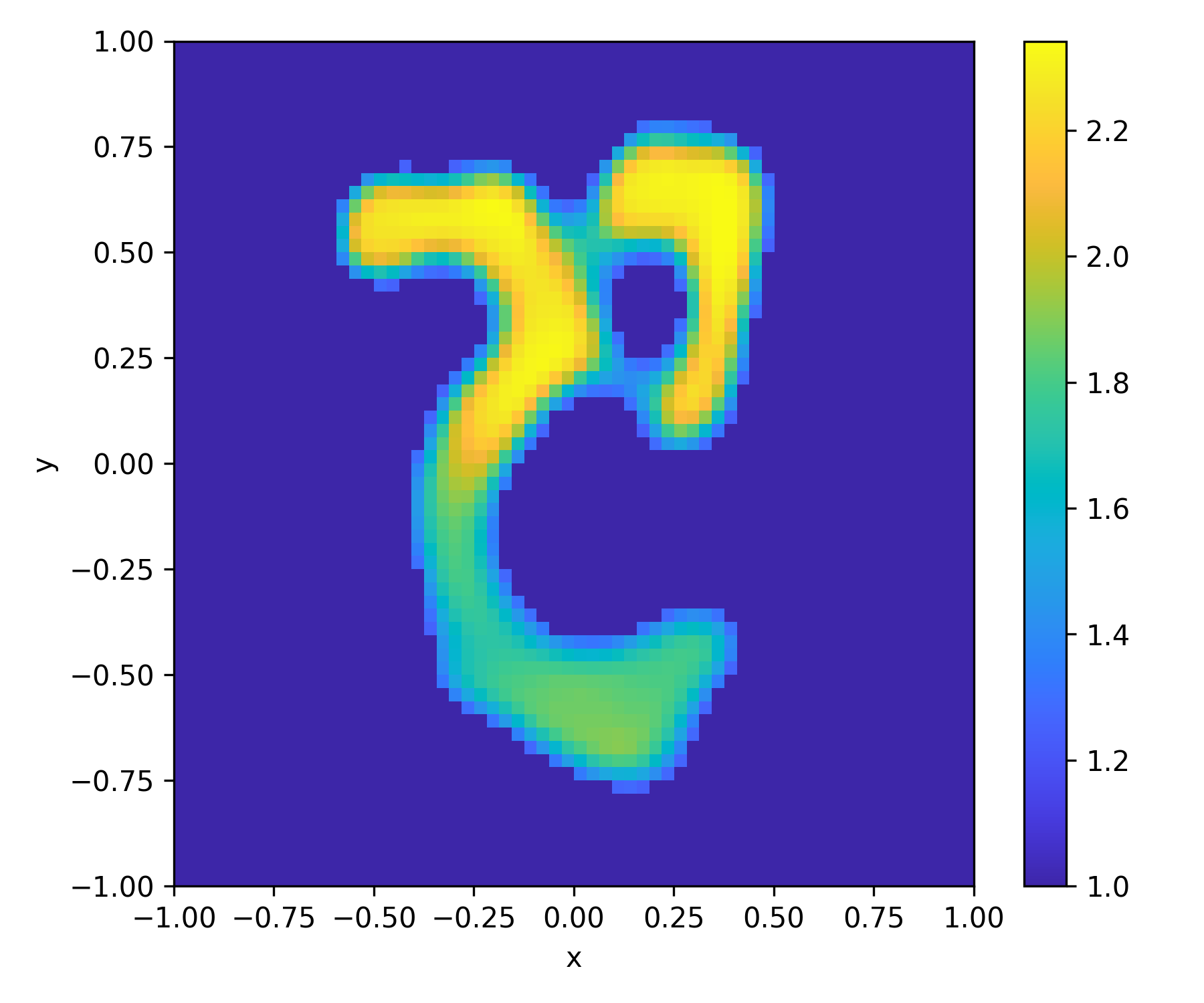}&
				\includegraphics[width=0.15\textwidth]{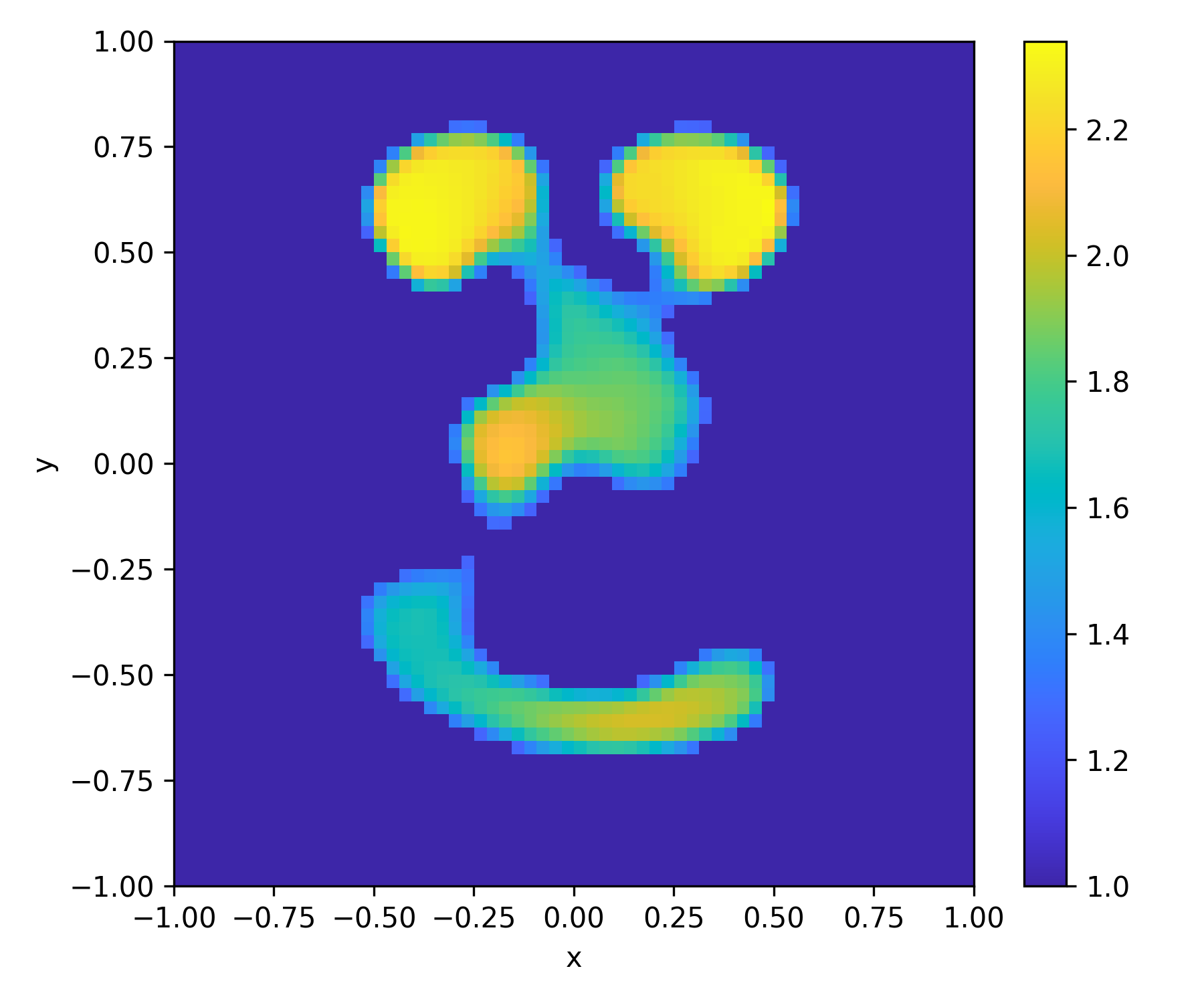}&
				\includegraphics[width=0.15\textwidth]{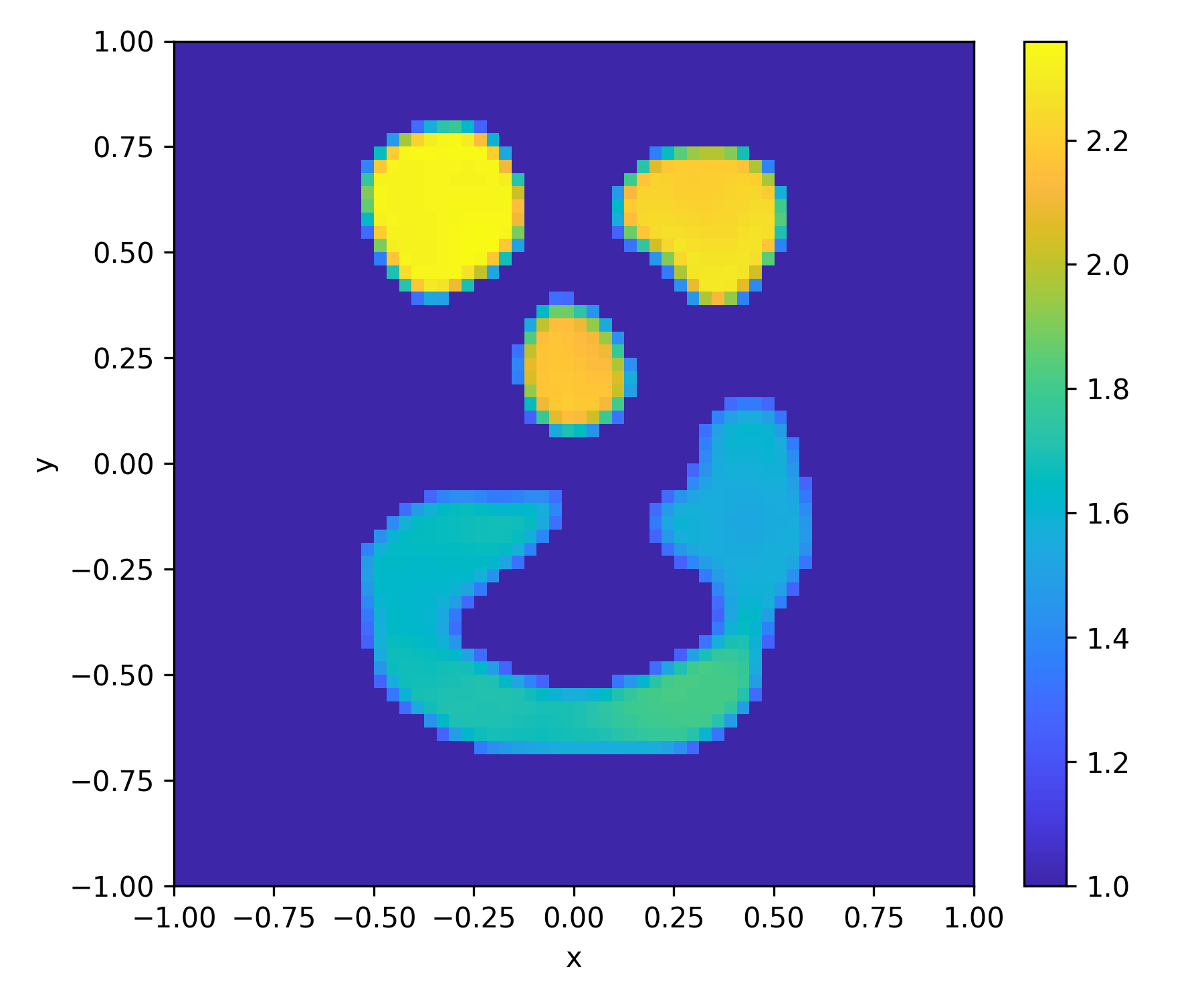}&
				\includegraphics[width=0.15\textwidth]{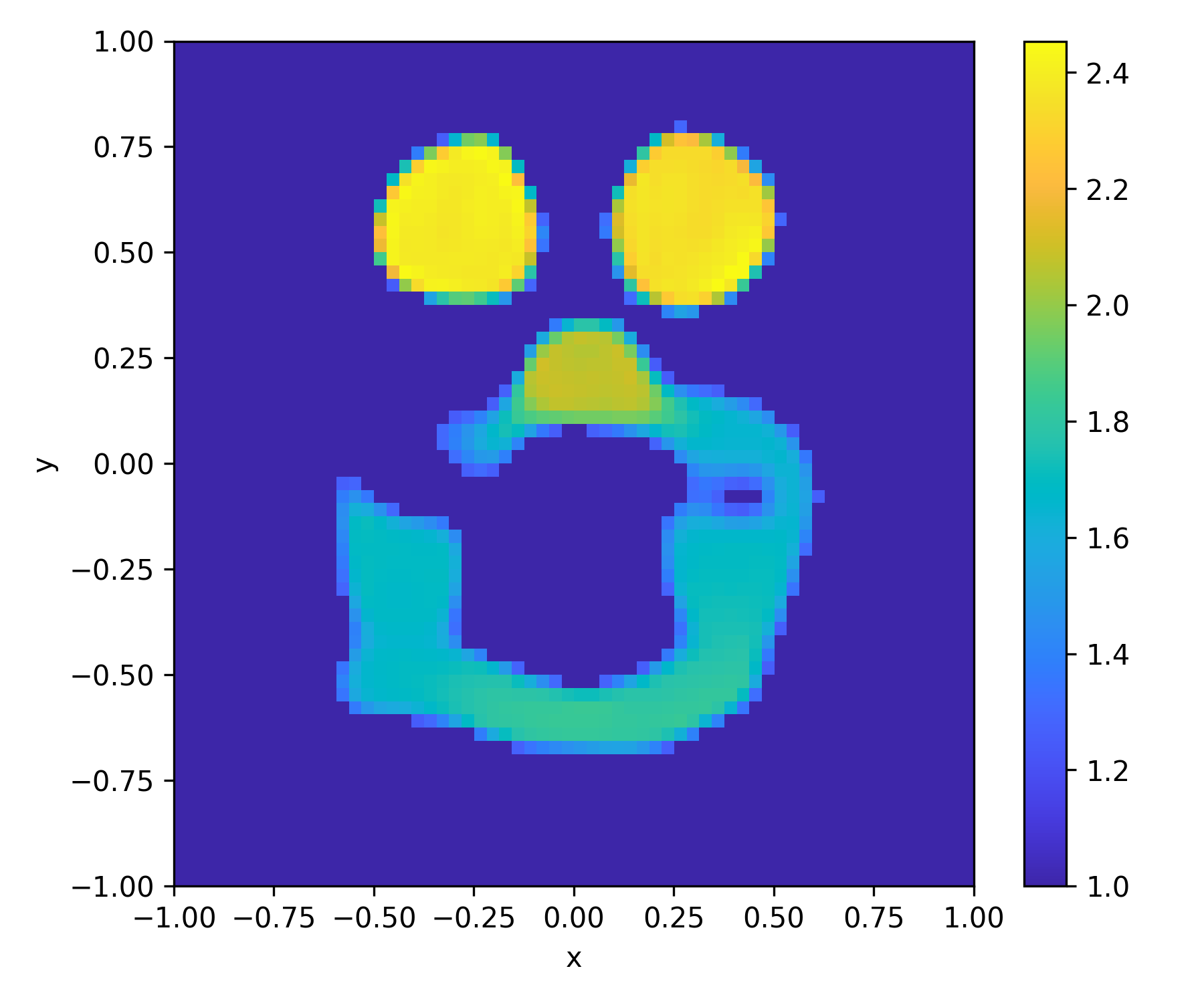}&
				\includegraphics[width=0.15\textwidth]{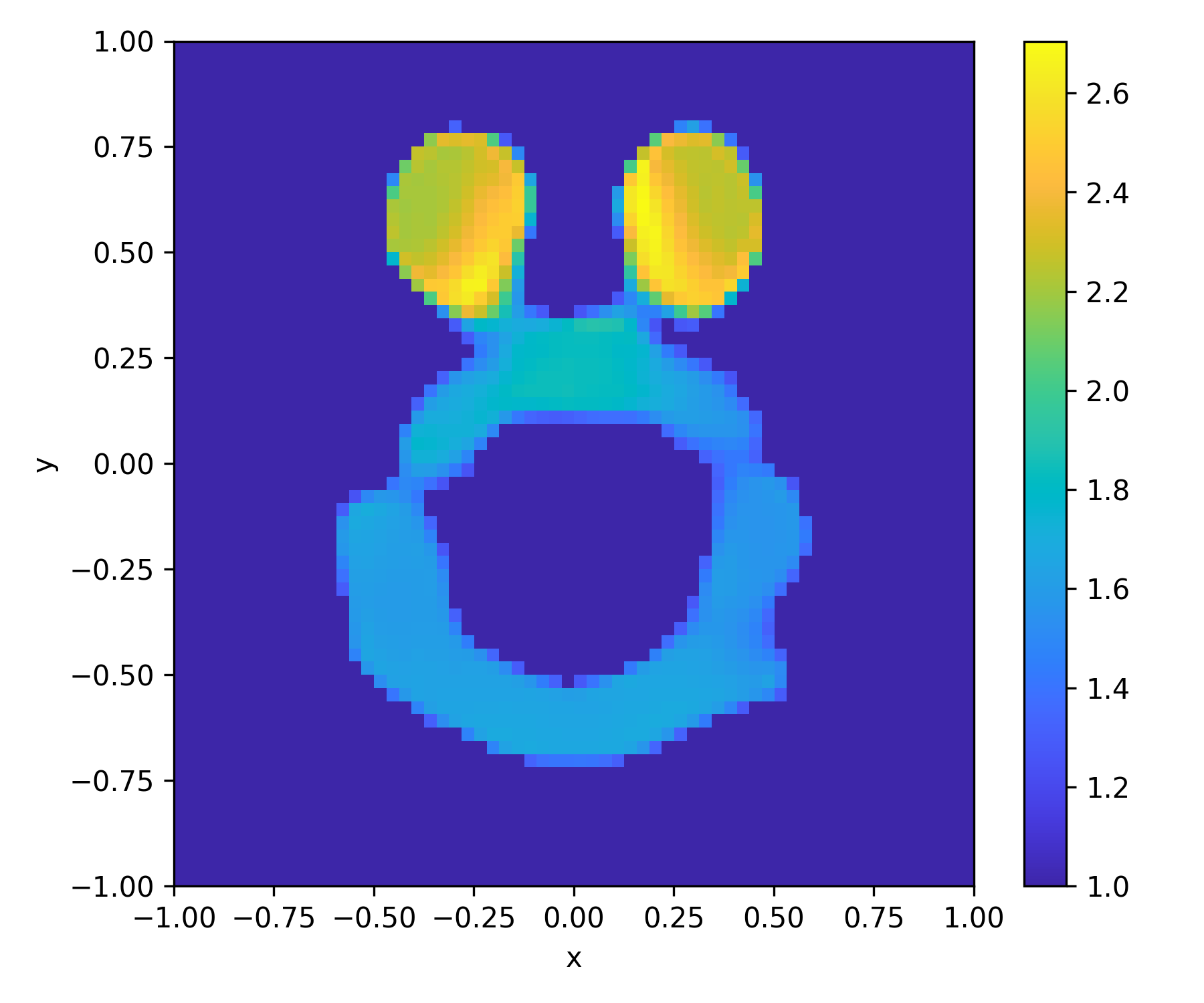}&
				\\
				40\%& & 
				\includegraphics[width=0.15\textwidth]{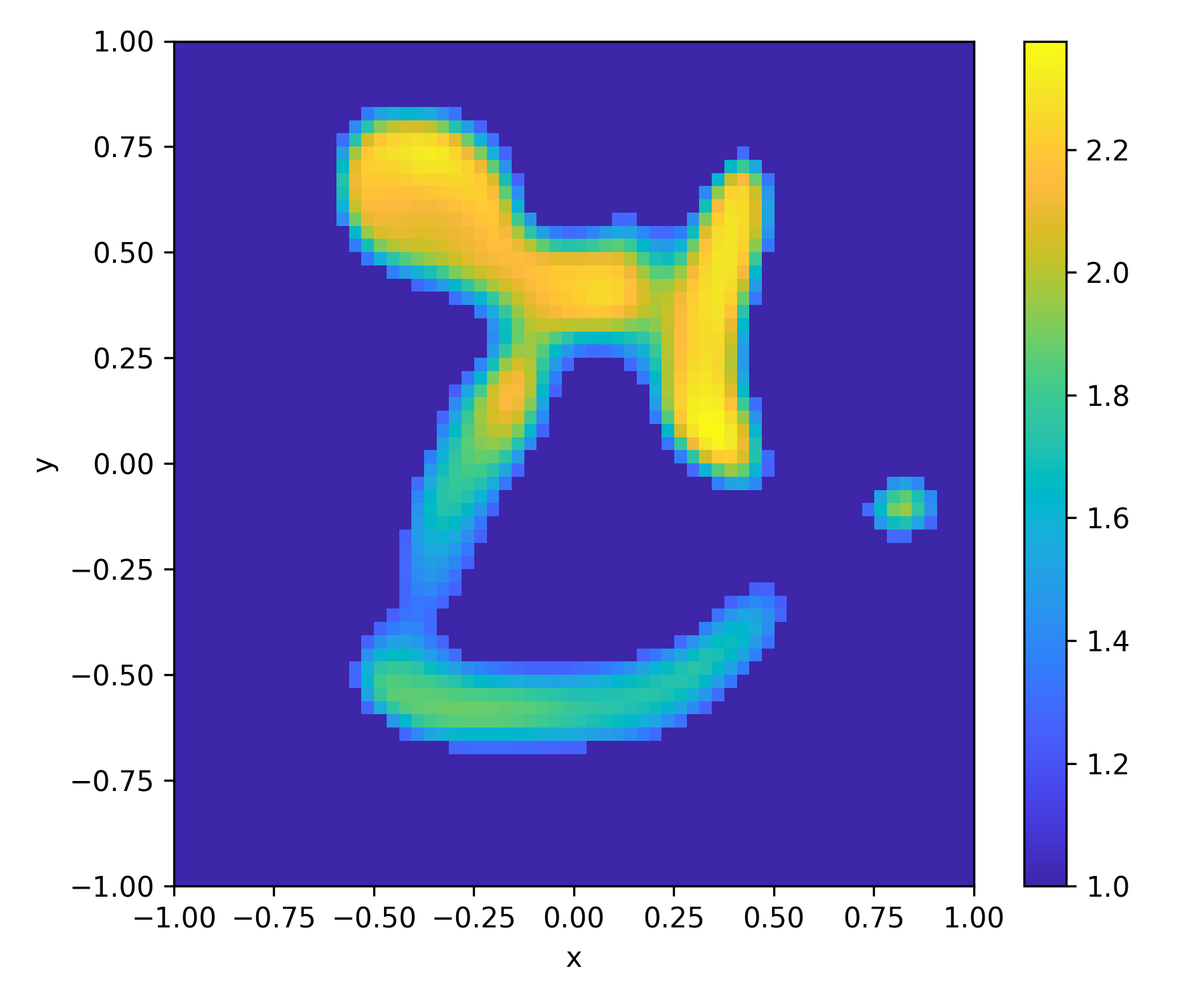}&
				\includegraphics[width=0.15\textwidth]{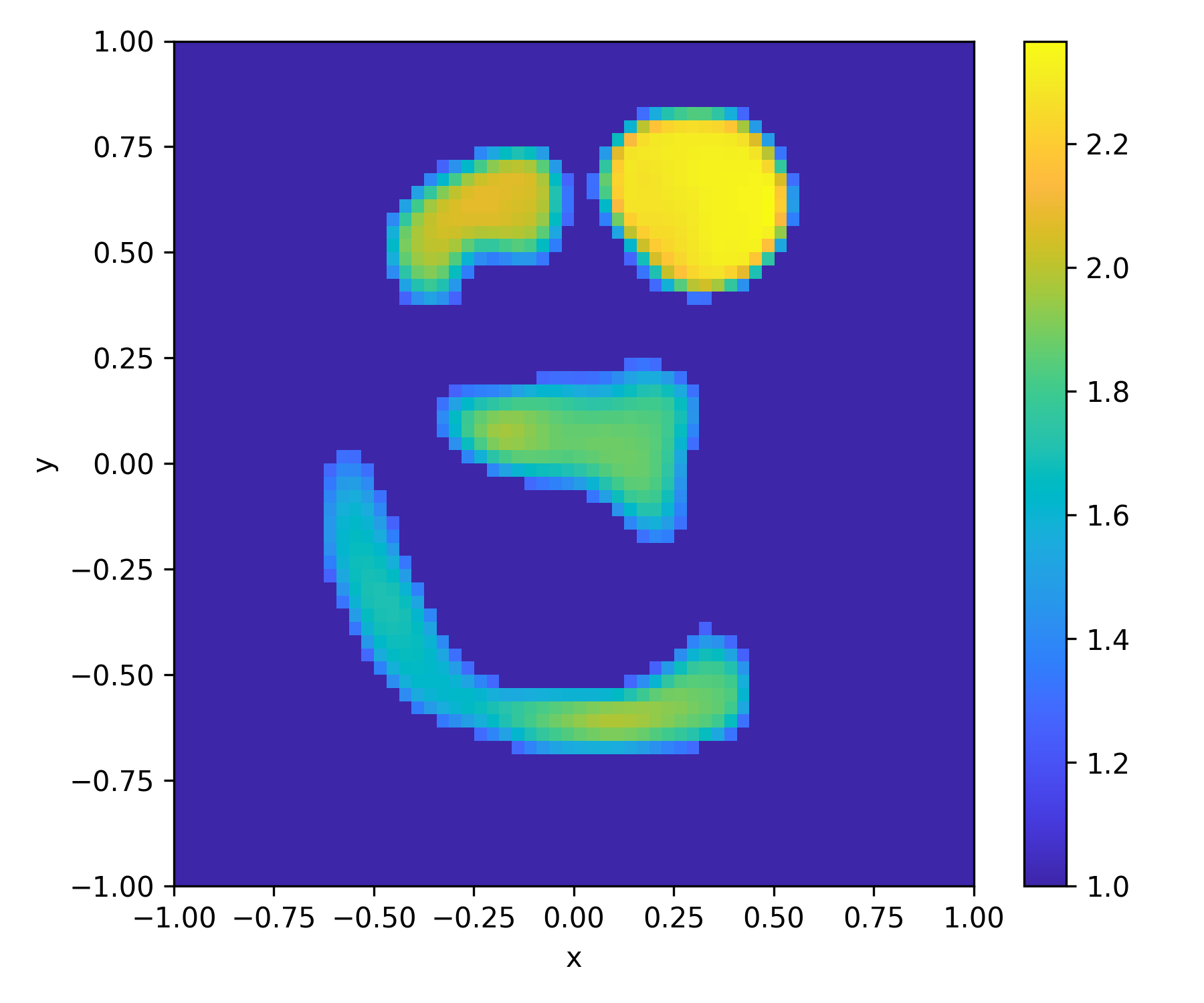}&
				\includegraphics[width=0.15\textwidth]{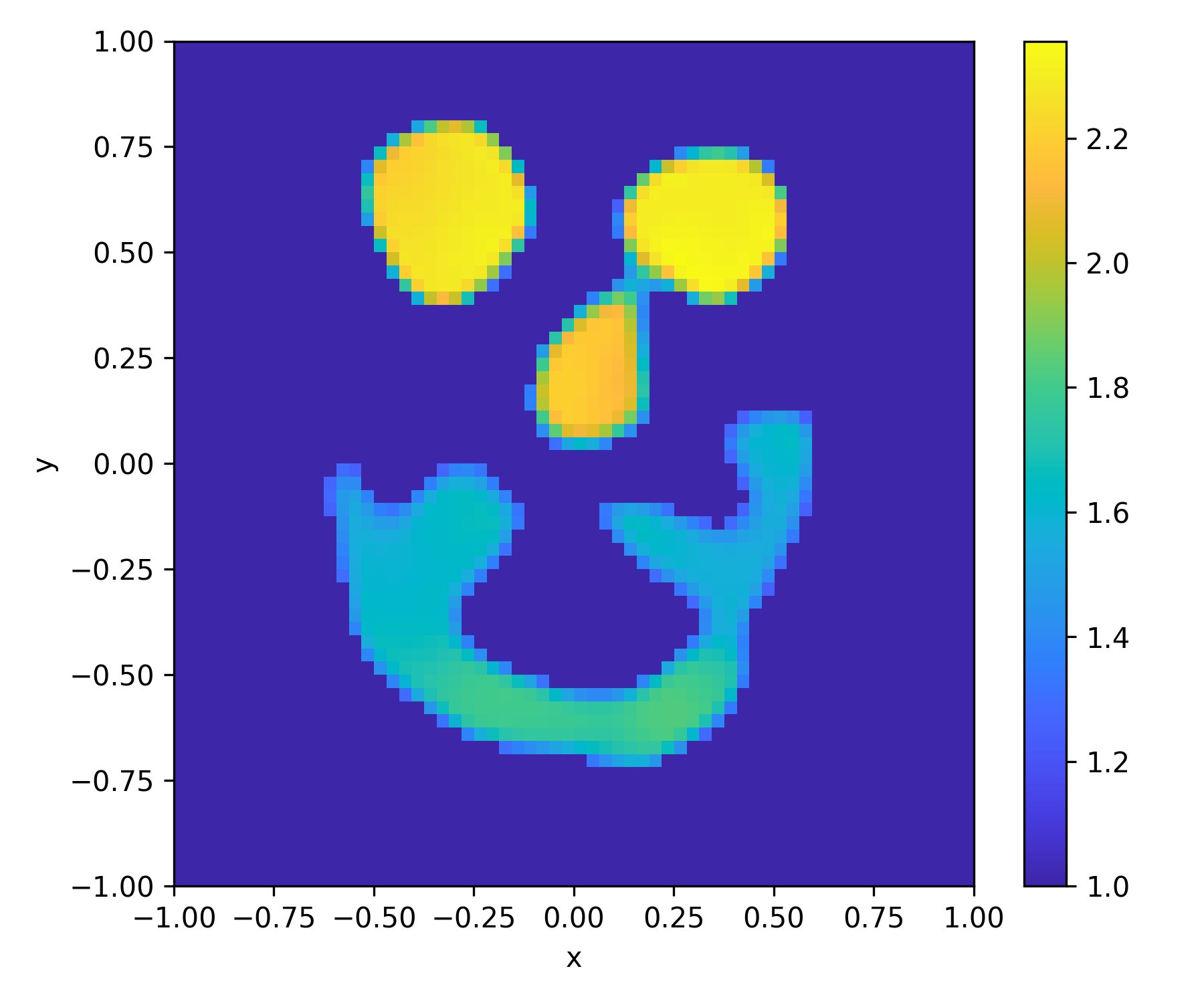}&
				\includegraphics[width=0.15\textwidth]{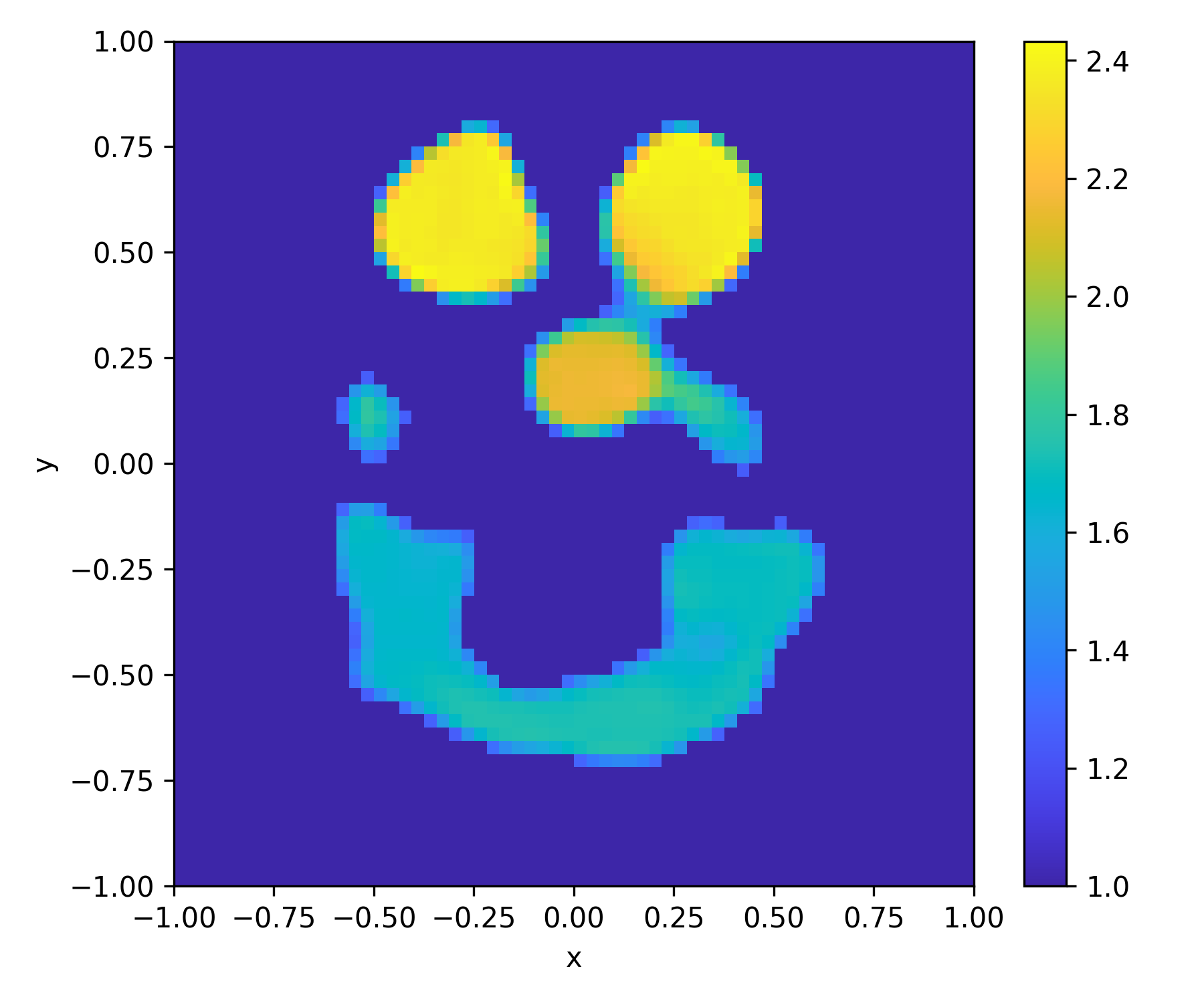}&
				\includegraphics[width=0.15\textwidth]{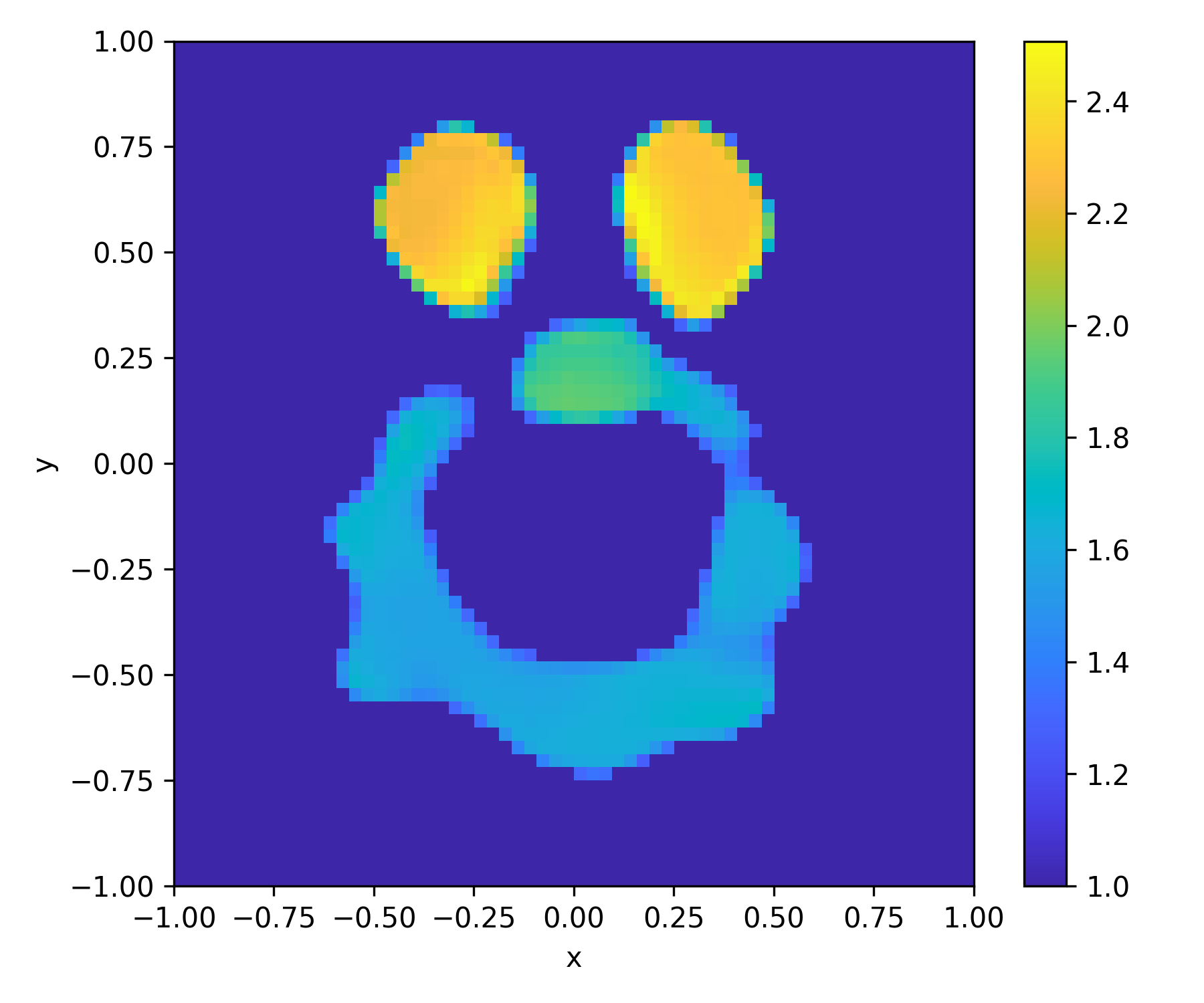}&
				\\
				
				15\%&\SetCell[r=2]{c}\includegraphics[width=0.15\textwidth]{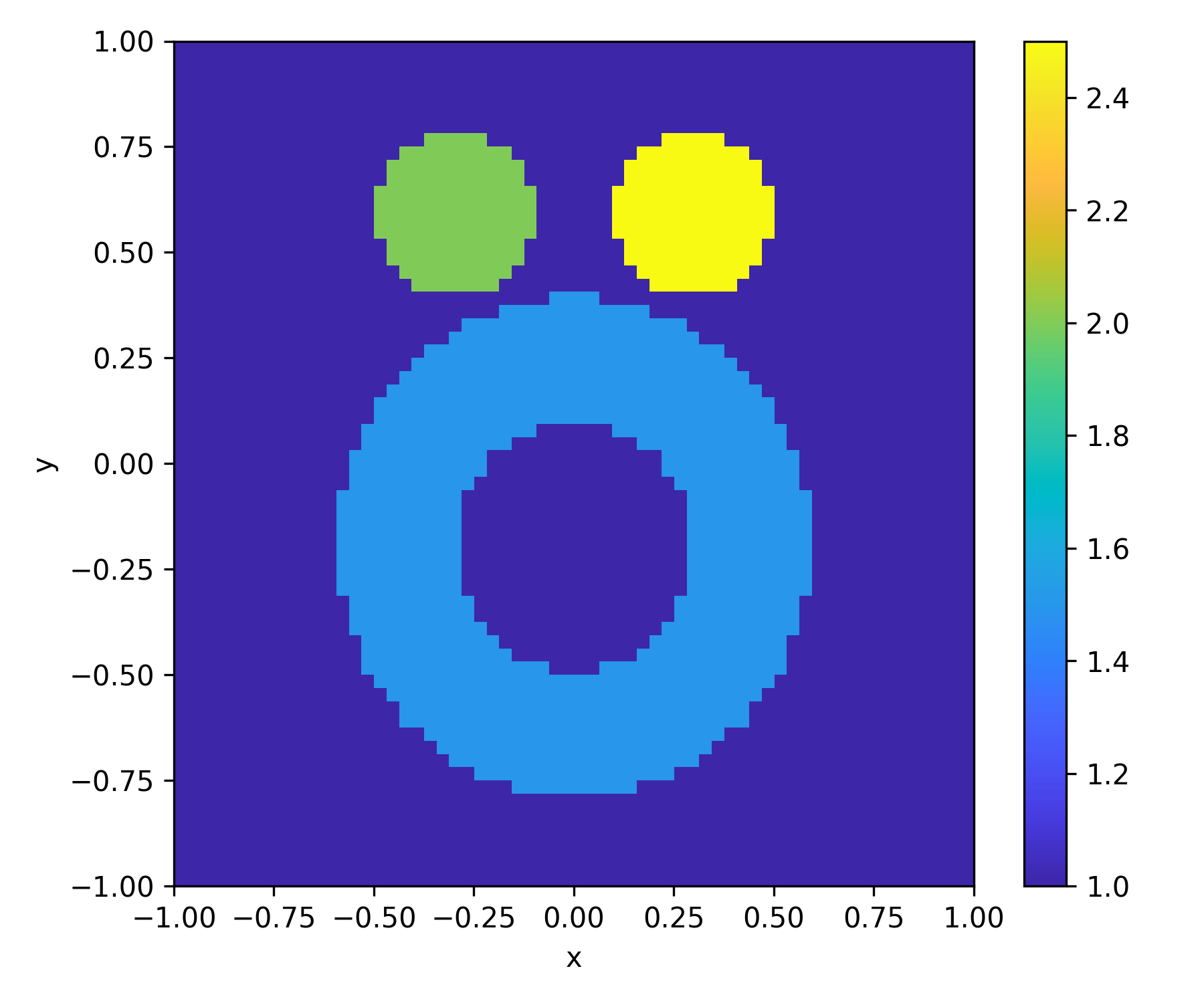}&
				\includegraphics[width=0.15\textwidth]{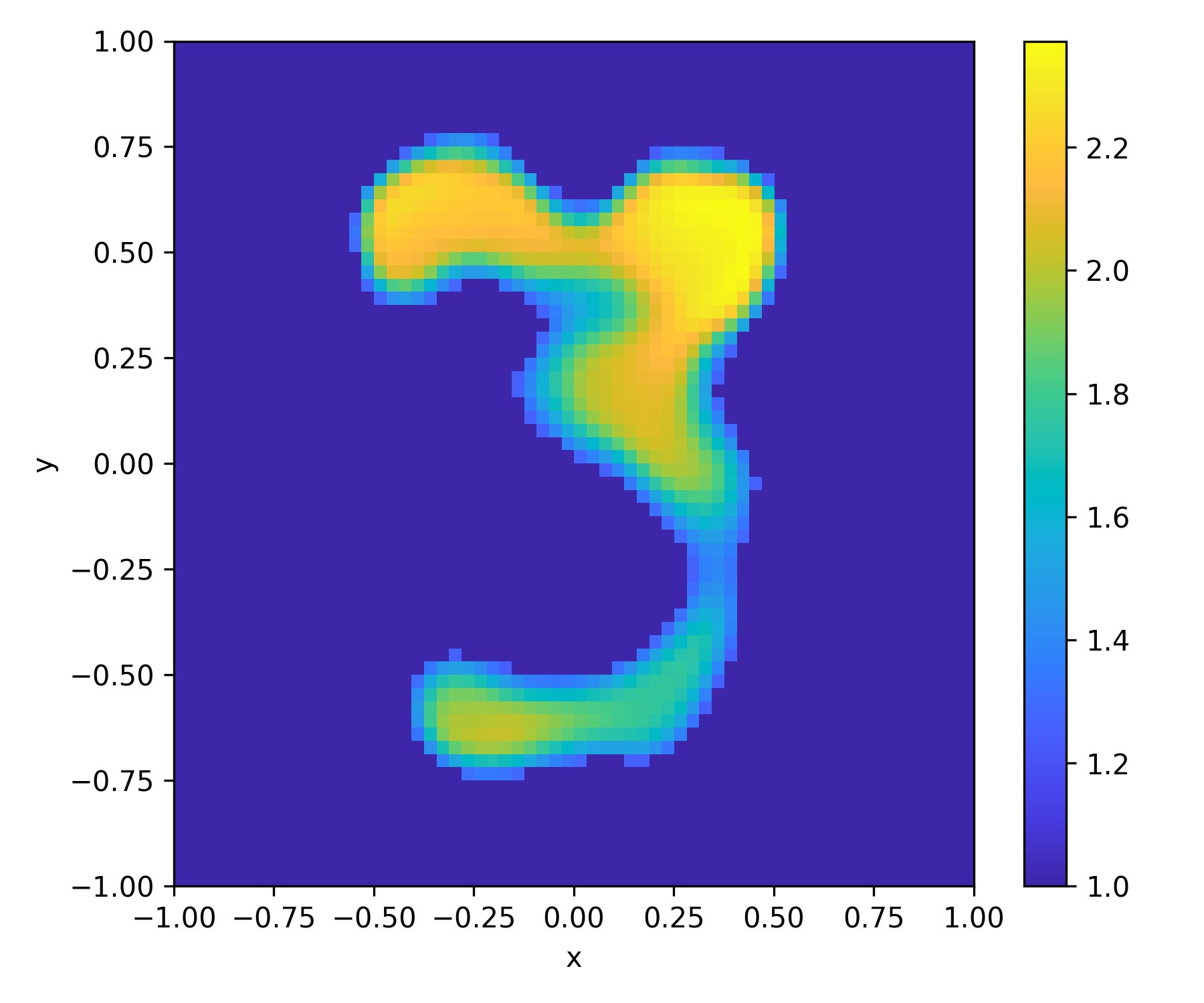}&
				\includegraphics[width=0.15\textwidth]{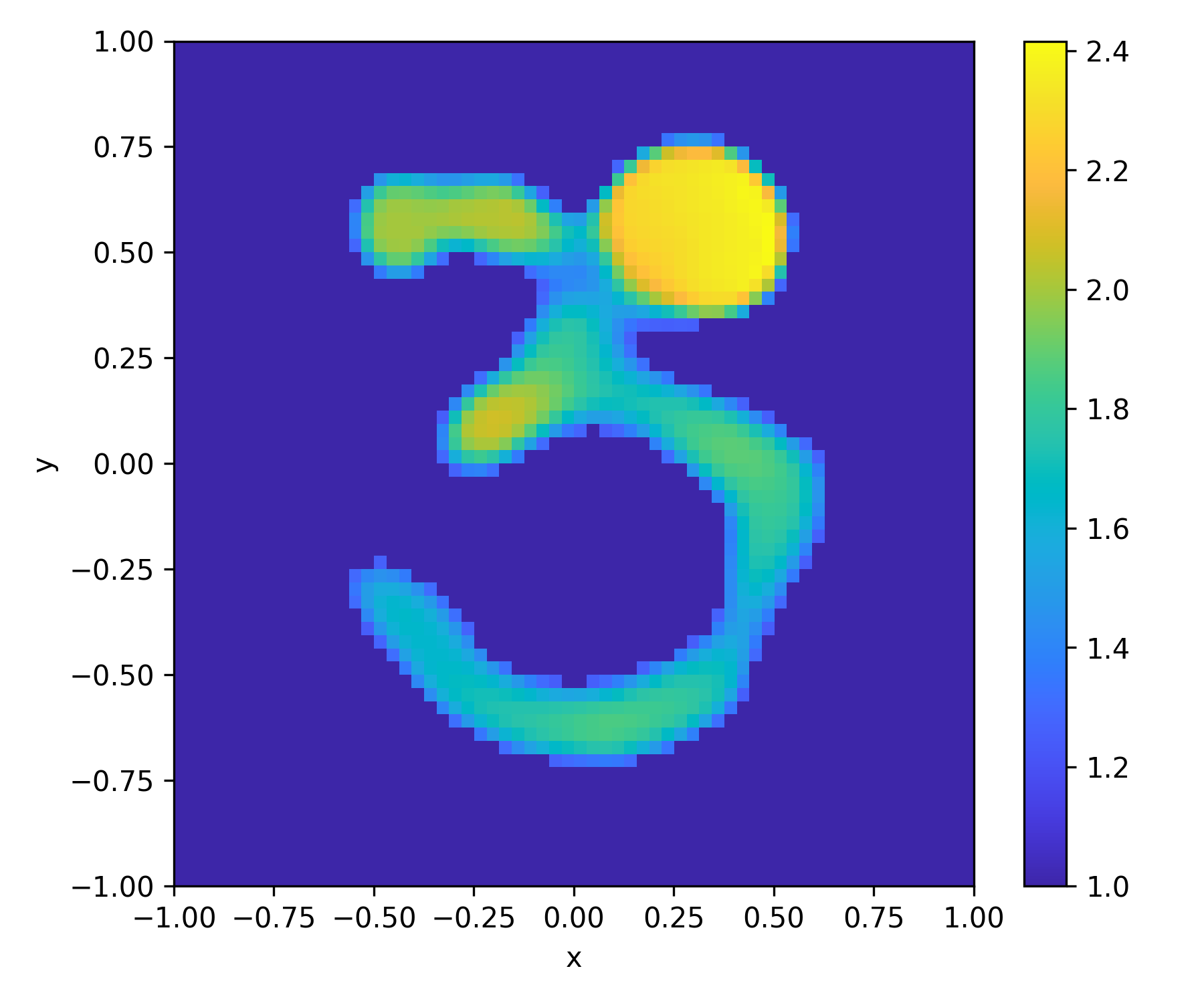}&
				\includegraphics[width=0.15\textwidth]{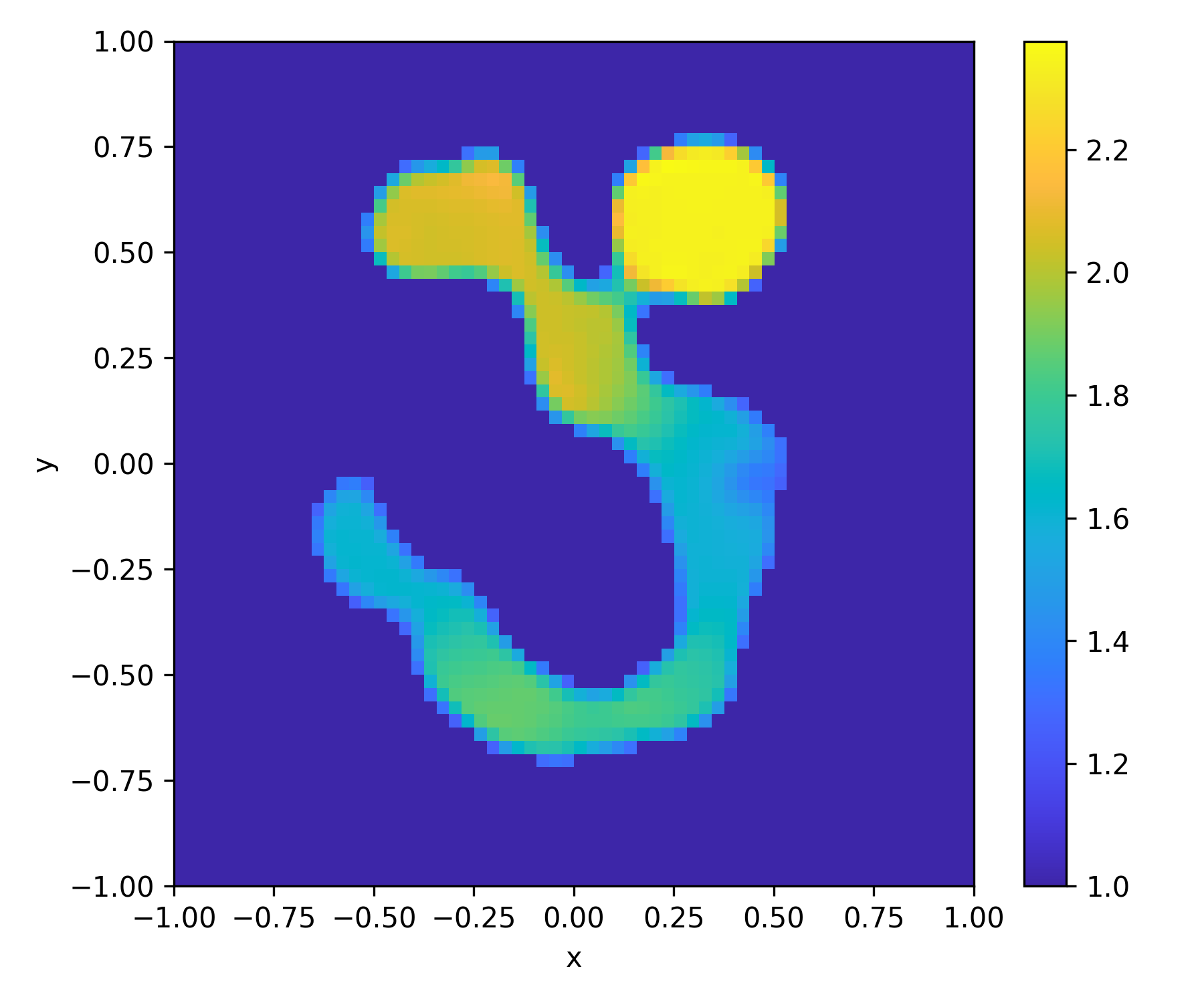}&
				\includegraphics[width=0.15\textwidth]{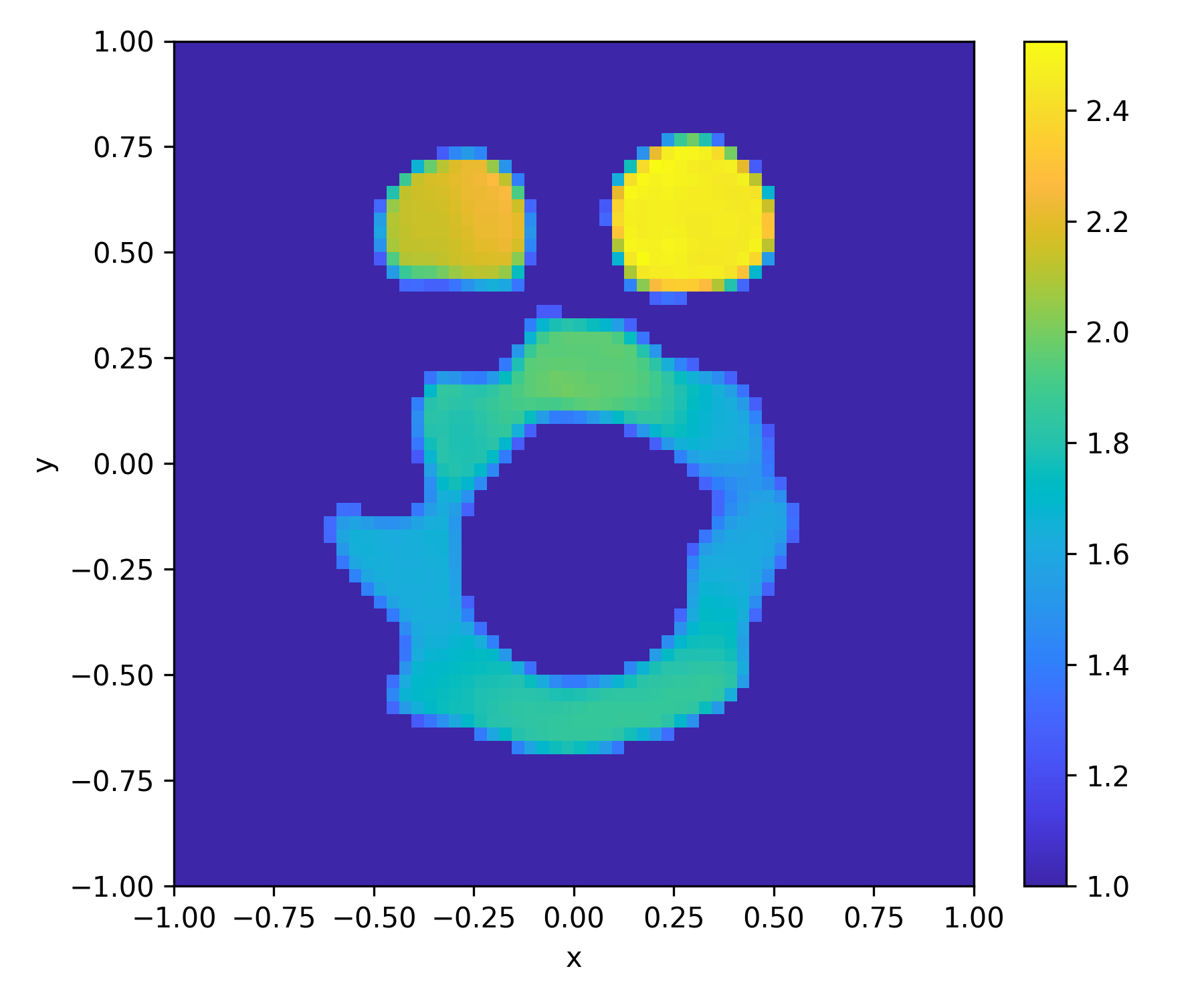}&
				\includegraphics[width=0.15\textwidth]{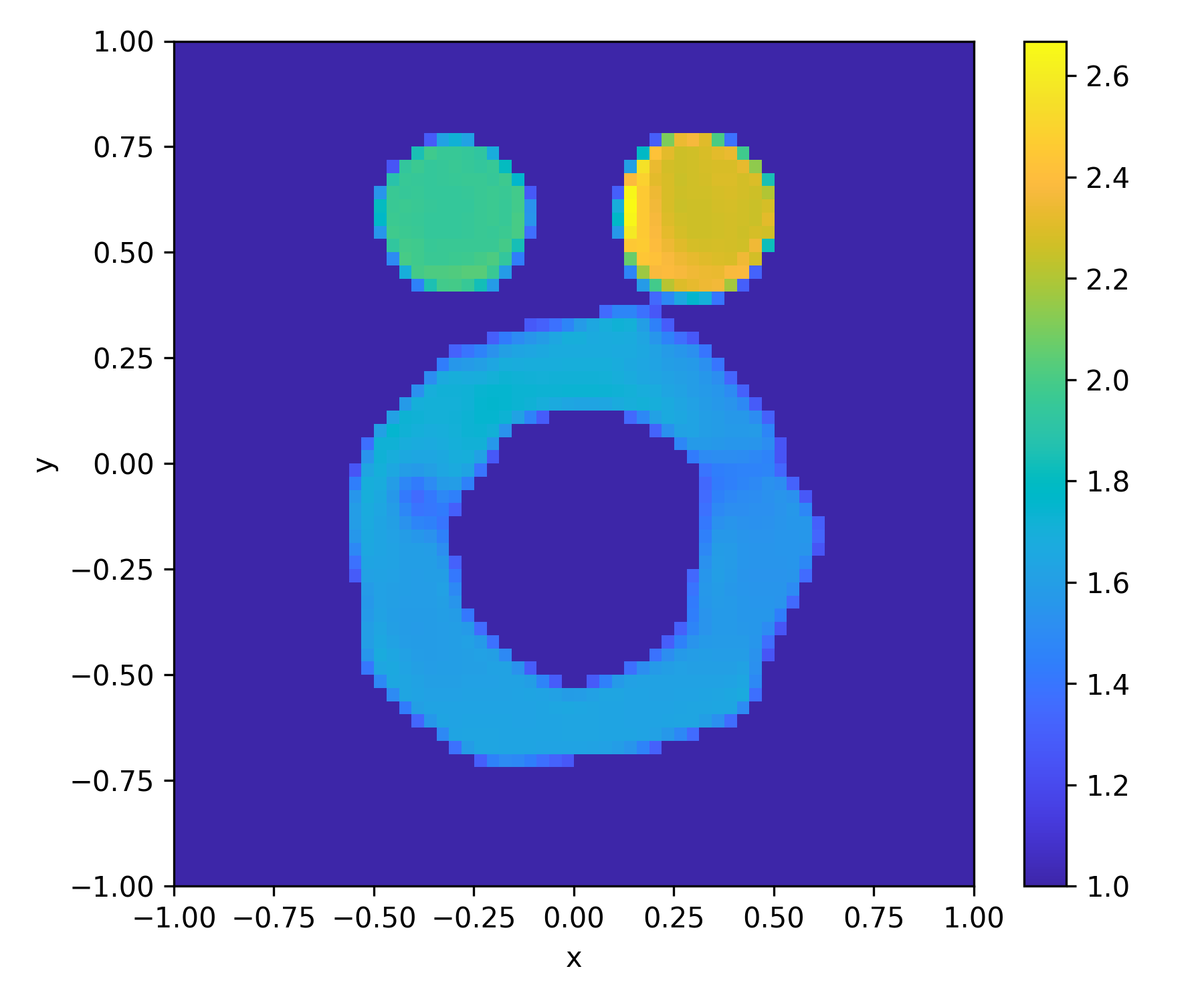}		
				\\
				40\%& &
				\includegraphics[width=0.15\textwidth]{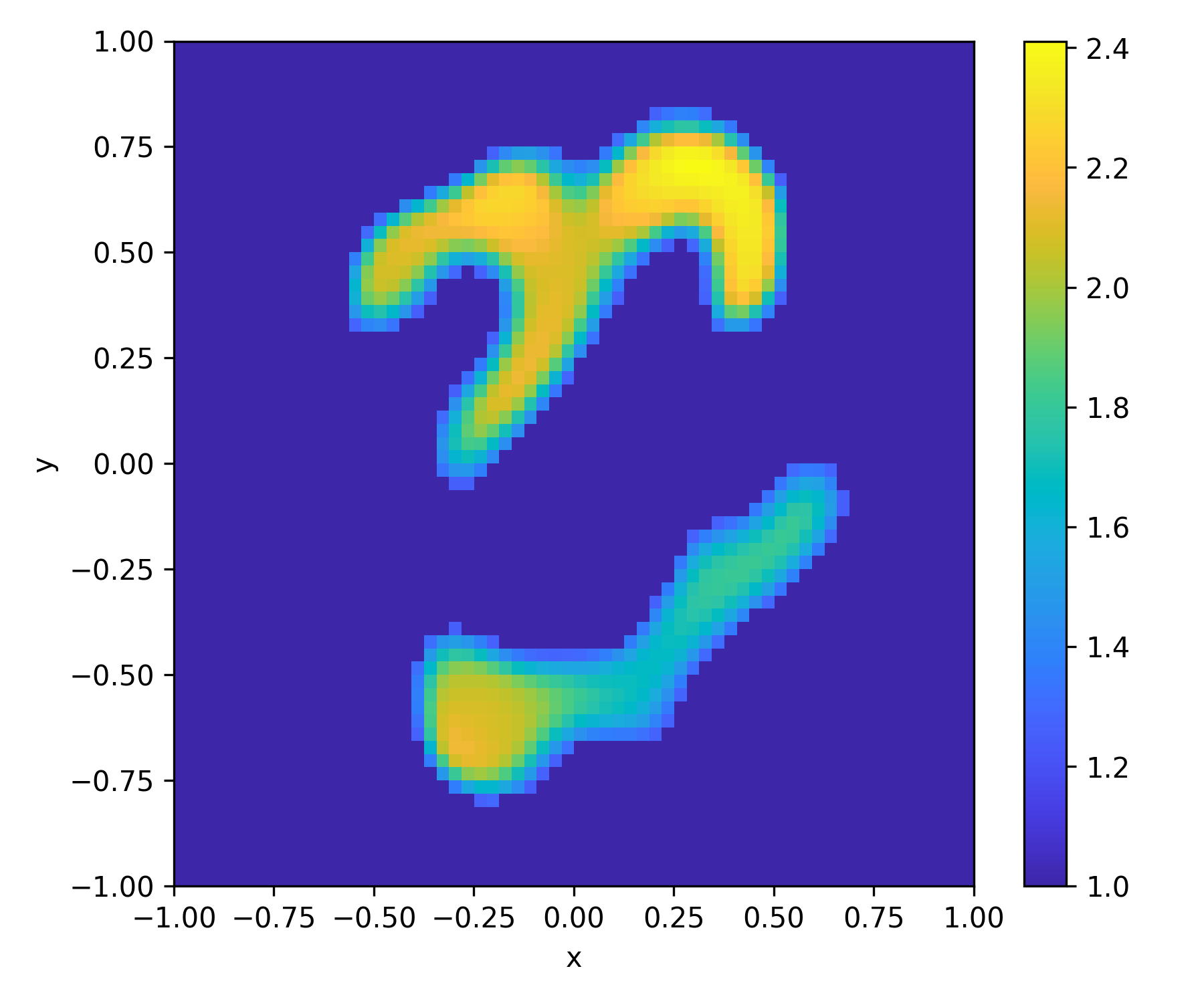}&
				\includegraphics[width=0.15\textwidth]{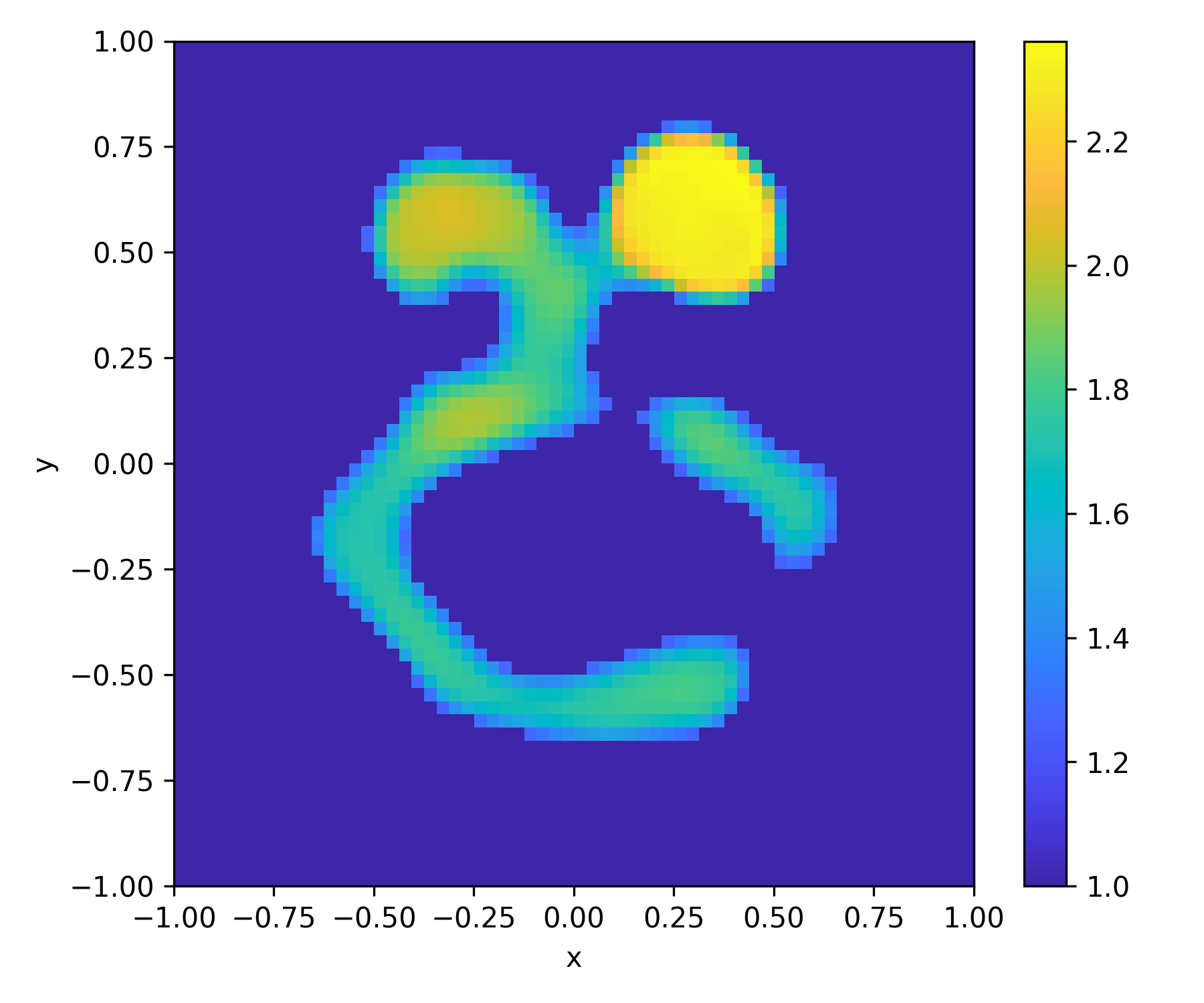}&
				\includegraphics[width=0.15\textwidth]{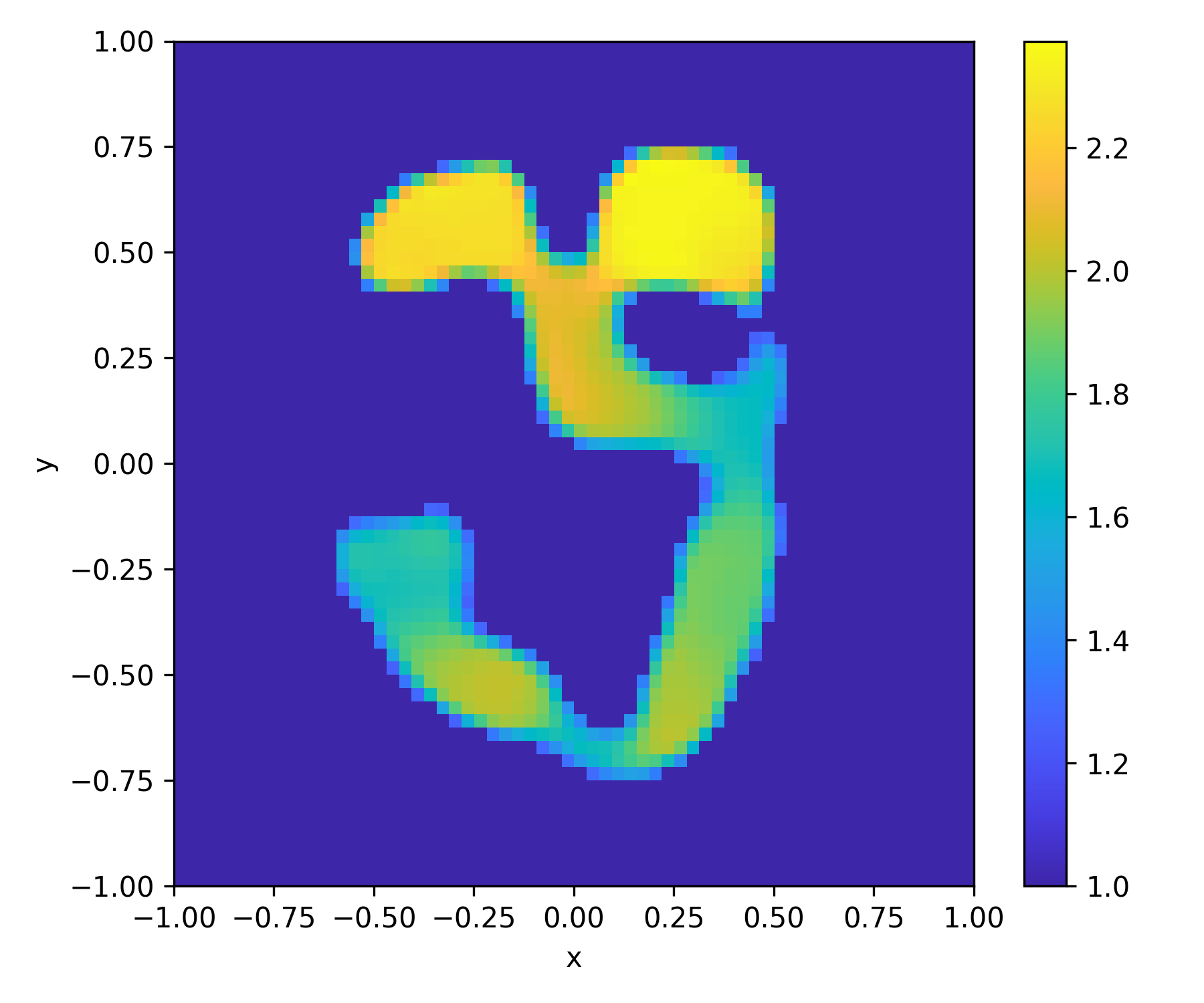}&
				\includegraphics[width=0.15\textwidth]{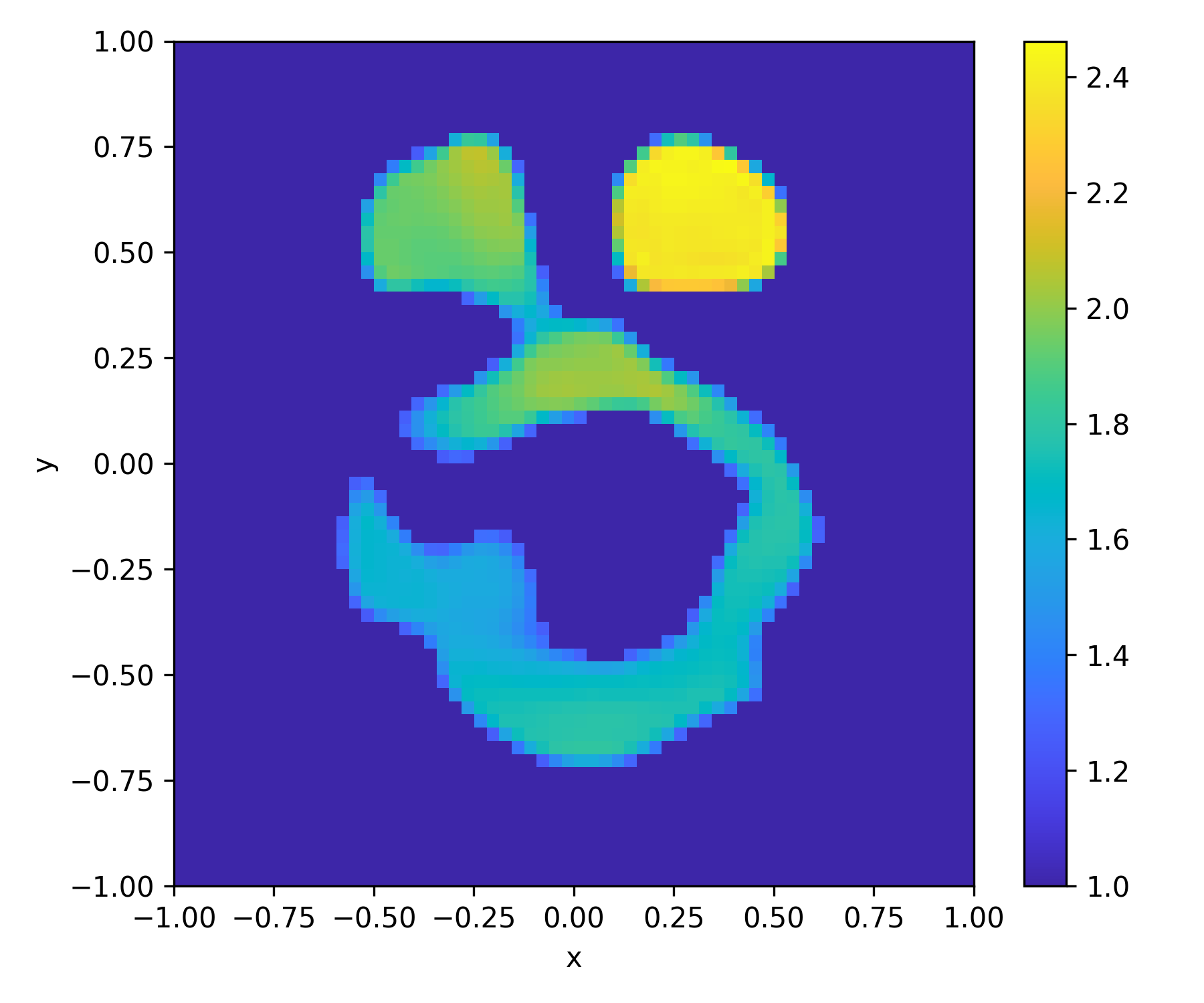}&
				\includegraphics[width=0.15\textwidth]{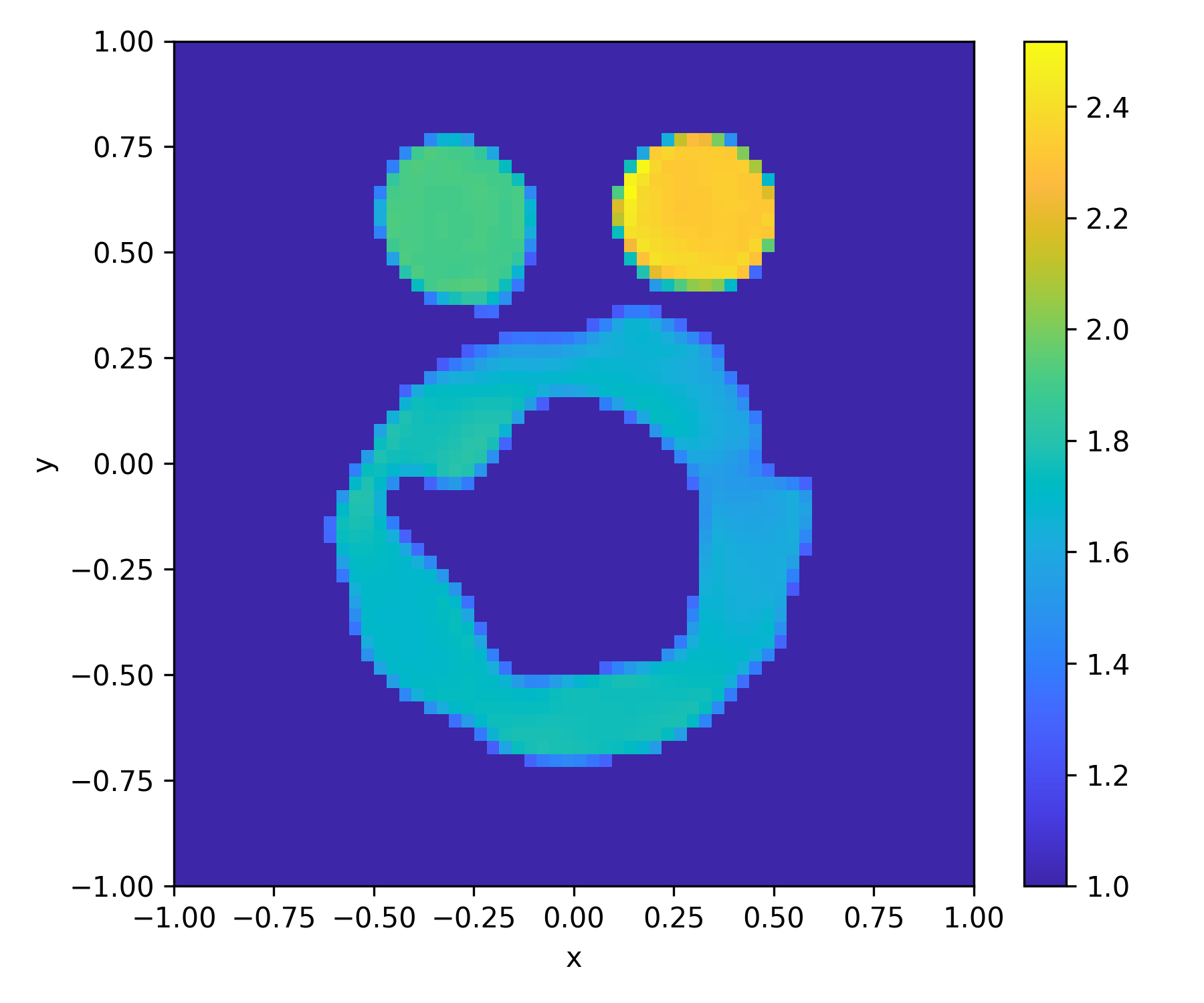}	 
			\end{tblr}
			\caption{Image reconstructions of two “Austria profiles” with $15\%$ and $40\%$  Gaussian noises by using the networks trained by the MNIST dataset.  From left to right: the ground-truth images, the reconstruction with 1,2,4,8, and 16 incident fields.}
			\label{tab:fig-Mnist-Austria}
		\end{center}
	\end{figure}	
 
	\subsubsection{Tests with Latin letters “D”, “S” and “M” with high noise levels}
	In this example,  we use the network trained by the MNIST dataset with $N_i=16$ to reconstruct Latin letters “D”, “S” and “M” with high noise levels $50\%$ and $100\%$. In addition, different from the configuration for the training data, we use 60 receivers equally distributed on a circular curve with a radius $5$ centered at the origin to collect the data for the scattered field. As shown in equation (\ref{conver_Ind}), the index function is convergent as the radius $R$ tends to infinity, which explains the new configuration can still achieve good reconstructions. Based on the results shown in Fig.\,\ref{tab:fig-DSM-HighNoise}, it can be seen that the proposed DSM-DL is very robust to high-level noise as the noise is smoothed in the DSM process. We can also observe that the reconstruction of the thinner part of the letters is more likely to be distorted, which is mainly because the thinner part contributes less to the scattered field. 
    	\begin{figure}[htp]\small
    	\begin{center}
    		\begin{tblr}
    			{width=0.8\linewidth,
    				colspec = {X[-1]X[c,h]X[c,h]X[c,h]},
    				stretch = 0,
    				rowsep = 0pt,}
    			True Contrast &
    			\includegraphics[width=0.15\textwidth]{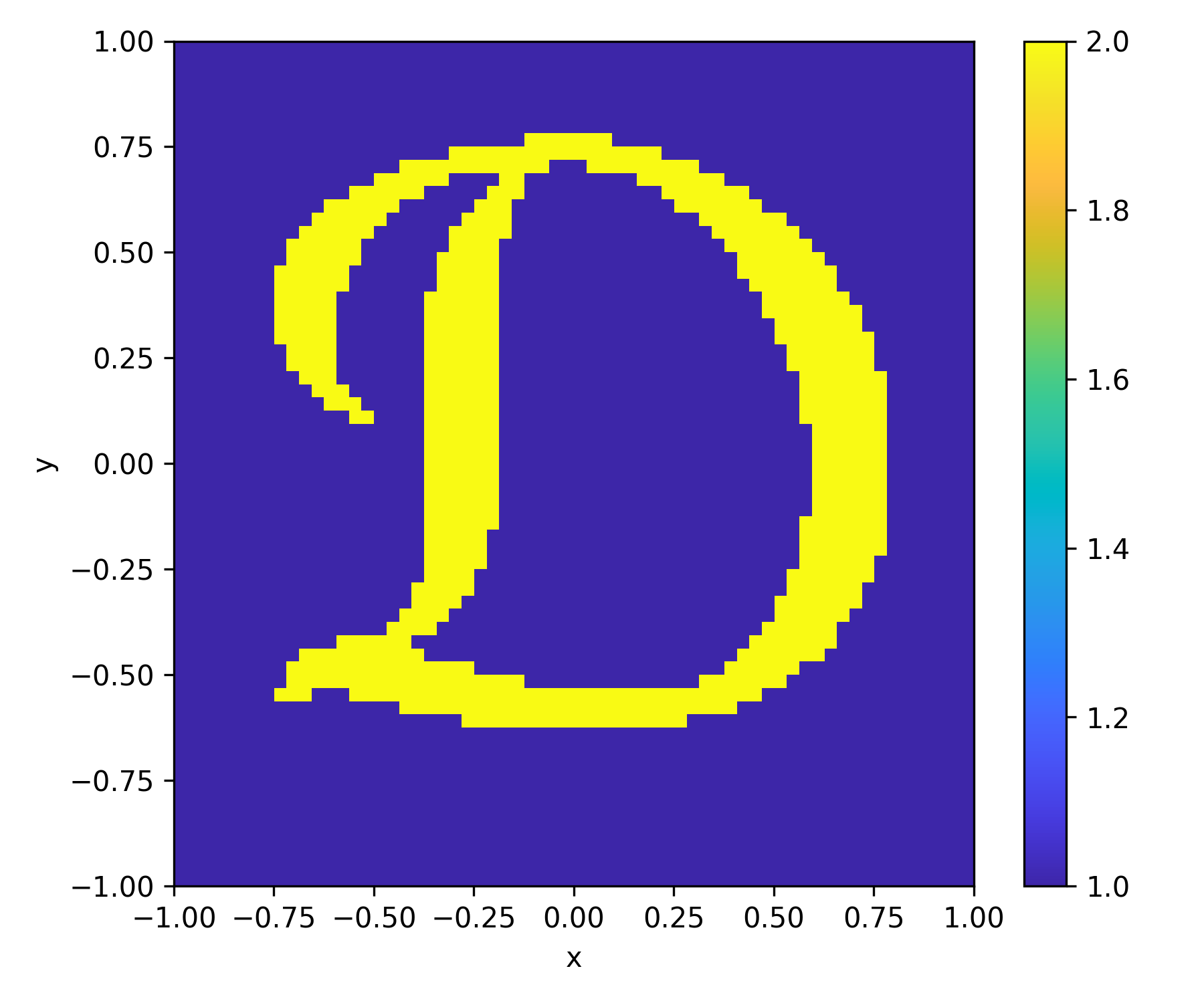}& 
    			\includegraphics[width=0.15\textwidth]{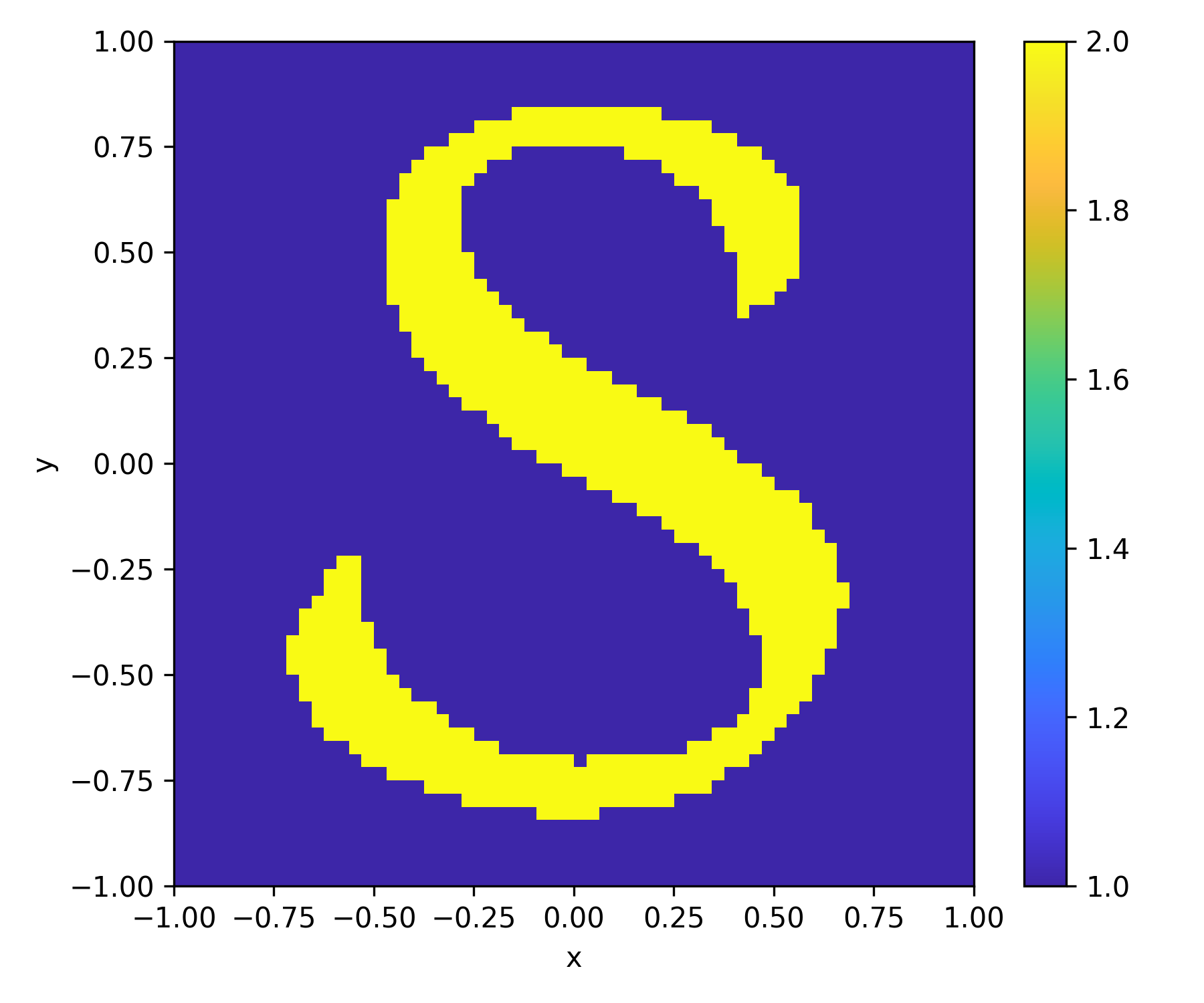}&
    			\includegraphics[width=0.15\textwidth]{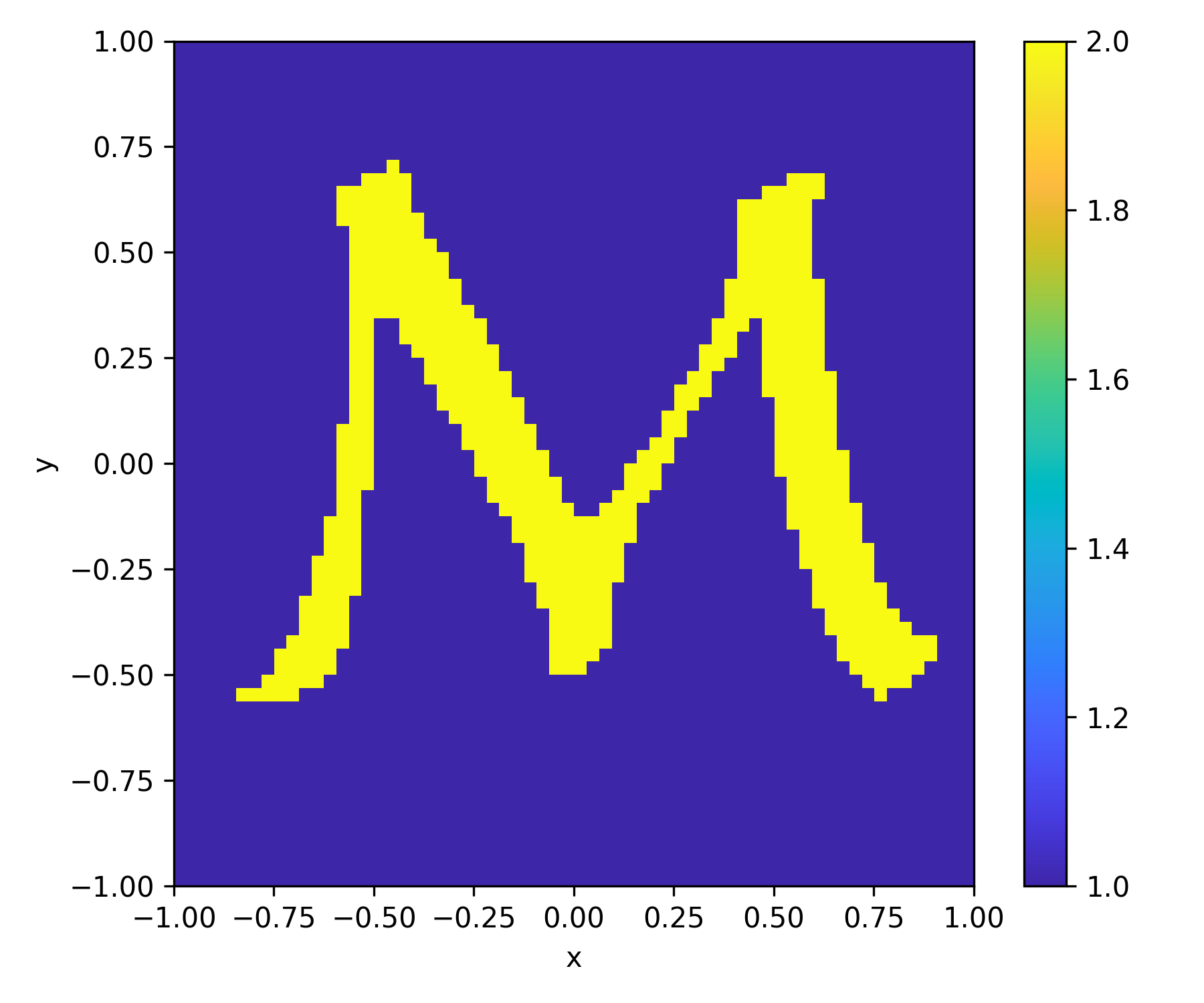}&
    			\\
    			$\delta=50\%, N_i = 16$&
    			\includegraphics[width=0.15\textwidth]{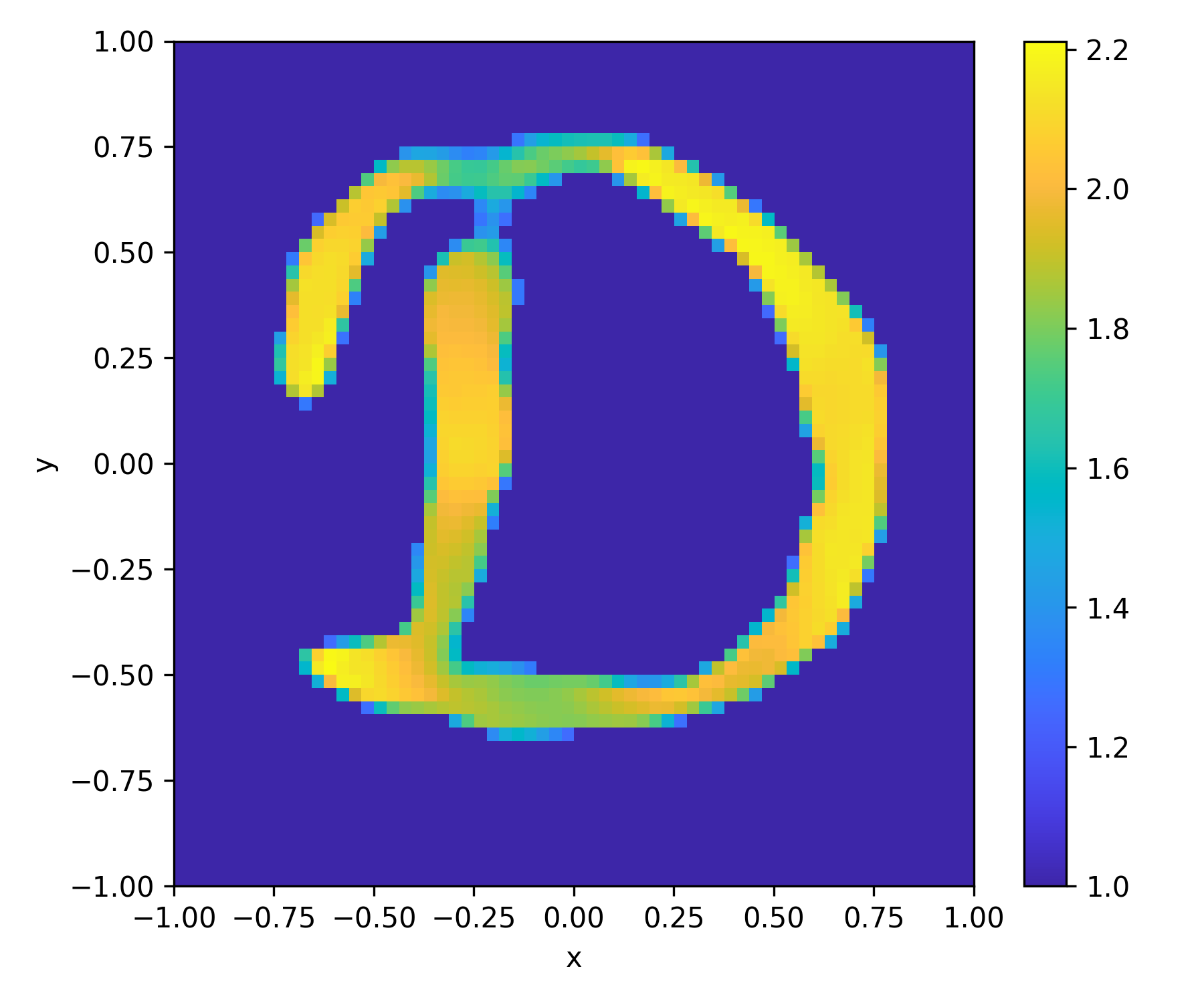}& 
    			\includegraphics[width=0.15\textwidth]{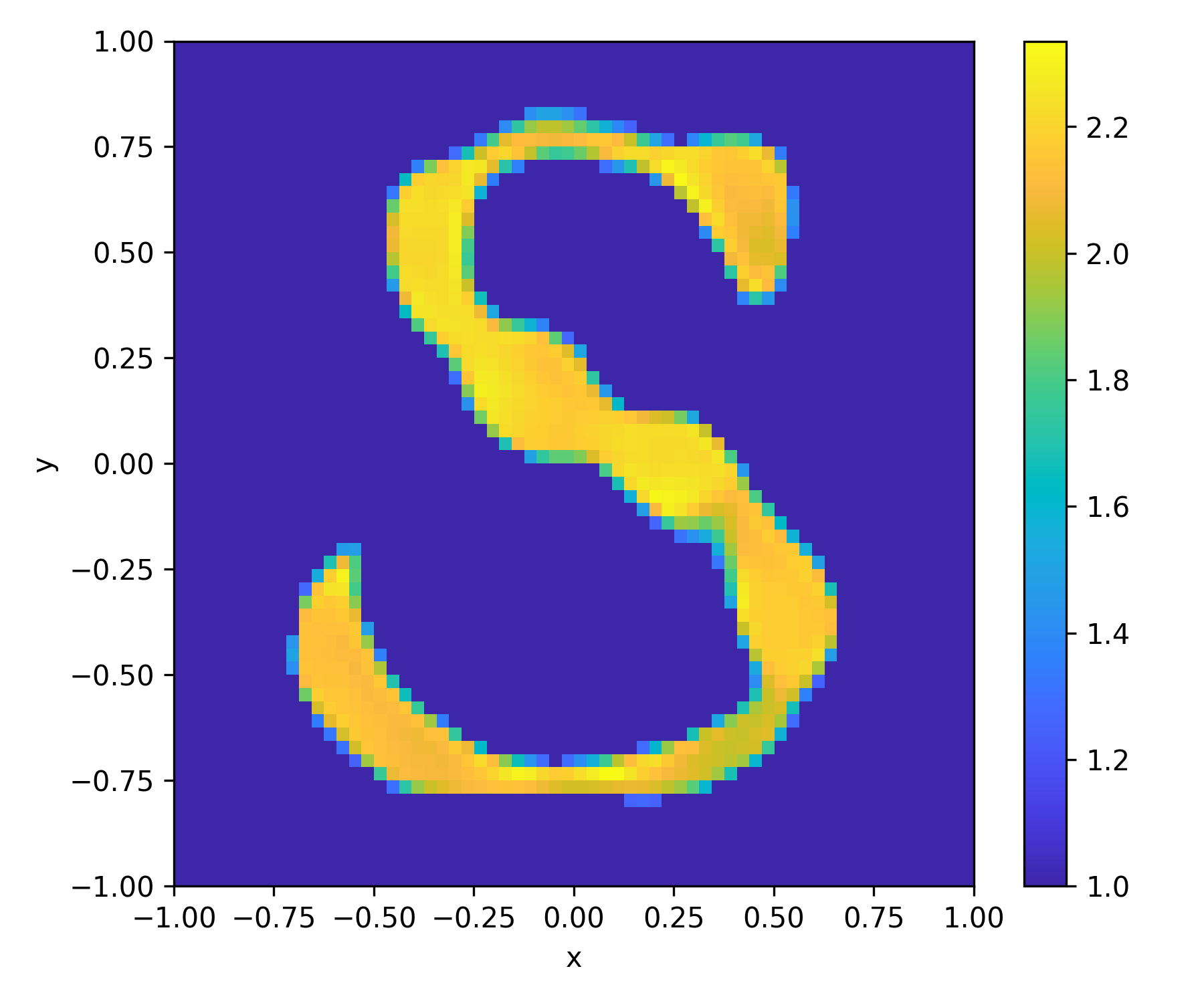}&
    			\includegraphics[width=0.15\textwidth]{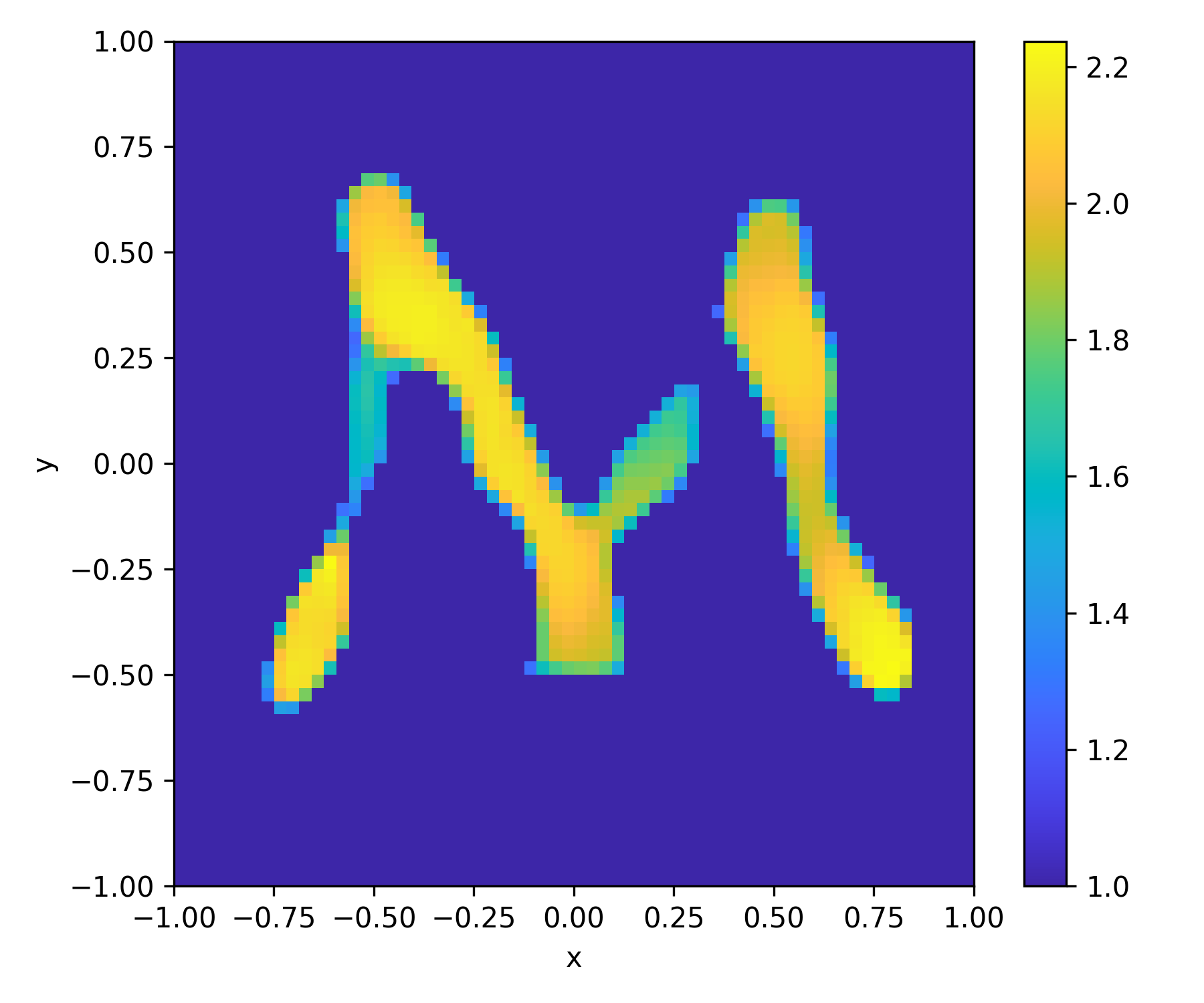}&
    			\\ 
    			$\delta=100\%, N_i=16$&
    		    \includegraphics[width=0.15\textwidth]{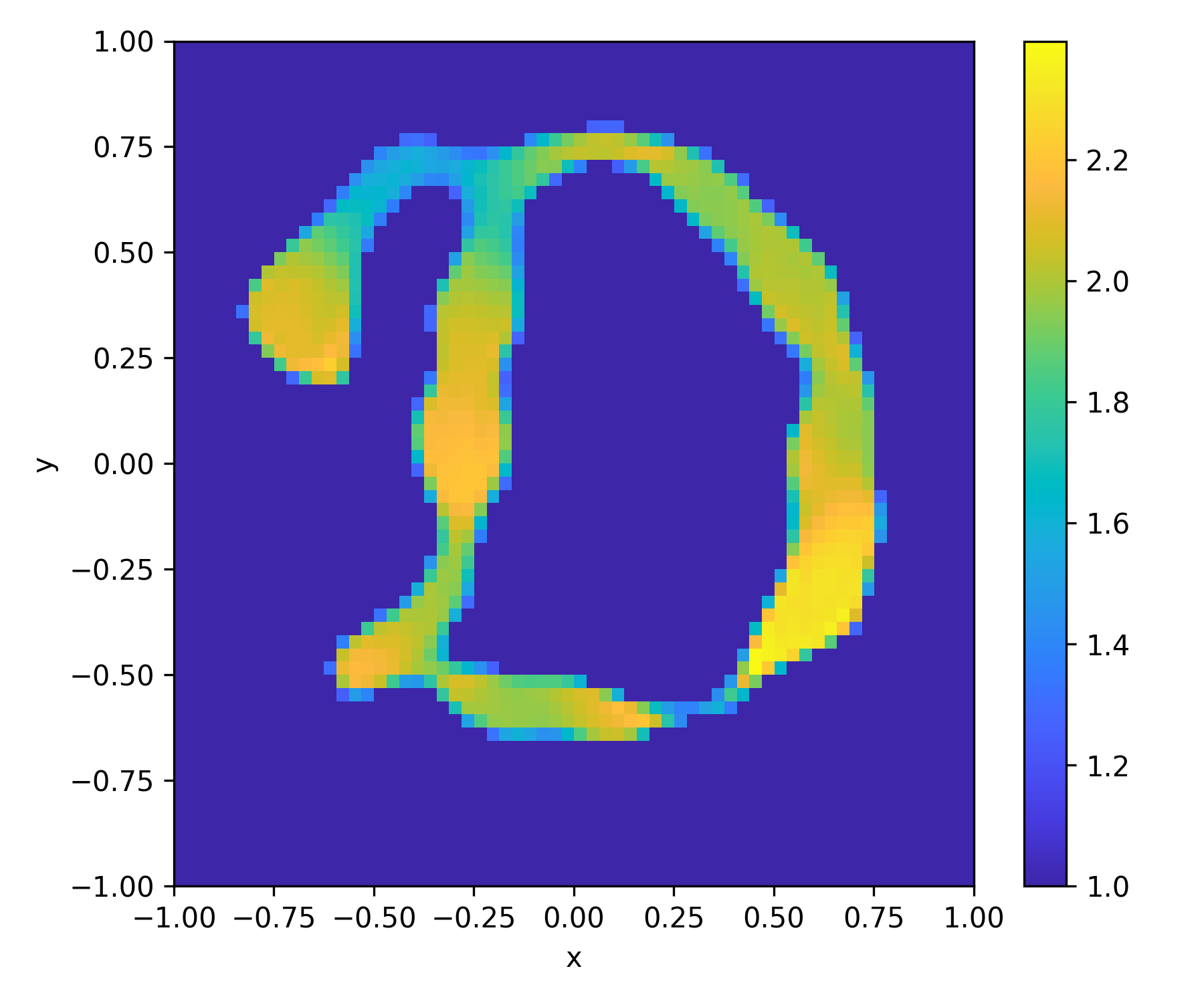}& 
    	    	\includegraphics[width=0.15\textwidth]{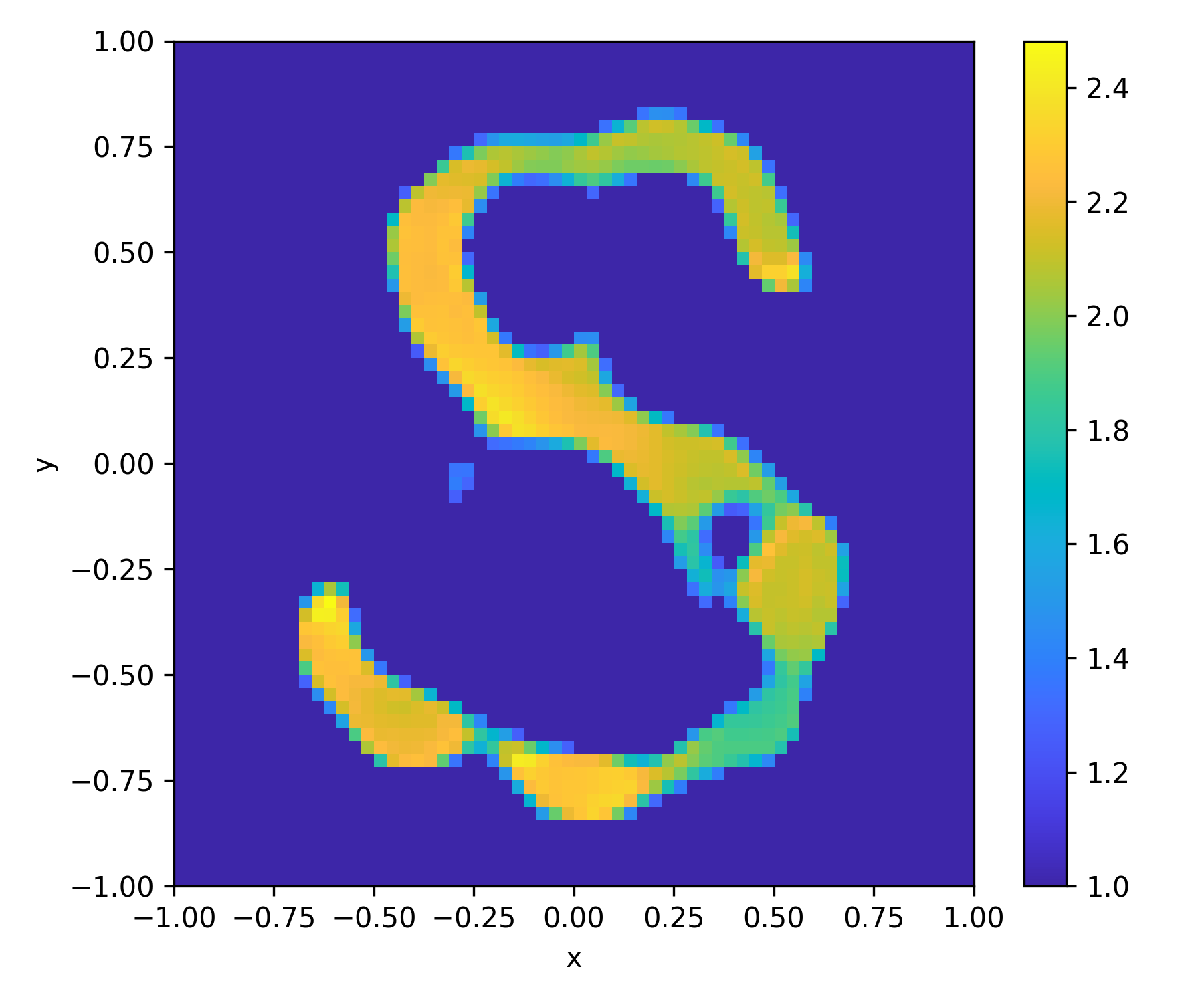}&
    	    	\includegraphics[width=0.15\textwidth]{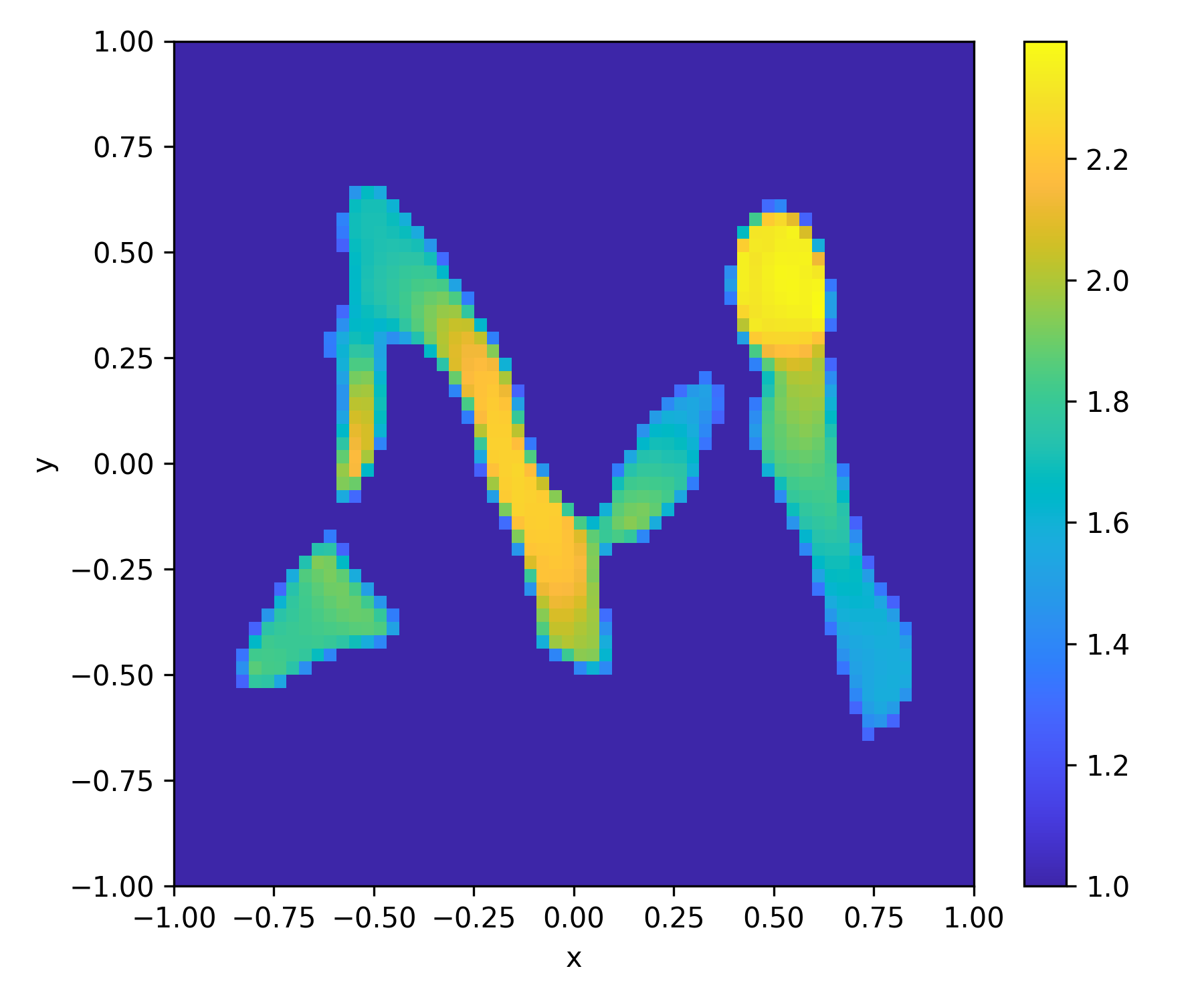}& 
    	    	\\
    	    	
    		\end{tblr}
    		\caption{Image reconstructions of Latin letters “D”, “S” and “M” with  $50\%$ and $100\%$ Gaussian noises by the networks trained by the MNIST dataset with $N_i=16$. We use 60 receivers equally distributed on a circular curve with a radius $5$ centered at the origin to collect the scattered field data.}
    		\label{tab:fig-DSM-HighNoise}
    	\end{center}
    \end{figure}

	\subsubsection{Tests with measured data from half circular curve}
	We recall that the inner product in equation (\ref{indexfunction}) is defined on the whole circular curve $S$, thus the original DSM requires enough receivers and it is not applicable if the receivers are all located on a half circular curve $\Gamma:=\{(r,\theta)\big\vert r=3; 0\le\theta<\pi\}\subset S$. However, based on the upcoming numerical experiments, it is exciting to see that the DSM-DL can be applied to the limited aperture case by considering the following function
	\begin{equation}
		\Phi^{\Gamma}_{p}(x):=|\langle u^{s}_{p},G(x,\cdot)\rangle_{L^{2}(\Gamma)}|, \quad p=1,2,\cdots,N_{i}. 
		\label{halfInd}
	\end{equation}
	In this numerical experiment, there are $16$ receivers located on $\Gamma$, and we employ the functions (\ref{halfInd}) as the inputs of the neural networks for both training and testing. The data augmentations are not used in this example as the rotation and symmetry properties of the index functions do not hold for this case. $20000$ images are used as training data and the other settings are the same as those in the previous MNIST example. The reconstructions for an image from the testing data and an “Austria Ring” are presented in Fig.\,\ref{tab:fig-Mnist-Half}, the relative L2 testing error and SSIM are presented in Table\,\ref{tab:error}. The results indicate that although the original DSM is not applicable in this case, the DSM-DL can still achieve very satisfactory results. However, the performance is not as good as that of using full measurement on the whole surface, this is reasonable since fewer measurement data is used.

	\begin{table}[htp]\small
		\begin{center}
			\begin{tabular}{ |c|c|c|c|c|c|c|c|}
				\hline
				Example&Metric&Noise Level& $N_{i}=1$& $N_{i}=2$ & $N_{i}=4$ & $N_{i}=8$ & $N_{i}=16$ \\
				\hline
				\multirow{4}*{MNIST}&\multirow{2}*{Relative L2 error}&15\% &0.2041& 0.1713 & 0.1417 & 0.1228& 0.1073 \\
				\cline{3-8}
				& &40\% &0.2324& 0.1925 & 0.1577 & 0.1367& 0.1216 \\
				\cline{2-8}
				&\multirow{2}*{SSIM}&15\% &0.6834& 0.7433 & 0.8077 & 0.8489& 0.8794\\
				\cline{3-8}
				& &40\% &0.6280& 0.7007& 0.7727 & 0.8190& 0.8504 \\
				\hline
				
				\multirow{4}*{MNIST$\big |_{\Gamma}$}&\multirow{2}*{Relative L2 error}&15\% &0.2384& 0.2138 & 0.1728 & 0.1472& 0.1250 \\
				\cline{3-8}
				& &40\% &0.2518& 0.2284 & 0.1965 & 0.1681& 0.1437 \\
				\cline{2-8}
				&\multirow{2}*{SSIM}&15\% &0.5865& 0.6391& 0.7444 & 0.7980& 0.8416\\
				\cline{3-8}
				& &40\% &0.5580& 0.6065& 0.6968 & 0.7541& 0.8042 \\
				\hline
			\end{tabular}
			\caption{Relative L2 testing error and SSIM for different examples with different noise levels and number of incidences, where the example MNIST$\big |_{\Gamma}$ refers to that the scattered fields are measured on the half circular curve $\Gamma$.}
			\label{tab:error}
		\end{center}
	\end{table}

   \begin{figure}[htp]\small
   	\begin{center}
   		\begin{tblr}
   			{colspec = {X[-1]X[c]X[c,h]X[c,h]X[c,h]X[c,h]X[c,h]},
   				stretch = 0,
   				rowsep = 0pt,}
   			Noise Level& Ground truth& $N_{i}$=1& $N_{i}$=2 &$N_{i}$=4&$N_{i}$=8& $N_{i}$=16\\
   			15\%&\SetCell[r=2]{c}
   			\includegraphics[width=0.15\textwidth]{figures/10-Mnist_exact1.png}& 
   			\includegraphics[width=0.15\textwidth]{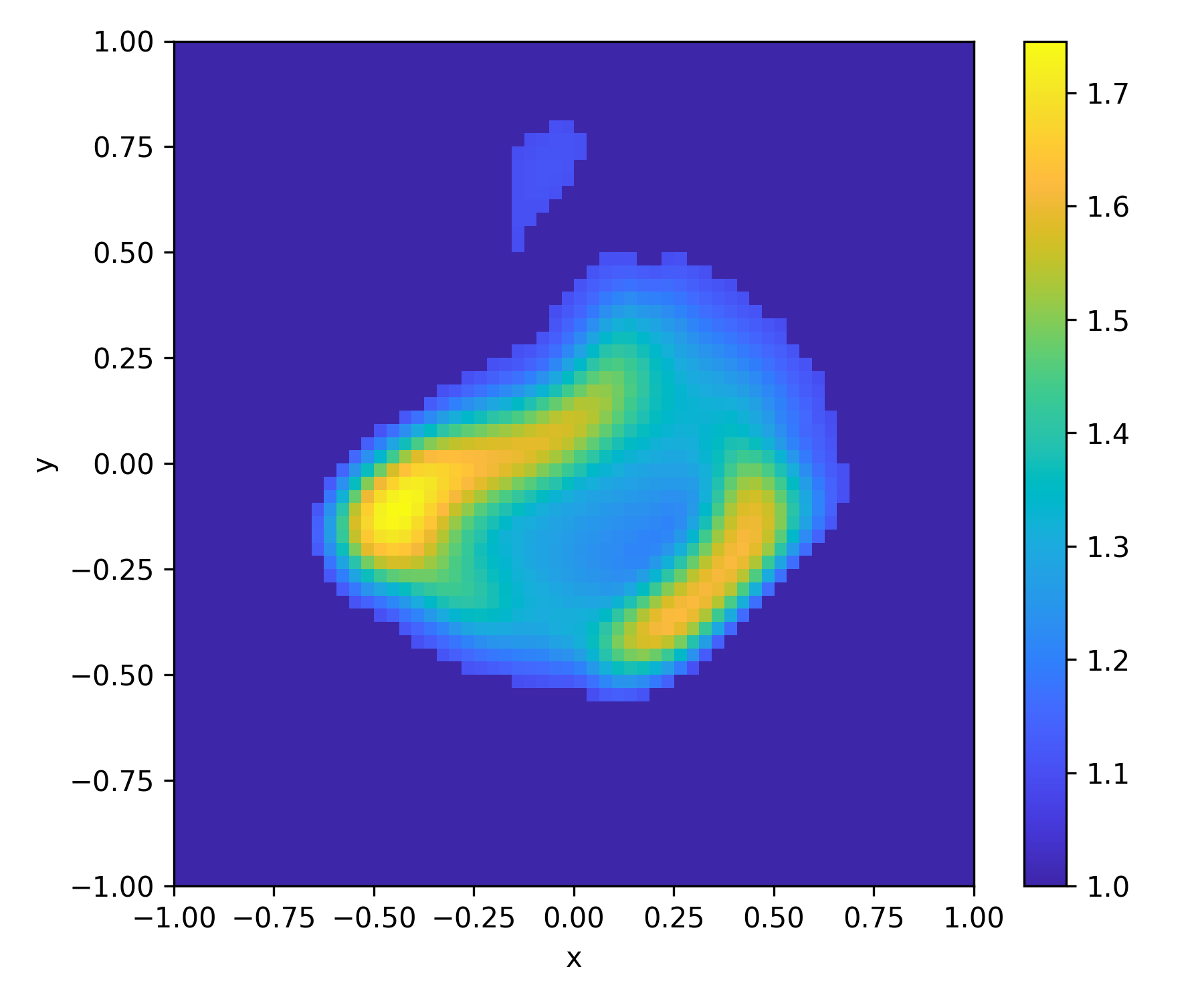}&
   			\includegraphics[width=0.15\textwidth]{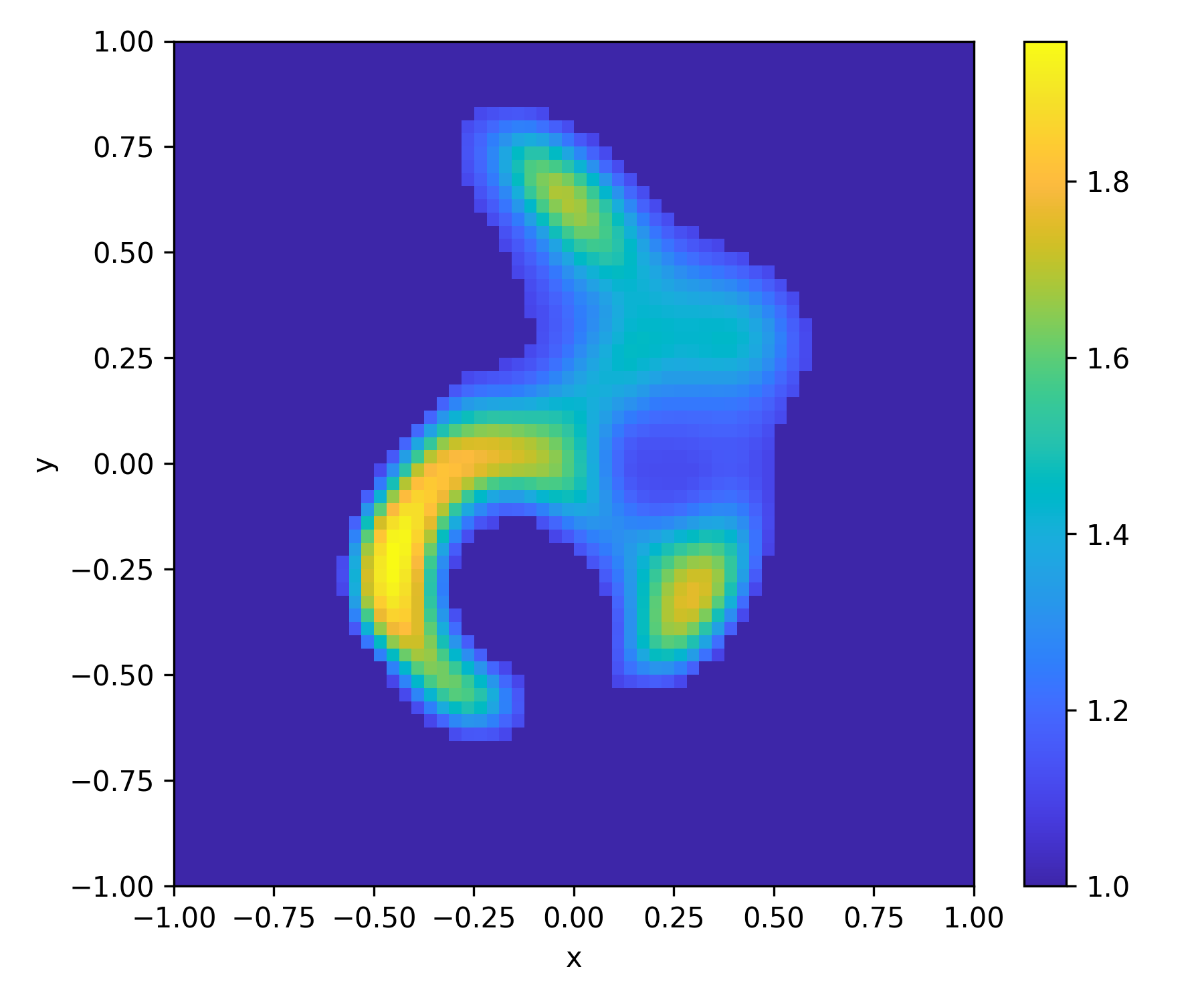}&
   			\includegraphics[width=0.15\textwidth]{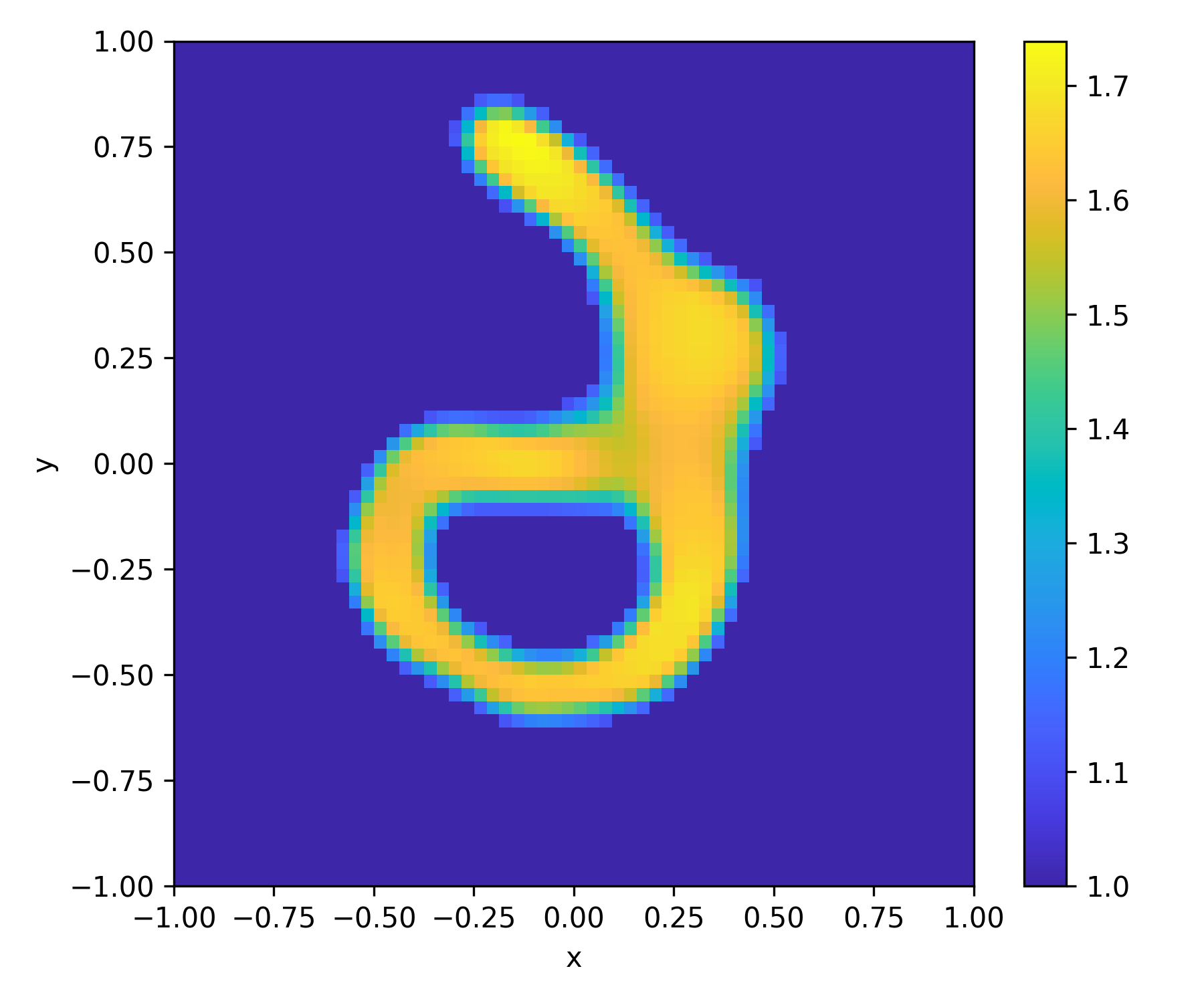}&
   			\includegraphics[width=0.15\textwidth]{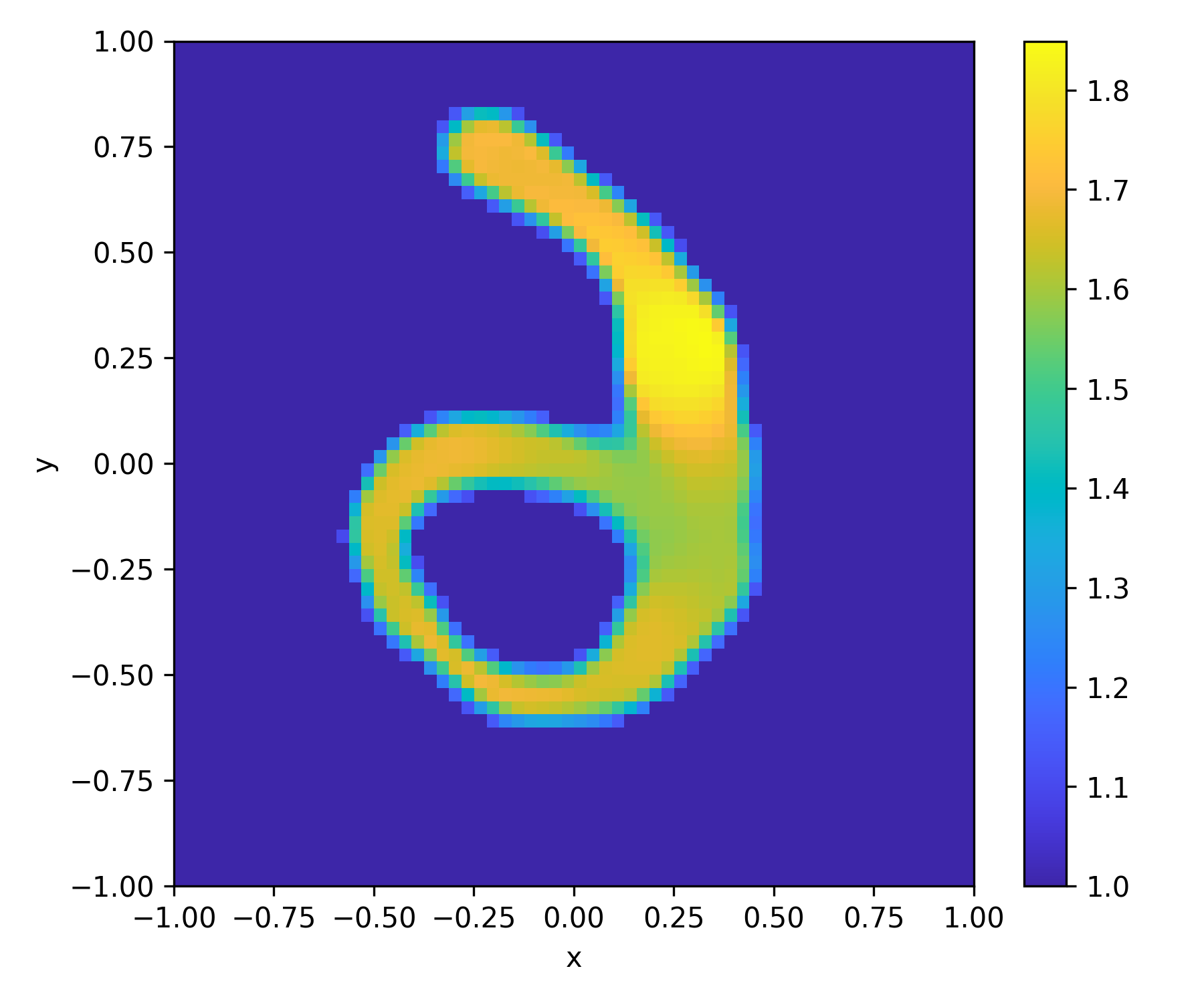}&
   			\includegraphics[width=0.15\textwidth]{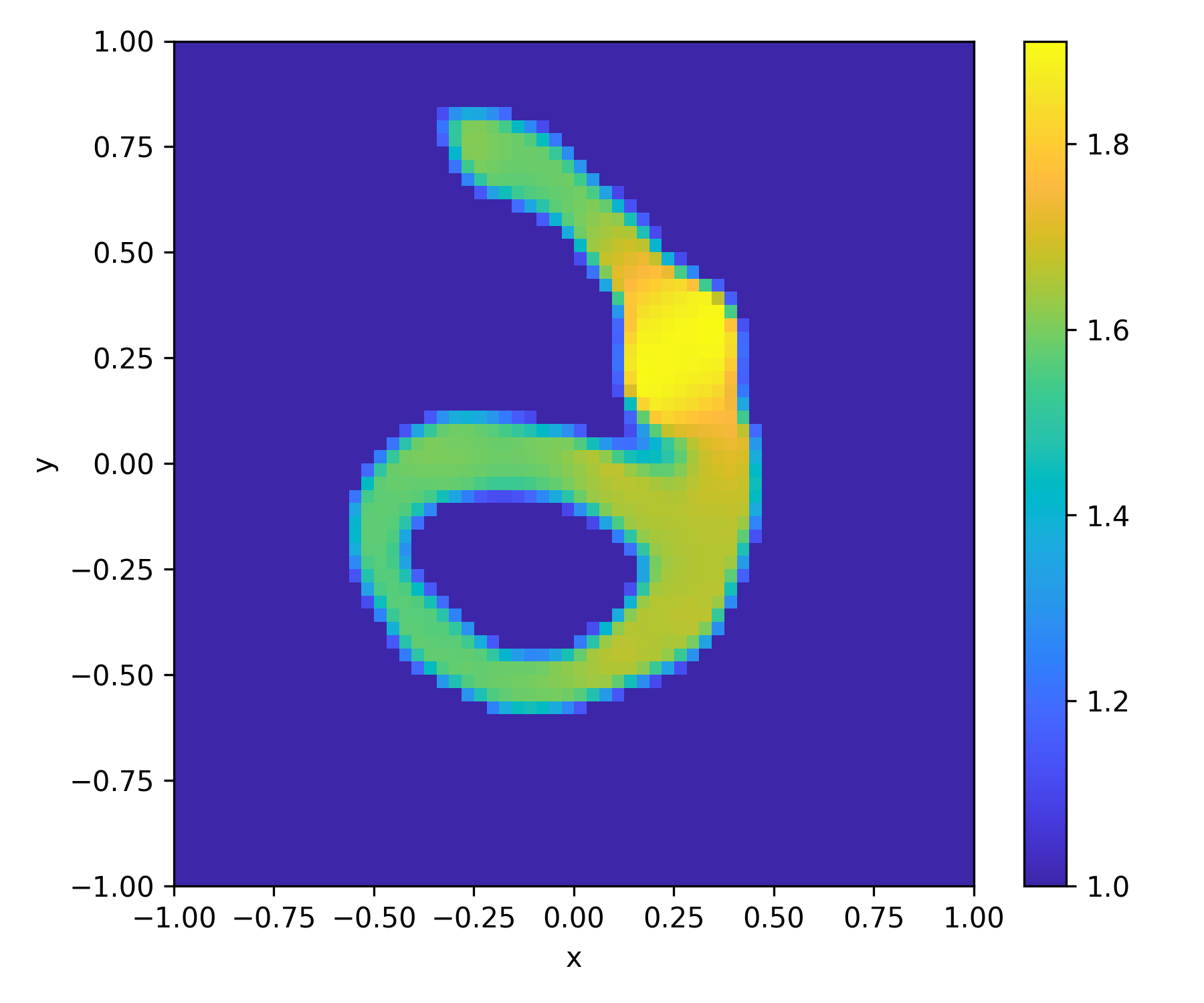}&
   			\\
   			40\%& & 
   			\includegraphics[width=0.15\textwidth]{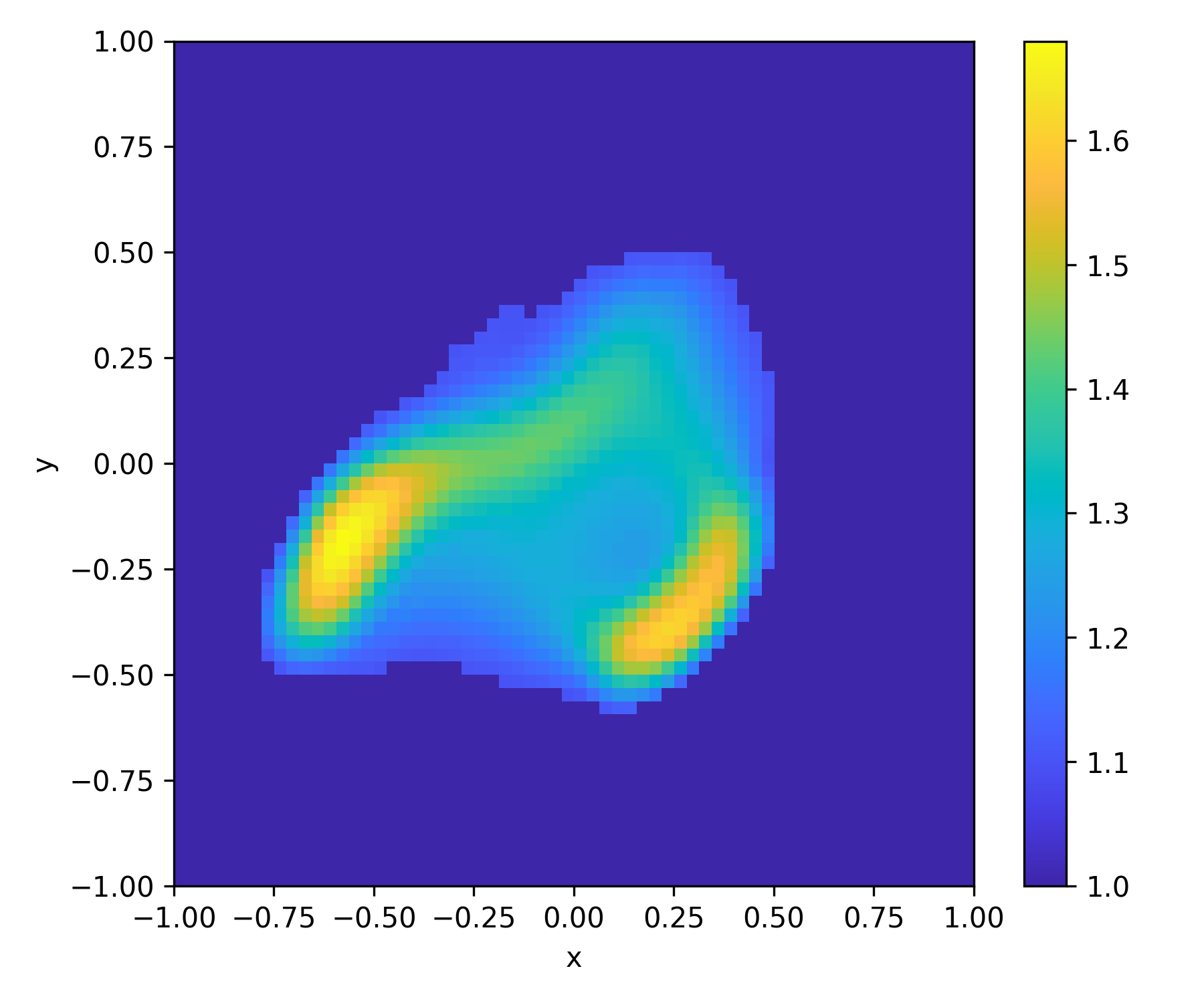}&
   		    \includegraphics[width=0.15\textwidth]{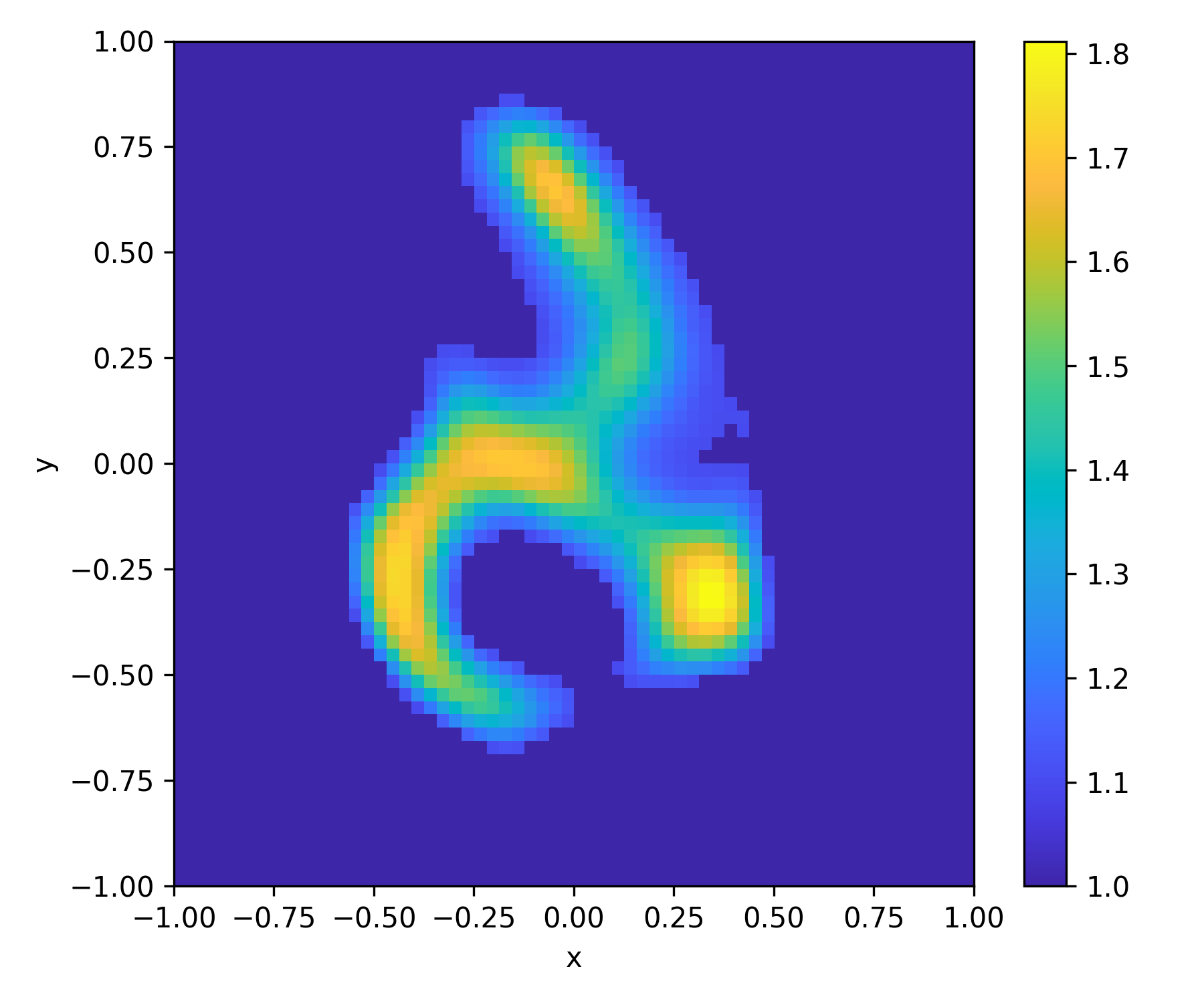}&
   	    	\includegraphics[width=0.15\textwidth]{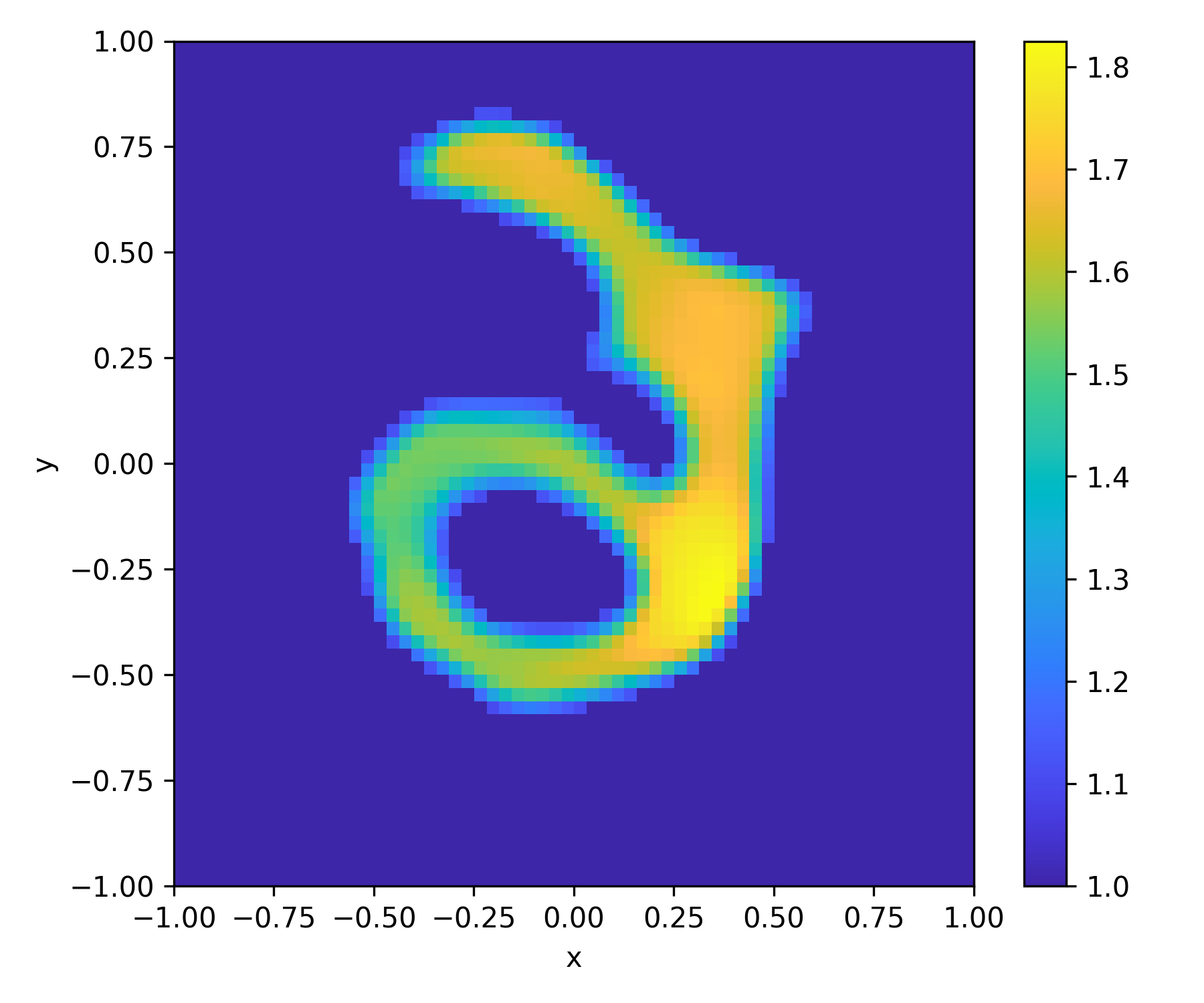}&
   	    	\includegraphics[width=0.15\textwidth]{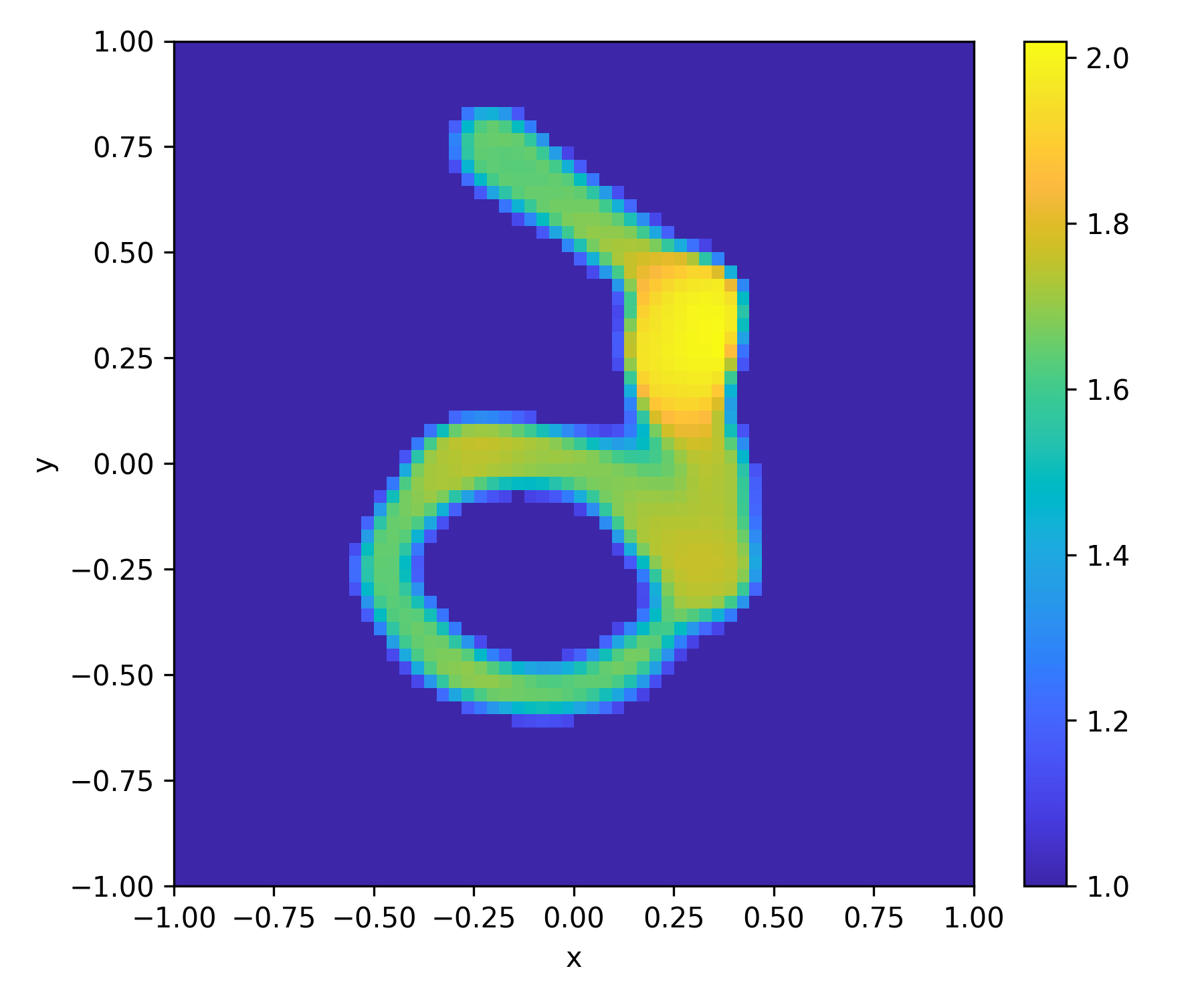}&
   	    	\includegraphics[width=0.15\textwidth]{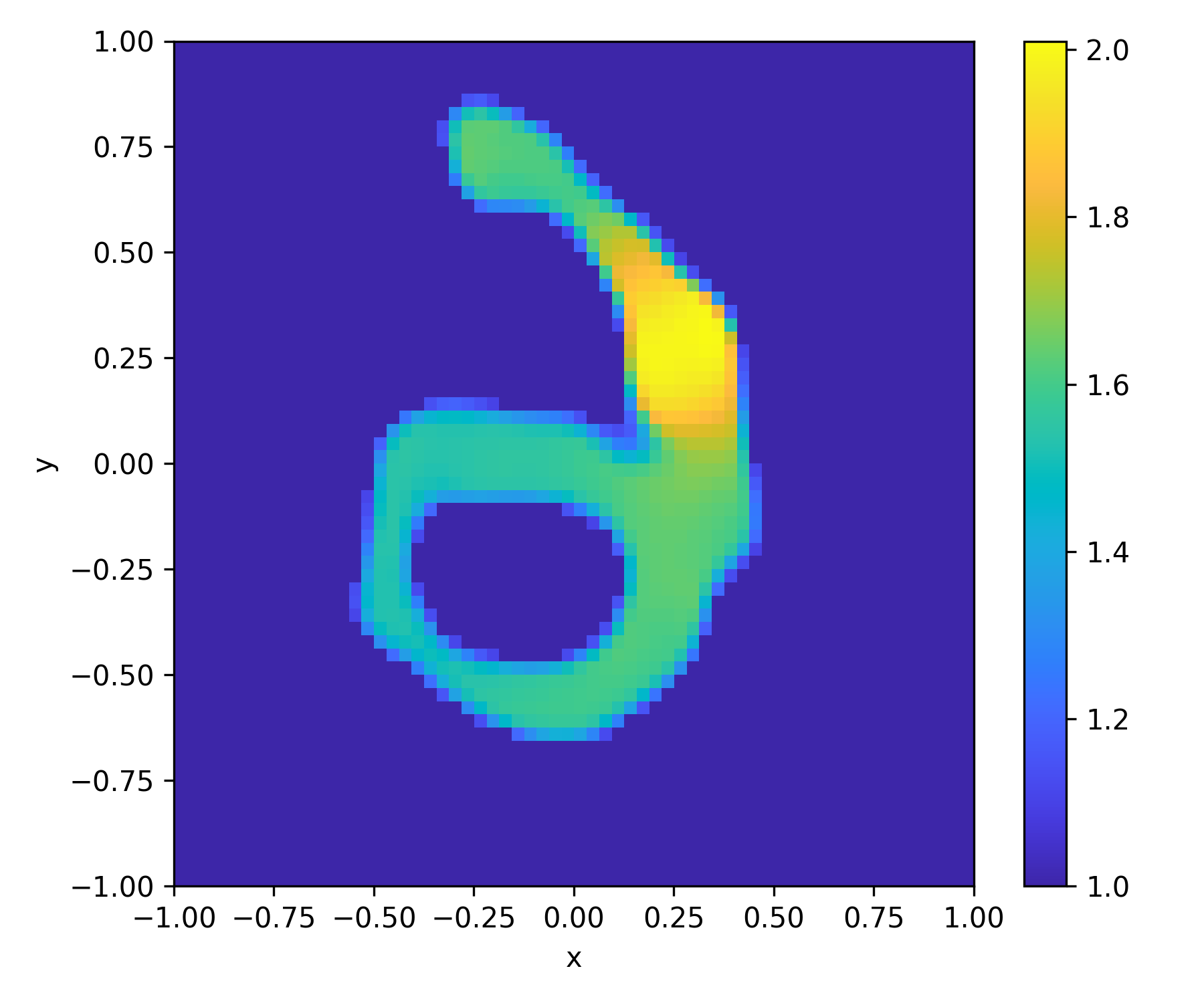}& 
   			\\
   			
   			15\%&\SetCell[r=2]{c}\includegraphics[width=0.15\textwidth]{figures/Austria_exact1.png}&
   			\includegraphics[width=0.15\textwidth]{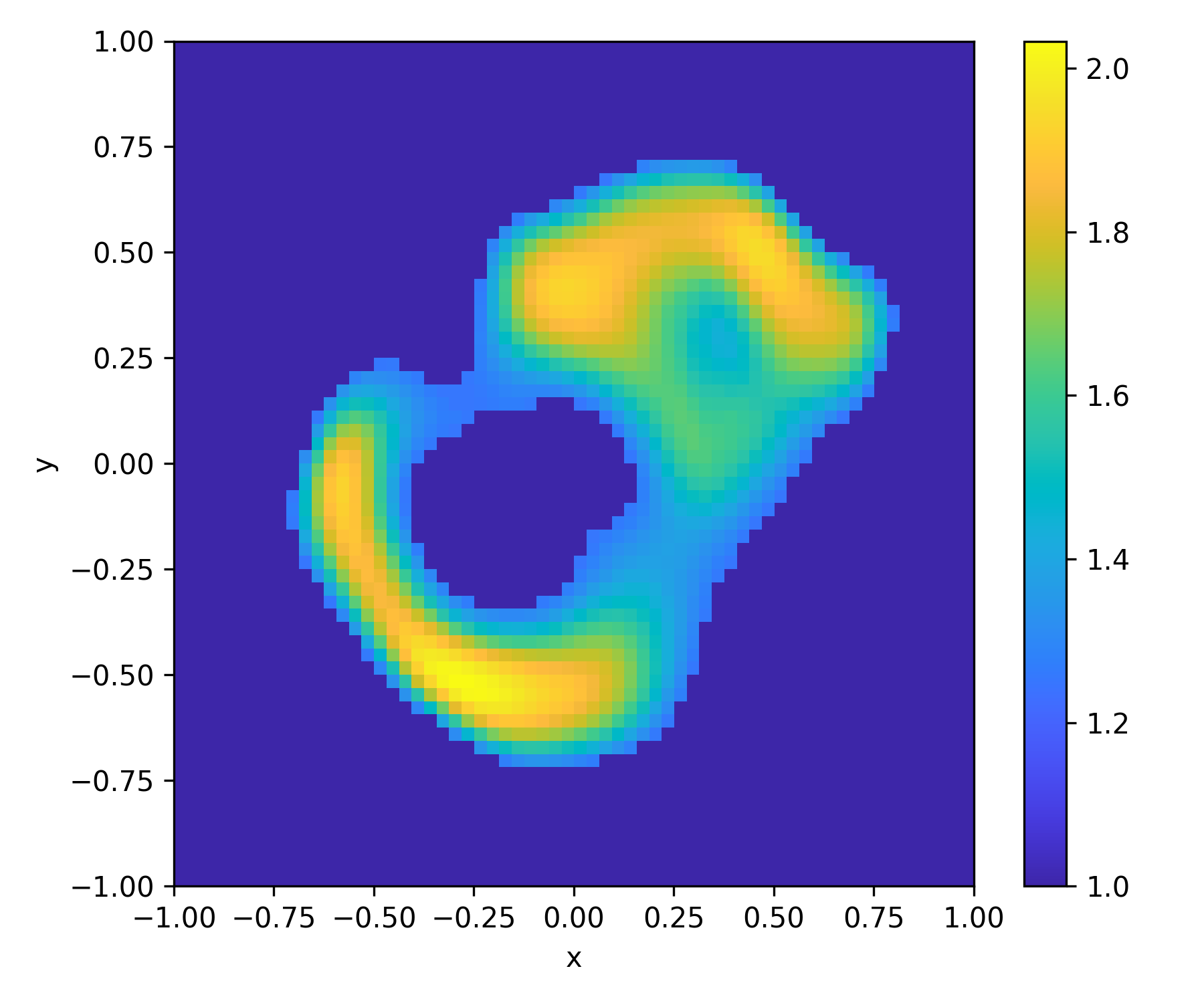}&
   			\includegraphics[width=0.15\textwidth]{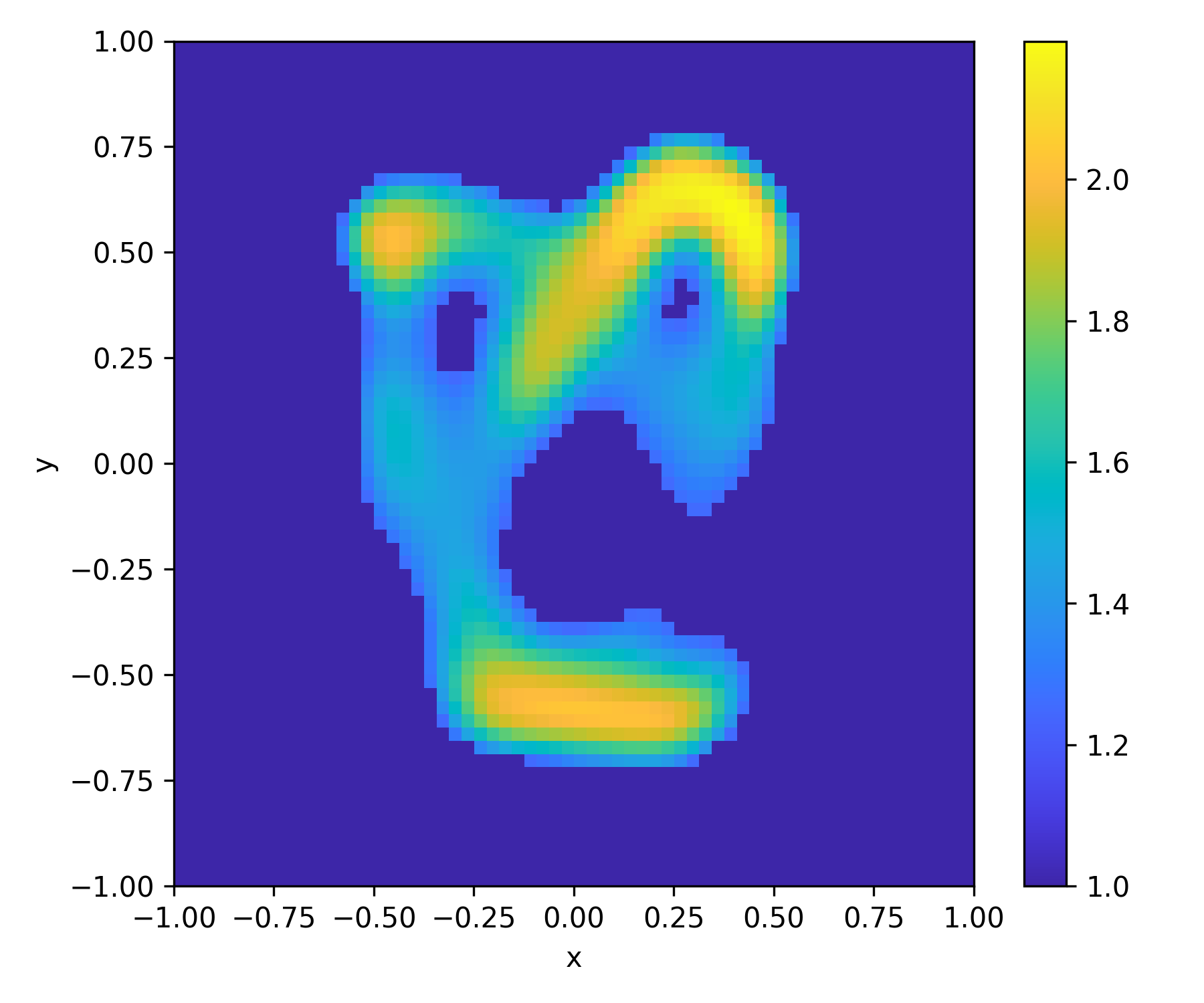}&
   			\includegraphics[width=0.15\textwidth]{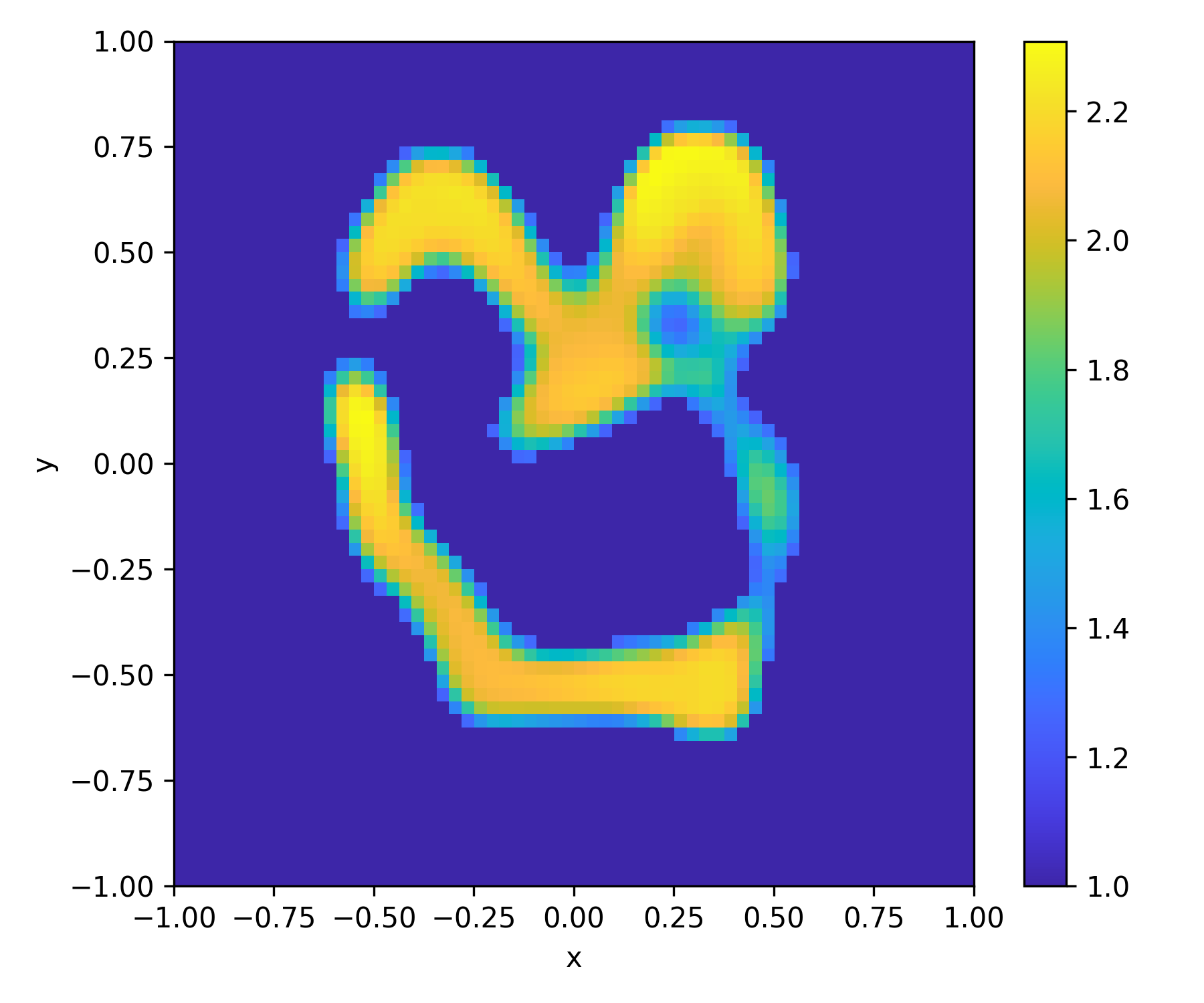}&
   			\includegraphics[width=0.15\textwidth]{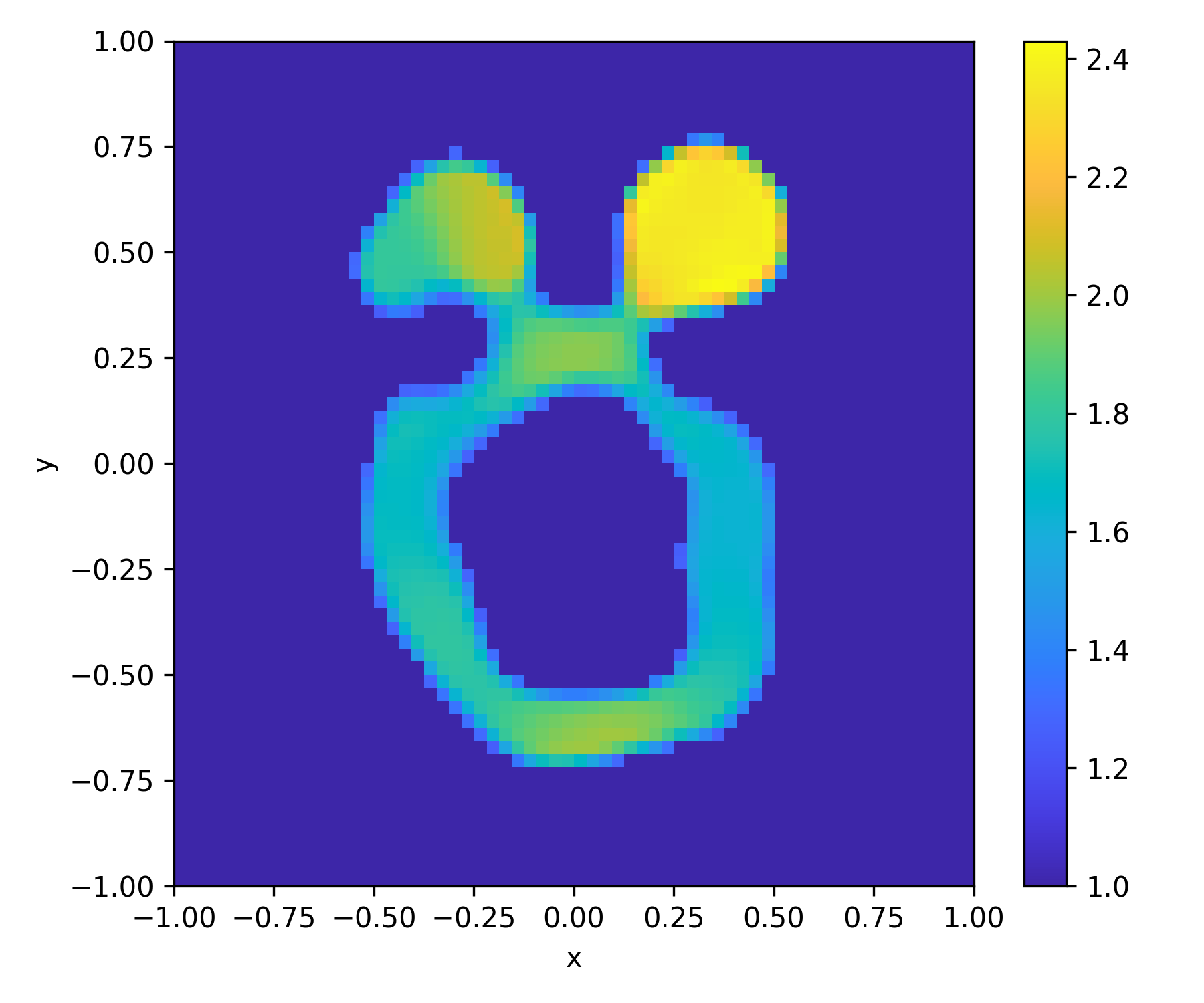}&
   			\includegraphics[width=0.15\textwidth]{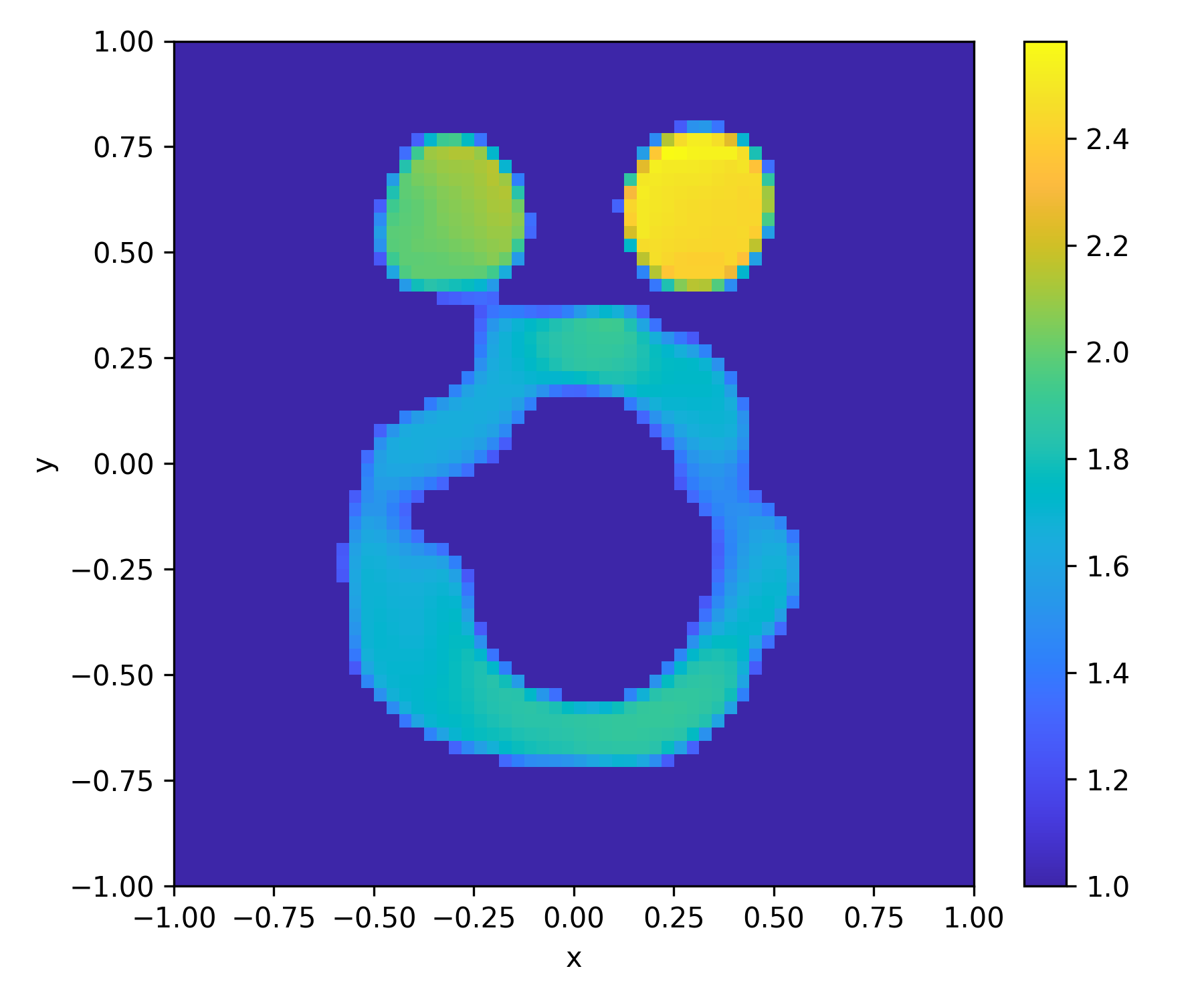}		
   			\\
   			40\%& &
   			\includegraphics[width=0.15\textwidth]{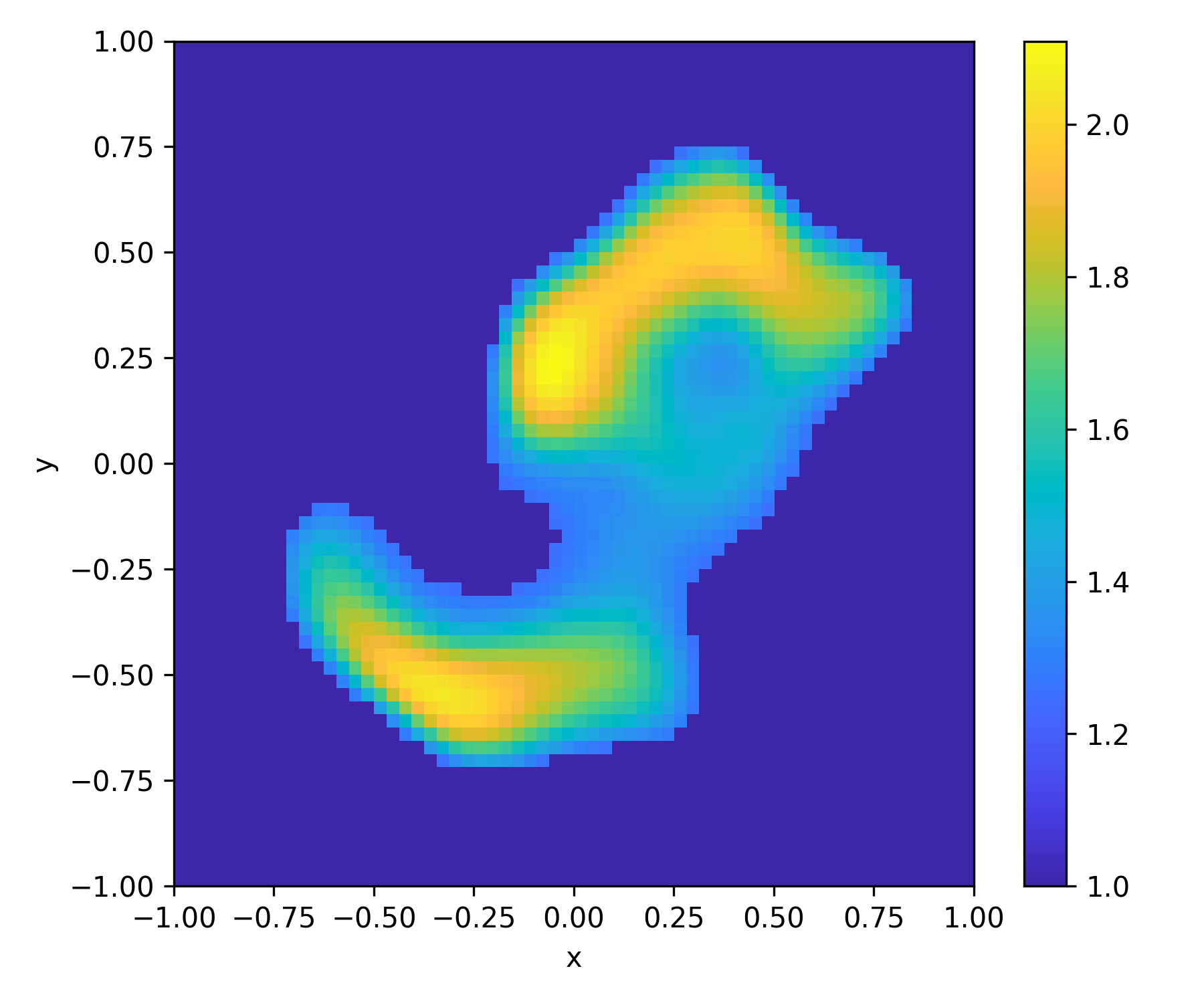}&
   			\includegraphics[width=0.15\textwidth]{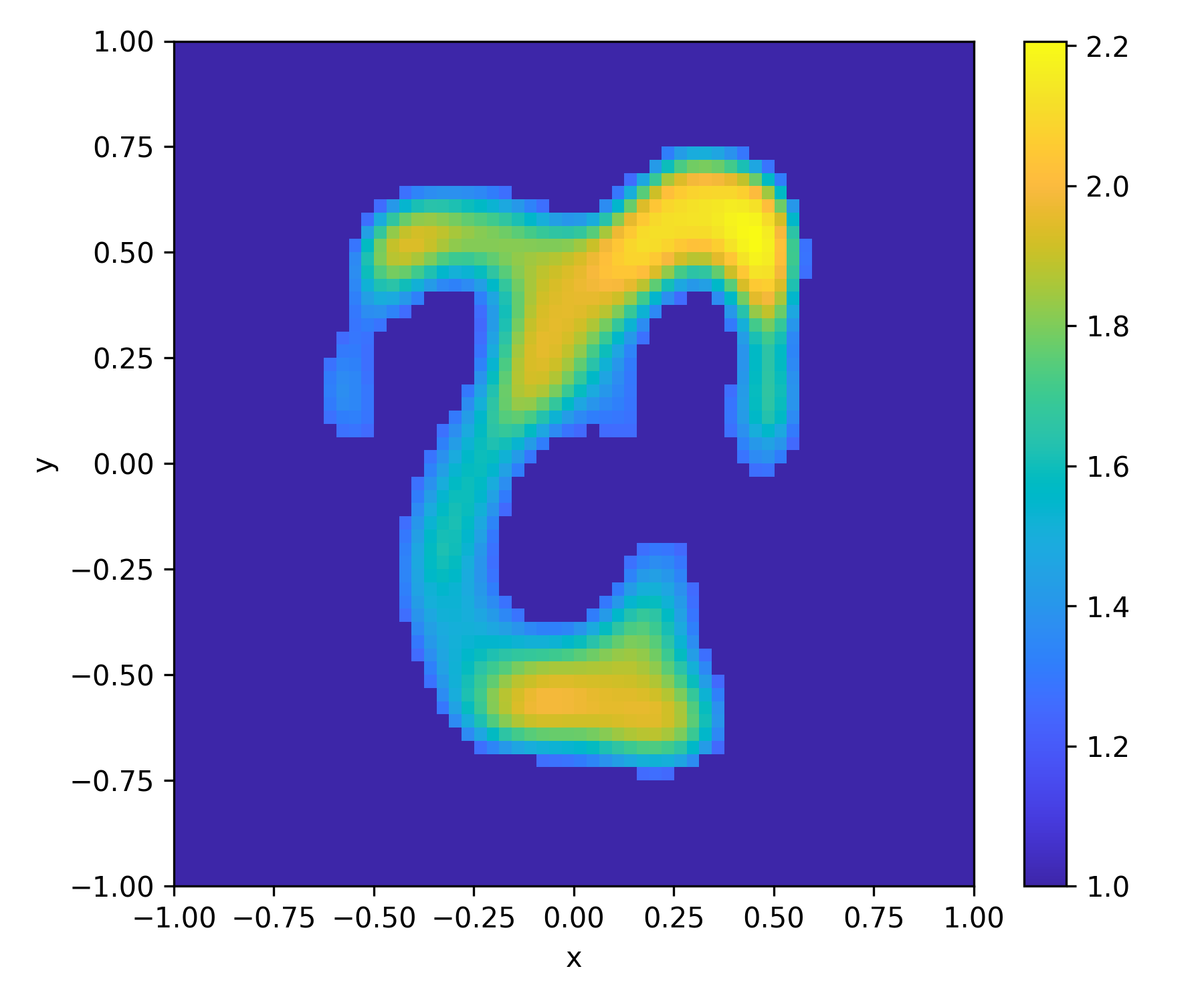}&
   			\includegraphics[width=0.15\textwidth]{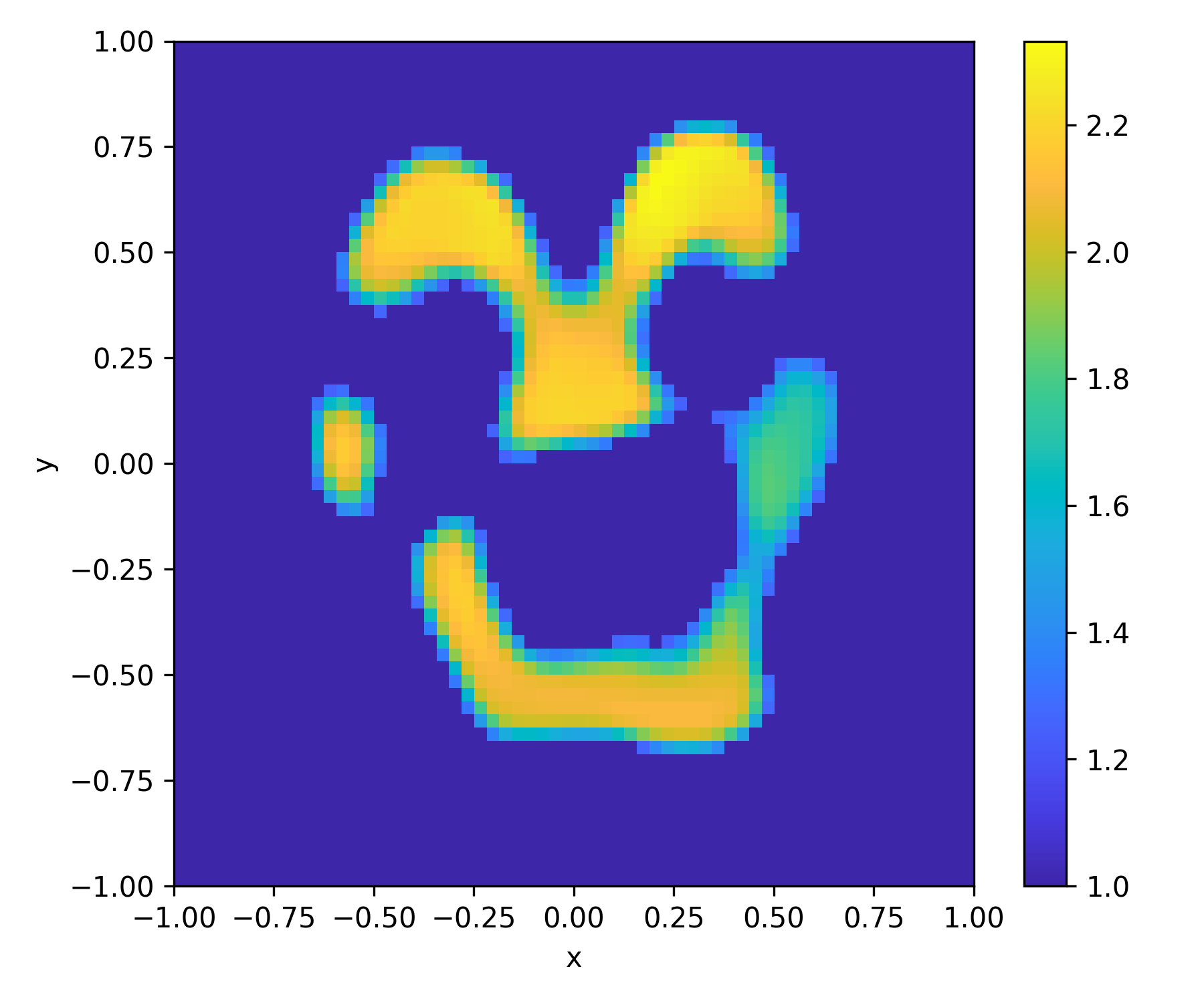}&
   			\includegraphics[width=0.15\textwidth]{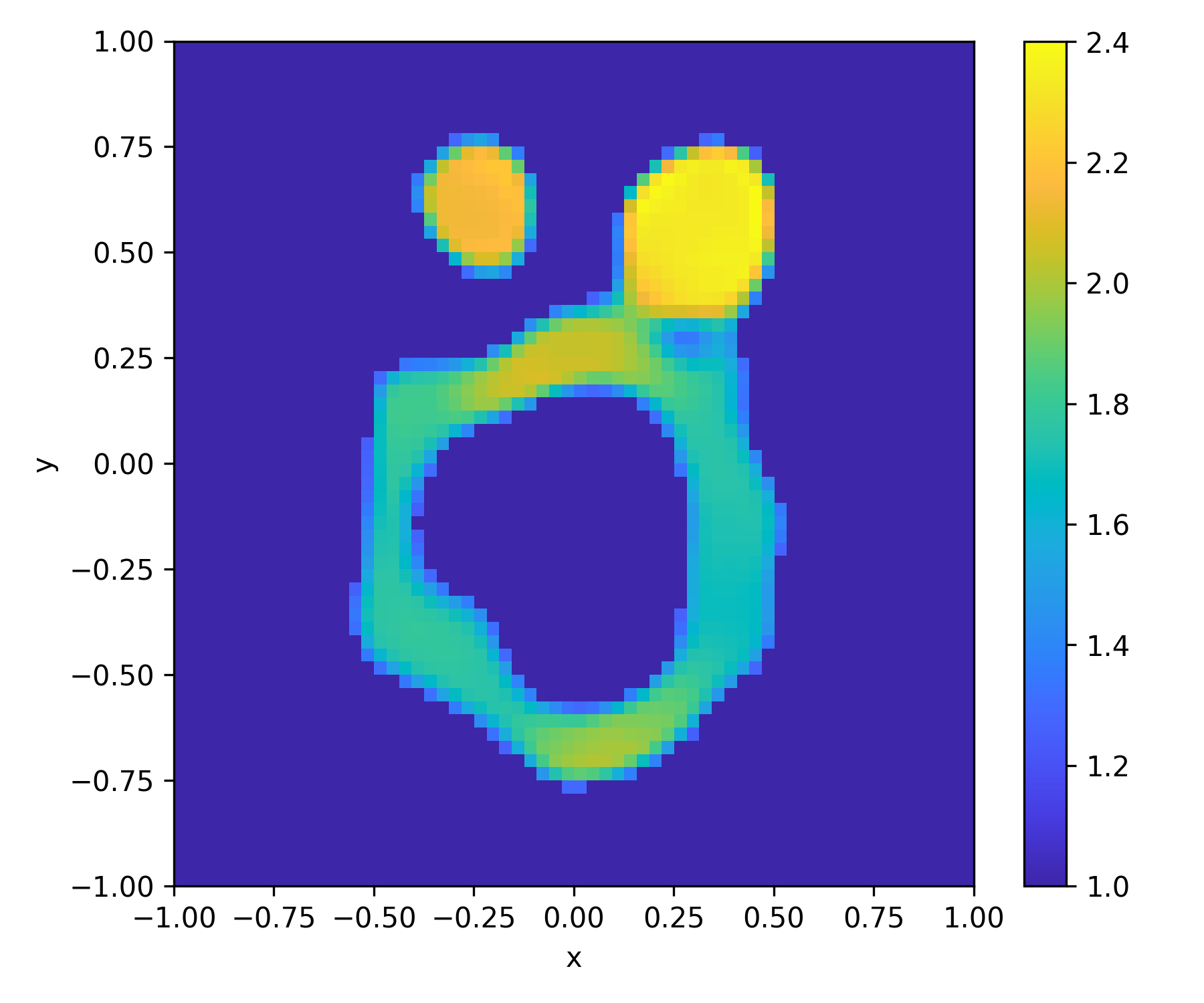}&
   			\includegraphics[width=0.15\textwidth]{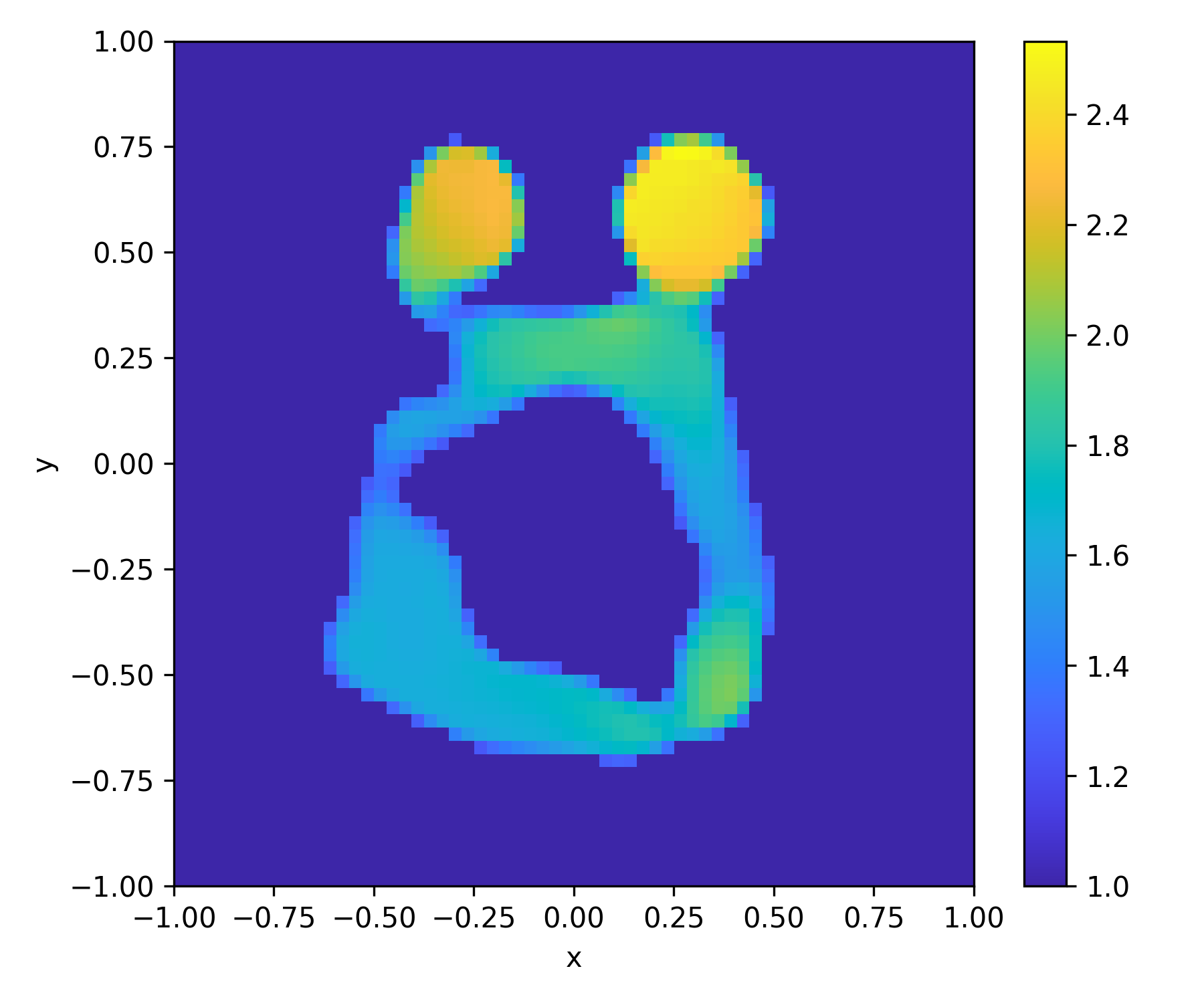}
   		\end{tblr}
   		\caption{Image reconstructions of an example from testing data and an "Austria Ring" with $15\%$ and $40\%$ Gaussian noises in the scattered fields by the networks trained by MNIST dataset. From left to right: the ground-truth images, the reconstruction with 1,2,4,8, and 16 incident fields.}
   		\label{tab:fig-Mnist-Half}
   	\end{center}
   \end{figure}

	\section{Conclusions}\label{secConcl}
 In this work, we propose a novel approach for solving inverse medium scattering problems, building upon the classical DSM method introduced in \cite{ito2012direct}. Our method combines deep learning with direct sampling to approximate and learn the relationship between index functions and true contrasts. The proposed DSM-DL has several attractive features. Firstly, our approach is highly efficient and parallelizable, while being robust to noise. Secondly, our deep neural network is able to fully leverage multiple scattering data, leading to significant improvements in reconstruction quality and noise resilience. Thirdly, even with only one incident wave, our method can produce satisfactory results, particularly for simple scatterer distributions. Numerical experiments confirm these features.

With the framework of the proposed DSM-DL in this paper, the DSM-DL can also be generalized to many other inverse problems for which proper DSMs have been developed, for example, in inverse obstacle scattering, inverse electromagnetic scattering, inverse elastic scattering, electrical impedance tomography, and inversion of radon transform.  In general, DSM-DL offers a promising methodology for tackling various inverse problems. 

	\bibliography{bibfile}

\begin{thebibliography}{10}

\bibitem{adler2017solving}
J.~Adler and O.~{\"O}ktem.
\newblock Solving ill-posed inverse problems using iterative deep neural
  networks.
\newblock {\em Inverse Problems}, 33(12):124007, 2017.

\bibitem{ammari2009layer}
H.~Ammari, H.~Kang, and H.~Lee.
\newblock {\em Layer potential techniques in spectral analysis}.
\newblock Number 153. American Mathematical Society, 2009.

\bibitem{arridge2019solving}
S.~Arridge, P.~Maass, O.~{\"O}ktem, and C.-B. Sch{\"o}nlieb.
\newblock Solving inverse problems using data-driven models.
\newblock {\em Acta Numerica}, 28:1--174, 2019.

\bibitem{bao2015inverse}
G.~Bao, P.~Li, J.~Lin, and F.~Triki.
\newblock Inverse scattering problems with multi-frequencies.
\newblock {\em Inverse Problems}, 31(9):093001, 2015.

\bibitem{buchanan2004marine}
J.~L. Buchanan, R.~P. Gilbert, A.~Wirgin, and Y.~S. Xu.
\newblock {\em Marine acoustics: direct and inverse problems}.
\newblock SIAM, Philadelphia, PA, 2004.

\bibitem{cakoni2011linear}
F.~Cakoni, D.~Colton, and P.~Monk.
\newblock {\em The linear sampling method in inverse electromagnetic
  scattering}.
\newblock SIAM, 2011.

\bibitem{cakoni2014qualitative}
F.~Cakoni and D.~L. Colton.
\newblock {\em A qualitative approach to inverse scattering theory}, volume
  767.
\newblock Springer, New York, 2014.

\bibitem{chen2013reverse}
J.~Chen, Z.~Chen, and G.~Huang.
\newblock Reverse time migration for extended obstacles: acoustic waves.
\newblock {\em Inverse Problems}, 29(8):085005, 2013.

\bibitem{chen2009subspace}
X.~Chen.
\newblock Subspace-based optimization method for solving inverse-scattering
  problems.
\newblock {\em IEEE Transactions on Geoscience and Remote Sensing},
  48(1):42--49, 2009.

\bibitem{xudong_book_2018}
X.~Chen.
\newblock {\em Computational methods for electromagnetic inverse scattering}.
\newblock Wiley-IEEE Press, Singapore, 2018.

\bibitem{chen2020review}
X.~Chen, Z.~Wei, M.~Li, and P.~Rocca.
\newblock A review of deep learning approaches for inverse scattering problems.
\newblock {\em Progress In Electromagnetics Research}, 167:67--81, 2020.

\bibitem{cheney2001linear}
M.~Cheney.
\newblock The linear sampling method and the music algorithm.
\newblock {\em Inverse problems}, 17(4):591, 2001.

\bibitem{chow2021direct}
Y.~T. Chow, F.~Han, and J.~Zou.
\newblock A direct sampling method for the inversion of the radon transform.
\newblock {\em SIAM Journal on Imaging Sciences}, 14(3):1004--1038, 2021.

\bibitem{chow2022direct}
Y.~T. Chow, F.~Han, and J.~Zou.
\newblock A direct sampling method for simultaneously recovering
  electromagnetic inhomogeneous inclusions of different nature.
\newblock {\em Journal of Computational Physics}, 470:111584, 2022.

\bibitem{chow2015direct}
Y.~T. Chow, K.~Ito, K.~Liu, and J.~Zou.
\newblock Direct sampling method for diffusive optical tomography.
\newblock {\em SIAM Journal on Scientific Computing}, 37(4):A1658--A1684, 2015.

\bibitem{chow2014direct}
Y.~T. Chow, K.~Ito, and J.~Zou.
\newblock A direct sampling method for electrical impedance tomography.
\newblock {\em Inverse Problems}, 30(9):095003, 2014.

\bibitem{colton1998inverse}
D.~L. Colton and R.~Kress.
\newblock {\em Inverse acoustic and electromagnetic scattering theory},
  volume~93.
\newblock Springer, 1998.

\bibitem{gao2022artificial}
Y.~Gao, H.~Liu, X.~Wang, and K.~Zhang.
\newblock On an artificial neural network for inverse scattering problems.
\newblock {\em Journal of Computational Physics}, 448:110771, 2022.

\bibitem{guo2021construct}
R.~Guo and J.~Jiang.
\newblock Construct deep neural networks based on direct sampling methods for
  solving electrical impedance tomography.
\newblock {\em SIAM Journal on Scientific Computing}, 43(3):B678--B711, 2021.

\bibitem{harris2020orthogonality}
I.~Harris and D.-L. Nguyen.
\newblock Orthogonality sampling method for the electromagnetic inverse
  scattering problem.
\newblock {\em SIAM Journal on Scientific Computing}, 42(3):B722--B737, 2020.

\bibitem{huang2020deep}
Y.~Huang, R.~Song, K.~Xu, X.~Ye, C.~Li, and X.~Chen.
\newblock Deep learning-based inverse scattering with structural similarity
  loss functions.
\newblock {\em IEEE Sensors Journal}, 21(4):4900--4907, 2020.

\bibitem{ioffe2015batch}
S.~Ioffe and C.~Szegedy.
\newblock Batch normalization: Accelerating deep network training by reducing
  internal covariate shift.
\newblock In {\em International conference on machine learning}, pages
  448--456. PMLR, 2015.

\bibitem{ito2012direct}
K.~Ito, B.~Jin, and J.~Zou.
\newblock A direct sampling method to an inverse medium scattering problem.
\newblock {\em Inverse Problems}, 28(2):025003, 2012.

\bibitem{ito2013direct}
K.~Ito, B.~Jin, and J.~Zou.
\newblock A direct sampling method for inverse electromagnetic medium
  scattering.
\newblock {\em Inverse Problems}, 29(9):095018, 2013.

\bibitem{ito2022least}
K.~Ito, Y.~Liang, and J.~Zou.
\newblock Least-squares method for recovering multiple medium parameters.
\newblock {\em Inverse Problems}, 38(12):125004, 2022.

\bibitem{ji2018direct}
X.~Ji, X.~Liu, and Y.~Xi.
\newblock Direct sampling methods for inverse elastic scattering problems.
\newblock {\em Inverse Problems}, 34(3):035008, 2018.

\bibitem{jiang2021learn}
J.~Jiang, Y.~Li, and R.~Guo.
\newblock Learn an index operator by {CNN} for solving diffusive optical
  tomography: a deep direct sampling method.
\newblock {\em arXiv preprint arXiv:2104.07703}, 2021.

\bibitem{khoo2019switchnet}
Y.~Khoo and L.~Ying.
\newblock Switchnet: a neural network model for forward and inverse scattering
  problems.
\newblock {\em SIAM Journal on Scientific Computing}, 41(5):A3182--A3201, 2019.

\bibitem{kirsch2002music}
A.~Kirsch.
\newblock The music-algorithm and the factorization method in inverse
  scattering theory for inhomogeneous media.
\newblock {\em Inverse problems}, 18(4):1025, 2002.

\bibitem{kirsch2011introduction}
A.~Kirsch.
\newblock {\em An introduction to the mathematical theory of inverse problems},
  volume 120.
\newblock Springer, 2011.

\bibitem{langer2010investigation}
S.~Langer.
\newblock Investigation of preconditioning techniques for the iteratively
  regularized gauss--newton method for exponentially ill-posed problems.
\newblock {\em SIAM Journal on Scientific Computing}, 32(5):2543--2559, 2010.

\bibitem{le2022sampling}
T.~Le, D.-L. Nguyen, V.~Nguyen, and T.~Truong.
\newblock Sampling type method combined with deep learning for inverse
  scattering with one incident wave.
\newblock {\em arXiv preprint arXiv:2207.10011}, 2022.

\bibitem{li2020nett}
H.~Li, J.~Schwab, S.~Antholzer, and M.~Haltmeier.
\newblock {NETT}: Solving inverse problems with deep neural networks.
\newblock {\em Inverse Problems}, 36(6):065005, 2020.

\bibitem{li2013direct}
J.~Li and J.~Zou.
\newblock A direct sampling method for inverse scattering using far-field data.
\newblock {\em Inverse Problems and Imaging}, 7(3):757--775, 2013.

\bibitem{li2022reconstruction}
K.~Li, B.~Zhang, and H.~Zhang.
\newblock Reconstruction of inhomogeneous media by iterative reconstruction
  algorithm with learned projector.
\newblock {\em arXiv preprint arXiv:2207.13032}, 2022.

\bibitem{li2020fourier}
Z.~Li, N.~Kovachki, K.~Azizzadenesheli, B.~Liu, K.~Bhattacharya, A.~Stuart, and
  A.~Anandkumar.
\newblock Fourier neural operator for parametric partial differential
  equations.
\newblock {\em arXiv preprint arXiv:2010.08895}, 2020.

\bibitem{lu2021learning}
L.~Lu, P.~Jin, G.~Pang, Z.~Zhang, and G.~E. Karniadakis.
\newblock Learning nonlinear operators via deeponet based on the universal
  approximation theorem of operators.
\newblock {\em Nature Machine Intelligence}, 3(3):218--229, 2021.

\bibitem{mccann2017convolutional}
M.~T. McCann, K.~H. Jin, and M.~Unser.
\newblock Convolutional neural networks for inverse problems in imaging: A
  review.
\newblock {\em IEEE Signal Processing Magazine}, 34(6):85--95, 2017.

\bibitem{potthast2006survey}
R.~Potthast.
\newblock A survey on sampling and probe methods for inverse problems.
\newblock {\em Inverse Problems}, 22(2):R1, 2006.

\bibitem{potthast2005singular}
R.~Potthast and I.~Stratis.
\newblock The singular sources method for an inverse transmission problem.
\newblock {\em Computing}, 75(2):237--255, 2005.

\bibitem{qu2019locating}
F.~Qu and H.~Zhang.
\newblock Locating a complex inhomogeneous medium with an approximate
  factorization method.
\newblock {\em Inverse Problems}, 35(4):045001, 2019.

\bibitem{raissi2019physics}
M.~Raissi, P.~Perdikaris, and G.~E. Karniadakis.
\newblock Physics-informed neural networks: A deep learning framework for
  solving forward and inverse problems involving nonlinear partial differential
  equations.
\newblock {\em Journal of Computational physics}, 378:686--707, 2019.

\bibitem{ronneberger2015u}
O.~Ronneberger, P.~Fischer, and T.~Brox.
\newblock U-net: Convolutional networks for biomedical image segmentation.
\newblock In {\em International Conference on Medical image computing and
  computer-assisted intervention}, pages 234--241. Springer, 2015.

\bibitem{shorten2019survey}
C.~Shorten and T.~M. Khoshgoftaar.
\newblock A survey on image data augmentation for deep learning.
\newblock {\em Journal of big data}, 6(1):1--48, 2019.

\bibitem{tanyu2022deep}
D.~N. Tanyu, J.~Ning, T.~Freudenberg, N.~Heilenk{\"o}tter, A.~Rademacher,
  U.~Iben, and P.~Maass.
\newblock Deep learning methods for partial differential equations and related
  parameter identification problems.
\newblock {\em arXiv preprint arXiv:2212.03130}, 2022.

\bibitem{tripura2022wavelet}
T.~Tripura and S.~Chakraborty.
\newblock Wavelet neural operator: a neural operator for parametric partial
  differential equations.
\newblock {\em arXiv preprint arXiv:2205.02191}, 2022.

\bibitem{van2021forward}
P.~M. van~den Berg.
\newblock {\em Forward and inverse scattering algorithms based on contrast
  source integral equations}.
\newblock John Wiley \& Sons, 2021.

\bibitem{van1997contrast}
P.~M. Van Den~Berg and R.~E. Kleinman.
\newblock A contrast source inversion method.
\newblock {\em Inverse problems}, 13(6):1607, 1997.

\bibitem{wang2003multiscale}
Z.~Wang, E.~P. Simoncelli, and A.~C. Bovik.
\newblock Multiscale structural similarity for image quality assessment.
\newblock In {\em The Thrity-Seventh Asilomar Conference on Signals, Systems \&
  Computers, 2003}, volume~2, pages 1398--1402. IEEE, 2003.

\bibitem{wei2018deep}
Z.~Wei and X.~Chen.
\newblock Deep-learning schemes for full-wave nonlinear inverse scattering
  problems.
\newblock {\em IEEE Transactions on Geoscience and Remote Sensing},
  57(4):1849--1860, 2018.

\bibitem{yao2019two}
H.~M. Yao, E.~Wei, and L.~Jiang.
\newblock Two-step enhanced deep learning approach for electromagnetic inverse
  scattering problems.
\newblock {\em IEEE Antennas and Wireless Propagation Letters},
  18(11):2254--2258, 2019.

\bibitem{zhdanov2002geophysical}
M.~S. Zhdanov.
\newblock {\em Geophysical inverse theory and regularization problems},
  volume~36.
\newblock Elsevier, 2002.

\end{thebibliography}
	
\end{document}